\documentclass[english,leqno]{article}

\usepackage{latexsym}
\usepackage{linearb}
\usepackage[LGR, OT1]{fontenc}
\usepackage{amsmath,amsthm}
\usepackage{amssymb}
\usepackage{MnSymbol}
\usepackage{graphicx}
\usepackage{marvosym}
\usepackage{eufrak}
\usepackage[margin=1.25in,dvips]{geometry}
\usepackage{color}
\usepackage{comment}
\usepackage{mathrsfs}\usepackage[LGR,T1]{fontenc}
\usepackage[latin9]{inputenc}
\usepackage{geometry}
\usepackage{babel}
\usepackage{enumitem}
\usepackage{amsthm}
\usepackage{amstext}
\usepackage{amssymb}
\usepackage{esint}
\usepackage{float}
\usepackage[unicode=true,pdfusetitle,
 bookmarks=true,bookmarksnumbered=false,bookmarksopen=false,
 breaklinks=false,pdfborder={0 0 1},backref=false,colorlinks=false]
 {hyperref}
\makeatletter

\newcommand{\noun}[1]{\textsc{#1}}
\DeclareRobustCommand{\greektext}{%
  \fontencoding{LGR}\selectfont\def\encodingdefault{LGR}}
\DeclareRobustCommand{\textgreek}[1]{\leavevmode{\greektext #1}}
\DeclareFontEncoding{LGR}{}{}
\DeclareTextSymbol{\~}{LGR}{126}

\numberwithin{equation}{section}
\numberwithin{figure}{section}
\usepackage{enumitem}		
\newcommand{\lyxaddress}[1]{
\par {\raggedright #1
\vspace{1.4em}
\noindent\par}
}
  \theoremstyle{remark}
  \newtheorem*{rem*}{\protect\remarkname}
  \theoremstyle{definition}
  \newtheorem{defn}{\protect\definitionname}[section]
  \theoremstyle{plain}
  \newtheorem{prop}{\protect\propositionname}[section]
  \theoremstyle{plain}
  \newtheorem{lem}{\protect\lemmaname}[section]
  \theoremstyle{plain}
  \newtheorem{thm}{\protect\theoremname}[section]

\newtheorem{innercustomcor}{Corollary}
\newenvironment{customcor}[1]
  {\renewcommand\theinnercustomcor{#1}\innercustomcor}
  {\endinnercustomcor}

\newtheorem{innercustomthm}{Theorem}
\newenvironment{customthm}[1]
  {\renewcommand\theinnercustomthm{#1}\innercustomthm}
  {\endinnercustomthm}

\newtheorem{innercustomprop}{Proposition}
\newenvironment{customprop}[1]
  {\renewcommand\theinnercustomprop{#1}\innercustomprop}
  {\endinnercustomprop}

\makeatother

\usepackage{babel}
  \providecommand{\definitionname}{Definition}
  \providecommand{\lemmaname}{Lemma}
  \providecommand{\propositionname}{Proposition}
  \providecommand{\remarkname}{Remark}
\providecommand{\theoremname}{Theorem}

\begin{document}
\title{The Einstein--null dust system in spherical symmetry with an inner mirror: structure of the maximal development and Cauchy stability}

\author{Georgios Moschidis}

\maketitle

\lyxaddress{Princeton University, Department of Mathematics, Fine Hall, Washington
Road, Princeton, NJ 08544, United States, \tt gm6@math.princeton.edu}
\begin{abstract}
In this paper, we study the evolution of asymptotically AdS initial
data for the spherically symmetric Einstein--massless Vlasov system
for $\Lambda<0$, with reflecting boundary conditions imposed on timelike
infinity $\mathcal{I}$, in the case when the Vlasov field is supported
only on radial geodesics. This system is equivalent to the spherically
symmetric Einstein--null dust system, allowing for both ingoing and
outgoing dust. In general, solutions to this system break down in
finite time (independent of the size of the initial data); we highlight
this fact by showing that, at the first point where the ingoing dust
reaches the axis of symmetry, solutions become $C^{0}$ inextendible,
although the spacetime metric remains regular up to that point. 

One way to overcome this ``trivial'' obstacle to well-posedness
is to place an inner mirror on a timelike hypersurface of the form
$\{r=r_{0}\}$, $r_{0}>0$, and study the evolution on the exterior
domain $\{r\ge r_{0}\}$. In this setting, we prove the existence
and uniqueness of maximal developments for general smooth and asymptotically
AdS initial data sets, and study the basic geometric properties of
these developments. Furthermore, we establish the well-posedness and
Cauchy stabilty of solutions with respect to a ``rough'' initial
data norm, measuring the concentration of energy at scales proportional
to the mirror radius $r_{0}$.

The above well-posedness and Cauchy stability estimates are used in
our companion paper \cite{MoschidisNullDust} for the proof of the
AdS instability conjecture for the Einstein--null dust system. However,
the results of the present paper might also be of independent interest.
\end{abstract}
\tableofcontents{}

\section{Introduction}

In the presence of a cosmological constant $\Lambda<0$, the simplest
solution of the \emph{vacuum Einstein equations} 
\begin{equation}
Ric_{\text{\textgreek{m}\textgreek{n}}}-\frac{1}{2}Rg_{\text{\textgreek{m}\textgreek{n}}}+\Lambda g_{\text{\textgreek{m}\textgreek{n}}}=0\label{eq:VacuumEinsteinEqtns}
\end{equation}
is Anti-de Sitter spacetime $(\mathcal{M}_{AdS},g_{AdS})$ (see \cite{HawkingEllis1973}).
In the standard polar coordinate chart on $\mathcal{M}_{AdS}\simeq\mathbb{R}^{3+1}$,
$g_{AdS}$ is expressed as: 
\begin{equation}
g_{AdS}=-\big(1-\frac{1}{3}\Lambda r^{2}\big)dt^{2}+\big(1-\frac{1}{3}\Lambda r^{2}\big)^{-1}dr^{2}+r^{2}g_{\mathbb{S}^{2}}.\label{eq:AdSMetricPolarCoordinates}
\end{equation}

Solutions $(\mathcal{M},g)$ to (\ref{eq:VacuumEinsteinEqtns}) which
are \emph{asymptotically AdS}, i.\,e.~possess an asymptotic region
where $g$ is close to (\ref{eq:AdSMetricPolarCoordinates}) for $r\gg1$,
fail to be globally hyperbolic: They are conformally equivalent to
spacetimes $(\widetilde{\mathcal{M}},\tilde{g})$ with a timelike
boundary $\mathcal{I}$ (see \cite{HawkingEllis1973}). Thus, the
correct setting to study the dynamics of asymptotically AdS solutions
to (\ref{eq:VacuumEinsteinEqtns}) is that of an \emph{initial-boundary
value problem}, with suitable boundary conditions prescribed asymptotically
on the conformal boundary $\mathcal{I}$. The well-posedness of the
initial-boundary value problem for (\ref{eq:VacuumEinsteinEqtns})
for a certain class of boundary conditions on $\mathcal{I}$ was first
addressed by Friedrich in \cite{Friedrich1995}. Well-posedness for
more general boundary conditions and matter fields in spherical symmetry
was obtained in \cite{HolzSmul2012,HolzWarn2015}; see also \cite{HolzegelLukSmuleviciWarnick,Friedrich2014}
for a discussion on the issue of general boundary conditions on $\mathcal{I}$
without any symmetry reductions.

A fundamental question arising in the study of the long time dynamics
of (\ref{eq:VacuumEinsteinEqtns}) is whether $(\mathcal{M}_{AdS},g_{AdS})$
is stable under perturbations of its initial data. In 2006, Dafermos
and Holzegel \cite{DafHol,DafermosTalk} proposed the so-called \emph{AdS
instability conjecture}, stating that there exist arbitrarily small
perturbations to the initial data of $(\mathcal{M}_{AdS},g_{AdS})$
which, under evolution by (\ref{eq:VacuumEinsteinEqtns}) with a reflecting
boundary condition on $\mathcal{I}$, lead to the formation of black
hole regions (and, thus, $(\mathcal{M}_{AdS},g_{AdS})$ is non-linearly
unstable). This conjecture was also supported by results of Anderson
\cite{AndersonAdS}.

Starting from the seminal work of Bizon and Rostworowski \cite{bizon2011weakly},
a vast amount of numerical and heuristic works have been devoted to
the study of the AdS instability conjecture (see, e.\,g.,~\cite{DiasHorSantos,BuchelEtAl2012,DiasEtAl,MaliborskiEtAl,BalasubramanianEtAl,CrapsEtAl2014,CrapsEtAl2015,BizonMalib,DimitrakopoulosEtAl,GreenMailardLehnerLieb,HorowitzSantos,DimitrakopoulosEtAl2015,DimitrakopoulosEtAl2016}).
The bulk of these works is dedicated to the study of this conjecture
in the context of the \emph{spherically symmetric Einstein-scalar
field system}. This is a simpler model of (\ref{eq:VacuumEinsteinEqtns}),
being a well-posed $1+1$ hyperbolic system with non-trivial dynamics
resembling the qualitative properties of (\ref{eq:VacuumEinsteinEqtns})
(see the discussion in \cite{Christodoulou1999}).%
\footnote{In $3+1$ dimensions, (\ref{eq:VacuumEinsteinEqtns}) has trivial
dynamics in spherical symmetry, as a consequence of Birkhoff's theorem
(see \cite{Birkhoff}). In higher dimensions, however, there exist
symmetry classes compatible with AdS asymptotics, under which (\ref{eq:VacuumEinsteinEqtns})
is reduced to a $1+1$ hyperbolic system with non-trivial dynamics,
such as the biaxial Bianchi IX symmetry class in $4+1$ dimensions
introduced by Bizon--Chmaj--Schmidt \cite{BizonChmajSchmidt2005}.%
}

An even simpler model for which the AdS instability conjecture can
be addressed is the \emph{Einstein--massless Vlasov }system (see \cite{Andreasson2011,Rein1995})
\begin{equation}
\begin{cases}
Ric_{\text{\textgreek{m}\textgreek{n}}}-\frac{1}{2}Rg_{\text{\textgreek{m}\textgreek{n}}}+\Lambda g_{\text{\textgreek{m}\textgreek{n}}}=8\pi T_{\text{\textgreek{m}\textgreek{n}}}[f],\\
p^{\text{\textgreek{a}}}\partial_{x^{\text{\textgreek{a}}}}f-\text{\textgreek{G}}_{\text{\textgreek{b}\textgreek{g}}}^{\text{\textgreek{a}}}p^{\text{\textgreek{b}}}p^{\text{\textgreek{g}}}\partial_{p^{\text{\textgreek{a}}}}f=0,\\
supp(f)\subseteq\big\{ g_{\text{\textgreek{a}\textgreek{b}}}(x)p^{\text{\textgreek{a}}}p^{\text{\textgreek{b}}}=0\big\}\subset T\mathcal{M}
\end{cases}\label{eq:EinsteinMasslessVlasov}
\end{equation}
 in \emph{spherical symmetry}. This system can be further reduced
to the case when the Vlasov field $f$ is supported only on \emph{radial
}null geodesics. The resulting system, which we will call the \emph{spherically
symmetric Einstein--radial massless Vlasov }system, is a singular
reduction of (\ref{eq:EinsteinMasslessVlasov}) (see \cite{Rendall1997})
and is equivalent to the \emph{spherically symmetric Einstein--null
dust} system, allowing for both ingoing and outgoing null dust. This
system was first studied by Poisson and Israel in \cite{PoissonIsrael1990}. 

The spherically symmetric Einstein--null dust system suffers from
a severe break down occuring once the support of the ingoing null
dust reaches the axis of symmetry (see the discussion in the next
section for more details). This ``trivial'' obstacle to well-posedness
can be overcome by restricting the evolution of the system in the
exterior of an \emph{inner mirror }with spherical radius $r_{0}>0$. 

In the present paper, we will establish a number of well-posedness
results for the spherically symmetric Einstein--null dust system with
reflecting boundary conditions on $\mathcal{I}$ and an inner mirror,
which are necessary for addressing the AdS instability conjecture
in this setting. The proof of the AdS instability conjecture for this
model is then established in our companion paper \cite{MoschidisNullDust}. 

In particular, for any mirror radius $r_{0}>0$ and any smooth, asymptotically
AdS, spherically symmetric and radial initial data set $\mathcal{S}_{in}$
for (\ref{eq:EinsteinMasslessVlasov}), we will establish two types
of results, under reflecting boundary conditions on both $r=r_{0}$
and $\mathcal{I}$:

\begin{enumerate}

\item We will establish the existence and uniqueness of the maximal
development $(\mathcal{M},g;f)$ of $\mathcal{S}_{in}$ (restricted
to the domain $\{r\ge r_{0}\}$) and study the basic geometric properties
of $(\mathcal{M},g;f)$. See Theorem \ref{thm:MaxDevelopmentIntro}.

\item We will establish a Cauchy stability result for $\mathcal{S}_{in}$
in a rough initial data norm measuring the concentration of the energy
of $\mathcal{S}_{in}$ in spherical annuli of radius $\sim r_{0}$.
See Theorem \ref{thm:CauchyIntro}. In particular, this result will
provide a Cauchy stability statement for $(\mathcal{M}_{AdS},g_{AdS})$\emph{
uniformly in} $r_{0}$; this fact will be important for addressing
the AdS instability conjecture in this setting in our companion paper
\cite{MoschidisNullDust}. 

\end{enumerate}

We will now proceed to discuss these results in more detail.

\subsection{\label{sub:The-maximal-development}Theorem 1: The maximal development
for the Einstein--null dust system with reflecting boundary conditions
on $r=r_{0}$ and $\mathcal{I}$}

\textgreek{T}he system (\ref{eq:EinsteinMasslessVlasov}), reduced
to the case where $(\mathcal{M},g)$ and $f$ are spherically symmetric
and $f$ is supported only on radial null geodesics, is equivalent
to the spherically symmetric Einstein--null dust system with two dusts
(see \cite{Andreasson2011,Rendall1997,PoissonIsrael1990}). In double
null coordinates $(u,v)$, where the metric $g$ takes the form 
\begin{equation}
g=-\text{\textgreek{W}}^{2}dudv+r^{2}g_{\mathbb{S}^{2}},
\end{equation}
this system can be expressed in terms of $r,\text{\textgreek{W}}$
and the renormalised components $\bar{\text{\textgreek{t}}}=r^{2}T_{vv}$,
$\text{\textgreek{t}}=r^{2}T_{uu}$ of the energy momentum tensor
$T$ as: 
\begin{equation}
\begin{cases}
\partial_{u}\partial_{v}(r^{2}) & =-\frac{1}{2}(1-\Lambda r^{2})\text{\textgreek{W}}^{2},\\
\partial_{u}\partial_{v}\log(\text{\textgreek{W}}^{2}) & =\frac{\text{\textgreek{W}}^{2}}{2r^{2}}\big(1+4\text{\textgreek{W}}^{-2}\partial_{u}r\partial_{v}r\big),\\
\partial_{v}(\text{\textgreek{W}}^{-2}\partial_{v}r) & =-4\pi r^{-1}\text{\textgreek{W}}^{-2}\bar{\text{\textgreek{t}}},\\
\partial_{u}(\text{\textgreek{W}}^{-2}\partial_{u}r) & =-4\pi r^{-1}\text{\textgreek{W}}^{-2}\text{\textgreek{t}},\\
\partial_{u}\bar{\text{\textgreek{t}}} & =0,\\
\partial_{v}\text{\textgreek{t}} & =0.
\end{cases}\label{eq:EinsteinNullDust}
\end{equation}

It is known that solutions to (\ref{eq:EinsteinNullDust}) with $\text{\textgreek{t}}=0$
(known in the literature as \emph{Vaidya} spacetimes, see e.\,g.~\cite{Hiscock1982})
develop a curvature singularity beyond the first point when the support
of the ingoing dust reaches the axis of symmetry $\text{\textgreek{g}}$;
see, e.\,g.~\cite{Hiscock1982,Lemos1992}. We will actually show
an even stronger ill-posedness result for (\ref{eq:EinsteinNullDust}):

\begin{customprop}{1}\label{prop:C0InextendibilityIntroduction} Any
spherically symmetric solution $(\mathcal{M},g;\text{\textgreek{t}},\bar{\text{\textgreek{t}}})$
of (\ref{eq:EinsteinNullDust}) with non-empty axis $\text{\textgreek{g}}$,
arising from smooth initial data on $\{u=0\}$ with $\text{\textgreek{t}}|_{u=0}=0$,
remains smooth up to the first point when a radial geodesic in the
support of $\bar{\text{\textgreek{t}}}$ reaches $\text{\textgreek{g}}$.
However, beyond that point, $(\mathcal{M},g;\text{\textgreek{t}},\bar{\text{\textgreek{t}}})$
is $C^{0}$ inextendible as a spherically symmetric solution to (\ref{eq:EinsteinNullDust}).
In particular, (\ref{eq:EinsteinNullDust}) is ill-posed in any ``reasonable''
initial data topology.

\end{customprop}

For a more detailed statement of Proposition \ref{prop:C0InextendibilityIntroduction}
and a precise definition of a $C^{0}$ solution of (\ref{eq:EinsteinNullDust}),
see Proposition~\ref{prop:Well-Posedness-Einstein-nulldust} and
Theorem~\ref{thm:IllPosedness} in Section \ref{sub:An-ill-posedness-result}
of the Appendix.
\begin{rem*}
Notice that Proposition \ref{prop:C0InextendibilityIntroduction}
yields a uniform upper bound on the time of existence $u_{*}$ of
solutions $(\mathcal{M},g;\text{\textgreek{t}},\bar{\text{\textgreek{t}}})$
to (\ref{eq:EinsteinNullDust}) for any characteristic initial data
set at $\{u=0\}$ for which $\bar{\text{\textgreek{t}}}$ is not identically
equal to $0$, with $u_{*}$ depending only on the distance of $supp(\bar{\text{\textgreek{t}}})$
from $\text{\textgreek{g}}$ initially (and not on the proximity of
the initial data to the trivial data, in any ``reasonable'' initial
data norm). In the case when the initial data on $\{u=0\}$ are close
to the trivial data, $(\mathcal{M},g)$ is globally $C^{\infty}$
extendible as spherically symmetric Lorentzian manifold beyond $\{u=u_{*}\}$,
despite the fact that $(\mathcal{M},g;\text{\textgreek{t}},\bar{\text{\textgreek{t}}})$
is $C^{0}$ inextendible as a solution to (\ref{eq:EinsteinNullDust})
(see Proposition~\ref{prop:Well-Posedness-Einstein-nulldust}).%
\footnote{We do not examine the question of whether $(\mathcal{M},g;\text{\textgreek{t}},\bar{\text{\textgreek{t}}})$
admits a low regularity extension as a solution of (the analogue of)
(\ref{eq:EinsteinNullDust}) outside spherical symmetry. In general,
the question of $C^{0}$ extendibility outside the regime of surface
symmetry is rather intricate; see \cite{SbierskiC0}.%
} We should also remark that Proposition\inputencoding{latin1}{~}\inputencoding{latin9}\ref{prop:C0InextendibilityIntroduction}
holds independently of the value of the cosmological constant $\Lambda$.%
\footnote{In the case $\Lambda<0$, Proposition \ref{prop:C0InextendibilityIntroduction}
implies that anti-de Sitter spacetime $(\mathcal{M}_{AdS},g_{AdS})$
is \emph{not} Cauchy stable for (\ref{eq:EinsteinNullDust}) for any
``reasonable'' initial data topology. Similarly for Minkowski spacetime
$(\mathbb{R}^{3+1},\text{\textgreek{h}})$ in the case $\Lambda=0$,
or de Sitter spacetime $(\mathcal{M}_{dS},g_{dS})$ in the case $\Lambda>0$.%
}
\end{rem*}
One way to overcome the obstacle to well-posedness raised by Proposition
\ref{prop:C0InextendibilityIntroduction} is to restrict the evolution
of the system (\ref{eq:EinsteinNullDust}) in the region $\{r\ge r_{0}\}$,
for some $r_{0}>0$, and impose a reflecting boundary condition on
the portion $\text{\textgreek{g}}_{0}$ of the curve $\{r=r_{0}\}$
which is timelike. Our first result concerns the well posedness and
the structure of the maximal development of the characteristic initial-boundary
value problem for the system (\ref{eq:EinsteinNullDust}) in this
setting:

\begin{customthm}{1}[rough version]\label{thm:MaxDevelopmentIntro}
For any $r_{0}>0$ and any smooth, asymptotically AdS initial data
set $(r,\text{\textgreek{W}}^{2};\text{\textgreek{t}},\bar{\text{\textgreek{t}}})|_{u=0}$
for (\ref{eq:EinsteinNullDust}), restricted to the region $\{r\ge r_{0}\}$,
there exists a unique, smooth, maximal future developmment $(r,\text{\textgreek{W}}^{2};\text{\textgreek{t}},\bar{\text{\textgreek{t}}})$
of $(r,\text{\textgreek{W}}^{2};\text{\textgreek{t}},\bar{\text{\textgreek{t}}})|_{u=0}$
on $\{r\ge r_{0}\}$, solving (\ref{eq:EinsteinNullDust}) with reflecting
boundary conditions on $\mathcal{I}$ and $\text{\textgreek{g}}_{0}$,
where $r|_{\text{\textgreek{g}}_{0}}=r_{0}$ and $\text{\textgreek{g}}_{0}$
coincides with the portion of the curve $\{r=r_{0}\}$ which is timelike
(fixing the gauge freedom by imposing a reflecting gauge condition
on both $\mathcal{I}$ and $\text{\textgreek{g}}_{0}$). The conformal
boundary $\mathcal{I}$ is future complete, and the hypersurface $\{r=r_{0}\}$
is timelike in the past $J^{-}(\mathcal{I})$ of $\mathcal{I}$. 

In the case when $(r,\text{\textgreek{W}}^{2};\text{\textgreek{t}},\bar{\text{\textgreek{t}}})$
has a non-empty future event horizon $\mathcal{H}^{+}$, $\mathcal{H}^{+}$
is smooth and future complete and the curve $\{r=r_{0}\}$ has a spacelike
portion which lies in the future of $\mathcal{H}^{+}$ (see Figure
\ref{fig:domain_intro}). A necessary condition for $\mathcal{H}^{+}$
to be non-empty is that the total mass $\tilde{m}|_{\mathcal{I}}$
satisfies 
\begin{equation}
\frac{2\tilde{m}|_{\mathcal{I}}}{r_{0}}>1-\frac{1}{3}\Lambda r_{0}^{2}.\label{eq:BoundHorizonIntro}
\end{equation}

\end{customthm}

For a more detailed statement of Theorem \ref{thm:MaxDevelopmentIntro},
see Section \ref{sub:Well-posedness}.

\begin{figure}[h] 
\centering 
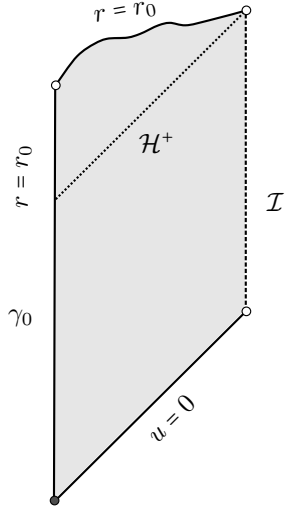 
\caption{Schematic depiction of the domain on which the maximal future development $(r,\Omega^2,\tau,\bar{\tau})$ of a smooth initial data set on $u=0$ (with reflecting boundary conditions on $\mathcal{I}$ and $\gamma_{0}$) is defined. Conformal infinity $\mathcal{I}$ is always complete in this setting. In the case when $J^{-}(\mathcal{I})$ does not cover all of the domain, the future event horizon $\mathcal{H}^{+}$ is non-empty and has infinite affine length. In this case, in addition to the mirror $\gamma_{0}$, there exists a spacelike piece of the boundary of the domain on which $r=r_{0}$. This boundary piece necessarily lies in the future of $\mathcal{H}^{+}$.}
\label{fig:domain_intro}
\end{figure}

\subsection{\label{sub:Cauchy-Intro}Theorem 2: Cauchy stability in a rough norm,
uniformly in $r_{0}$}

In view of the fact that $\mathcal{H}^{+}=\emptyset$ when (\ref{eq:BoundHorizonIntro})
holds, it can be readily deduced that, for fixed $r_{0}$, $(\mathcal{M}_{AdS},g_{AdS})$
is orbitally stable as a solution of (\ref{eq:EinsteinNullDust})
with reflecting boundary conditions on $\mathcal{I}$ and $\text{\textgreek{g}}_{0}$
under perturbations which are sufficiently small with respect to $r_{0}$.
Therefore, addressing the AdS instability conjecture in this setting
(a task that will take place in our companion paper \cite{MoschidisNullDust})
requires allowing $r_{0}$ to shrink to $0$ with the size of the
initial data. To this end, it will be necessary to obtain a Cauchy
stability statement for (\ref{eq:EinsteinNullDust}) with reflecting
boundary conditions on $\mathcal{I}$ and $\text{\textgreek{g}}_{0}$
which is \emph{uniform in $r_{0}$}.

For any $r_{0}>0$, let us define the distance function $dist_{r_{0}}(\cdot,\cdot)$
on the set of smooth and asymptotically AdS initial data for (\ref{eq:EinsteinNullDust})
so that, for any two initial data sets $\mathcal{S}=(r,\text{\textgreek{W}}^{2};\text{\textgreek{t}},\bar{\text{\textgreek{t}}})|_{u=0}$
and $\mathcal{S}'=(r',(\text{\textgreek{W}}')^{2};\text{\textgreek{t}}',\bar{\text{\textgreek{t}}}')|_{u=0}$:
\begin{align}
dist_{r_{0}}\big(\mathcal{S},\mathcal{S}'\big)\doteq\sup_{\bar{v}}\int_{u=0} & \Big|\frac{\bar{\text{\textgreek{t}}}(v)}{\big(|\text{\textgreek{r}}(v)-\text{\textgreek{r}}(\bar{v})|+\tan^{-1}(\sqrt{-\Lambda}r_{0})\big)\partial_{v}\text{\textgreek{r}}(v)}-\frac{\bar{\text{\textgreek{t}}}'(v)}{\big(|\text{\textgreek{r}}'(v)-\text{\textgreek{r}}'(\bar{v})|+\tan^{-1}(\sqrt{-\Lambda}r_{0})\big)\partial_{v}\text{\textgreek{r}}'(v)}\Big|\,(-\Lambda)dv+\label{eq:DistanceFunctionIntro}\\
 & +\sup_{u=0}\Big|\frac{2\tilde{m}}{r}-\frac{2\tilde{m}'}{r'}\Big|+\sqrt{-\Lambda}\big|\tilde{m}|_{r=+\infty}-\tilde{m}'|_{r=+\infty}\big|,\nonumber 
\end{align}
where 
\begin{equation}
\text{\textgreek{r}}(v)\doteq\tan^{-1}(\sqrt{-\Lambda}r)(v),
\end{equation}
and $\tilde{m}$ is the renormalised Hawking mass associated to $\mathcal{S}$,
defined as 
\begin{equation}
\tilde{m}\doteq\frac{r}{2}\Big(1-4\text{\textgreek{W}}^{-2}\partial_{u}r\partial_{v}r\Big)-\frac{1}{6}\Lambda r^{3}\label{eq:RebormalisedhawkingMassIntro}
\end{equation}
(with $\text{\textgreek{r}}'$, $\tilde{m}'$ defined similarly in
terms of $\mathcal{S}'$). 
\begin{rem*}
Denoting with $\mathcal{S}_{0}$ the trivial initial data set $(r_{vac},\text{\textgreek{W}}_{vac}^{2};0,0)$
in the standard gauge, where 
\begin{equation}
r_{vac}(v)\doteq\sqrt{-\frac{3}{\Lambda}}\tan\Big(\sqrt{-\frac{\Lambda}{12}}v\Big)
\end{equation}
and 
\begin{equation}
\text{\textgreek{W}}_{vac}^{2}=1-\frac{1}{3}\Lambda r_{vac}^{2},
\end{equation}
the distance $dist_{r_{0}}(\mathcal{S},\mathcal{S}_{0})$ of any initial
data set $\mathcal{S}$ from $\mathcal{S}_{0}$ measures the concentration
of the energy of $\mathcal{S}$ at scales comparable to $r_{0}$. 
\end{rem*}
Our second result is a Cauchy stability statement for (\ref{eq:EinsteinNullDust})
with reflecting boundary conditions on $r=r_{0}$ and $\mathcal{I}$,
in the initial data topology defined by (\ref{eq:DistanceFunctionIntro}),
\emph{independently of the precise value of $r_{0}$}:

\begin{customthm}{2}[rough version]\label{thm:CauchyIntro} 

Let $(r,\text{\textgreek{W}}^{2};\text{\textgreek{t}},\bar{\text{\textgreek{t}}})$
be a solution of (\ref{eq:EinsteinNullDust}) arising from a smooth
asymptotically AdS characteristic initial data set $\mathcal{S}_{in}$,
with reflecting boundary conditions on $\mathcal{I}$ and $\text{\textgreek{g}}_{0}$.
Let also $u_{*}>0$ be any value such that the domain $\mathcal{U}_{*}=\{0\le u\le u_{*}\}$
lies in the past of $\mathcal{I}$, i.\,e. 
\begin{equation}
\mathcal{U}_{*}\subset J^{-}(\mathcal{I}).
\end{equation}
 Then, for any $\text{\textgreek{e}}>0$, there exists a $\text{\textgreek{d}}>0$,
depending only on $\mathcal{U}$, $\text{\textgreek{e}}$ and $dist_{r_{0}}(\mathcal{S}_{in},\mathcal{S}_{0})$
(but \underline{not} $r_{0}$), with the following property: For
any other smooth asymptotically AdS initial data set $\mathcal{S}_{in}^{\prime}$
for (\ref{eq:EinsteinNullDust}) satisfying 
\begin{equation}
dist_{r_{0}}\big(\mathcal{S}_{in},\mathcal{S}_{in}^{\prime}\big)<\text{\textgreek{d}},\label{eq:CauchyNormIntro-1}
\end{equation}
the maximal development $(r',(\text{\textgreek{W}}')^{2};\text{\textgreek{t}}',\bar{\text{\textgreek{t}}}')$
satisfies 
\begin{equation}
\sup_{0\le\bar{u}\le u_{*}}dist_{r_{0}}\big(\mathcal{S}_{\bar{u}},\mathcal{S}_{\bar{u}}^{\prime}\big)<\text{\textgreek{e}},
\end{equation}
 where 
\begin{equation}
\mathcal{S}_{\bar{u}}\doteq(r',(\text{\textgreek{W}}')^{2};\text{\textgreek{t}}',\bar{\text{\textgreek{t}}}')|_{u=\bar{u}}
\end{equation}
(and similarly for $\mathcal{S}_{\bar{u}}^{\prime}$)

\end{customthm}

For a more precise statement of Theorem \ref{thm:CauchyIntro}, see
Section \ref{sub:Cauchy-stability}.
\begin{rem*}
Theorem \ref{thm:CauchyIntro} provides a Cauchy stability estimate
for the domain of outer communications of solutions to (\ref{eq:EinsteinNullDust}).
Restricting to the case when $\mathcal{S}_{in}=\mathcal{S}_{0}$ in
Theorem \ref{thm:CauchyIntro}, we thus readily obtain a Cauchy stability
estmate for the AdS spacetime $(\mathcal{M}_{AdS},g_{AdS})$ in the
topology defined by (\ref{eq:DistanceFunctionIntro}), indepedently
of the precise value of $r_{0}$. This fact will be important for
addressing the AdS instability conjecture in our companion paper \cite{MoschidisNullDust}.

In particular, for any $u_{*}>0$, the maximal development of any
initial data set $\mathcal{S}_{in}^{\prime}$ for (\ref{eq:EinsteinNullDust})
with reflecting boundary conditions on $\text{\textgreek{g}}_{0}$
and $\mathcal{I}$ will \underline{not} contain a trapped surface
for $0\le u\le u_{*}$, provided $dist_{r_{0}}(\mathcal{S}_{0},\mathcal{S}_{in}^{\prime})$
is small enough in terms of $u_{*}$ (independently of $r_{0}$).
\end{rem*}

\subsection{Outline of the paper}

This paper is organised as follows:

In Section \ref{sec:The-Einstein--Vlasov-system}, we will formulate
the spherically symmetric Einstein--radial massless Vlasov system
in double null coordinates, and we will introduce the notion of reflecting
boundary conditions for this system on timelike hypersurfaces.

In Section \ref{sec:CharacteristicInitialData}, we will introduce
the basic definitions related to the characteristic initial-bounary
value problem for the spherically symmetric Einstein--radial massless
Vlasov system.

In Section \ref{sec:The-main-results}, we will state in detail the
main results of this paper, namely the existence and uniqueness of
a maximal future development for the characteristic initial-bounary
value problem introduced in Section \ref{sec:CharacteristicInitialData},
as well as a Cauchy stability statement in a rough norm. The proofs
of these results will occupy sections \ref{sec:Well-posedness-for-the}
and \ref{sec:Proof-of-Cauchy}, respectively.

Finally, in Section \ref{sec:Ill-posedness} of the Appendix, we will
prove an ill-posedness result related to solutions of the spherically
symmetric Einstein--null dust system with non-empty axis of symmetry.

\subsection{Acknowledgements}

I would like to express my gratitude to my advisor Mihalis Dafermos
for many insightful discussions and crucial suggestions. I would also
like to thank Igor Rodnianski for numerous additional suggestions
and comments. Finally, I would like to thank DPMMS and King's College,
Cambridge, for their hospitality while this work was being written.

\section{\label{sec:The-Einstein--Vlasov-system}The spherically symmetric
Einstein--massless Vlasov system: the null dust reduction}

In this section, we will formulate the Einstein--massless Vlasov system,
reduced to the spherically symmetric and radial case, in a double
null coordinate gauge, following the conventions of \cite{DafermosRendall}.
We will also introduce the notion of reflecting boundary conditions
for this system on spherically symmetric timelike hypersurfaces, possibly
lying ``at infinity''.

\subsection{\label{sub:Spherically-symmetric-spacetimes}Spherically symmetric
spacetimes in a double null gauge}

In this paper, we will consider $3+1$ dimensional smooth Lorentzian
manifolds $(\mathcal{M},g)$ satisfying the following properties:

\begin{itemize}

\item{ $\mathcal{M}$ splits diffeomorphically as 
\begin{equation}
\mathcal{M}\simeq\mathcal{U}\times\mathbb{S}^{2}\label{eq:SphericallySymmetricmanifold}
\end{equation}
where $\mathcal{U}$ is an open domain of $\mathbb{R}^{2}$ with piecewise
Lipschitz boundary $\partial\mathcal{U}$.} 

\item{ In the Cartesian $(u,v)$ coordinates on $\mathcal{U}\subset\mathbb{R}^{2}$,
$g$ takes the double-null form 
\begin{equation}
g=-\text{\textgreek{W}}^{2}(u,v)dudv+r^{2}(u,v)g_{\mathbb{S}^{2}},\label{eq:SphericallySymmetricMetric}
\end{equation}
where $g_{\mathbb{S}^{2}}$ is the standard round metric on $\mathbb{S}^{2}$
and $\text{\textgreek{W}},r:\mathcal{U}\rightarrow(0,+\infty)$ are
smooth functions.} 

\item{ The function $r$ is bounded away from $0$ on $\mathcal{U}$:
\begin{equation}
\inf_{\mathcal{U}}r>0.\label{eq:NoAxisPresent}
\end{equation}
}

\end{itemize}

Note that (\ref{eq:SphericallySymmetricmanifold}) and (\ref{eq:SphericallySymmetricMetric})
imply that $(\mathcal{M},g)$ is time orientable. We will fix a time
orientation by requiring that $\partial_{u}+\partial_{v}$ is a future
directed vector field on $\mathcal{M}$.
\begin{rem*}
The form (\ref{eq:SphericallySymmetricMetric}) of the metic $g$
implies that the action of $SO(3)$ on $(\mathcal{M},g)$ through
rotations of the $\mathbb{S}^{2}$ factor of (\ref{eq:SphericallySymmetricmanifold})
is an isometric action and the function $r$ can be geometrically
defined as 
\begin{equation}
r(p)=\sqrt{\frac{Area(\mathcal{S}(p))}{4\pi}},\label{eq:DefinitionR}
\end{equation}
where $\mathcal{S}(p)$ is the $SO(3)$-orbit of $p\in\mathcal{M}$.
The condition (\ref{eq:NoAxisPresent}) implies that this $SO(3)$
action has no fixed points, i.\,e.~$(\mathcal{M},g)$ does not have
an axis of symmetry. 
\end{rem*}
For a Lorentzian manifold $(\mathcal{M},g)$ as above, we will also
define the \emph{Hawking mass} $m:\mathcal{M}\rightarrow\mathbb{R}$
as 
\begin{equation}
m=\frac{r}{2}\big(1-g(\nabla r,\nabla r)\big).
\end{equation}
 Viewed as a function on $\mathcal{U}$, $m$ can be expressed as:
\begin{equation}
m=\frac{r}{2}\big(1+4\text{\textgreek{W}}^{-2}\partial_{u}r\partial_{v}r\big).\label{eq:DefinitionHawkingMass}
\end{equation}
Note that (\ref{eq:DefinitionHawkingMass}) can be rearranged as 
\begin{equation}
\text{\textgreek{W}}^{2}=4\frac{(-\partial_{u}r)\partial_{v}r}{1-\frac{2m}{r}}.\label{eq:RelationHawkingMass}
\end{equation}

Any local coordinate chart $(y^{1},y^{2})$ on $\mathbb{S}^{2}$ yields
a local $(u,v,y^{1},y^{2})$ coordinate chart on $\mathcal{M}$. In
any such chart, the non-zero Christoffel symbols of \ref{eq:SphericallySymmetricMetric}
are computed as follows: 
\begin{gather}
\text{\textgreek{G}}_{uu}^{u}=\partial_{u}\log(\text{\textgreek{W}}^{2}),\hphantom{A}\text{\textgreek{G}}_{vv}^{v}=\partial_{u}\log(\text{\textgreek{W}}^{2}),\label{eq:ChristoffelSymbols}\\
\text{\textgreek{G}}_{AB}^{u}=\text{\textgreek{W}}^{-2}\partial_{v}(r^{2})(g_{\mathbb{S}^{2}})_{AB},\hphantom{\,}\text{\textgreek{G}}_{AB}^{v}=\text{\textgreek{W}}^{-2}\partial_{u}(r^{2})(g_{\mathbb{S}^{2}})_{AB},\nonumber \\
\text{\textgreek{G}}_{uB}^{A}=r^{-1}\partial_{u}r\text{\textgreek{d}}_{B}^{A},\hphantom{\,}\text{\textgreek{G}}_{vB}^{A}=r^{-1}\partial_{v}r\text{\textgreek{d}}_{B}^{A},\nonumber \\
\text{\textgreek{G}}_{BC}^{A}=(\text{\textgreek{G}}_{\mathbb{S}^{2}})_{BC}^{A},\nonumber 
\end{gather}
where the latin indices $A,B,C$ are associated to the spherical coordinates
$y^{1},y^{2}$, $\text{\textgreek{d}}_{B}^{A}$ is Kronecker delta
and $\text{\textgreek{G}}_{\mathbb{S}^{2}}$ are the Christoffel symbols
of $(\mathbb{S}^{2},g_{\mathbb{S}^{2}})$ in the $(y^{1},y^{2})$
chart.

For any pair of smooth maps $h_{1},h_{2}:\mathbb{R}\rightarrow\mathbb{R}$
which are strictly monotonic, we can introduce a new pair of double
null coordinates $(\bar{u},\bar{v})$ on $\mathcal{M}$ by defining
\begin{equation}
(\bar{u},\bar{v})=(h_{1}(u),h_{2}(v)).\label{eq:GeneralCoordinateTransformation}
\end{equation}
In these new coordinates, the metric $g$ takes the form 
\begin{equation}
g=-\bar{\text{\textgreek{W}}}^{2}(\bar{u},\bar{v})d\bar{u}d\bar{v}+r^{2}(\bar{u},\bar{v})g_{\mathbb{S}^{2}},\label{eq:SphericallySymmetricMetricNewGauge}
\end{equation}
where 
\begin{gather}
\bar{\text{\textgreek{W}}}^{2}(\bar{u},\bar{v})=\frac{1}{h_{1}^{\prime}h_{2}^{\prime}}\text{\textgreek{W}}^{2}\big(h_{1}^{-1}(\bar{u}),h_{2}^{-1}(\bar{v})\big),\label{eq:NewOmega}\\
r(\bar{u},\bar{v})=r\big(h_{1}^{-1}(\bar{u}),h_{2}^{-1}(\bar{v})\big).\label{eq:NewR}
\end{gather}
Throughout this paper, we will frequently make use of such coordinate
transformations, without renaming the coordinates each time. 

Note that $m$ is invariant under coordinate transformations as above;
that is to say, for $(\bar{u},\bar{v})$ defined by (\ref{eq:GeneralCoordinateTransformation}):
\begin{equation}
m(\bar{u},\bar{v})=m\big(h_{1}^{-1}(\bar{u}),h_{2}^{-1}(\bar{v})\big).
\end{equation}

\subsection{\label{sub:VlasovEquations}The massless Vlasov equation: The radial
reduction}

Let $(\mathcal{M},g)$ be as in Section \ref{sub:Spherically-symmetric-spacetimes}.
For any local coordinate chart $(x^{0},x^{1},x^{2},x^{3})$ on $\mathcal{M}$,
the associated momentum coordinate system $(p^{0},p^{1},p^{2},p^{3})$
on each fiber of $T\mathcal{M}$ is defined with respect to the coordinate
frame $\{\partial_{x^{a}}\}_{a=0}^{3}$. The geodesic flow on $T\mathcal{M}$
is then described by the following first order system: 
\begin{equation}
\begin{cases}
\frac{dx^{\text{\textgreek{a}}}}{dt} & =p^{\text{\textgreek{a}}}\\
\frac{dp^{\text{\textgreek{a}}}}{dt} & =-\text{\textgreek{G}}_{\text{\textgreek{b}\textgreek{g}}}^{\text{\textgreek{a}}}(x)p^{\text{\textgreek{a}}}p^{\text{\textgreek{b}}},
\end{cases}\label{eq:GeodesicFlow}
\end{equation}
where the Greek lowercase indices run from $0$ to $3$ and $\text{\textgreek{G}}_{\text{\textgreek{b}\textgreek{g}}}^{\text{\textgreek{a}}}$
are the Christoffel symbols associated to the $\{x^{a}\}_{a=0}^{3}$
coordinate chart.

The spherical symmetry of $(\mathcal{M},g)$ implies that the quantity
\begin{equation}
\mathcal{L}=r^{2}g_{AB}p^{A}p^{B},
\end{equation}
evaluated in the $p^{A}$-momentum coordinates associated to the $y^{A}$-spherical
coordinates in the double null coordinate chart $(u,v,y^{1},y^{2})$
on $\mathcal{M}$ (see Section \ref{sub:Spherically-symmetric-spacetimes}),
is constant along the geodesic flow. The quantity $\mathcal{L}$ associated
to any geodesic $\text{\textgreek{g}}$ of $\mathcal{M}$ is called
the \emph{angular momentum }of the geodesic. Geodesics for which $\mathcal{L}=0$
are called \emph{radial. }Geodesics which are null, future directed
and radial fall into two categories: the ingoing ones (for which $p^{v}=0$)
and the outgoing ones (for which $p^{u}=0$).

Let $f\ge0$ be a measure on $T\mathcal{M}$ which is constant along
the geodesic flow, i.\,e.~satisfies in any local coordinate chart
$(x^{0},x^{1},x^{2},x^{3})$ on $\mathcal{M}$ (with associated momentum
coordinates $(p^{0},p^{1},p^{2},p^{3})$): 
\begin{equation}
p^{\text{\textgreek{a}}}\partial_{x^{\text{\textgreek{a}}}}f-\text{\textgreek{G}}_{\text{\textgreek{b}\textgreek{g}}}^{\text{\textgreek{a}}}p^{\text{\textgreek{b}}}p^{\text{\textgreek{g}}}\partial_{p^{\text{\textgreek{a}}}}f=0,\label{eq:VlasovEquation}
\end{equation}
where $\text{\textgreek{G}}_{\text{\textgreek{b}\textgreek{g}}}^{\text{\textgreek{a}}}$
are the Christoffel symbols of $g$ in the chart $(x^{0},x^{1},x^{2},x^{3})$.
We will call $f$ a \emph{massless Vlasov field} if it is supported
on the set $P\subset T\mathcal{M}$ of null vectors, i.\,e.~on the
set 
\begin{equation}
g_{\text{\textgreek{a}\textgreek{b}}}(x)p^{\text{\textgreek{a}}}p^{\text{\textgreek{b}}}=0.
\end{equation}

The \emph{energy momentum }tensor $T_{\text{\textgreek{a}\textgreek{b}}}$
of $f$ is a symmetric $(0,2)$-tensor on $\mathcal{M}$ defined formally
by the epression 
\begin{equation}
T_{\text{\textgreek{a}\textgreek{b}}}(x)=\int_{\text{\textgreek{p}}^{-1}(x)}p_{\text{\textgreek{a}}}p_{\text{\textgreek{b}}}f,\label{eq:EnergyMomentumTensor}
\end{equation}
where $\text{\textgreek{p}}^{-1}(x)$ denotes the fiber of $T\mathcal{M}$
over $x\in\mathcal{M}$ and
\begin{equation}
p_{\text{\textgreek{g}}}=g_{\text{\textgreek{g}}\text{\textgreek{d}}}(x)p^{\text{\textgreek{d}}}.
\end{equation}

\begin{rem*}
In this paper, we will only consider distributions $f$ for which
the expression (\ref{eq:EnergyMomentumTensor}) is finite and depends
smoothly on $x\in\mathcal{M}$. 
\end{rem*}
We will consider only distributions $f$ which are spherically symmetric,
i.\,e.~invariant under the action of $SO(3)$ on $\mathcal{M}$.
In that case, in any $(u,v,y^{1},y^{2})$ local coordinate chart as
in Section \ref{sub:Spherically-symmetric-spacetimes}, the energy-momentum
tensor $T$ is of the form 
\begin{equation}
T=T_{uu}(u,v)du^{2}+2T_{uv}(u,v)dudv+T_{vv}(u,v)dv^{2}+T_{AB}(u,v)dx^{A}dx^{B}.\label{eq:SphericallySymmetricTensor}
\end{equation}

A \emph{radial }massless Vlasov field $f$ is a massless Vlasov field
$f$ supported only on radial null geodesics. In view of the separation
of radial null geodesics into ingoing and outgoing, a spherically
symmetric, radial massless Vlasov field $f$ takes the following form
in any $(u,v,y^{1},y^{2})$ local coordinate chart as in Section \ref{sub:Spherically-symmetric-spacetimes}
(with associated momentum coordinates $(p^{u},p^{v},p^{1},p^{2})$):
\begin{equation}
f(u,v,y^{1},y^{2};p^{u},p^{v},p^{1},p^{2})=\big(\bar{f}_{in}(u,v;p^{u})+\bar{f}_{out}(u,v;p^{v})\big)\text{\textgreek{d}}\big(\sqrt{(g_{\mathbb{S}^{2}})_{AB}p^{A}p^{B}}\big)\text{\textgreek{d}}(\text{\textgreek{W}}^{2}p^{u}p^{v}),\label{eq:RadialVlasovField}
\end{equation}
where $\bar{f}_{in},\bar{f}_{out}\ge0$ and $\text{\textgreek{d}}$
is the Dirac delta funcion on $\mathbb{R}$.%
\footnote{Note that the arguments of the $\text{\textgreek{d}}$ functions in
the right hand side of (\ref{eq:RadialVlasovField}) are invariannt
under transormations of the spherically symmetric double null coordinate
system.%
} In this case, we can compute the components of the energy momentum
tensor (\ref{eq:EnergyMomentumTensor}) as follows: 
\begin{gather}
T_{uu}(u,v)=\int_{0}^{+\infty}\text{\textgreek{W}}^{4}(p^{v})^{2}\bar{f}_{out}(u,v;p^{v})\, r^{2}\frac{dp^{v}}{p^{v}},\label{eq:T_uuComponent}\\
T_{vv}(u,v)=\int_{0}^{+\infty}\text{\textgreek{W}}^{4}(p^{u})^{2}\bar{f}_{in}(u,v;p^{u})\, r^{2}\frac{dp^{u}}{p^{u}}\label{eq:T_vvComponent}
\end{gather}
and 
\begin{equation}
T_{uv}=T_{AB}=0.
\end{equation}

\begin{rem*}
In this paper, we will only consider the case when $\bar{f}_{in},\bar{f}_{out}$
are smooth functions which are compactly supported in their last argument.
\end{rem*}
The expression (\ref{eq:RadialVlasovField}) implies that equation
(\ref{eq:VlasovEquation}) is equivalent to the following system for
$\bar{f}_{in}$ and $\bar{f}_{out}$:
\begin{gather}
\partial_{u}(\text{\textgreek{W}}^{4}r^{4}p^{u}\bar{f}_{in})+p^{u}\partial_{p^{u}}(\text{\textgreek{W}}^{4}r^{4}p^{u}\bar{f}_{in})=0,\label{eq:IngoingEquation}\\
\partial_{v}(\text{\textgreek{W}}^{4}r^{4}p^{v}\bar{f}_{out})+p^{v}\partial_{p^{v}}(\text{\textgreek{W}}^{4}r^{4}p^{v}\bar{f}_{out})=0.\label{eq:OutgoingEquation}
\end{gather}
The relations (\ref{eq:IngoingEquation})--(\ref{eq:OutgoingEquation})
imply that: 
\begin{gather}
\partial_{v}(r^{2}T_{uu})=0,\label{eq:EquationT_uu}\\
\partial_{u}(r^{2}T_{vv})=0.\label{eq:EquationT_vv}
\end{gather}

\begin{rem*}
Under a double null coordinate transformation of the form (\ref{eq:GeneralCoordinateTransformation}),
$\bar{f}_{in}$ and $\bar{f}_{out}$ transform as 
\begin{equation}
\bar{f}_{in}^{(new)}(h_{1}(u),h_{2}(v);h_{1}^{\prime}(u)p)=\bar{f}_{in}\big(u,v;p\big)\label{eq:NewIngoingVlasov}
\end{equation}
and 
\begin{equation}
\bar{f}_{out}^{(new)}(h_{1}(u),h_{2}(v);h_{2}^{\prime}(v)p)=\bar{f}_{out}\big(u,v;p\big),\label{eq:NewOutgoingVlasov}
\end{equation}
respectively.
\end{rem*}

\subsection{\label{sub:The-Einstein-equations}The spherically symmetric Einstein--radial
massless Vlasov and Einstein--null dust system}

The \emph{Einstein--Vlasov} system with cosmological constant $\Lambda<0$,
for a smooth Lorentzian manifold $(\mathcal{M},g)$ and a non-negative
measure $f$ on $T\mathcal{M}$, takes the form 
\begin{equation}
\begin{cases}
Ric_{\text{\textgreek{m}\textgreek{n}}}(g)-\frac{1}{2}R(g)g_{\text{\textgreek{m}\textgreek{n}}}+\text{\textgreek{L}}g_{\text{\textgreek{m}\textgreek{n}}}=8\text{\textgreek{p}}T_{\text{\textgreek{m}\textgreek{n}}},\\
p^{\text{\textgreek{a}}}\partial_{x^{\text{\textgreek{a}}}}f-\text{\textgreek{G}}_{\text{\textgreek{b}\textgreek{g}}}^{\text{\textgreek{a}}}p^{\text{\textgreek{b}}}p^{\text{\textgreek{g}}}\partial_{p^{\text{\textgreek{a}}}}f=0,
\end{cases}\label{eq:EinsteinVlasovEquations}
\end{equation}
where $T_{\text{\textgreek{m}\textgreek{n}}}$ is expressed in terms
of $f$ by (\ref{eq:EnergyMomentumTensor}) (see \cite{Andreasson2011}).

Reducing (\ref{eq:EinsteinVlasovEquations}) to the case where $(\mathcal{M},g)$
is a spherically symmetric spacetime as in Section \ref{sub:Spherically-symmetric-spacetimes}
and $f$ is a radial massless Vlasov field (i.\,e.~has the form
(\ref{eq:RadialVlasovField})), we obtain the following system for
$(r,\text{\textgreek{W}}^{2},\bar{f}_{in},\bar{f}_{out})$:
\begin{align}
\partial_{u}\partial_{v}(r^{2})= & -\frac{1}{2}(1-\Lambda r^{2})\text{\textgreek{W}}^{2},\label{eq:RequationFinal}\\
\partial_{u}\partial_{v}\log(\text{\textgreek{W}}^{2})= & \frac{\text{\textgreek{W}}^{2}}{2r^{2}}\big(1+4\text{\textgreek{W}}^{-2}\partial_{u}r\partial_{v}r\big),\label{eq:OmegaEquationFinal}\\
\partial_{v}(\text{\textgreek{W}}^{-2}\partial_{v}r)= & -4\pi rT_{vv}\text{\textgreek{W}}^{-2},\label{eq:ConstrainVFinal}\\
\partial_{u}(\text{\textgreek{W}}^{-2}\partial_{u}r)= & -4\pi rT_{uu}\text{\textgreek{W}}^{-2},\label{eq:ConstraintUFinal}\\
\partial_{u}(\text{\textgreek{W}}^{4}r^{4}p^{u}\bar{f}_{in})= & -p^{u}\partial_{p^{u}}(\text{\textgreek{W}}^{4}r^{4}p^{u}\bar{f}_{in}),\label{eq:IngoingVlasovFinal}\\
\partial_{v}(\text{\textgreek{W}}^{4}r^{4}p^{v}\bar{f}_{out})= & -p^{v}\partial_{p^{v}}(\text{\textgreek{W}}^{4}r^{4}p^{v}\bar{f}_{out}),\label{eq:OutgoingVlasovFinal}
\end{align}
where $T_{uu},T_{vv}$ are expressed in terms of $\bar{f}_{out},\bar{f}_{in}$
by (\ref{eq:T_uuComponent}), (\ref{eq:T_vvComponent}), respectively
(for the spherically symmetric reduction of the \emph{massive} Einstein--Vlasov
system in double null coordinates, see \cite{DafermosRendall}). Notice
that the system (\ref{eq:RequationFinal})--(\ref{eq:OutgoingVlasovFinal})
reduces to the following system for $(r,\text{\textgreek{W}}^{2},T_{uu},T_{vv})$:
\begin{align}
\partial_{u}\partial_{v}(r^{2})= & -\frac{1}{2}(1-\Lambda r^{2})\text{\textgreek{W}}^{2},\label{eq:RequationFinal-2}\\
\partial_{u}\partial_{v}\log(\text{\textgreek{W}}^{2})= & \frac{\text{\textgreek{W}}^{2}}{2r^{2}}\big(1+4\text{\textgreek{W}}^{-2}\partial_{u}r\partial_{v}r\big),\label{eq:OmegaEquationFinal-2}\\
\partial_{v}(\text{\textgreek{W}}^{-2}\partial_{v}r)= & -4\pi rT_{vv}\text{\textgreek{W}}^{-2},\label{eq:ConstrainVFinal-1}\\
\partial_{u}(\text{\textgreek{W}}^{-2}\partial_{u}r)= & -4\pi rT_{uu}\text{\textgreek{W}}^{-2},\label{eq:ConstraintUFinal-1}\\
\partial_{u}(r^{2}T_{vv})= & 0,\label{eq:IngoingConservationClosed}\\
\partial_{v}(r^{2}T_{uu})= & 0.\label{eq:OutgoingConservationClosed}
\end{align}

\begin{rem*}
The system (\ref{eq:RequationFinal-2})--(\ref{eq:OutgoingConservationClosed})
is the Einstein--null dust system with both ingoing and outgoing dust
(see \cite{Rendall1997}). 
\end{rem*}
Let us define the \emph{renormalised Hawking mass} by the relation
\begin{equation}
\tilde{m}\doteq m-\frac{1}{6}\Lambda r^{3}.\label{eq:RenormalisedHawkingMass}
\end{equation}
From equations (\ref{eq:RequationFinal-2})--(\ref{eq:OutgoingConservationClosed}),
we can formally obtain the following system for $(r,\tilde{m},T_{uu},T_{vv})$
(valid in the region where $\partial_{u}r<0<\partial_{v}r$ and $1-\frac{2m}{r}>0$):

\begin{align}
\partial_{u}\log\big(\frac{\partial_{v}r}{1-\frac{2m}{r}}\big)= & -4\pi r^{-1}\frac{r^{2}T_{uu}}{-\partial_{u}r},\label{eq:DerivativeInUDirectionKappa}\\
\partial_{v}\log\big(\frac{-\partial_{u}r}{1-\frac{2m}{r}}\big)= & 4\pi r^{-1}\frac{r^{2}T_{vv}}{\partial_{v}r},\label{eq:DerivativeInVDirectionKappaBar}\\
\partial_{u}\partial_{v}r= & -\frac{2\tilde{m}-\frac{2}{3}\Lambda r^{3}}{r^{2}}\frac{(-\partial_{u}r)\partial_{v}r}{1-\frac{2m}{r}},\label{eq:EquationRForProof}\\
\partial_{u}\tilde{m}= & -2\pi\frac{\big(1-\frac{2m}{r}\big)}{-\partial_{u}r}r^{2}T_{uu},\label{eq:DerivativeTildeUMass}\\
\partial_{v}\tilde{m}= & 2\pi\frac{\big(1-\frac{2m}{r}\big)}{\partial_{v}r}r^{2}T_{vv}\label{eq:DerivativeTildeVMass}\\
\partial_{u}(r^{2}T_{vv})= & 0,\label{eq:ConservationT_vv}\\
\partial_{v}(r^{2}T_{uu})= & 0.\label{eq:ConservationT_uu}
\end{align}

\subsection{\label{sub:Reflective-boundary-conditions}Reflection of radial null
geodesics and the reflecting boundary condition}

Let $(\mathcal{M},g)$ be as in Section \ref{sub:Spherically-symmetric-spacetimes}.
Recall that $\mathcal{M}\simeq\mathcal{U}\times\mathbb{S}^{2}$ (see
(\ref{eq:SphericallySymmetricmanifold})), where $\mathcal{U}\subset\mathbb{R}^{2}$
has piecewise Lipschitz boundary $\partial\mathcal{U}$. Let $\partial_{tim}\mathcal{U}$
be the subset of $\partial\mathcal{U}$ consisting of a union of connected,
timelike Lipschitz curves with respect to the comparison metric 
\begin{equation}
g_{comp}=-dudv\label{eq:ComparisonUVMetric}
\end{equation}
 on $\mathbb{R}^{2}$. Recall that a connected Lipschitz curve $\text{\textgreek{g}}$
in $\mathbb{R}^{2}$ is said to be timelike with respect to (\ref{eq:ComparisonUVMetric})
if, for every point $p=(u_{0},v_{0})\in\text{\textgreek{g}}$, we
have 
\begin{equation}
\text{\textgreek{g}}\backslash p\subset I^{-}(p)\cup I^{+}(p)\doteq\big(\{u<u_{0}\}\cap\{v<v_{0}\}\big)\cup\big(\{u>u_{0}\}\cap\{v>v_{0}\}\big).
\end{equation}

Let us also fix $w:\mathcal{U}\cup\partial_{tim}\mathcal{U}\rightarrow\mathbb{R}$
to be a smooth boundary defining function of $\partial_{tim}\mathcal{U}$,
i.\,e.~$w|_{\partial_{tim}\mathcal{U}}=0$, $dw|_{\partial_{tim}\mathcal{U}}\neq0$
and $w>0$ on $\mathcal{U}$. We can split $\partial_{tim}\mathcal{U}$
into its ``left'' and ``right'' components as 
\begin{equation}
\partial_{tim}\mathcal{U}=\partial_{tim}^{\vdash}\mathcal{U}\cup\partial_{tim}^{\dashv}\mathcal{U},
\end{equation}
where 
\begin{gather*}
\partial_{tim}^{\vdash}\mathcal{U}=\big\{(u_{0},v_{0})\in\partial_{tim}\mathcal{U}\,:\,\partial_{v}w(u_{0},v_{0})>0\big\},\\
\partial_{tim}^{\dashv}\mathcal{U}=\big\{(u_{0},v_{0})\in\partial_{tim}\mathcal{U}\,:\,\partial_{v}w(u_{0},v_{0})<0\big\}.
\end{gather*}

\begin{rem*}
Notice that any future directed radial null geodesic of $\mathcal{M}=\mathcal{U}\times\mathbb{S}^{2}$
with a future limiting point on $\partial_{tim}^{\vdash}\mathcal{U}\times\mathbb{S}^{2}$
(in the ambient $\mathbb{R}^{2}\times\mathbb{S}^{2}$ topology of
$clos(\mathcal{U})\times\mathbb{S}^{2}$) is necessarily ingoing.
Similarly, future directed radial null geodesics ``terminating''
at $\partial_{tim}^{\dashv}\mathcal{U}\times\mathbb{S}^{2}$ are necessarily
outgoing.
\end{rem*}
\begin{figure}[h] 
\centering 
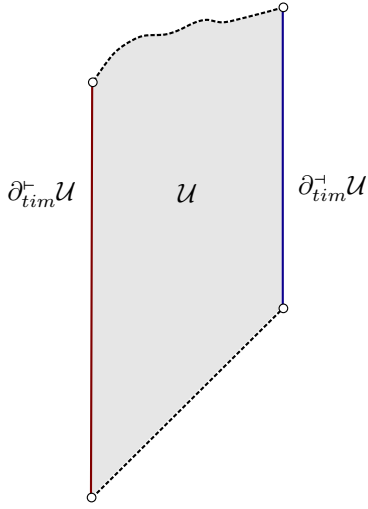 
\caption{Let $\mathcal{U}\subset \mathbb{R}^{2}$ be a domain as depicted above. The timelike part $\partial_{tim}\mathcal{U}$ of the boundary $\partial \mathcal{U}$ splits as the union of a "left" component $\partial^{\vdash}_{tim}\mathcal{U}$ and a "right" component $\partial^{\dashv}_{tim}\mathcal{U}$. While, in general, it is not necessary that $\partial^{\vdash}_{tim}\mathcal{U}$ and  $\partial^{\dashv}_{tim}\mathcal{U}$ are straight line segments, in the rest of the paper we will only consider domains $\mathcal{U}$ with this property (see Definition \ref{def:DevelopmentSets}).}
\end{figure}

We will define the reflection of radial null geodesics of $\mathcal{M}$
on $\partial_{tim}\mathcal{U}\times\mathbb{S}^{2}$ as follows:
\begin{defn}
Let $\text{\textgreek{z}}_{0}:[0,c)\rightarrow\mathcal{M}$, $c\le+\infty$,
be a future directed radial null geodesic of $\mathcal{M}=\mathcal{U}\times\mathbb{S}^{2}$
parametrised so that $\nabla_{\dot{\text{\textgreek{z}}}_{0}}\dot{\text{\textgreek{z}}}_{0}=0$,
such that, in the ambient $\mathbb{R}^{2}\times\mathbb{S}^{2}$ manifold,
the limit $\lim_{t\rightarrow c^{-}}\text{\textgreek{z}}_{0}(t)\doteq(u_{0},v_{0},y_{0}^{1},y_{0}^{2})$
exists and belongs to $\partial_{tim}\mathcal{U}\times\mathbb{S}^{2}$.
Then, the reflection of $\text{\textgreek{z}}_{0}$ on $\partial_{tim}\mathcal{U}\times\mathbb{S}^{2}$
at $(u_{0},v_{0},x_{0}^{1},x_{0}^{2})$ is defined as the unique inextendible,
future directed and radial null geodesic $\bar{\text{\textgreek{z}}}_{0}:(a,b)\rightarrow\mathcal{M}$
of $(\mathcal{M},g)$, $-\infty\le a<b\le+\infty$, parametrised so
that $\nabla_{\dot{\bar{\text{\textgreek{z}}}}_{0}}\dot{\bar{\text{\textgreek{z}}}}_{0}=0$,
satisfying the following conditions:

\begin{enumerate}

\item $\bar{\text{\textgreek{z}}}_{0}$ emanates from $(u_{0},v_{0},x_{0}^{1},x_{0}^{2})$,
i.\,e.:

\begin{equation}
\lim_{t\rightarrow a^{+}}\bar{\text{\textgreek{z}}}_{0}(t)=\lim_{t\rightarrow c^{-}}\text{\textgreek{z}}_{0}(t)=(u_{0},v_{0},x_{0}^{1},x_{0}^{2}).\label{eq:SameReflectingPoint}
\end{equation}

\item For any $0<w_{*}\ll1$, defining $t[w_{*}]\in[0,c)$ and $\bar{t}[w_{*}]\in(a,b)$
implicitly by 
\[
w\big(\text{\textgreek{z}}_{0}(t[w_{*}])\big)=w_{*}
\]
and 
\[
w\big(\bar{\text{\textgreek{z}}}_{0}(\bar{t}[w_{*}])\big)=w_{*},
\]
the following relation holds:
\begin{equation}
\lim_{w_{*}\rightarrow0^{+}}\frac{\frac{d}{dt}\big(w\circ\bar{\text{\textgreek{z}}}_{0}\big)\big|_{t=\bar{t}[w_{*}]}}{\frac{d}{dt}\big(w\circ\text{\textgreek{z}}_{0}\big)\big|_{t=t[w_{*}]}}=-1.\label{eq:ReflectionParametrisation}
\end{equation}

\end{enumerate}
\end{defn}
Note that the above definition is independent of the specific choice
of the boundary defining function $w$ and the coordinate chart on
$\mathcal{U}$, while, in view of (\ref{eq:ReflectionParametrisation}),
the parametrisation of the reflected geodesic is completely determined
in the case $c<+\infty$, and determined up to a translation $t\rightarrow t+t_{0}$
in the case $c=+\infty$. Notice also that the reflection of an ingoing
radial null geodesic is an outgoing one, and vice-versa.
\begin{rem*}
In the next sections, we will only consider the reflection of radial
null geodesics on parts of $\partial_{tim}\mathcal{U}$ for which
either $r-r_{0}$ (for some $r_{0}>0$) or $1/r$ is a boundary defining
function.\end{rem*}
\begin{defn}
A radial massless Vlasov field $f$ on $T\mathcal{M}$ will be said
to satisfy the \emph{reflecting boundary condition} on $\partial_{tim}\mathcal{U}\times\mathbb{S}^{2}$
if it is invariant under reflections of radial null geodesics on $\partial_{tim}\mathcal{U}\times\mathbb{S}^{2}$. 
\end{defn}
It can be readily verified that a radial massless Vlasov field satisfies
the reflecting boundary condition on $\partial_{tim}\mathcal{U}\times\mathbb{S}^{2}$
if and only if its ingoing and outgoing components $\bar{f}_{in},\bar{f}_{out}$
(see \ref{eq:RadialVlasovField}) satisfy the following boundary conditions
(as a consequence of (\ref{eq:SameReflectingPoint}), (\ref{eq:ReflectionParametrisation})):
\begin{itemize}
\item For any $(u_{0},v_{0})\in\partial_{tim}^{\vdash}\mathcal{U}$ and
any $p>0$:
\begin{equation}
\lim_{h\rightarrow0^{+}}\Bigg(\frac{\bar{f}_{out}\big(u_{0},v_{0}+h;\,\frac{-\partial_{u}w}{\partial_{v}w}(u_{0},v_{0})\cdot\text{\textgreek{W}}^{-2}(u_{0},v_{0}+h)\cdot p\big)}{\bar{f}_{in}\big(u_{0}-h,v_{0};\,\text{\textgreek{W}}^{-2}(u_{0}-h,v_{0})\cdot p\big)}\Bigg)=1.\label{eq:LeftBoundaryCondition}
\end{equation}

\item For any $(u_{1},v_{1})\in\partial_{tim}^{\dashv}\mathcal{U}$ and
any $p>0$:
\begin{equation}
\lim_{h\rightarrow0^{+}}\Bigg(\frac{\bar{f}_{in}\big(u_{1}+h,v_{1};\,\frac{-\partial_{v}w}{\partial_{u}w}(u_{1},v_{1})\cdot\text{\textgreek{W}}^{-2}(u_{1}+h,v_{1})\cdot p\big)}{\bar{f}_{out}\big(u_{1},v_{1}-h;\,\text{\textgreek{W}}^{-2}(u_{1},v_{1}-h)\cdot p\big)}\Bigg)=1.\label{eq:RightBoundaryCondition}
\end{equation}

\end{itemize}
Note that the relations (\ref{eq:LeftBoundaryCondition}) and (\ref{eq:RightBoundaryCondition})
for $\bar{f}_{in},\bar{f}_{out}$ imply the following boundary relations
for the components (\ref{eq:T_uuComponent})--(\ref{eq:T_vvComponent})
of the energy momentum tensor $T$: 
\begin{itemize}
\item For any $(u_{0},v_{0})\in\partial_{tim}^{\vdash}\mathcal{U}$:
\begin{equation}
\lim_{h\rightarrow0^{+}}\frac{r^{2}T_{uu}(u_{0},v_{0}+h)}{r^{2}T_{vv}(u_{0}-h,v_{0})}=\Big(\frac{-\partial_{u}w}{\partial_{v}w}(u_{0},v_{0})\Big)^{2}.\label{eq:LeftBoundaryConditionT}
\end{equation}

\item For any $(u_{1},v_{1})\in\partial_{tim}^{\dashv}\mathcal{U}$:
\begin{equation}
\lim_{h\rightarrow0^{+}}\frac{r^{2}T_{vv}(u_{1}+h,v_{1})}{r^{2}T_{uu}(u_{1},v_{1}-h)}=\Big(\frac{-\partial_{v}w}{\partial_{u}w}(u_{0},v_{0})\Big)^{2}.\label{eq:RightBoundaryConditionT}
\end{equation}

\end{itemize}

\section{\label{sec:CharacteristicInitialData}The boundary--characteristic
initial value problem}

In this Section, we will formulate the asymptotically AdS initial-bounary
value problem for the system (\ref{eq:RequationFinal})--(\ref{eq:OutgoingVlasovFinal})
with reflecting boundary conditions on $\mathcal{I}$ and the timelike
portion of $\{r=r_{0}\}$ and $\mathcal{I}$.

\subsection{Asymptotically AdS characteristic initial data}

We will define the notion of a boundary-characteristic initial data
set for (\ref{eq:RequationFinal})--(\ref{eq:OutgoingVlasovFinal})
as follows:
\begin{defn}
\label{def:TypeII}For any $v_{1}<v_{2}$ and any $r_{0}>0$, let
$r_{/}:[v_{1},v_{2})\rightarrow[r_{0},+\infty)$, $\text{\textgreek{W}}_{/}:[v_{1},v_{2})\rightarrow(0,+\infty)$
and $\bar{f}_{in/},\bar{f}_{out/}:[v_{1},v_{2})\times(0,+\infty)\rightarrow[0,+\infty)$
be smooth functions, such that 
\begin{equation}
\lim_{v\rightarrow v_{2}}r_{/}(v)=+\infty\label{eq:RGoesToInfinity}
\end{equation}
and 
\begin{equation}
r_{/}(v_{1})=r_{0}.
\end{equation}
Let also $(\partial_{u}r)_{/}:[v_{1},v_{2})\rightarrow(-\infty,0)$
be defined by 
\begin{equation}
(\partial_{u}r)_{/}(v)=\frac{1}{r_{/}(v)}\Big(-r_{0}\partial_{v}r_{/}(v_{1})-\frac{1}{4}\int_{v1}^{v}(1-\Lambda r_{/}^{2}(\bar{v}))\text{\textgreek{W}}_{/}^{2}(\bar{v})\, d\bar{v}\Big).\label{eq:TransversalDerivativeU-1}
\end{equation}
 We will call $(r_{/},\text{\textgreek{W}}_{/}^{2},\bar{f}_{in/},\bar{f}_{out/})$
an \emph{asymptotically AdS boundary-characteristic initial data set}
on $[v_{1},v_{2})$ for the system (\ref{eq:RequationFinal})--(\ref{eq:OutgoingVlasovFinal})
satisfying the reflecting gauge condition at $r=r_{0},+\infty$ if:

\begin{itemize}

\item{ The constraint equation (\ref{eq:ConstrainVFinal}) is satisfied
by $(r_{/},\text{\textgreek{W}}_{/}^{2})$, i.\,e: 
\begin{equation}
\partial_{v}(\text{\textgreek{W}}_{/}^{-2}\partial_{v}r_{/})=-4\pi r_{/}(T_{vv})_{/}\text{\textgreek{W}}_{/}^{-2},\label{eq:ConstraintVDef-1}
\end{equation}
 where 
\begin{equation}
(T_{vv})_{/}(v)\doteq\int_{0}^{+\infty}\text{\textgreek{W}}_{/}^{4}(v)(p^{u})^{2}\bar{f}_{in/}(v;p^{u})\, r_{/}^{2}(v)\frac{dp^{u}}{p^{u}}.\label{eq:EnergyMomentumIntialRight-1}
\end{equation}
}

\item{ The functions $\bar{f}_{out/},\bar{f}_{in/}$ satisfy the
following conditions at $v=v_{1},v_{2}$ for any $p>0$:
\begin{equation}
\frac{\bar{f}_{out/}\big(v_{1};\,\frac{-(\partial_{u}r)_{/}}{\partial_{v}r_{/}}(v_{1})\cdot\text{\textgreek{W}}_{/}^{-2}(v_{1})\cdot p\big)}{\bar{f}_{in/}\big(v_{1};\,\text{\textgreek{W}}_{/}^{-2}(v_{1})\cdot p\big)}=1\label{eq:LeftBoundaryConditionInitialData}
\end{equation}
and
\begin{equation}
\lim_{h\rightarrow0^{+}}\Bigg(\frac{\bar{f}_{in/}\big(v_{2}-h;\,\frac{\partial_{v}r_{/}}{-(\partial_{u}r)_{/}}(v_{2}-h)\cdot\text{\textgreek{W}}_{/}^{-2}(v_{2}-h)\cdot p\big)}{\bar{f}_{out/}\big(v_{2}-h;\,\text{\textgreek{W}}_{/}^{-2}(v_{2}-h)\cdot p\big)}\Bigg)=1.\label{eq:RightBoundaryConditionInitialData}
\end{equation}
}

\item{ The initial transversal derivative $(\partial_{u}r)_{/}$
satisfies 
\begin{equation}
\lim_{v\rightarrow v_{2}^{-}}\frac{(\partial_{u}r)_{/}}{\partial_{v}r_{/}}=1.\label{eq:GaugeInfinityInitialData}
\end{equation}
}

\item{ The function $\bar{f}_{out/}$ solves (\ref{eq:OutgoingVlasovFinal}),
i.\,e: 
\begin{gather}
\partial_{v}\big(\text{\textgreek{W}}_{/}^{4}(v)r_{/}^{4}(v)p^{v}\bar{f}_{out/}(v,p^{v})\big)+p^{v}\partial_{p^{v}}\big(\text{\textgreek{W}}_{/}^{4}(v)r_{/}^{4}(v)p^{v}\bar{f}_{out/}(v,p^{v})\big)=0.\label{eq:OutgoingEquationCombatibility-1}
\end{gather}
}

\end{itemize}\end{defn}
\begin{rem*}
The constraint equation (\ref{eq:ConstraintVDef-1}) implies that
\begin{equation}
\partial_{v}(\text{\textgreek{W}}_{/}^{-2}\partial_{v}r_{/})\le0.
\end{equation}
Therefore, in view of (\ref{eq:RGoesToInfinity}), we can bound for
all $v\in[v_{1},v_{2})$: 
\begin{equation}
\partial_{v}r_{/}(v)>0.\label{eq:NonTrappedInitialData}
\end{equation}

In Section \ref{sub:A-wider-class-Initial-data}, we will introduce
two more classes of characteristic initial data for the system (\ref{eq:RequationFinal})--(\ref{eq:OutgoingVlasovFinal}).
\end{rem*}
We will also define the initial Hawking mass $m_{/}$ and the initial
renormalised Hawking mass $\tilde{m}_{/}$, in accordance with (\ref{eq:RelationHawkingMass}),
(\ref{eq:RenormalisedHawkingMass}), as follows:
\begin{defn}
Let $(r_{/},\text{\textgreek{W}}_{/}^{2},\bar{f}_{in/},\bar{f}_{out/})$
be an asymptotically AdS boundary-characteristic initial data set
on $[v_{1},v_{2})$ with reflecting gauge conditions at $r=r_{0},+\infty$.
We will define the initial Hawking mass $m_{/}$ and initial renormalised
Hawking mass $\tilde{m}_{/}$ on $[v_{1}.v_{2})$ by the relations
\begin{equation}
m_{/}\doteq\frac{r_{/}}{2}\big(1-4\text{\textgreek{W}}_{/}^{-2}(\partial_{u}r)_{/}\partial_{v}r_{/}\big),\label{eq:DefinitionHawkingMassCharacteristic-1}
\end{equation}
and 
\begin{equation}
\tilde{m}_{/}\doteq m_{/}-\frac{1}{6}\Lambda r_{/}^{3}.\label{eq:RenormalisedHawkingMassCharacteristic}
\end{equation}

\end{defn}

\subsection{Developments with reflecting boundary conditions on $\text{\textgreek{g}}_{0}$
and $\mathcal{I}$}

In the following sections, we will only consider solutions $(r,\text{\textgreek{W}}^{2},\bar{f}_{in},\bar{f}_{out})$
to (\ref{eq:RequationFinal})--(\ref{eq:OutgoingVlasovFinal}) satisfying
a reflecting gauge condition on $\partial_{tim}\mathcal{U}$. This
condition fixes $\partial_{tim}\mathcal{U}$ to be a union of vertical
straight lines in the $(u,v)$-plane. This motivates defining the
following class of domains $\mathcal{U}\subset\mathbb{R}^{2}$:

\begin{figure}[h] 
\centering 
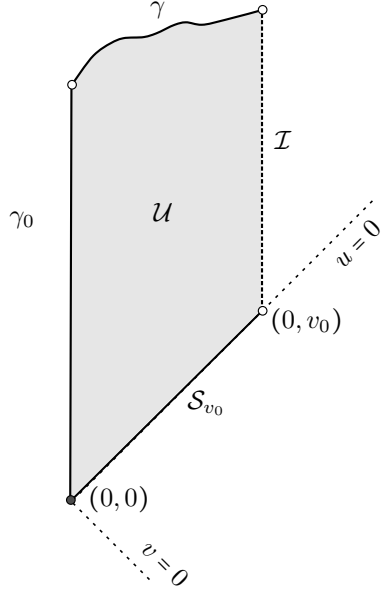 
\caption{A typical domain $\mathcal{U}\in\mathscr{U}_{v_{0}}$ would be as depicted above. In the case when the boundary set $\gamma$ is empty, it is necessary that both $\gamma_{0}$ and $\mathcal{I}$ are unbounded (i.\,e.~extend all the way to $u+v=\infty$).}
\end{figure}
\begin{defn}
\label{def:DevelopmentSets}For any $v_{0}>0$, let $\mathscr{U}_{v_{0}}$
be the set of all connected open domains $\mathcal{U}$ of the $(u,v)$-plane
with piecewise Lipschitz boundary $\partial\mathcal{U}$, such that
$\partial\mathcal{U}$ splits as the following union of Lipschitz
curves 
\begin{equation}
\partial\mathcal{U}=\mathcal{S}_{v_{0}}\cup\text{\textgreek{g}}_{0}\cup\mathcal{I}\cup clos(\text{\textgreek{g}}),\label{eq:BoundaryOfU}
\end{equation}
where 
\begin{equation}
\mathcal{S}_{v_{0}}=\{0\}\times[0,v_{0}],
\end{equation}
\begin{equation}
\text{\textgreek{g}}_{0}=\{u=v\}\cap\{0\le u<u^{(\text{\textgreek{g}}_{0})}\},\label{eq:AxisForm}
\end{equation}
\begin{equation}
\mathcal{I}=\{u=v-v_{0}\}\cap\{0\le u<u^{(\mathcal{I})}\}\label{eq:InfinityForm}
\end{equation}
(for some $0<u^{(\text{\textgreek{g}}_{0})},u^{(\mathcal{I})}\le+\infty$)
and $\text{\textgreek{g}}:(x_{1},x_{2})\rightarrow\mathbb{R}^{2}$
is an achronal (with respect to the reference Lorentzian metric (\ref{eq:ComparisonUVMetric}))
curve, which is allowed to be empty. The closure $clos(\text{\textgreek{g}})$
of $\text{\textgreek{g}}$ in (\ref{eq:BoundaryOfU}) is considered
with respect to the usual topology of $\mathbb{R}^{2}$\end{defn}
\begin{rem*}
Definition \ref{def:DevelopmentSets} implies that $\mathcal{U}$
is necessarily contained in the future domain of dependence of $\mathcal{S}_{v_{0}}\cup\text{\textgreek{g}}_{0}\cup\mathcal{I}$
(with respect to the comparison metric (\ref{eq:ComparisonUVMetric})).
In the case when $\text{\textgreek{g}}$ in (\ref{eq:BoundaryOfU})
is empty, it is necessary that both $\text{\textgreek{g}}_{0}$ and
$\mathcal{I}$ are unbounded in the future, i.\,e.~extend to $u+v=+\infty$.
\end{rem*}
We will now proceed to define the notion of a future development of
an asymptotically AdS boundary-characteristic initial data set for
the system (\ref{eq:RequationFinal})--(\ref{eq:OutgoingVlasovFinal})
with reflecting boundary conditions on $\text{\textgreek{g}}_{0}$
and $\mathcal{I}$:
\begin{defn}
\label{def:Development} For any $v_{0}>0$ and $r_{0}>0$, let $(r_{/},\text{\textgreek{W}}_{/}^{2},\bar{f}_{in/},\bar{f}_{out/})$
be a smooth, asymptotically AdS boundary-characteristic initial data
set on $[0,v_{0})$ for the system (\ref{eq:RequationFinal})--(\ref{eq:OutgoingVlasovFinal})
satisfying the reflecting gauge condition at $r=r_{0},+\infty$ (see
Definition \ref{def:TypeII}). A \emph{future development} of $(r_{/},\text{\textgreek{W}}_{/}^{2},\bar{f}_{in/},\bar{f}_{out/})$
will consist of a domain $\mathcal{U}\in\mathscr{U}_{v_{0}}$ (see
Definition \ref{def:DevelopmentSets}) and a quadruple of smooth functions
$r:\mathcal{U}\rightarrow(r_{0},+\infty)$, $\text{\textgreek{W}}^{2}:\mathcal{U}\rightarrow(0,+\infty)$
and $\bar{f}_{in},\bar{f}_{out}:\mathcal{U}\times(0,+\infty)\rightarrow[0,+\infty)$
satisfying the following conditions:

\begin{enumerate}

\item $(r,\text{\textgreek{W}}^{2},\bar{f}_{in},\bar{f}_{out})$
satisfy the given initial conditions on $\mathcal{S}_{v_{0}}=\{0\}\times[0,v_{0})$,
i.\,e.: 
\begin{equation}
(r,\text{\textgreek{W}}^{2},\bar{f}_{in},\bar{f}_{out})|_{\mathcal{S}_{v_{0}}}=(r_{/},\text{\textgreek{W}}_{/}^{2},\bar{f}_{in/},\bar{f}_{out/}).\label{eq:InitialDataRightMaximal}
\end{equation}

\item$(r,\text{\textgreek{W}}^{2},\bar{f}_{in},\bar{f}_{out})$ solve
(\ref{eq:RequationFinal})--(\ref{eq:OutgoingVlasovFinal}) on $\mathcal{U}$.

\item The following gauge conditions are satisfied on $\text{\textgreek{g}}_{0}$
and $\mathcal{I}$: 
\begin{equation}
\partial_{u}r|_{\text{\textgreek{g}}_{0}}=-\partial_{v}r|_{\text{\textgreek{g}}_{0}}\label{eq:GaugeMirrorMaximal}
\end{equation}
 and 
\begin{equation}
\partial_{u}(1/r)|_{\mathcal{I}}=-\partial_{v}(1/r)|_{\mathcal{I}}.\label{eq:GaugeInfinityMaximal}
\end{equation}

\item $(r,\bar{f}_{in},\bar{f}_{out})$ satisfy on $\mathcal{I}$
the boundary conditions 
\begin{equation}
(1/r)|_{\mathcal{I}}=0\label{eq:InfinityRMaximal}
\end{equation}
 and 
\begin{equation}
\lim_{h\rightarrow0^{+}}\Bigg(\frac{\bar{f}_{in}\big(u_{*}+h,v_{*};\,\text{\textgreek{W}}^{-2}(u_{*}+h,v_{*})\cdot p\big)}{\bar{f}_{out}\big(u_{*},v_{*}-h;\,\text{\textgreek{W}}^{-2}(u_{*},v_{*}-h)\cdot p\big)}\Bigg)=1,\label{eq:ReflectionInfinityMaximal}
\end{equation}
for all $(u_{*},v_{*})\in\mathcal{I}$ and $p>0$.

\item $(r,\bar{f}_{in},\bar{f}_{out})$ satisfy on $\text{\textgreek{g}}_{0}$
the boundary conditions 
\begin{equation}
r|_{\text{\textgreek{g}}_{0}}=r_{0}\label{eq:MirrorRMaximal}
\end{equation}
 and 
\begin{equation}
\bar{f}_{out}\big(u_{*},v_{*};\, p\big)=\bar{f}_{in}\big(u_{*},v_{*};\, p\big),\label{eq:ReflectionMirrorMaximal}
\end{equation}
for all $(u_{*},v_{*})\in\text{\textgreek{g}}_{0}$ and $p>0$.

\end{enumerate}\end{defn}
\begin{rem*}
Notice that the relations (\ref{eq:GaugeMirrorMaximal}) and (\ref{eq:GaugeInfinityMaximal})
folow from the boundary conditions (\ref{eq:MirrorRMaximal}) and
(\ref{eq:InfinityRMaximal}), combined with the form (\ref{eq:AxisForm})
and (\ref{eq:InfinityForm}) of $\text{\textgreek{g}}_{0}$ and $\mathcal{I}$,
respectively. However, the relations (\ref{eq:GaugeMirrorMaximal})
and (\ref{eq:GaugeInfinityMaximal}) should be viewed as \emph{gauge
conditions,} fixing, in conjuction with (\ref{eq:MirrorRMaximal})
and (\ref{eq:InfinityRMaximal}), the form (\ref{eq:AxisForm}) and
(\ref{eq:InfinityForm}) of $\text{\textgreek{g}}_{0}$ and $\mathcal{I}$.\end{rem*}
\begin{defn}
If $\mathscr{D}=(\mathcal{U};r,\text{\textgreek{W}}^{2},\bar{f}_{in},\bar{f}_{out})$
and $\mathscr{D}^{\prime}=(\mathcal{U}^{\prime};r^{\prime},(\text{\textgreek{W}}^{\prime})^{2},\bar{f}_{in}^{\prime},\bar{f}_{out}^{\prime})$
are two future developments of the same initial data $(r_{/},\text{\textgreek{W}}_{/}^{2},\bar{f}_{in/},\bar{f}_{out/})$,
we will say that $\mathscr{D}^{\prime}$ is an extension of $\mathscr{D}$,
writing $\mathscr{D}\subseteq\mathscr{D}^{\prime}$, if $\mathcal{U}\subseteq\mathcal{U}^{\prime}$
and the restriction $(r^{\prime},(\text{\textgreek{W}}^{\prime})^{2},\bar{f}_{in}^{\prime},\bar{f}_{out}^{\prime})|_{\mathcal{U}}$
of $(r^{\prime},(\text{\textgreek{W}}^{\prime})^{2},\bar{f}_{in}^{\prime},\bar{f}_{out}^{\prime})$
on $\mathcal{U}$ satisfies 
\begin{equation}
(r^{\prime},(\text{\textgreek{W}}^{\prime})^{2},\bar{f}_{in}^{\prime},\bar{f}_{out}^{\prime})|_{\mathcal{U}}\equiv(r,\text{\textgreek{W}}^{2},\bar{f}_{in},\bar{f}_{out}).
\end{equation}
.\end{defn}
\begin{rem*}
It can be readily deduced from Propositions \ref{Prop:LocalExistenceTypeII},
\ref{Prop:LocalExistenceTypeIII} and \ref{prop:LocalExistenceTypeI}
in Section \ref{sub:Local-well-Posedness} that, if $\mathscr{D}=(\mathcal{U};r,\text{\textgreek{W}}^{2},\bar{f}_{in},\bar{f}_{out})$
and $\mathscr{D}^{\prime}=(\mathcal{U}^{\prime};r^{\prime},(\text{\textgreek{W}}^{\prime})^{2},\bar{f}_{in}^{\prime},\bar{f}_{out}^{\prime})$
are two future developments of the same initial data $(r_{/},\text{\textgreek{W}}_{/}^{2},\bar{f}_{in/},\bar{f}_{out/})$,
then 
\begin{equation}
(r,\text{\textgreek{W}}^{2},\bar{f}_{in},\bar{f}_{out})|_{\mathcal{U}\cap\mathcal{U}^{\prime}}=(r^{\prime},(\text{\textgreek{W}}^{\prime})^{2},\bar{f}_{in}^{\prime},\bar{f}_{out}^{\prime})|_{\mathcal{U}\cap\mathcal{U}^{\prime}}.
\end{equation}

\end{rem*}

\section{\label{sec:The-main-results}Precise statement of the main results}

In this section, we will provide a detailed statement of the main
results of this paper. These results are used as a starting point
in the proof of the instability of AdS for the system (\ref{eq:RequationFinal})--(\ref{eq:OutgoingVlasovFinal})
in our companion paper \cite{MoschidisNullDust}.

\subsection{\label{sub:Well-posedness}Existence, uniqueness and the basic properties
of the maximal future development}

Our first result concerns the existence, uniqueness and the basic
properties of the maximal future development of any asymptotically
AdS boundary-characteristic initial data set for (\ref{eq:RequationFinal})--(\ref{eq:OutgoingVlasovFinal}):

\begin{customthm}{1}[precise version]\label{thm:maximalExtension}

For any $v_{0}>0$ and any $r_{0}>0$, let $(r_{/},\text{\textgreek{W}}_{/}^{2},\bar{f}_{in/},\bar{f}_{out/})$
be a smooth asymptotically AdS boundary-characteristic initial data
set on $[0,v_{0})$ for the system (\ref{eq:RequationFinal})--(\ref{eq:OutgoingVlasovFinal})
satisfying the reflecting gauge condition at $r=r_{0},+\infty$, according
to Definition \ref{def:TypeII}, such that the quantities $\frac{\text{\textgreek{W}}_{/}^{2}}{1-\frac{1}{3}\Lambda r_{/}^{2}},r_{/}^{2}(T_{vv})_{/}$
and $\tan^{-1}r_{/}$ extend smoothly on $v=v_{0}$. Then, there exists
a unique future development $(\mathcal{U};r,\text{\textgreek{W}}^{2},\bar{f}_{in},\bar{f}_{out})$
of $(r_{/},\text{\textgreek{W}}_{/}^{2},\bar{f}_{in/},\bar{f}_{out/})$
with reflecting boundary condition on $\text{\textgreek{g}}_{0},\mathcal{I}$
(see Definition \ref{def:Development}) which is \underline{maximal},
that is to say, any other future development $(\mathcal{U}^{\prime};r^{\prime},(\text{\textgreek{W}}^{\prime})^{2},\bar{f}_{in}^{\prime},\bar{f}_{out}^{\prime})$
of $(r_{/},\text{\textgreek{W}}_{/}^{2},\bar{f}_{in/},\bar{f}_{out/})$
with $r^{\prime}\ge r_{0}$ everywhere on $\mathcal{U}^{\prime}$satisfies
$\mathcal{U}^{\prime}\subseteq\mathcal{U}$ and 
\begin{equation}
(r^{\prime},(\text{\textgreek{W}}^{\prime})^{2},\bar{f}_{in}^{\prime},\bar{f}_{out}^{\prime})\equiv(r,\text{\textgreek{W}}^{2},\bar{f}_{in},\bar{f}_{out})|_{\mathcal{U}'}.
\end{equation}

The maximal future development $(\mathcal{U};r,\text{\textgreek{W}}^{2},\bar{f}_{in},\bar{f}_{out})$
satisfies the following properties (for the definition of the curves
$\text{\textgreek{g}}_{0},\mathcal{I},\text{\textgreek{g}}$, see
Definition \ref{def:DevelopmentSets}):

\begin{enumerate}

\item We have 
\begin{equation}
\partial_{u}r<0,\label{eq:NegativeDerivativeRMaximal}
\end{equation}
\begin{equation}
\big(1-\frac{2m}{r}\big)\big|_{J^{-}(\mathcal{I})\cup J^{-}(\text{\textgreek{g}}_{0})}>0\label{eq:NonTrappingMaximal}
\end{equation}
 and 
\begin{equation}
\partial_{v}r|_{J^{-}(\mathcal{I})\cup J^{-}(\text{\textgreek{g}}_{0})}>0,\label{eq:D_vRPositiveMaximal}
\end{equation}
where 
\begin{equation}
J^{-}(\mathcal{I})=\big\{0\le u<\sup_{\mathcal{I}}u\big\}\cap\mathcal{U}\label{eq:PastOfInfinity}
\end{equation}
is the causal past of $\mathcal{I}$ and 
\begin{equation}
J^{-}(\text{\textgreek{g}}_{0})=\big\{0\le v<\sup_{\text{\textgreek{g}}_{0}}v\big\}\cap\mathcal{U}
\end{equation}
is the causal past of $\text{\textgreek{g}}_{0}$ (with respect to
the reference Lorenztian metric (\ref{eq:ComparisonUVMetric})).

\item The renormalised Hawking mass $\tilde{m}$ is conserved on
$\text{\textgreek{g}}_{0}$ and $\mathcal{I}$, i.\,e.: 
\begin{equation}
\tilde{m}|_{\text{\textgreek{g}}_{0}}=\tilde{m}_{/}(0)\label{eq:ConstantMassMirror}
\end{equation}
and 
\begin{equation}
\tilde{m}|_{\mathcal{I}}=\lim_{v\rightarrow v_{0}^{-}}\tilde{m}_{/}(v).\label{eq:ConstantMassInfinity}
\end{equation}

\item The conformal infinity $\mathcal{I}$ is complete, i.\,e.
$\text{\textgreek{W}}^{2}/(1-\frac{1}{3}\Lambda r^{2})$ has a finite
limit on $\mathcal{I}$ and: 
\begin{equation}
\int_{\mathcal{I}}\sqrt{\frac{\text{\textgreek{W}}^{2}}{1-\frac{1}{3}\Lambda r^{2}}}\Big|_{\mathcal{I}}\, du=+\infty.\label{eq:CompleConformalInfinity}
\end{equation}

\item In the case $\mathcal{U}\backslash J^{-}(\mathcal{I})\neq\emptyset$,
the future event horizon 
\begin{equation}
\mathcal{H}^{+}=\mathcal{U}\cap\partial J^{-}(\mathcal{I})\label{eq:DefinitionHorizon}
\end{equation}
has the following properties:

\begin{enumerate}

\item $\mathcal{H}^{+}$ has infinite affine length, i.\,e.: 
\begin{equation}
\int_{\mathcal{H}^{+}}\text{\textgreek{W}}^{2}\, dv=+\infty.\label{eq:InfiniteLengthHorizon}
\end{equation}

\item All the matter falls inside the black hole, i.\,e. 
\begin{equation}
\sup_{\mathcal{H}^{+}}r=r_{S}\label{eq:UpperBoundRHorizon}
\end{equation}
and
\begin{equation}
\inf_{\mathcal{H}^{+}}\Big(1-\frac{2m}{r}\Big)=0,\label{eq:TrappingAsymptoticallyHorizon}
\end{equation}
where $r_{S}$ defined by the relation 
\begin{equation}
1-2\frac{\lim_{v\rightarrow v_{0}^{-}}\tilde{m}_{/}(v)}{r_{S}}-\frac{1}{3}\Lambda r_{S}^{2}=0.\label{eq:DefinitionRs}
\end{equation}

\end{enumerate}

\item In the case $\mathcal{H}^{+}\neq\emptyset$, the curve $\text{\textgreek{g}}_{0}$
is bounded and contains points lying to the future of $\mathcal{H}^{+}$,
i.\,e.~satisfies 
\begin{equation}
\text{\textgreek{g}}_{0}\nsubseteq J^{-}(\mathcal{I}).\label{eq:MirrorExtendsBeyondHorizon}
\end{equation}

\item In the case $\mathcal{H}^{+}\neq\emptyset$, the curve $\text{\textgreek{g}}$
is non-empty and $r$ extends continuously on $\text{\textgreek{g}}$
with $r|_{\text{\textgreek{g}}_{0}}=r_{0}$. Moreover, there is no
Cauchy horizon ``emanating from timelike infinity'': for any point
$(u_{1},v_{1})\in\text{\textgreek{g}}$, the line $\{v=v_{1}\}$ intersects
$\mathcal{I}$. In other words, there is no point in $\text{\textgreek{g}}$
which lies on the curve $\{v=v_{\mathcal{I}}\}$, where $(u_{\mathcal{I}},v_{\mathcal{I}})$
is the future limit point of $\mathcal{I}$.

\end{enumerate}

\end{customthm}

For the pro\textgreek{o}f of Theorem \ref{thm:maximalExtension},
see Section \ref{sub:ProofOfProp}.
\begin{rem*}
Note that, in view of (\ref{eq:UpperBoundRHorizon}), (\ref{eq:DefinitionRs})
and the fact that $r>r_{0}$ on $\mathcal{U}$, a necessary condition
for $\mathcal{H}^{+}$ to be non-empty is that the total mass $\lim_{v\rightarrow v_{0}}\tilde{m}_{/}(v)$
and the mirror radus $r_{0}$ satisfy 
\begin{equation}
2\frac{\lim_{v\rightarrow v_{0}^{-}}\tilde{m}_{/}(v)}{r_{0}}>1-\frac{1}{3}\Lambda r_{0}^{2}.\label{eq:LowerBoundMass}
\end{equation}
From the proof of Theorem \ref{thm:maximalExtension}, it follows
that, in the case when $\mathcal{H}^{+}\neq\emptyset$, $(\mathcal{U};r,\text{\textgreek{W}}^{2},\bar{f}_{in},\bar{f})$
approaches (in a suitable sense) the Schwarzschild--AdS solution near
timelike infinity $(u_{\mathcal{I}},v_{\mathcal{I}})$.
\end{rem*}
\begin{figure}[h] 
\centering 
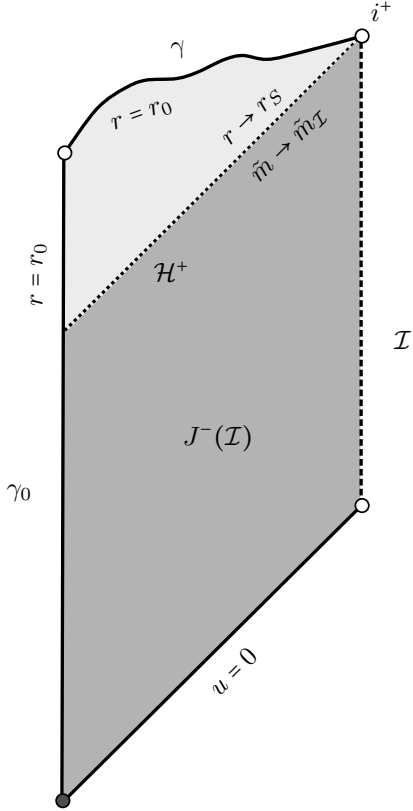 
\caption{Schematic depiction of the  maximal future development $(\mathcal{U}; r,\Omega^2,\tau,\bar{\tau})$ of a  smooth asymptotically AdS boundary-characteristic initial data set $(r_{/},\Omega_{/}^{2},\bar{f}_{in/},\bar{f}_{out/})$.  In the case when $J^{-}(\mathcal{I})$ does not cover all of $\mathcal{U}$, the future event horizon $\mathcal{H}^{+}$ is non-empty and has infinite affine length. In this case, the curve $\gamma$ is non-empty, and does \underline{not} contain a Cauchy horizon component emanating from $i^{+}$. The final mass of the event horizon is equal to the (conserved) renormalised Hawking mass on $\mathcal{I}$.}
\end{figure}
\begin{defn}
\label{def:MaximalDevelopment}The development $(\mathcal{U};r,\text{\textgreek{W}}^{2},\bar{f}_{in},\bar{f}_{out})$
introduced by Theorem \ref{thm:maximalExtension} will be called the
\emph{maximal future development} of the asymptotically AdS boundary-characteristic
initial data set $(r_{/},\text{\textgreek{W}}_{/}^{2},\bar{f}_{in/},\bar{f}_{out/})$. \end{defn}
\begin{rem*}
In a similar way, we can uniquely define the \emph{maximal past development}
$(\mathcal{U};r,\text{\textgreek{W}}^{2},\bar{f}_{in},\bar{f}_{out})$
of $(r_{/},\text{\textgreek{W}}_{/}^{2},\bar{f}_{in/},\bar{f}_{out/})$,
satisfying the properties outlined by Theorem \ref{thm:maximalExtension}
after performing a ``time reversal'' transformation $(u,v)\rightarrow(-v,-u)$.
Notice that such a coordinate transformation turns an asymptotically
AdS boundary-characteristic initial data set on $u=0$ into an asymptotically
AdS boundary-characteristic initial data set on $v=0$. However, Theorem
\ref{thm:maximalExtension} also holds (with exactly the same proof)
for such initial data sets.
\end{rem*}

\subsection{\label{sub:Cauchy-stability}Cauchy stability in a rough norm, uniformly
in $r_{0}$}

Our next result is a Cauchy stability statement for the domain of
outer communications of solutions $(r,\text{\textgreek{W}}^{2},\bar{f}_{in},\bar{f}_{out})$
to (\ref{eq:RequationFinal})--(\ref{eq:OutgoingVlasovFinal}) in
a rough initial data topology:

\begin{customthm}{2}[precise version]\label{thm:CauchyStability}

For any $v_{1}<v_{2}$ and any $0<r_{0}<(-\Lambda)^{-1/2}$, let $\mathcal{S}_{i}=(r_{/i},\text{\textgreek{W}}_{/i}^{2},\bar{f}_{in/i},\bar{f}_{out/i})$,
$i=1,2$, be two smooth asymptotically AdS boundary-characteristic
initial data sets on $[v_{1},v_{2})$ for the system (\ref{eq:RequationFinal})--(\ref{eq:OutgoingVlasovFinal})
satisfying the reflecting gauge condition at $r=r_{0},+\infty$, according
to Definition \ref{def:TypeII}, such that the quantities $\frac{\text{\textgreek{W}}_{/i}^{2}}{1-\frac{1}{3}\Lambda r_{/i}^{2}},r_{/i}^{2}(T_{vv})_{/i}$
and $\tan^{-1}r_{/i}$ extend smoothly on $v=v_{2}$. Assume, furthermore,
that the following conditions hold:

\begin{enumerate}

\item For some $u_{0}>0$, the maximal future development $(\mathcal{U}_{1};r_{1},\text{\textgreek{W}}_{1}^{2},\bar{f}_{in1},\bar{f}_{out1})$
of $\mathcal{S}_{1}$ satisfies 
\begin{equation}
\mathcal{W}_{u_{0}}\doteq\{0<u<u_{0}\}\cap\{u+v_{1}<v<u+v_{2}\}\subset\mathcal{U}_{1}
\end{equation}
and 
\begin{align}
C_{0}\doteq\sup_{\mathcal{W}_{u_{0}}}\Bigg\{\Big|\log\big(\frac{\text{\textgreek{W}}_{1}^{2}}{1-\frac{1}{3}\Lambda r_{1}^{2}}\big)\Big|+\Big|\log\Big(\frac{2\partial_{v}r_{1}}{1-\frac{2m_{1}}{r_{1}}}\Big)\Big|+\Big|\log\Big(\frac{1-\frac{2m_{1}}{r_{1}}}{1-\frac{1}{3}\Lambda r_{1}^{2}}\Big)\Big|+\sqrt{-\Lambda}|\tilde{m}_{1}|\Bigg\}+\label{eq:UpperBoundNonTrappingForCauchyStability}\\
+\sup_{\bar{u}}\int_{\{u=\bar{u}\}\cap\mathcal{W}_{u_{0}}}r_{1}\frac{(T_{vv})_{1}}{\partial_{v}r_{1}}\, dv+\sup_{\bar{v}}\int_{\{v=\bar{v}\}\cap\mathcal{W}_{u_{0}}}r_{1}\frac{(T_{uu})_{1}}{-\partial_{u}r_{1}}\, du & <+\infty.\nonumber 
\end{align}

\item The pair of initial data $\mathcal{S}_{i}$, $i=1,2$, satisfy
\begin{align}
\sup_{v\in[v_{1},v_{2})}\Bigg\{\Big|\log\big(\frac{\text{\textgreek{W}}_{/1}^{2}}{1-\frac{1}{3}\Lambda r_{/1}^{2}}\big)-\log\big(\frac{\text{\textgreek{W}}_{/2}^{2}}{1-\frac{1}{3}\Lambda r_{/2}^{2}}\big)\Big|+\Big|\log\Big(\frac{2\partial_{v}r_{/1}}{1-\frac{2m_{/1}}{r_{/1}}}\Big)-\log\Big(\frac{2\partial_{v}r_{/2}}{1-\frac{2m_{/2}}{r_{/2}}}\Big)\Big|+\label{eq:GaugeDifferenceBoundCauchystability}\\
+\Big|\log\Big(\frac{1-\frac{2m_{/_{1}}}{r_{/1}}}{1-\frac{1}{3}\Lambda r_{/1}^{2}}\Big)-\log\Big(\frac{1-\frac{2m_{/_{2}}}{r_{/2}}}{1-\frac{1}{3}\Lambda r_{/2}^{2}}\Big)\Big|+\sqrt{-\Lambda}|\tilde{m}_{/1}-\tilde{m}_{/2}|\Bigg\}(v) & \le\text{\textgreek{d}}\nonumber 
\end{align}
and 
\begin{align}
\sup_{v\in[v_{1},v_{2}]}(-\Lambda)\int_{v_{1}}^{v_{2}}\Bigg|\frac{r_{1/}^{2}(T_{vv})_{1/}(\bar{v})}{\big(|\text{\textgreek{r}}_{1/}(\bar{v})-\text{\textgreek{r}}_{1/}(v)|+\tan^{-1}\big(\sqrt{-\frac{\Lambda}{3}}r_{0}\big)\big)\partial_{v}\text{\textgreek{r}}_{1/}(\bar{v})}-\label{eq:DifferenceBoundCauchyStability}\\
\hphantom{\sup_{v\in[v_{1},v_{2}]}(-\Lambda)\int_{v_{1}}^{v_{2}}}-\frac{r_{2/}^{2}(T_{vv})_{2/}(\bar{v})}{\big(|\text{\textgreek{r}}_{2/}(\bar{v})-\text{\textgreek{r}}_{2/}(v)|+\tan^{-1}\big(\sqrt{-\frac{\Lambda}{3}}r_{0}\big)\big)\partial_{v}\text{\textgreek{r}}_{2/}(\bar{v})} & \Bigg|\, d\bar{v}\le\text{\textgreek{d}},
\end{align}
where, for some fixed large absolute constant $C_{1}$, the parameter
$\text{\textgreek{d}}$ satisfies 
\begin{equation}
0\le\text{\textgreek{d}}\le\text{\textgreek{d}}_{0}\doteq\exp\big(-\exp\big(C_{1}(1+C_{0})\frac{u_{0}}{v_{2}-v_{1}}\big)\big)\label{eq:SmallnessDeltaForCauchyStability}
\end{equation}
and $\text{\textgreek{r}}_{/}$ is defined by the relation 
\begin{equation}
\text{\textgreek{r}}_{/}(v)\doteq\tan^{-1}\big(\sqrt{-\frac{\Lambda}{3}}r_{/}(v)\big).\label{eq:DefinitionInitialRho}
\end{equation}

\end{enumerate}

Then, the maximal development $(\mathcal{U}_{2};r_{2},\text{\textgreek{W}}_{2}^{2},\bar{f}_{in2},\bar{f}_{out2})$
of $\mathcal{S}_{2}$ satisfies 
\begin{equation}
\mathcal{W}_{u_{0}}\subset\mathcal{U}_{2}\label{eq:InclusionOtherdomain}
\end{equation}
and 
\begin{align}
\sup_{\mathcal{W}_{u_{0}}}\Bigg\{\Big|\log\big(\frac{\text{\textgreek{W}}_{1}^{2}}{1-\frac{1}{3}\Lambda r_{1}^{2}}\big)-\log\big(\frac{\text{\textgreek{W}}_{2}^{2}}{1-\frac{1}{3}\Lambda r_{2}^{2}}\big)\Big|+\Big|\log\Big(\frac{2\partial_{v}r_{1}}{1-\frac{2m_{1}}{r_{1}}}\Big)-\log\Big(\frac{2\partial_{v}r_{2}}{1-\frac{2m_{2}}{r_{2}}}\Big)\Big|+\label{eq:UpperBoundNonTrappingForCauchyStability-1}\\
+\Big|\log\Big(\frac{1-\frac{2m_{1}}{r_{1}}}{1-\frac{1}{3}\Lambda r_{1}^{2}}\Big)-\log\Big(\frac{1-\frac{2m_{2}}{r_{2}}}{1-\frac{1}{3}\Lambda r_{2}^{2}}\Big)\Big|+\sqrt{-\Lambda}|\tilde{m}_{1}-\tilde{m}_{2}|\Bigg\}+\nonumber \\
+\sup_{\bar{u}}\int_{\{u=\bar{u}\}\cap\mathcal{W}_{u_{0}}}\big|r_{1}(T_{vv})_{1}-r_{2}(T_{vv})_{2}\big|\, dv+\sup_{\bar{v}}\int_{\{v=\bar{v}\}\cap\mathcal{W}_{u_{0}}}\big|r_{1}(T_{uu})_{1}-r_{2}(T_{uu})_{2}\big|\, du & \le\nonumber \\
\le\exp\big(\exp\big( & C_{1}(1+C_{0})\big)\frac{u_{0}}{v_{2}-v_{1}}\big)\text{\textgreek{d}}.\nonumber 
\end{align}

\end{customthm}

For the proof of Theorem \ref{thm:CauchyStability}, see Section \ref{sub:Proof-of-Cauchy-General}.
\begin{rem*}
By repeating the proof of Theorem \ref{thm:CauchyStability}, the
Cauchy stability estimate (\ref{eq:UpperBoundNonTrappingForCauchyStability-1})
also holds in the case when $(\mathcal{U}_{i};r_{i},\text{\textgreek{W}}_{i}^{2},\bar{f}_{in;i},\bar{f}_{out;i})$,
$i=1,2$, are the maximal \underline{past} developments of $\mathcal{S}_{i}$,
i.\,e.~when $\mathcal{W}_{u_{0}}$ is replaced by 
\begin{equation}
\mathcal{W}_{u_{0}}^{(-)}\doteq\{-u_{0}\le u<0\}\cap\{u+v_{1}<v<u+v_{2}\}
\end{equation}
and (\ref{eq:UpperBoundNonTrappingForCauchyStability}) holds on $\mathcal{W}_{u_{0}}^{(-)}$
in place of $\mathcal{W}_{u_{0}}$ (see the remark below Definition
\ref{def:MaximalDevelopment}). 
\end{rem*}
For any $r_{0}>0$ and any $v_{0}>0$, let us define the following
``norm'' on the space of smooth asymptotically AdS boundary-characteristic
initial data sets $\mathcal{S}=(r_{/},\text{\textgreek{W}}_{/}^{2},\bar{f}_{in/},\bar{f}_{out/})$
on $[0,v_{0})$ for the system (\ref{eq:RequationFinal})--(\ref{eq:OutgoingVlasovFinal}):
\begin{align}
||(r_{/},\text{\textgreek{W}}_{/}^{2},\bar{f}_{in/},\bar{f}_{out/})||_{\mathcal{C\mathcal{S}}}\doteq\sqrt{-\Lambda} & \sup_{0\le v<v_{0}}|\tilde{m}_{/}(v)|+(-\Lambda)\sup_{0\le v<v_{0}}\int_{0}^{v_{0}}\frac{1}{\text{\textgreek{r}}_{/}(v)-\text{\textgreek{r}}_{/}(\bar{v})+\text{\textgreek{r}}_{/}(0)}\Big(\frac{r_{/}^{2}(T_{vv})_{/}}{\partial_{v}\text{\textgreek{r}}_{/}}\Big)(\bar{v})\, d\bar{v}\Bigg\}+\label{eq:GeometricNorm}\\
 & +\sup_{0\le v<v_{0}}\max\big\{\frac{2\tilde{m}_{/}}{r_{/}},0\big\},\nonumber 
\end{align}
where $\text{\textgreek{r}}_{/}$ is defined by (\ref{eq:DefinitionInitialRho}).
\begin{rem*}
Note that (\ref{eq:GeometricNorm}) is invariant under gauge transformations,
as well as scale transformations of the form $(u,v)\rightarrow(\text{\textgreek{l}}u,\text{\textgreek{l}}v)$,
$(r,\tilde{m},\Lambda)\rightarrow(\text{\textgreek{l}}r,\text{\textgreek{l}}\tilde{m},\text{\textgreek{l}}^{-2}\Lambda)$,
$r_{0}\rightarrow\text{\textgreek{l}}r_{0}$, $(\bar{f}_{in},\bar{f}_{out})\rightarrow(\text{\textgreek{l}}^{-4}\bar{f}_{in},\text{\textgreek{l}}^{-4}\bar{f}_{out})$.
Furthermore, notice that $||(r_{/},\text{\textgreek{W}}_{/}^{2},\bar{f}_{in/},\bar{f}_{out/})||_{\mathcal{C\mathcal{S}}}=0$
if and only if $\bar{f}_{in/}=0$ and $\bar{f}_{out/}=0$, i.\,e.~if
$(r_{/},\text{\textgreek{W}}_{/}^{2},\bar{f}_{in/},\bar{f}_{out/})$
is the initial data set for the pure AdS spacetime $(\mathcal{M}_{AdS},g_{AdS})$
on $\{r\ge r_{0}\}$. The dependence of (\ref{eq:GeometricNorm})
in terms of the Vlasov fields $(\bar{f}_{in/},\bar{f}_{out/})$ is
only through the ingoing energy momentum component $(T_{vv})_{/}$.
\end{rem*}
Specialising to the case when $\mathcal{S}_{1}$ is the trivial initial
data set $\mathcal{S}_{AdS}$, 
\begin{equation}
\mathcal{S}_{AdS}=(r_{AdS/},\text{\textgreek{W}}_{AdS/}^{2},0,0),
\end{equation}
 in Theorem \ref{thm:CauchyStability}, we obtain the following Cauchy
stability statement for $(\mathcal{M}_{AdS},g_{AdS})$ with respect
to the topology defined by (\ref{eq:GeometricNorm}), which is independent
of the inner mirror radius $r_{0}$:

\begin{customcor}{1}\label{cor:CauchyStabilityOfAdS}

For any (possibly large) $l_{*}>0$, there exists a (small) $\text{\textgreek{e}}_{0}>0$
and a constant $C_{l_{*}}>0$ depending only on $l_{*}$, so that
the following statement holds: For any $v_{0}>0$ and $0<r_{0}<(-\Lambda)^{-1/2}$,
if $(r_{/},\text{\textgreek{W}}_{/}^{2},\bar{f}_{in/},\bar{f}_{out/})$
is a smooth asymptotically AdS boundary-characteristic initial data
set on $[0,v_{0})$ for the system (\ref{eq:RequationFinal})--(\ref{eq:OutgoingVlasovFinal})
satisfying the reflecting gauge condition at $r=r_{0},+\infty$, according
to Definition \ref{def:TypeII}, such that the quantities $\frac{\text{\textgreek{W}}_{/}^{2}}{1-\frac{1}{3}\Lambda r_{/}^{2}},r_{/}^{2}(T_{vv})_{/}$
and $\tan^{-1}r_{/}$ extend smoothly on $v=v_{0}$ and moreover 
\begin{equation}
||(r_{/},\text{\textgreek{W}}_{/}^{2},\bar{f}_{in/},\bar{f}_{out/})||_{\mathcal{C}\mathcal{S}}<\text{\textgreek{e}}\label{eq:SmallnessForCauchyStability}
\end{equation}
for some $0<\text{\textgreek{e}}\le\text{\textgreek{e}}_{0}$, then
the maximal development $(\mathcal{U};r,\text{\textgreek{W}}^{2},\bar{f}_{in},\bar{f}_{out})$
satisfies 
\begin{equation}
\mathcal{W}_{l_{*}}\doteq\{0<u\le l_{*}v_{0}\}\cap\{u<v<u+v_{0}\}\subset\mathcal{U}\label{eq:InclusionInMaximalDomain}
\end{equation}
and 
\begin{equation}
\sqrt{-\Lambda}\sup_{\mathcal{W}_{l_{*}}}|\tilde{m}|+\sup_{\mathcal{W}_{l_{*}}}\log\Bigg(\frac{1-\frac{1}{3}\Lambda r^{2}}{1-\max\{\frac{2m}{r},0\}}\Bigg)+\sup_{\bar{u}}\int_{\{u=\bar{u}\}\cap\mathcal{W}_{l_{*}}}\frac{rT_{vv}}{\partial_{v}r}\, dv+\sup_{\bar{v}}\int_{\{v=\bar{v}\}\cap\mathcal{W}_{l_{*}}}\frac{rT_{uu}}{(-\partial_{u}r)}\, du<C_{l_{*}}\text{\textgreek{e}}.\label{eq:SmallnessCauchyStability}
\end{equation}

\end{customcor}

For the proof of Corollary \ref{cor:CauchyStabilityOfAdS}, see Section
\ref{sub:Proof-of-Cauchy-AdS}.

\section{\label{sec:Well-posedness-for-the}Well-posedness and structure of
the maximal development}

The aim of this Section is the proof of Theorem \ref{thm:maximalExtension}.
To this end, we will first introduce, in Section \ref{sub:A-wider-class-Initial-data},
a number of characteristic initial value problems for (\ref{eq:RequationFinal})--(\ref{eq:OutgoingVlasovFinal}),
in addition to the asymptotically AdS boundary-characteristic initial
value problem introduced by Definition \ref{def:TypeII}. We will
then establish the well-posedness of these initial value problems
in Section \ref{sub:Local-well-Posedness}. The results of Section
\ref{sub:Local-well-Posedness}, combined with a number of continuation
criteria that will be established in Section \ref{sub:Continuation-criteria},
will allow us to construct the maximal future development of a general
smooth, asymptotically AdS boundary-characteristic initial data set
and complete the proof of Theorem \ref{thm:maximalExtension}. This
will be achieved in Section \ref{sub:ProofOfProp}.

\subsection{\label{sub:A-wider-class-Initial-data}Auxiliary types of characterisitic
initial data sets}

In this Section, we will define some auxiliary types of characteristic
initial data sets for (\ref{eq:RequationFinal})--(\ref{eq:OutgoingVlasovFinal}),
in addition to the one introduced by Definition \ref{def:TypeII}.
\begin{defn}
\label{def:TypeI} For any $u_{1}<u_{2}$, $v_{1}<v_{2}$ and any
$r_{0}>0$, let $r_{\backslash}:[u_{1},u_{2}]\rightarrow(r_{0},+\infty)$,
$\text{\textgreek{W}}_{\backslash}:[u_{1},u_{2}]\rightarrow[0,+\infty)$,
$\bar{f}_{in\backslash},\bar{f}_{out\backslash}:[u_{1},u_{2}]\times(0,+\infty)\rightarrow[0,+\infty)$,
$r_{/}:[v_{1},v_{2}]\rightarrow(r_{0},+\infty)$, $\text{\textgreek{W}}_{/}:[v_{1},v_{2}]\rightarrow(0,+\infty)$
and $\bar{f}_{in/},\bar{f}_{out/}:[v_{1},v_{2}]\times(0,+\infty)\rightarrow[0,+\infty)$
be smooth functions, such that 
\begin{equation}
r_{\backslash}(u_{1})=r_{/}(v_{1}),\label{eq:EqualRJunction}
\end{equation}
\begin{equation}
\text{\textgreek{W}}_{\backslash}(u_{1})=\text{\textgreek{W}}_{/}(v_{1}),\label{eq:EqualOmegaJunction}
\end{equation}
and, for all $p\in(0,+\infty)$: 
\begin{equation}
\bar{f}_{in\backslash}(u_{1},p)=\bar{f}_{in/}(v_{1},p)\label{eq:EqualIngoingJunction}
\end{equation}
and 
\begin{equation}
\bar{f}_{out\backslash}(u_{1},p)=\bar{f}_{out/}(v_{1},p).\label{eq:EqualOutgoingJunction}
\end{equation}
 We will call $(r_{\backslash},\text{\textgreek{W}}_{\backslash},\bar{f}_{in\backslash},\bar{f}_{out\backslash})$
and $(r_{\backslash},\text{\textgreek{W}}_{\backslash},\bar{f}_{in\backslash},\bar{f}_{out\backslash})$
a\emph{ characteristic initial data} \emph{set} for the system (\ref{eq:RequationFinal})--(\ref{eq:OutgoingVlasovFinal})
if the pairs $(r_{\backslash},\text{\textgreek{W}}_{\backslash})$
and $(r_{/},\text{\textgreek{W}}_{/})$ satisfy the constraint equations
\begin{align}
\partial_{u}(\text{\textgreek{W}}_{\backslash}^{-2}\partial_{u}r_{\backslash})= & -4\pi r_{\backslash}(T_{uu})_{\backslash}\text{\textgreek{W}}_{\backslash}^{-2},\label{eq:ConstraintUDef}\\
\partial_{v}(\text{\textgreek{W}}_{/}^{-2}\partial_{v}r_{/})= & -4\pi r_{/}(T_{vv})_{/}\text{\textgreek{W}}_{/}^{-2},\label{eq:ConstraintVDef}
\end{align}
 where 
\begin{gather}
(T_{uu})_{\backslash}(u)\doteq\int_{0}^{+\infty}\text{\textgreek{W}}_{\backslash}^{4}(u)(p^{v})^{2}\bar{f}_{out\backslash}(u;p^{v})\, r_{\backslash}^{2}(u)\frac{dp^{v}}{p^{v}},\label{eq:EnergyMomentumInitialLeft}\\
(T_{vv})_{/}(v)\doteq\int_{0}^{+\infty}\text{\textgreek{W}}_{/}^{4}(v)(p^{u})^{2}\bar{f}_{in/}(v;p^{u})\, r_{/}^{2}(v)\frac{dp^{u}}{p^{u}},\label{eq:EnergyMomentumIntialRight}
\end{gather}
while $\bar{f}_{in\backslash},\bar{f}_{out/}$ solve (\ref{eq:IngoingVlasovFinal})
and (\ref{eq:OutgoingVlasovFinal}) along $\{v=v_{1}\}$ and $\{u=u_{1}\}$
respectively, i.\,e.~ 
\begin{gather}
\partial_{u}\big(\text{\textgreek{W}}_{\backslash}^{4}(u)r_{\backslash}^{4}(u)p^{u}\bar{f}_{in\backslash}(u,p^{u})\big)+p^{u}\partial_{p^{u}}\big(\text{\textgreek{W}}_{\backslash}^{4}(u)r_{\backslash}^{4}(u)p^{u}\bar{f}_{in\backslash}(u,p^{u})\big)=0,\label{eq:IngoingEquationCombatibility}\\
\partial_{v}\big(\text{\textgreek{W}}_{/}^{4}(v)r_{/}^{4}(v)p^{v}\bar{f}_{out/}(v,p^{v})\big)+p^{v}\partial_{p^{v}}\big(\text{\textgreek{W}}_{/}^{4}(v)r_{/}^{4}(v)p^{v}\bar{f}_{out/}(v,p^{v})\big)=0.\label{eq:OutgoingEquationCombatibility}
\end{gather}
\end{defn}
\begin{rem*}
Let $(r,\text{\textgreek{W}}^{2},\bar{f}_{in},\bar{f}_{out})$ be
a solution of the system (\ref{eq:RequationFinal})--(\ref{eq:OutgoingVlasovFinal})
on a closed subset of $\mathcal{V}$ of $[u_{1},u_{2}]\times[v_{1},v_{2}]$
containing $\big([u_{1},u_{2}]\times\{v_{1}\}\big)\cup\big(\{u_{2}\}\times[v_{1},v_{2}]\big)$,
such that 
\begin{equation}
(r,\text{\textgreek{W}}^{2},\bar{f}_{in},\bar{f}_{out})|_{[u_{1},u_{2}]\times\{v_{1}\}}=(r_{\backslash},\text{\textgreek{W}}_{\backslash}^{2},\bar{f}_{in\backslash},\bar{f}_{out\backslash})
\end{equation}
and 
\begin{equation}
(r,\text{\textgreek{W}}^{2},\bar{f}_{in},\bar{f}_{out})|_{\{u_{1}\}\times[v_{1},v_{2}]}=(r_{/},\text{\textgreek{W}}_{/}^{2},\bar{f}_{in/},\bar{f}_{out/}).
\end{equation}
Then, the transversal derivatives of $r$ across $[u_{1},u_{2}]\times\{v_{1}\}$
and $\{u_{2}\}\times[v_{1},v_{2}]$ can be computed in terms of $(r_{\backslash},\text{\textgreek{W}}_{\backslash}^{2})$
and $(r_{/},\text{\textgreek{W}}_{/}^{2})$ by integrating equation
(\ref{eq:RequationFinal}), i.\,e. for all $u\in[u_{1},u_{2}]$:
\begin{equation}
(r\partial_{v}r)(u,v_{1})=r_{/}\partial_{v}r_{/}(v_{1})-\frac{1}{4}\int_{u_{1}}^{u}(1-\Lambda r_{\backslash}^{2}(\bar{u}))\text{\textgreek{W}}_{\backslash}^{2}(\bar{u})\, d\bar{u}\label{eq:TransversalDerivativeV}
\end{equation}
and, for all $v\in[v_{1},v_{2}]$: 
\begin{equation}
(r\partial_{u}r)(u_{1},v)=r_{\backslash}\partial_{u}r_{\backslash}(u_{1})-\frac{1}{4}\int_{v1}^{v}(1-\Lambda r_{/}^{2}(\bar{v}))\text{\textgreek{W}}_{/}^{2}(\bar{v})\, d\bar{v}.\label{eq:TransversalDerivativeU}
\end{equation}
\end{rem*}
\begin{defn}
\label{def:TypeIII}For any $u_{1}<u_{2}$, $v_{1}<v_{2}$ and any
$r_{0}>0$, let $r_{\backslash}:[u_{1},u_{2}]\rightarrow[r_{0},+\infty)$,
$\text{\textgreek{W}}_{\backslash}:[u_{1},u_{2}]\rightarrow[0,+\infty)$,
$\bar{f}_{in\backslash},\bar{f}_{out\backslash}:[u_{1},u_{2}]\times(0,+\infty)\rightarrow[0,+\infty)$,
$r_{/}:[v_{1},v_{2}]\rightarrow(r_{0},+\infty)$, $\text{\textgreek{W}}_{/}:[v_{1},v_{2}]\rightarrow(0,+\infty)$
and $\bar{f}_{in/},\bar{f}_{out/}:[v_{1},v_{2}]\times(0,+\infty)\rightarrow[0,+\infty)$
be smooth functions, satisfying (\ref{eq:EqualRJunction})--(\ref{eq:EqualOutgoingJunction}).
We will call $(r_{\backslash},\text{\textgreek{W}}_{\backslash},\bar{f}_{in\backslash},\bar{f}_{out\backslash})$
and $(r_{\backslash},\text{\textgreek{W}}_{\backslash},\bar{f}_{in\backslash},\bar{f}_{out\backslash})$
a\emph{ boundary-double characteristic initial data} \emph{set} for
the system (\ref{eq:RequationFinal})--(\ref{eq:OutgoingVlasovFinal})
with boundary at $r=r_{0}$ and satisfying the reflecting gauge condition
at $r=r_{0}$ if they satisfy the assumptions (\ref{eq:ConstraintUDef})--(\ref{eq:OutgoingEquationCombatibility})
of Definition \ref{def:TypeI}, and moreover 
\begin{equation}
r_{\backslash}(u_{2})=r_{0},
\end{equation}
\begin{equation}
\partial_{u}r_{\backslash}(u_{2})<0
\end{equation}
and 
\begin{equation}
(\partial_{v}r)_{\backslash}(u_{2})=-\partial_{u}r_{\backslash}(u_{2}),\label{eq:PositiveD_vDerivativeOnMirror}
\end{equation}
where $(\partial_{v}r)_{/}$ is defined by (\ref{eq:TransversalDerivativeV}),
i.\,e., for any $u_{1}\le u\le u_{2}$: 
\begin{equation}
r_{\backslash}(\partial_{v}r)_{\backslash}(u)\doteq r_{/}\partial_{v}r_{/}(v_{1})-\frac{1}{4}\int_{u_{1}}^{u}(1-\Lambda r_{\backslash}^{2}(\bar{u}))\text{\textgreek{W}}_{\backslash}^{2}(\bar{u})\, d\bar{u}.\label{eq:DefinitionTransversalD_vDerivative}
\end{equation}
\end{defn}
\begin{rem*}
Note that (\ref{eq:PositiveD_vDerivativeOnMirror}) and (\ref{eq:DefinitionTransversalD_vDerivative})
imply that 
\begin{equation}
(\partial_{v}r)_{\backslash}(u)>0\label{eq:PositiveDerivativeD_vonIngoing}
\end{equation}
for all $u\in[u_{1},u_{2}]$. In particular, 
\begin{equation}
\partial_{v}r_{/}(v_{1})>0.\label{eq:PositivityOfOutgoingDerivativeForTypeIII}
\end{equation}

\end{rem*}

\subsection{\label{sub:Local-well-Posedness}Local existence and uniqueness}

In this Section, we will establish the local well-posedness of the
initial value problems for (\ref{eq:RequationFinal})--(\ref{eq:OutgoingVlasovFinal})
associated to the types of initial data sets introduced by Definitions
\ref{def:TypeII}, \ref{def:TypeI} and \ref{def:TypeIII}. We will
in fact establish the well-posedness of these initial value problems
in the rough topology defined by (\ref{eq:GeometricNorm}).

The next result is a well-posedness result for the initial data introduced
by Definition \ref{def:TypeII}.
\begin{prop}
\label{Prop:LocalExistenceTypeII}Let $C_{0}\gg1$ be a (large) constant.
For any $v_{1}<v_{2}$ and any $r_{0}>0$, let $(r_{/},\text{\textgreek{W}}_{/}^{2},\bar{f}_{in/},\bar{f}_{out/})$
be a smooth asymptotically AdS boundary-characteristic initial data
set on $[v_{1},v_{2})$ for the system (\ref{eq:RequationFinal})--(\ref{eq:OutgoingVlasovFinal})
satisfying the reflecting gauge condition at $r=r_{0},+\infty$, according
to Definition \ref{def:TypeII}, such that the quantities $\frac{\text{\textgreek{W}}_{/}^{2}}{1-\frac{1}{3}\Lambda r_{/}^{2}},r_{/}^{2}(T_{vv})_{/}$
and $\tan^{-1}r_{/}$ extend smoothly on $v=v_{2}$. Let us also set
\begin{equation}
M\doteq\sup_{v\in[v_{1},v_{2})}\Bigg\{\Big|\log\big(\frac{\text{\textgreek{W}}_{/}^{2}}{1-\frac{1}{3}\Lambda r_{/}^{2}}\big)\Big|+\Big|\log\Big(\frac{2\partial_{v}r_{/}}{1-\frac{2m_{/}}{r_{/}}}\Big)\Big|+\Big|\log\Big(\frac{1-\frac{2m_{/}}{r_{/}}}{1-\frac{1}{3}\Lambda r_{/}^{2}}\Big)\Big|+\sqrt{-\Lambda}|\tilde{m}_{/}|\Bigg\}(v)+\int_{v_{1}}^{v_{2}}r_{/}(T_{vv})_{/}\, d\bar{v}\label{eq:UpperBoundInitialData}
\end{equation}
and, for any $0<\text{\textgreek{d}}<1$: 
\begin{equation}
v_{in}(\text{\textgreek{d}})\doteq\sup\Bigg\{0\le v_{*}\le v_{2}-v_{1}:\mbox{ }\sup_{v\in[v_{1},v_{2}]}\int_{\max\{v-v_{*},v_{1}\}}^{\min\{v+v_{*},v_{2}\}}\frac{r_{/}^{2}(T_{vv})_{/}(\bar{v})}{|\bar{v}-v|+r_{0}}\, d\bar{v}<\text{\textgreek{d}}\Bigg\},\label{eq:SmallnessInL1OfT}
\end{equation}
where 
\begin{equation}
m_{/}(v)=\frac{r_{/}}{2}\big(1+4\text{\textgreek{W}}_{/}^{-2}(\partial_{u}r)_{/}\partial_{v}r_{/}\big)(v),
\end{equation}
 
\begin{equation}
\tilde{m}_{/}(v)=m(v)-\frac{1}{6}\Lambda r_{/}^{3}(v)
\end{equation}
and $(\partial_{u}r)_{/}$ is defined according to (\ref{eq:TransversalDerivativeU-1}).
Then, provided 
\begin{equation}
u_{0}<\frac{v_{in}(2e^{-C_{0}^{2}M}M)}{e^{C_{0}^{2}((-\Lambda)(v_{2}-v_{1})^{2}+1)}},\label{eq:U0UpperBound}
\end{equation}
the following holds: Setting 
\begin{equation}
\mathcal{I}_{u_{0}}\doteq\{u=v-v_{2}\}\cap\{0<u<u_{0}\},
\end{equation}
\begin{equation}
\text{\textgreek{g}}_{0;u_{0}}\doteq\{u=v-v_{1}\}\cap\{0<u<u_{0}\}
\end{equation}
and 
\begin{equation}
\mathcal{W}\doteq\{0<u<u_{0}\}\cap\{u+v_{1}<v<u+v_{2}\},
\end{equation}
 there exist unique smooth functions $r:\mathcal{W}\cup\text{\textgreek{g}}_{0;u_{0}}\rightarrow(r_{0},+\infty)$,
$\text{\textgreek{W}}:\mathcal{W}\cup\text{\textgreek{g}}_{0;u_{0}}\rightarrow(0,+\infty)$
and $\bar{f}_{in},\bar{f}_{out}:\mathcal{W}\cup\text{\textgreek{g}}_{0;u_{0}}\times(0,+\infty)\rightarrow[0,+\infty)$
solving equations (\ref{eq:RequationFinal})--(\ref{eq:OutgoingVlasovFinal})
on $\mathcal{W}$ (with $T_{uu},T_{vv}$ expressed by (\ref{eq:T_uuComponent}),
(\ref{eq:T_vvComponent})), such that:

\begin{enumerate}

\item The functions $r,\text{\textgreek{W}}^{2},\bar{f}_{in},\bar{f}_{out}$
satisfy the given initial conditions on $\{0\}\times[v_{1},v_{2})$,
i.\,e.: 
\begin{equation}
(r,\text{\textgreek{W}}^{2},\bar{f}_{in},\bar{f}_{out})|_{\{0\}\times[v_{1},v_{2})}=(r_{/},\text{\textgreek{W}}_{/}^{2},\bar{f}_{in/},\bar{f}_{out/}).\label{eq:InitialDataRight-1}
\end{equation}

\item The functions $(r,\bar{f}_{in},\bar{f}_{out})$ satisfy on
$\text{\textgreek{g}}_{0;u_{0}}$ the boundary conditions 
\begin{equation}
r|_{\text{\textgreek{g}}_{0;u_{0}}}=r_{0}\label{eq:MirrorLocalExistence}
\end{equation}
 and 
\begin{equation}
\bar{f}_{out}\big(u_{*},v_{*};\, p\big)=\bar{f}_{in}\big(u_{*},v_{*};\, p\big)\label{eq:ReflectionMirrorLocalExistence}
\end{equation}
for all $(u_{*},v_{*})\in\text{\textgreek{g}}_{0;u_{0}}$ and $p>0$,
as well as the reflecting gauge condition 
\begin{equation}
\partial_{u}r|_{\text{\textgreek{g}}_{0;u_{0}}}=-\partial_{v}r|_{\text{\textgreek{g}}_{0;u_{0}}}.\label{eq:GaugeMirrorLocalExistence}
\end{equation}

\item The functions $(r,\bar{f}_{in},\bar{f}_{out})$ satisfy on
$\mathcal{I}_{u_{0}}$ the boundary conditions: 
\begin{equation}
(1/r)|_{\mathcal{I}_{u_{0}}}=0\label{eq:InfinityRLocalExistence}
\end{equation}
 and 
\begin{equation}
\lim_{h\rightarrow0^{+}}\Bigg(\frac{\bar{f}_{in}\big(u_{*}+h,v_{*};\,\text{\textgreek{W}}^{-2}(u_{*}+h,v_{*})\cdot p\big)}{\bar{f}_{out}\big(u_{*},v_{*}-h;\,\text{\textgreek{W}}^{-2}(u_{*},v_{*}-h)\cdot p\big)}\Bigg)=1\label{eq:ReflectionInfinityLocalExistence}
\end{equation}
for all $(u_{*},v_{*})\in\mathcal{I}_{u_{0}}$ and $p>0$, as well
as the reflecting gauge condition 
\begin{equation}
\partial_{u}(1/r)|_{\mathcal{I}_{u_{0}}}=-\partial_{v}(1/r)|_{\mathcal{I}_{u_{0}}}.\label{eq:GaugeInfinityLocalExistence}
\end{equation}

\item The function $r$ satisfies 
\begin{equation}
\sup_{\mathcal{W}}\partial_{u}r<0.\label{eq:SignConditionDuR-1-1-2}
\end{equation}

\item The following estimates hold on $\mathcal{W}$: 
\begin{align}
\sup_{\mathcal{W}}\Bigg\{\Big|\log\big(\frac{\text{\textgreek{W}}^{2}}{1-\frac{1}{3}\Lambda r^{2}}\big)\Big|+\Big|\log\Big(\frac{2\partial_{v}r}{1-\frac{2m}{r}}\Big)\Big|+\Big|\log\Big(\frac{1-\frac{2m}{r}}{1-\frac{1}{3}\Lambda r^{2}}\Big)\Big|+\sqrt{-\Lambda}|\tilde{m}|\Bigg\}+\label{eq:BoundToShowInitialValueForContinuation-1}\\
+\sup_{\bar{u}}\int_{\{u=\bar{u}\}\cap\mathcal{W}}rT_{vv}\, dv & \le C_{0}M,\nonumber 
\end{align}
\begin{equation}
\sup_{\bar{v}}\int_{\{v=\bar{v}\}\cap\mathcal{W}}rT_{uu}\, du+\sup_{\bar{u}}\sup_{v\in[v_{1}+\bar{u},v_{2}+\bar{u}]}\int_{\max\{v-u_{0},v_{1}+\bar{u}\}}^{\min\{v+u_{0},v_{2}+\bar{u}\}}\frac{r^{2}(T_{vv})(\bar{u},\bar{v})}{|\bar{v}-v|+r_{0}}\, d\bar{v}\le e^{-C_{0}M}M
\end{equation}
and 
\begin{equation}
\sup_{\mathcal{W}}r^{2}T_{uu},\sup_{\mathcal{W}}r^{2}T_{vv}=\sup_{[v_{1},v_{2}]}r_{/}^{2}(T_{vv})_{/}.\label{eq:BoundForenergyMomentum}
\end{equation}

\end{enumerate}\end{prop}
\begin{rem*}
By repeating the proof of Proposition \ref{Prop:LocalExistenceTypeII}
without any significant change, one also infers the existence and
of a smooth development of $(r_{/},\text{\textgreek{W}}_{/}^{2},\bar{f}_{in/},\bar{f}_{out/})$
\emph{backwards} in time, i.\,e.~the existence and uniqueness of
a smooth solution $(r,\text{\textgreek{W}}^{2},\bar{f}_{in},\bar{f}_{out})$
to the system (\ref{eq:RequationFinal})--(\ref{eq:OutgoingVlasovFinal})
on 
\begin{equation}
\mathcal{W}^{-}\doteq\{-u_{0}<u<0\}\cap\{u+v_{1}<v<u+v_{2}\}
\end{equation}
satisfying (the analogues of) (\ref{eq:InitialDataRight-1})--(\ref{eq:BoundForenergyMomentum}).
\end{rem*}
\begin{figure}[h] 
\centering 
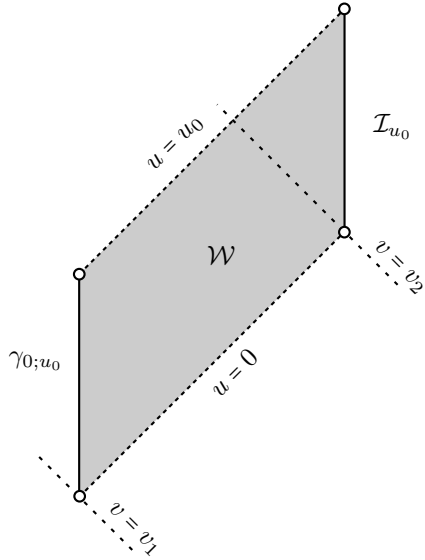 
\caption{Schematic depiction of the domain $\mathcal{W}$ in the statement of Proposition \ref{Prop:LocalExistenceTypeII}.}
\end{figure}
\begin{proof}
It suffices to establish the existence and uniqueness of a smooth
solution $(r,m,T_{uu},T_{vv})$ to the system (\ref{eq:DerivativeInUDirectionKappa})--(\ref{eq:ConservationT_uu})
on $\mathcal{W}\cup\text{\textgreek{g}}_{0;u_{0}}$, satisfying initial
conditions 
\begin{equation}
(r,T_{vv})|_{\{0\}\times[v_{1},v_{2})}=(r_{/},(T_{vv})_{/})\label{eq:InitialRight}
\end{equation}
and 
\begin{equation}
\tilde{m}(0,v_{1})=\tilde{m}_{/}(0,v_{1})\label{eq:InitialMass}
\end{equation}
the gauge conditions (\ref{eq:GaugeMirrorLocalExistence}) , (\ref{eq:GaugeInfinityLocalExistence})
and the boundary conditions (\ref{eq:MirrorLocalExistence}), (\ref{eq:InfinityRLocalExistence}),
\begin{equation}
\frac{(\partial_{v}r)^{2}r^{2}T_{uu}}{(\partial_{u}r)^{2}r^{2}T_{vv}}\Bigg|_{\text{\textgreek{g}}_{0;u_{0}}}=1\label{eq:BoundaryConditionTAxis}
\end{equation}
and
\begin{equation}
\frac{(\partial_{v}r)^{2}r^{2}T_{uu}}{(\partial_{u}r)^{2}r^{2}T_{vv}}\Bigg|_{\mathcal{I}_{u_{0}}}=1.\label{eq:BoundaryConditionTInfinity}
\end{equation}
Given the existence and uniqueness of such a smooth solution $(r,m,T_{uu},T_{vv})$
to the system (\ref{eq:DerivativeInUDirectionKappa})--(\ref{eq:ConservationT_uu}),
by solving equations (\ref{eq:IngoingVlasovFinal}) and (\ref{eq:OutgoingVlasovFinal})
for $\bar{f}_{in},\bar{f}_{out}$ on $\mathcal{W}$ with initial data
\begin{equation}
(\bar{f}_{in},\bar{f}_{out})|_{\{0\}\times[v_{1},v_{2})}=(\bar{f}_{in/},\bar{f}_{out/})
\end{equation}
and boundary conditions (\ref{eq:ReflectionMirrorLocalExistence})
and (\ref{eq:ReflectionInfinityLocalExistence}), and using the formula
(\ref{eq:RelationHawkingMass}) for $\text{\textgreek{W}}^{2}$, one
then obtains the existence and uniqueness of a smooth solution $(r,\text{\textgreek{W}}^{2},\bar{f}_{in},\bar{f}_{out})$
of the system (\ref{eq:RequationFinal})--(\ref{eq:OutgoingVlasovFinal})
satisfying (\ref{eq:InitialDataRight-1}), (\ref{eq:GaugeMirrorLocalExistence}),
(\ref{eq:GaugeInfinityLocalExistence}), (\ref{eq:MirrorLocalExistence}),
(\ref{eq:InfinityRLocalExistence}), (\ref{eq:ReflectionMirrorLocalExistence})
and (\ref{eq:ReflectionInfinityLocalExistence}). 

Let us introduce a new set of renormalised variables 
\begin{gather}
\text{\textgreek{r}}=\tan^{-1}(\sqrt{-\frac{\Lambda}{3}}r),\label{eq:RhoVariable}\\
\text{\textgreek{k}}=2\frac{\partial_{v}r}{1-\frac{2m}{r}},\label{eq:KappaVariable}\\
\bar{\text{\textgreek{k}}}=2\frac{-\partial_{u}r}{1-\frac{2m}{r}}\label{eq:KappaBarVariable}\\
\bar{\text{\textgreek{t}}}=r^{2}T_{vv},\label{eq:Tbar}\\
\text{\textgreek{t}}=r^{2}T_{uu}.\label{eq:TVariable}
\end{gather}
We will also define, for convenience, the functions $F_{1},F_{2}:(0,+\infty)\times\mathbb{R}\rightarrow\mathbb{R}$
by the relations 
\begin{equation}
F_{1}(x;y)=8\pi\frac{1}{x(1-\frac{2y}{x}+x^{2})},
\end{equation}
 and
\begin{equation}
F_{2}(x;y)=-\frac{1}{2}\frac{y}{x}\frac{(1+3x^{2})(1-\frac{2y}{x}+x^{2})}{x(1+x^{2})^{2}}.\label{eq:RightHandSideForRho}
\end{equation}
Note that, for any $x_{0}>0$, any $y\in\mathbb{R}$ such that $(1-\frac{2y}{x_{0}}+x_{0}^{2})>0$,
any $x\ge x_{0}$ and any integer $k\ge0$, we can bound:
\begin{equation}
x^{k+1}\big|\partial_{x}^{k}F_{1}(x,y)\big|+x^{k+1}\big|\partial_{x}^{k}F_{2}(x,y)\big|\le C_{k}\Big(1+\frac{|y|^{k}}{x^{k}}(1-\frac{2y}{x}+x^{2})^{-k-1}+\frac{|y|^{2}}{x^{2}}\Big)\label{eq:BoundednessF_1,F_2,F_3}
\end{equation}
Switching to the variables (\ref{eq:RhoVariable})--(\ref{eq:TVariable})
and using the relations (\ref{eq:DefinitionHawkingMass}) and (\ref{eq:RenormalisedHawkingMass})
for $\tilde{m}$, the system (\ref{eq:DerivativeInUDirectionKappa})--(\ref{eq:ConservationT_uu})
transforms into: 
\begin{align}
\partial_{u}\log(\text{\textgreek{k}})= & -\sqrt{-\frac{\Lambda}{3}}F_{1}\big(\tan\text{\textgreek{r}};\sqrt{-\frac{\Lambda}{3}}\tilde{m}\big)\bar{\text{\textgreek{k}}}^{-1}\text{\textgreek{t}},\label{eq:DerivativeInUDirectionKappaRenormalised}\\
\partial_{v}\log(\bar{\text{\textgreek{k}}})= & \sqrt{-\frac{\Lambda}{3}}F_{1}\big(\tan\text{\textgreek{r}};\sqrt{-\frac{\Lambda}{3}}\tilde{m}\big)\text{\textgreek{k}}^{-1}\bar{\text{\textgreek{t}}},\label{eq:DerivativeInVDirectionKappaBarRenormalised}\\
\partial_{u}\partial_{v}\text{\textgreek{r}}= & (-\frac{\Lambda}{3})F_{2}\big(\tan\text{\textgreek{r}};\sqrt{-\frac{\Lambda}{3}}\tilde{m}\big)\text{\textgreek{k}}\bar{\text{\textgreek{k}}},\label{eq:EquationRForRhoRenormalised}\\
\partial_{u}\tilde{m}= & -4\pi\bar{\text{\textgreek{k}}}^{-1}\text{\textgreek{t}},\label{eq:DerivativeInUtildeMRenormalised}\\
\partial_{u}\bar{\text{\textgreek{t}}}= & 0,\label{eq:ConservationTauBar}\\
\partial_{v}\text{\textgreek{t}}= & 0,\label{eq:ConservationTau}
\end{align}
 The initial condition (\ref{eq:InitialRight}), the gauge conditions
(\ref{eq:GaugeMirrorLocalExistence}), (\ref{eq:GaugeInfinityLocalExistence})
and the boundary conditions (\ref{eq:MirrorLocalExistence}), (\ref{eq:InfinityRLocalExistence})
(\ref{eq:BoundaryConditionTAxis}) and (\ref{eq:BoundaryConditionTInfinity})
are then replaced by: 
\begin{equation}
(\text{\textgreek{r}},\text{\textgreek{k}},\tilde{m},\bar{\text{\textgreek{t}}})|_{\{0\}\times[v_{1},v_{2})}=\Big(\tan^{-1}\big(\sqrt{-\frac{\Lambda}{3}}r_{/}\big),\frac{2\partial_{v}r_{/}}{1-\frac{2m_{/}}{r_{/}}},\tilde{m}_{/},r_{/}^{2}(T_{vv})_{/}\Big),\label{eq:InitialDataRecastRight}
\end{equation}
\begin{equation}
\text{\textgreek{r}}|_{\text{\textgreek{g}}_{0;u_{0}}}=\text{\textgreek{r}}_{0},\,\text{\textgreek{r}}|_{\mathcal{I}_{u_{0}}}=\frac{\pi}{2},\label{eq:RhoBoundary}
\end{equation}
 
\begin{equation}
\frac{\text{\textgreek{k}}}{\bar{\text{\textgreek{k}}}}\Big|_{\text{\textgreek{g}}_{0;u_{0}}}=1,\textnormal{}\frac{\text{\textgreek{k}}}{\bar{\text{\textgreek{k}}}}\Big|_{\mathcal{I}_{u_{0}}}=1\label{eq:KappaBoundary}
\end{equation}
 and
\begin{equation}
\frac{\text{\textgreek{t}}}{\bar{\text{\textgreek{t}}}}\Big|_{\text{\textgreek{g}}_{0;u_{0}}}=1,\mbox{ }\frac{\text{\textgreek{t}}}{\bar{\text{\textgreek{t}}}}\Big|_{\mathcal{I}_{u_{0}}}=1,\label{eq:TauBoundary}
\end{equation}
where 
\begin{equation}
\text{\textgreek{r}}_{0}\doteq\tan^{-1}(\sqrt{-\frac{\Lambda}{3}}r_{0})>0.\label{eq:Rho_0}
\end{equation}

\medskip{}

\noindent \emph{Remark.} Note that equations (\ref{eq:DerivativeInVDirectionKappaBar})
and (\ref{eq:EquationRForProof}) yield equation (\ref{eq:DerivativeTildeVMass})
for $\tilde{m}$. The relations (\ref{eq:DerivativeTildeVMass}),
(\ref{eq:DerivativeInUtildeMRenormalised}) and (\ref{eq:BoundaryConditionTInfinity})
imply that $\tilde{m}$ is conserved on $\mathcal{I}_{u_{0}}$, i.\,e.
\begin{equation}
\tilde{m}|_{\mathcal{I}_{u_{0}}}=\lim_{v\rightarrow v_{2}^{-}}\tilde{m}(0,v_{2})=\lim_{v\rightarrow v_{2}^{-}}\tilde{m}_{/}(0,v).\label{eq:ConservationTildeMAxis}
\end{equation}

\medskip{}

The proof of Proposition \ref{Prop:LocalExistenceTypeII} will follow
if we establish that the system (\ref{eq:DerivativeInUDirectionKappaRenormalised})--(\ref{eq:ConservationTau})
admits a unique smooth solution $(\text{\textgreek{r}},\text{\textgreek{k}},\bar{\text{\textgreek{k}}}\text{\textgreek{t}},\bar{\text{\textgreek{t}}},\tilde{m})$
on $\mathcal{W}\cup\text{\textgreek{g}}_{0;u_{0}}$ satisfying (\ref{eq:InitialDataRecastRight})--(\ref{eq:ConservationTildeMAxis})
and the estimates 
\begin{align}
\sup_{\mathcal{W}}\Big\{|\log(\text{\textgreek{k}})|+|\log(\bar{\text{\textgreek{k}}})|+\Big|\log\big(\sqrt{-\frac{12}{\Lambda}}\partial_{v}\text{\textgreek{r}}\big)\Big|+\Big|\log\big(-\sqrt{-\frac{12}{\Lambda}}\partial_{u}\text{\textgreek{r}}\big)\Big|+\label{eq:BoundToShowInitialValueFirstStep}\\
+\Big|\log\Big(\frac{1-\frac{2\sqrt{-\frac{\Lambda}{3}}\tilde{m}}{\tan\text{\textgreek{r}}}+\tan^{2}\text{\textgreek{r}}}{1+\tan^{2}\text{\textgreek{r}}}\Big)\Big|+\sqrt{-\Lambda}|\tilde{m}|\Big\}+\sup_{0\le\bar{u}\le u_{0}}\int_{\mathcal{W}\cap\{u=\bar{u}\}}\sqrt{-\Lambda} & \frac{\bar{\text{\textgreek{t}}}}{\text{\textgreek{r}}}\, dv\le\frac{1}{2}C_{0}M\nonumber 
\end{align}
(where $M$ is defined by (\ref{eq:UpperBoundInitialData})), 
\begin{equation}
\sup_{v_{1}\le\bar{v}\le v_{2}+u_{0}}\int_{\mathcal{W}\cap\{v=\bar{v}\}}\sqrt{-\Lambda}\frac{\text{\textgreek{t}}}{\text{\textgreek{r}}}\, du+\sup_{0\le\bar{u}\le u_{0}}\sup_{v\in[v_{1}+\bar{u},v_{2}+\bar{u}]}\int_{\max\{v-u_{0},v_{1}+\bar{u}\}}^{\min\{v+u_{0},v_{2}+\bar{u}\}}\frac{\bar{\text{\textgreek{t}}}}{|\bar{v}-v|+r_{0}}\, d\bar{v}\le e^{-C_{0}^{\frac{3}{2}}M}M\label{eq:BoundToShowSmallnessInDvR}
\end{equation}
 and 
\begin{equation}
\sup_{\mathcal{W}}\max\{\text{\textgreek{t}},\bar{\text{\textgreek{t}}}\}\le\sup_{v_{1}\le v\le v_{2}}r_{/}^{2}(T_{vv})_{/},\label{eq:UpperBoundForTauTauBar}
\end{equation}
 as well as the bound 
\begin{equation}
\text{\textgreek{r}}_{0}<\text{\textgreek{r}}|_{\mathcal{W}}<\frac{\pi}{2}.\label{eq:UpperBoundRho}
\end{equation}
Note that (\ref{eq:UpperBoundRho}) readily follows once (\ref{eq:BoundToShowInitialValueFirstStep})
has been established, in view of the fact that $\partial_{u}\text{\textgreek{r}}<0$
(as a consequence of (\ref{eq:BoundToShowInitialValueFirstStep}))
and $r_{/}|_{[v_{1},v_{2})}<+\infty$.

The proof of the existence and uniqueness of a smooth solution to
(\ref{eq:DerivativeInUDirectionKappaRenormalised})--(\ref{eq:ConservationTau})
satisfying (\ref{eq:InitialDataRecastRight})--(\ref{eq:ConservationTildeMAxis})
and the estimates (\ref{eq:BoundToShowInitialValueFirstStep})--(\ref{eq:UpperBoundForTauTauBar})
will consisit of two steps:

\medskip{}

\noindent \emph{Step 1.} We will first show that there exists a (weak)
$C^{0}$ solution $(\text{\textgreek{r}},\text{\textgreek{k}},\bar{\text{\textgreek{k}}},\text{\textgreek{t}},\bar{\text{\textgreek{t}}},\tilde{m})$
of the system (\ref{eq:DerivativeInUDirectionKappaRenormalised})--(\ref{eq:ConservationTau})
on $\mathcal{W}$ satisfying (\ref{eq:InitialDataRecastRight})--(\ref{eq:ConservationTildeMAxis})
and (\ref{eq:BoundToShowInitialValueFirstStep})--(\ref{eq:UpperBoundForTauTauBar})
and extending continuously on $\text{\textgreek{g}}_{0;u_{0}},\mathcal{I}_{u_{0}}$.
Notice that, in view of our assumption that the quantities $\frac{\text{\textgreek{W}}_{/}^{2}}{1-\frac{1}{3}\Lambda r_{/}^{2}},r_{/}^{2}(T_{vv})_{/}$
and $\tan^{-1}r_{/}$ extend smoothly up to $v=v_{2}$, the initial
data (\ref{eq:InitialDataRecastRight}) for $(\text{\textgreek{r}},\text{\textgreek{k}},\tilde{m},\bar{\text{\textgreek{t}}})$
on $\{0\}\times[v_{1},v_{2})$ extend smoothly up to $v=v_{2}$. Using
the bounds (\ref{eq:BoundToShowInitialValueFirstStep}) and (\ref{eq:UpperBoundForTauTauBar}),
it can be then readily inferred that $(\text{\textgreek{r}},\text{\textgreek{k}},\bar{\text{\textgreek{k}}},\text{\textgreek{t}},\bar{\text{\textgreek{t}}},\tilde{m})$
is in fact a classical $C^{\infty}$ solution on $\mathcal{W}$ extending
smoothly on $\text{\textgreek{g}}_{0;u_{0}}$, by commuting equations
(\ref{eq:DerivativeInUDirectionKappaRenormalised})--(\ref{eq:ConservationTau})
with $\partial_{u},\partial_{v}$ and treating the commuted equations
as linear equations in the highest order terms.

The proof will follow by the usual iteration argument. Let us define
inductively the sequence of $C^{1}$ functions $\big\{(\text{\textgreek{r}}_{n},\text{\textgreek{k}}_{n},\bar{\text{\textgreek{k}}}_{n},\text{\textgreek{t}}_{n},\bar{\text{\textgreek{t}}}_{n},\tilde{m}_{n})\Big\}_{n\in\mathbb{N}}$
on $\mathcal{W}$ by solving for each $n\in\mathbb{N}$: 
\begin{align}
\partial_{u}\log(\text{\textgreek{k}}_{n})= & -\sqrt{-\frac{\Lambda}{3}}F_{1}\big(\tan\text{\textgreek{r}}_{n-1};\sqrt{-\frac{\Lambda}{3}}\tilde{m}_{n-1}\big)\bar{\text{\textgreek{k}}}_{n-1}^{-1}\text{\textgreek{t}}_{n-1},\label{eq:DerivativeInUDirectionKappaInduction}\\
\partial_{v}\log(\bar{\text{\textgreek{k}}}_{n})= & \sqrt{-\frac{\Lambda}{3}}F_{1}\big(\tan\text{\textgreek{r}}_{n-1};\sqrt{-\frac{\Lambda}{3}}\tilde{m}_{n-1}\big)\text{\textgreek{k}}_{n-1}^{-1}\bar{\text{\textgreek{t}}}_{n-1},\label{eq:DerivativeInVDirectionKappaBarInduction}\\
\partial_{u}\partial_{v}\text{\textgreek{r}}_{n}= & (-\frac{\Lambda}{3})F_{2}\big(\tan\text{\textgreek{r}}_{n-1};\sqrt{-\frac{\Lambda}{3}}\tilde{m}_{n-1}\big)\text{\textgreek{k}}_{n-1}\bar{\text{\textgreek{k}}}_{n-1},\label{eq:EquationRForRhoInduction}\\
\partial_{u}\tilde{m}_{n}= & -4\pi\bar{\text{\textgreek{k}}}_{n}^{-1}\text{\textgreek{t}}_{n-1},\label{eq:DerivativeInUtildeMInduction}\\
\partial_{u}\bar{\text{\textgreek{t}}}_{n}= & 0,\label{eq:ConservationTauBarInduction}\\
\partial_{v}\text{\textgreek{t}}_{n}= & 0,\label{eq:ConservationTauInduction}
\end{align}
 where $(\text{\textgreek{r}}_{n},\text{\textgreek{k}}_{n},\bar{\text{\textgreek{k}}}_{n},\bar{\text{\textgreek{t}}}_{n},\text{\textgreek{t}}_{n},\tilde{m}_{n})$
satisfy 
\begin{equation}
(\text{\textgreek{r}}_{n},\text{\textgreek{k}}_{n},\tilde{m},\bar{\text{\textgreek{t}}}_{n})|_{\{0\}\times[v_{1},v_{2}]}=\Big(\tan^{-1}\big(\sqrt{-\frac{\Lambda}{3}}r_{/}\big),\frac{2\partial_{v}r_{/}}{1-\frac{2m_{/}}{r_{/}}},\tilde{m}_{/},r_{/}^{2}(T_{vv})_{/}\Big),\label{eq:InitialDataSequenceInduct}
\end{equation}
\begin{equation}
\text{\textgreek{r}}_{n}|_{\text{\textgreek{g}}_{0}}=\text{\textgreek{r}}_{0},\,\text{\textgreek{r}}_{n}|_{\mathcal{I}_{u_{0}}}=\frac{\pi}{2},\label{eq:RhoBoundaryInduct}
\end{equation}
 
\begin{equation}
\frac{\text{\textgreek{k}}_{n}}{\bar{\text{\textgreek{k}}}_{n}}\Big|_{\text{\textgreek{g}}_{0}}=1,\textnormal{}\frac{\text{\textgreek{k}}_{n}}{\bar{\text{\textgreek{k}}}_{n}}\Big|_{\mathcal{I}_{u_{0}}}=1,\label{eq:KappaBoundaryInduct}
\end{equation}
\begin{equation}
\frac{\text{\textgreek{t}}_{n}}{\bar{\text{\textgreek{t}}}_{n}}\Big|_{\text{\textgreek{g}}_{0}}=1,\mbox{ }\frac{\text{\textgreek{t}}_{n}}{\bar{\text{\textgreek{t}}}_{n}}\Big|_{\mathcal{I}_{u_{0}}}=1\label{eq:TauBoundaryInduct}
\end{equation}
and
\begin{equation}
\tilde{m}_{n}|_{\mathcal{I}_{u_{0}}}=\lim_{v\rightarrow v_{2}}\tilde{m}_{/}(0,v).\label{eq:ConservationTildeMAxisInduct}
\end{equation}
We also use the convention that, 
\begin{equation}
\text{\textgreek{k}}_{0}=\bar{\text{\textgreek{k}}}_{0}=\text{\textgreek{t}}_{0}=\bar{\text{\textgreek{t}}}_{0}=0.
\end{equation}

\medskip{}

\noindent \emph{Remark. }Notice that, under the assumption that $(\text{\textgreek{r}}_{n-1},\text{\textgreek{k}}_{n-1},\bar{\text{\textgreek{k}}}_{n-1},\text{\textgreek{t}}_{n-1},\bar{\text{\textgreek{t}}}_{n-1},\tilde{m}_{n-1})$
are $C^{1}$ functions on $\mathcal{W}$ satisying $\inf_{\mathcal{W}}\text{\textgreek{r}}_{n-1}>0$,
(\ref{eq:BoundednessF_1,F_2,F_3}) and (\ref{eq:UpperBoundForTauTauBar})
imply that $(\text{\textgreek{r}}_{n},\text{\textgreek{k}}_{n},\bar{\text{\textgreek{k}}}_{n},\text{\textgreek{t}}_{n},\bar{\text{\textgreek{t}}}_{n},\tilde{m}_{n})$
(obtained by solving (\ref{eq:DerivativeInUDirectionKappaInduction})--(\ref{eq:ConservationTauInduction})
with (\ref{eq:InitialDataSequenceInduct})--(\ref{eq:ConservationTildeMAxisInduct}))
are $C^{1}$ functions on $\mathcal{W}$ (despite the fact that $\tan\text{\textgreek{r}}_{n-1}$,
appearing as an argument of $F_{1},F_{2},F_{3}$ in the right hand
side of (\ref{eq:DerivativeInUDirectionKappaInduction})--(\ref{eq:ConservationTauInduction}),
is unbounded on $\mathcal{W}$).

\medskip{}

We will show that, as $n\rightarrow\infty$, the sequence $(\text{\textgreek{r}}_{n},\text{\textgreek{k}}_{n},\bar{\text{\textgreek{k}}}_{n},\text{\textgreek{t}}_{n},\bar{\text{\textgreek{t}}}_{n},\tilde{m}_{n})_{n\in\mathbb{N}}$
converges in the $C^{0}\big(\mathcal{W}\big)$ norm to a solution
$(\text{\textgreek{r}},\text{\textgreek{k}},\bar{\text{\textgreek{k}}},\text{\textgreek{t}},\bar{\text{\textgreek{t}}},\tilde{m})$
of (\ref{eq:DerivativeInUDirectionKappaRenormalised})--(\ref{eq:ConservationTau})
on $\mathcal{W}$ satisfying (\ref{eq:InitialDataRecastRight})--(\ref{eq:ConservationTildeMAxis})
and (\ref{eq:BoundToShowInitialValueFirstStep})--(\ref{eq:UpperBoundForTauTauBar}).
To this end, it suffices to show that, for any $n\in\mathbb{N}$,
under the assumption that for any $1\le k\le n-1$
\begin{align}
\sup_{\mathcal{W}}\Big\{|\log(\text{\textgreek{k}}_{k})|+|\log(\bar{\text{\textgreek{k}}}_{k})|+\Big|\log\big(\sqrt{-\frac{12}{\Lambda}}\partial_{v}\text{\textgreek{r}}_{k}\big)\Big|+\Big|\log\big(-\sqrt{-\frac{12}{\Lambda}}\partial_{u}\text{\textgreek{r}}_{k}\big)\Big|+\label{eq:BoundForInductionPreviousStep}\\
+\Big|\log\Big(\frac{1-\frac{2\sqrt{-\frac{\Lambda}{3}}\tilde{m}_{k}}{\tan\text{\textgreek{r}}_{k}}+\tan^{2}\text{\textgreek{r}}_{k}}{1+\tan^{2}\text{\textgreek{r}}_{k}}\Big)\Big|+\sqrt{-\Lambda}|\tilde{m}_{k}|\Big\}+\sup_{0\le\bar{u}\le u_{0}}\int_{\mathcal{W}\cap\{u=\bar{u}\}} & \sqrt{-\Lambda}\frac{\bar{\text{\textgreek{t}}}_{k}}{\text{\textgreek{r}}_{k}}\, dv\le\frac{1}{2}C_{0}M\nonumber 
\end{align}
\begin{equation}
\sup_{v_{1}\le\bar{v}\le v_{2}+u_{0}}\int_{\mathcal{W}\cap\{v=\bar{v}\}}\sqrt{-\Lambda}\frac{\text{\textgreek{t}}_{k}}{\text{\textgreek{r}}_{k}}\, du+\sup_{0\le\bar{u}\le u_{0}}\sup_{v\in[v_{1}+\bar{u},v_{2}+\bar{u}]}\int_{\max\{v-u_{0},v_{1}+\bar{u}\}}^{\min\{v+u_{0},v_{2}+\bar{u}\}}\frac{\bar{\text{\textgreek{t}}}_{k}}{|\bar{v}-v|+r_{0}}\, d\bar{v}\le e^{-C_{0}^{\frac{3}{2}}M}M\label{eq:SmallnessPreviousStep}
\end{equation}
\begin{equation}
\sup_{\mathcal{W}}\max\{\text{\textgreek{t}}_{k},\bar{\text{\textgreek{t}}}_{k}\}\le2\sup_{v_{1}\le v\le v_{2}}r_{/}^{2}(T_{vv})_{/},\label{eq:UpperBoundForTauTauBarPreviousStep}
\end{equation}
 and 
\begin{equation}
\text{\textgreek{r}}_{k}\ge\text{\textgreek{r}}_{0},\label{eq:LowerBoundRhoPreviousStep}
\end{equation}
the bounds (\ref{eq:BoundForInductionPreviousStep})--(\ref{eq:LowerBoundRhoPreviousStep})
also hold for $(\text{\textgreek{r}}_{n},\text{\textgreek{k}}_{n},\bar{\text{\textgreek{k}}}_{n},\text{\textgreek{t}}_{n},\bar{\text{\textgreek{t}}}_{n},\tilde{m}_{n})$,
i.\,e.:
\begin{align}
\sup_{\mathcal{W}}\Big\{|\log(\text{\textgreek{k}}_{n})|+|\log(\bar{\text{\textgreek{k}}}_{n})|+\Big|\log\big(\sqrt{-\frac{12}{\Lambda}}\partial_{v}\text{\textgreek{r}}_{n}\big)\Big|+\Big|\log\big(-\sqrt{-\frac{12}{\Lambda}}\partial_{u}\text{\textgreek{r}}_{n}\big)\Big|+\label{eq:BoundForInductionNextStep}\\
+\Big|\log\Big(\frac{1-\frac{2\sqrt{-\frac{\Lambda}{3}}\tilde{m}_{n}}{\tan\text{\textgreek{r}}_{n}}+\tan^{2}\text{\textgreek{r}}_{n}}{1+\tan^{2}\text{\textgreek{r}}_{n}}\Big)\Big|+\sqrt{-\Lambda}|\tilde{m}_{n}|\Big\}+\sup_{0\le\bar{u}\le u_{0}}\int_{\mathcal{W}\cap\{u=\bar{u}\}} & \sqrt{-\Lambda}\frac{\bar{\text{\textgreek{t}}}_{n}}{\text{\textgreek{r}}_{n}}\, dv\le\frac{1}{2}C_{0}M\nonumber 
\end{align}
\begin{equation}
\sup_{v_{1}\le\bar{v}\le v_{2}+u_{0}}\int_{\mathcal{W}\cap\{v=\bar{v}\}}\sqrt{-\Lambda}\frac{\text{\textgreek{t}}_{n}}{\text{\textgreek{r}}_{n}}\, du+\sup_{0\le\bar{u}\le u_{0}}\sup_{v\in[v_{1}+\bar{u},v_{2}+\bar{u}]}\int_{\max\{v-u_{0},v_{1}+\bar{u}\}}^{\min\{v+u_{0},v_{2}+\bar{u}\}}\frac{\bar{\text{\textgreek{t}}}_{n}}{|\bar{v}-v|+r_{0}}\, d\bar{v}\le e^{-C_{0}^{\frac{3}{2}}M}M,\label{eq:SmallnessNextStep}
\end{equation}
 
\begin{equation}
\sup_{\mathcal{W}}\max\{\text{\textgreek{t}}_{n},\bar{\text{\textgreek{t}}}_{n}\}\le2\sup_{v_{1}\le v\le v_{2}}r_{/}^{2}(T_{vv})_{/},\label{eq:UpperBoundForTauTauBarNextStep}
\end{equation}
and
\begin{equation}
\text{\textgreek{r}}_{n}\ge\text{\textgreek{r}}_{0}\label{eq:LowerBoundRhoNextStep}
\end{equation}
and, moreover, for $n\ge4$: 
\begin{equation}
\mathfrak{D}_{n}\le\frac{1}{2}\max\big\{\mathfrak{D}_{n-1},\mathfrak{D}_{n-2}\big\},\label{eq:DifferenceRate}
\end{equation}
where 
\begin{equation}
\mathfrak{D}_{n}\doteq\sup_{\mathcal{W}}\max\Big\{|\text{\textgreek{r}}_{n}-\text{\textgreek{r}}_{n-1}|,\Big|\log\text{\textgreek{k}}_{n}-\log\text{\textgreek{k}}_{n-1}\Big|,\Big|\log\bar{\text{\textgreek{k}}}_{n}-\log\bar{\text{\textgreek{k}}}_{n-1}\Big|,\sqrt{-\Lambda}|\tilde{m}_{n}-\tilde{m}_{n-1}|,|\text{\textgreek{t}}_{n}-\text{\textgreek{t}}_{n-1}|,|\bar{\text{\textgreek{t}}}_{n}-\bar{\text{\textgreek{t}}}_{n-1}|\Big\}.\label{eq:DifferenceNorm}
\end{equation}
Notice that, when $n=2$, the bounds (\ref{eq:BoundForInductionPreviousStep})--(\ref{eq:LowerBoundRhoNextStep})
can be readily obtained from (\ref{eq:DerivativeInUDirectionKappaInduction})--(\ref{eq:ConservationTauInduction}),
(\ref{eq:InitialDataSequenceInduct})--(\ref{eq:ConservationTildeMAxisInduct})
and (\ref{eq:UpperBoundInitialData})--(\ref{eq:U0UpperBound}), provided
$C_{0}\gg1$.

Integrating equation (\ref{eq:ConservationTauBarInduction}) along
the lines $\{v=const\}$ and (\ref{eq:ConservationTauInduction})
along the lines $\{u=const\}$, and using the boundary conditions
(\ref{eq:TauBoundaryInduct}), we obtain for any $(u,v)\in\mathcal{W}$
\begin{align}
\bar{\text{\textgreek{t}}}_{n} & (u,v)=\bar{\text{\textgreek{t}}}_{n}(0,v_{\dashv}[u,v])=r_{/}^{2}(T_{vv})_{/}(v_{\dashv}[u,v]),\label{eq:EqualityForEnergyMomentum}\\
\text{\textgreek{t}}_{n} & (u,v)=\bar{\text{\textgreek{t}}}_{n}(0,v_{\vdash}[u,v])=r_{/}^{2}(T_{vv})_{/}(v_{\vdash}[u,v]),\nonumber 
\end{align}
where 
\begin{equation}
v_{\dashv}[u,v]=\begin{cases}
v, & \mbox{if }v\le v_{2}\\
v-v_{2}+v_{1}, & \mbox{if }v>v_{2}
\end{cases}\label{eq:RightPastCurve}
\end{equation}
and 
\begin{equation}
v_{\vdash}[u,v]=v_{1}+u.\label{eq:LeftPastCurve}
\end{equation}
In particular, 
\begin{equation}
\bar{\text{\textgreek{t}}}_{n},\text{\textgreek{t}}_{n}\le\sup_{v_{1}\le v\le v_{2}}r_{/}^{2}(T_{vv})_{/}\label{eq:BoundForEnergyMomentum}
\end{equation}
 
\begin{equation}
\bar{\text{\textgreek{t}}}_{n}=\bar{\text{\textgreek{t}}}_{n-1}\mbox{ and }\text{\textgreek{t}}_{n}=\text{\textgreek{t}}_{n-1}.\label{eq:EqualityEnergyMomentum}
\end{equation}

Integrating equation (\ref{eq:DerivativeInUDirectionKappaInduction})
along $\{0\le u\le u_{0}\}\cap\{v=v_{*}\}$ for any $v_{1}\le v_{*}\le v_{2}$,
we obtain for any point $(u,v_{*})\in\mathcal{W}\cap\{v\le v_{2}\}$:
\begin{equation}
\log\big(\text{\textgreek{k}}_{n}(u,v_{*})\big)=\log\big(\text{\textgreek{k}}_{n}(0,v_{*})\big)-\int_{0}^{u}\sqrt{-\frac{\Lambda}{3}}F_{1}\big(\tan\text{\textgreek{r}}_{n-1};\sqrt{-\frac{\Lambda}{3}}\tilde{m}_{n-1}\big)\bar{\text{\textgreek{k}}}_{n-1}^{-1}\text{\textgreek{t}}_{n-1}(\bar{u},v_{*})\, d\bar{u}.\label{eq:FirstRelationKappaIntermediateRegion}
\end{equation}
In view of the bounds (\ref{eq:BoundForInductionPreviousStep}), (\ref{eq:SmallnessPreviousStep})
and (\ref{eq:BoundednessF_1,F_2,F_3}), using also (\ref{eq:UpperBoundInitialData})
and (\ref{eq:LowerBoundRhoPreviousStep}), the relation (\ref{eq:FirstRelationKappaIntermediateRegion})
yields for any point $(u,v)\in\mathcal{W}\cap\{v\le v_{2}\}$: 
\begin{equation}
\big|\log\big(\text{\textgreek{k}}_{n}(u,v)\big)\big|\le M(1+e^{-\frac{1}{2}C_{0}^{\frac{3}{2}}M}).\label{eq:UpperBoundK_nSmallerRegion}
\end{equation}
Furthermore, subtracting from equation (\ref{eq:FirstRelationKappaIntermediateRegion})
the same equation for $n-1$ in place of $n$, and using (\ref{eq:InitialDataSequenceInduct}),
(\ref{eq:BoundednessF_1,F_2,F_3}), (\ref{eq:BoundForInductionPreviousStep}),
(\ref{eq:SmallnessPreviousStep}), (\ref{eq:LowerBoundRhoPreviousStep})
and (\ref{eq:EqualityEnergyMomentum}), we infer that for any point
$(u,v)\in\mathcal{W}\cap\{v\le v_{2}\}$: 
\begin{equation}
\big|\log\big(\text{\textgreek{k}}_{n}(u,v)\big)-\log\big(\text{\textgreek{k}}_{n-1}(u,v)\big)\big|\le C_{0}e^{-\frac{1}{2}C_{0}^{\frac{3}{2}}M}M\mathfrak{D}_{n-1}.\label{eq:DifferenceKappaIntermediateRegion}
\end{equation}

Integrating equation (\ref{eq:DerivativeInVDirectionKappaBarInduction})
along $\{u=u_{*}\}\cap\{v_{1}+u_{*}\le v\le v_{2}+u_{*}\}$ for any
$0\le u_{*}\le u_{0}$, we obtain for any point $(u_{*},v)\in\mathcal{W}$:
\begin{align}
\log\big(\bar{\text{\textgreek{k}}}_{n}(u_{*},v)\big) & =\log\big(\bar{\text{\textgreek{k}}}_{n}(u_{*},v_{1}+u_{*})\big)+\int_{v_{1}+u_{*}}^{v}\sqrt{-\frac{\Lambda}{3}}F_{1}\big(\tan\text{\textgreek{r}}_{n-1};\sqrt{-\frac{\Lambda}{3}}\tilde{m}_{n-1}\big)\text{\textgreek{k}}_{n-1}^{-1}\bar{\text{\textgreek{t}}}_{n-1}(u_{*},\bar{v})\, d\bar{v}=\label{eq:FirstRelationKappaIntermediateRegion-1}\\
 & =\log\big(\bar{\text{\textgreek{k}}}_{n}(u_{*},v_{1}+u_{*})\big)+\int_{v_{1}+u_{*}}^{\min\{v_{2},v\}}\sqrt{-\frac{\Lambda}{3}}F_{1}\big(\tan\text{\textgreek{r}}_{n-1};\sqrt{-\frac{\Lambda}{3}}\tilde{m}_{n-1}\big)\text{\textgreek{k}}_{n-1}^{-1}\bar{\text{\textgreek{t}}}_{n-1}(u_{*},\bar{v})\, d\bar{v}+\nonumber \\
 & \hphantom{=++}+\int_{\min\{v_{2},v\}}^{v}\sqrt{-\frac{\Lambda}{3}}F_{1}\big(\tan\text{\textgreek{r}}_{n-1};\sqrt{-\frac{\Lambda}{3}}\tilde{m}_{n-1}\big)\text{\textgreek{k}}_{n-1}^{-1}\bar{\text{\textgreek{t}}}_{n-1}(u_{*},\bar{v})\, d\bar{v}=\nonumber \\
 & =\log\big(\bar{\text{\textgreek{k}}}_{n}(u_{*},v_{1}+u_{*})\big)+\int_{v_{1}+u_{*}}^{\min\{v_{2},v\}}\int_{0}^{u_{*}}\partial_{u}\Bigg\{\sqrt{-\frac{\Lambda}{3}}F_{1}\big(\tan\text{\textgreek{r}}_{n-1};\sqrt{-\frac{\Lambda}{3}}\tilde{m}_{n-1}\big)\text{\textgreek{k}}_{n-1}^{-1}\bar{\text{\textgreek{t}}}_{n-1}(\bar{u},\bar{v})\Bigg\}\, d\bar{u}d\bar{v}+\nonumber \\
 & \hphantom{=++}+\int_{v_{1}+u_{*}}^{\min\{v_{2},v\}}\sqrt{-\frac{\Lambda}{3}}F_{1}\big(\tan\text{\textgreek{r}}_{n-1};\sqrt{-\frac{\Lambda}{3}}\tilde{m}_{n-1}\big)\text{\textgreek{k}}_{n-1}^{-1}\bar{\text{\textgreek{t}}}_{n-1}(0,\bar{v})\, d\bar{v}+\nonumber \\
 & \hphantom{=++}+\int_{\min\{v_{2},v\}}^{v}\sqrt{-\frac{\Lambda}{3}}F_{1}\big(\tan\text{\textgreek{r}}_{n-1};\sqrt{-\frac{\Lambda}{3}}\tilde{m}_{n-1}\big)\text{\textgreek{k}}_{n-1}^{-1}\bar{\text{\textgreek{t}}}_{n-1}(u_{*},\bar{v})\, d\bar{v}.\nonumber 
\end{align}
In view of (\ref{eq:KappaBoundaryInduct}) and (\ref{eq:UpperBoundK_nSmallerRegion}),
we can bound 
\begin{equation}
\Big|\log\big(\bar{\text{\textgreek{k}}}_{n}(u_{*},v_{1}+u_{*})\big)\Big|\le M(1+e^{-\frac{1}{2}C_{0}^{\frac{3}{2}}M}).\label{eq:BoundKbarAxis}
\end{equation}
In view of (\ref{eq:DerivativeInUDirectionKappaInduction}), (\ref{eq:DerivativeInUtildeMInduction})
and (\ref{eq:ConservationTauBarInduction}) for $n-1$ in place of
$n$, as well as (\ref{eq:BoundForInductionPreviousStep}), (\ref{eq:SmallnessPreviousStep}),
(\ref{eq:LowerBoundRhoPreviousStep}) and (\ref{eq:BoundednessF_1,F_2,F_3}),
we can readily bound for all $(u_{*},v)\in\mathcal{W}$: 
\begin{align}
\int_{v_{1}+u_{*}}^{\min\{v_{2},v\}}\int_{0}^{u_{*}}\Bigg|\partial_{u}\Bigg\{ & \sqrt{-\frac{\Lambda}{3}}F_{1}\big(\tan\text{\textgreek{r}}_{n-1};\sqrt{-\frac{\Lambda}{3}}\tilde{m}_{n-1}\big)\text{\textgreek{k}}_{n-1}^{-1}\bar{\text{\textgreek{t}}}_{n-1}(\bar{u},\bar{v})\Bigg\}\Bigg|\, d\bar{u}d\bar{v}\le\label{eq:BoundDerivativeF_1}\\
 & \le C_{1}e^{10C_{0}M}(-\Lambda)\int_{v_{1}+u_{*}}^{\min\{v_{2},v\}}\int_{0}^{u_{*}}\Big(\frac{\bar{\text{\textgreek{t}}}_{n-1}}{\tan^{2}\text{\textgreek{r}}_{n-1}}+\frac{\text{\textgreek{t}}_{n-1}\bar{\text{\textgreek{t}}}_{n-1}}{\tan\text{\textgreek{r}}_{n-1}}\Big)\, d\bar{u}d\bar{v}\le\nonumber \\
 & \le C_{1}e^{20C_{0}M}(-\Lambda)\Bigg(\int_{v_{1}+u_{*}}^{\min\{v_{2},v\}}\big(\frac{1}{\text{\textgreek{r}}_{n-1}(u_{*},\bar{v})}-\frac{1}{\text{\textgreek{r}}_{n-1}(0,\bar{v})}\big)\bar{\text{\textgreek{t}}}_{n-1}(0,\bar{v})\, d\bar{v}+\nonumber \\
 & \hphantom{C_{1}e^{20C_{0}M}(-\Lambda)\Bigg(\int}+\Big(\sup_{0\le\bar{u}\le u_{*}}\int_{v_{1}+u_{*}}^{\min\{v_{2},v\}}\frac{\bar{\text{\textgreek{t}}}_{n-1}}{\text{\textgreek{r}}_{n-1}}(\bar{u},\bar{v})\, d\bar{v}\Big)\Big(\sup_{v_{1}+u_{*}\le\bar{v}\le\min\{v_{2},v\}}\int_{0}^{u_{*}}\frac{\text{\textgreek{t}}_{n-1}}{\text{\textgreek{r}}_{n-1}}(\bar{u},\bar{v})\, d\bar{u}\Big)\Bigg)\le\nonumber \\
 & \le C_{1}e^{20C_{0}M}\Bigg(e^{10C_{0}M}\int_{v_{1}+u_{*}}^{\min\{v_{2},v\}}\frac{u_{*}}{(r_{0}+\bar{v}-v_{1}-u_{*})(r_{0}+\bar{v}-v_{1})}\bar{\text{\textgreek{t}}}_{n-1}(0,\bar{v})\, d\bar{v}+e^{10C_{0}M}\cdot e^{-C_{0}^{\frac{3}{2}}M}M\Bigg)\le\nonumber \\
 & \le C_{1}e^{30C_{0}M}\Bigg(\int_{v_{1}+u_{*}}^{v_{1}+e^{C_{0}^{3/2}M}u_{*}}\frac{u_{*}}{(r_{0}+\bar{v}-v_{1}-u_{*})(r_{0}+\bar{v}-v_{1})}\bar{\text{\textgreek{t}}}_{n-1}(0,\bar{v})\, d\bar{v}+\nonumber \\
 & \hphantom{\le C_{1}e^{30C_{0}M}\Big(ss}+\int_{v_{1}+e^{C_{0}^{3/2}M}u_{*}}^{\min\{v_{2},v\}}\frac{u_{*}}{(r_{0}+\bar{v}-v_{1}-u_{*})(r_{0}+\bar{v}-v_{1})}\bar{\text{\textgreek{t}}}_{n-1}(0,\bar{v})\, d\bar{v}+e^{-C_{0}^{\frac{3}{2}}M}M\Bigg)\le\nonumber \\
 & \le C_{1}e^{30C_{0}M}\Bigg(\int_{v_{1}+u_{*}}^{v_{1}+e^{C_{0}^{3/2}M}u_{*}}\frac{\bar{\text{\textgreek{t}}}_{n-1}(0,\bar{v})}{(r_{0}+\bar{v}-v_{1})}\, d\bar{v}+\nonumber \\
 & \hphantom{\le C_{1}e^{30C_{0}M}\Big(ss}+e^{-C_{0}^{\frac{3}{2}}M}e^{10C_{0}M}\int_{v_{1}+e^{C_{0}^{3/2}}u_{*}}^{\min\{v_{2},v\}}\frac{\bar{\text{\textgreek{t}}}_{n-1}(0,\bar{v})}{r_{0}+\bar{v}-v_{1}}\, d\bar{v}+e^{-C_{0}^{\frac{3}{2}}M}M\Bigg).\nonumber 
\end{align}
for some absolute constant $C_{1}>0$.%
\footnote{Note that in passing from the second to the third line of (\ref{eq:BoundDerivativeF_1}),
we have used the fact that $\partial_{u}\bar{\text{\textgreek{t}}}=0$
and $\partial_{v}\text{\textgreek{t}}=0$.%
} Thus, in view of (\ref{eq:SmallnessInL1OfT}) and (\ref{eq:U0UpperBound}),
(\ref{eq:BoundDerivativeF_1}) implies (provided $C_{0}\gg C_{1}$)
that 
\begin{equation}
\Bigg|\int_{v_{1}+u_{*}}^{\min\{v_{2},v\}}\int_{0}^{u_{*}}\partial_{u}\Bigg\{\sqrt{-\frac{\Lambda}{3}}F_{1}\big(\tan\text{\textgreek{r}}_{n-1};\sqrt{-\frac{\Lambda}{3}}\tilde{m}_{n-1}\big)\text{\textgreek{k}}_{n-1}^{-1}\bar{\text{\textgreek{t}}}_{n-1}(\bar{u},\bar{v})\Bigg\}\, d\bar{u}d\bar{v}\Bigg|\le10e^{-\frac{1}{2}C_{0}^{\frac{3}{2}}M}M.\label{eq:BoundDerivativeF_1Integrated}
\end{equation}
In view of (\ref{eq:InitialDataSequenceInduct}), (\ref{eq:UpperBoundInitialData})
and (\ref{eq:DerivativeInVDirectionKappaBarInduction}) for $n-1$
in place of $n$ and $u=0$, we can estimate 
\begin{equation}
\Big|\int_{v_{1}+u_{*}}^{\min\{v_{2},v\}}\sqrt{-\frac{\Lambda}{3}}F_{1}\big(\tan\text{\textgreek{r}}_{n-1};\sqrt{-\frac{\Lambda}{3}}\tilde{m}_{n-1}\big)\text{\textgreek{k}}_{n-1}^{-1}\bar{\text{\textgreek{t}}}_{n-1}(0,\bar{v})\, d\bar{v}\Big|\le2\sup_{v_{1}\le v\le v_{2}}\big|\log\big(\bar{\text{\textgreek{k}}}(0,v)\big)\big|\le2M.\label{eq:BoundForInitialDataForDifference}
\end{equation}
Furthermore, inequality 
\begin{equation}
v-\min\{v_{2},v\}\le u_{*}\le u_{0}
\end{equation}
combined with (\ref{eq:SmallnessPreviousStep}), (\ref{eq:BoundForInductionPreviousStep}),
(\ref{eq:EqualityForEnergyMomentum}) and (\ref{eq:BoundednessF_1,F_2,F_3})
yields 
\begin{equation}
\Bigg|\int_{\min\{v_{2},v\}}^{v}\sqrt{-\frac{\Lambda}{3}}F_{1}\big(\tan\text{\textgreek{r}}_{n-1};\sqrt{-\frac{\Lambda}{3}}\tilde{m}_{n-1}\big)\text{\textgreek{k}}_{n-1}^{-1}\bar{\text{\textgreek{t}}}_{n-1}(u_{*},\bar{v})\, d\bar{v}\Bigg|\le10e^{-\frac{1}{2}C_{0}^{\frac{3}{2}}M}M.\label{eq:BoundLastTermOnIntegrandForKappaBar}
\end{equation}
Combining (\ref{eq:BoundKbarAxis}), (\ref{eq:BoundDerivativeF_1Integrated}),
(\ref{eq:BoundForInitialDataForDifference}) and (\ref{eq:BoundLastTermOnIntegrandForKappaBar}),
the relation (\ref{eq:FirstRelationKappaIntermediateRegion-1}) readily
yields for any $(u,v)\in\mathcal{W}$: 
\begin{equation}
\Big|\log\big(\bar{\text{\textgreek{k}}}_{n}(u,v)\big)\Big|\le3M(1+10e^{-\frac{1}{2}C_{0}^{\frac{3}{2}}M}).\label{eq:BoundKappaBarEveywhereInW}
\end{equation}

Subtracting from (\ref{eq:FirstRelationKappaIntermediateRegion-1})
the same relation for $n-1$ in place of $n$ and using (similarly
as before) (\ref{eq:SmallnessInL1OfT}), (\ref{eq:U0UpperBound}),
(\ref{eq:BoundForInductionPreviousStep}), (\ref{eq:SmallnessPreviousStep}),
(\ref{eq:BoundednessF_1,F_2,F_3}) and (\ref{eq:DifferenceKappaIntermediateRegion})
(in some instances for $n-1$ in place of $n$), we obtain for any
$(u,v)\in\mathcal{W}$: 
\begin{equation}
\Big|\log\big(\bar{\text{\textgreek{k}}}_{n}(u,v)\big)-\log\big(\bar{\text{\textgreek{k}}}_{n-1}(u,v)\big)\Big|\le C_{0}Me^{-C_{0}^{\frac{3}{2}}M}\mathfrak{D}_{n-1}.\label{eq:BoundKappaBarEveywhereInWDifference}
\end{equation}

Integrating equation (\ref{eq:DerivativeInUDirectionKappaInduction})
along $\{v_{*}-v_{2}\le u\le u_{0}\}\cap\{v=v_{*}\}$ for any $v_{2}\le v_{*}\le v_{2}+u_{0}$,
we obtain for any point $(u,v)\in\mathcal{W}\cap\{v_{2}\le v\le v_{2}+u\}$:
\begin{equation}
\log\big(\text{\textgreek{k}}_{n}(u,v)\big)=\log\big(\text{\textgreek{k}}_{n}(v-v_{2},v)\big)-\int_{v-v_{2}}^{u}\sqrt{-\frac{\Lambda}{3}}F_{1}\big(\tan\text{\textgreek{r}}_{n-1};\sqrt{-\frac{\Lambda}{3}}\tilde{m}_{n-1}\big)\bar{\text{\textgreek{k}}}_{n-1}^{-1}\text{\textgreek{t}}_{n-1}(\bar{u},v)\, d\bar{u}.\label{eq:FirstRelationKappaAwayRegion}
\end{equation}
Thus, in view of (\ref{eq:BoundKappaBarEveywhereInW}) for $n-1$
in place of $n$, (\ref{eq:U0UpperBound}), (\ref{eq:BoundForInductionPreviousStep}),
(\ref{eq:SmallnessPreviousStep}) and (\ref{eq:BoundednessF_1,F_2,F_3}),
using also (\ref{eq:KappaBoundaryInduct}) and (\ref{eq:BoundKappaBarEveywhereInW})
for the first term of the right hand side of (\ref{eq:FirstRelationKappaAwayRegion}),
the relation (\ref{eq:FirstRelationKappaAwayRegion}) yields for any
point $(u,v)\in\mathcal{W}\cap\{v_{2}\le v\le v_{2}+u\}$: 
\begin{equation}
\big|\log\big(\text{\textgreek{k}}_{n}(u,v)\big)\big|\le3M(1+10e^{-\frac{1}{2}C_{0}^{\frac{3}{2}}M}).\label{eq:UpperBoundK_nAwayRegion}
\end{equation}
 Furthermore, subtracting from equation (\ref{eq:FirstRelationKappaAwayRegion})
the same equation for $n-1$ in place of $n$, and using (\ref{eq:EqualityForEnergyMomentum}),
(\ref{eq:BoundForInductionPreviousStep}), (\ref{eq:SmallnessPreviousStep}),
(\ref{eq:BoundednessF_1,F_2,F_3}), (\ref{eq:LowerBoundRhoPreviousStep}),
(\ref{eq:KappaBoundaryInduct}) and (\ref{eq:BoundKappaBarEveywhereInWDifference})
(in some instances for $n-1$ in place of $n$), we infer for any
point $(u,v)\in\mathcal{W}\cap\{v_{2}\le v\le v_{2}+u\}$: 
\begin{equation}
\big|\log\big(\text{\textgreek{k}}_{n}(u,v)\big)-\log\big(\text{\textgreek{k}}_{n-1}(u,v)\big)\big|\le C_{0}Me^{-C_{0}^{\frac{3}{2}}M}\big(\mathfrak{D}_{n-1}+\mathfrak{D}_{n-2}\big).\label{eq:DifferenceKappaAwayRegion}
\end{equation}
Combining (\ref{eq:UpperBoundK_nSmallerRegion}), (\ref{eq:DifferenceKappaIntermediateRegion}),
(\ref{eq:UpperBoundK_nAwayRegion}) and (\ref{eq:DifferenceKappaAwayRegion}),
we thus deduce that, for any $(u,v)\in\mathcal{W}$:
\begin{equation}
\big|\log\big(\text{\textgreek{k}}_{n}(u,v)\big)\big|\le3M(1+10e^{-\frac{1}{2}C_{0}^{\frac{3}{2}}M})\label{eq:UpperBoundK_nEverywhere}
\end{equation}
 and 
\begin{equation}
\big|\log\big(\text{\textgreek{k}}_{n}(u,v)\big)-\log\big(\text{\textgreek{k}}_{n-1}(u,v)\big)\big|\le C_{0}Me^{-C_{0}^{\frac{3}{2}}M}\big(\mathfrak{D}_{n-1}+\mathfrak{D}_{n-2}\big).\label{eq:DifferenceKappaEverywhere}
\end{equation}

Setting 
\begin{equation}
\varrho_{n}=\partial_{v}\text{\textgreek{r}}_{n}\label{eq:Varrho}
\end{equation}
and 
\begin{equation}
\bar{\varrho}_{n}=-\partial_{u}\text{\textgreek{r}}_{n},\label{eq:VarRhoBar}
\end{equation}
equation (\ref{eq:EquationRForRhoInduction}) is equivalent to the
system 
\begin{align}
\partial_{u}\varrho_{n}= & (-\frac{\Lambda}{3})F_{2}\big(\tan\text{\textgreek{r}}_{n-1};\sqrt{-\frac{\Lambda}{3}}\tilde{m}_{n-1}\big)\text{\textgreek{k}}_{n-1}\bar{\text{\textgreek{k}}}_{n-1},\label{eq:EquationVarRhoInduction}\\
\partial_{v}\bar{\varrho}_{n}= & -(-\frac{\Lambda}{3})F_{2}\big(\tan\text{\textgreek{r}}_{n-1};\sqrt{-\frac{\Lambda}{3}}\tilde{m}_{n-1}\big)\text{\textgreek{k}}_{n-1}\bar{\text{\textgreek{k}}}_{n-1},\label{eq:EquationVarRhoBarInduction}
\end{align}
while the initial data (\ref{eq:InitialDataSequenceInduct}) and the
boundary conditions (\ref{eq:RhoBoundaryInduct}) gives rise to the
conditions 
\begin{equation}
\sqrt{-\frac{12}{\Lambda}}\varrho_{n}(0,v)=\frac{2\partial_{v}r_{/}}{1-\frac{1}{3}\Lambda r_{/}^{2}}(v)\label{eq:InitialConditionVarRho}
\end{equation}
and 
\begin{equation}
\frac{\varrho_{n}}{\bar{\text{\ensuremath{\varrho}}}_{n}}\Big|_{\text{\textgreek{g}}_{0}}=1,\textnormal{}\frac{\varrho_{n}}{\bar{\text{\ensuremath{\varrho}}}_{n}}\Big|_{\mathcal{I}_{u_{0}}}=1.\label{eq:BoundaryVarRho}
\end{equation}
Notice the similarity of (\ref{eq:EquationVarRhoInduction}), (\ref{eq:EquationVarRhoBarInduction})
and (\ref{eq:BoundaryVarRho}) with (\ref{eq:DerivativeInUDirectionKappaInduction}),
(\ref{eq:DerivativeInVDirectionKappaBarInduction}) and (\ref{eq:KappaBoundaryInduct}),
respectively. In particular, by repeating the same arguments that
led to (\ref{eq:BoundKappaBarEveywhereInW}), (\ref{eq:BoundKappaBarEveywhereInWDifference}),
(\ref{eq:UpperBoundK_nEverywhere}) and (\ref{eq:DifferenceKappaEverywhere})
(treating the regions $\mathcal{W}\cap\{v\le v_{2}\}$ and $\mathcal{W}\cap\{v_{2}\le v\le v_{2}+u\}$
seperately, as was done in the case of $\text{\textgreek{k}}_{n}$),
we readily obtain for any $(u,v)\in\mathcal{W}$ (provided $\text{\textgreek{d}}_{M}$
is small enough in terms of $\text{\textgreek{r}}_{0},M$): 
\begin{equation}
\big|\log\big(\sqrt{-\frac{12}{\Lambda}}\varrho_{n}(u,v)\big)\big|\le3M(1+10e^{-\frac{1}{2}C_{0}^{\frac{3}{2}}M}),\label{eq:BoundVarRho}
\end{equation}
\begin{equation}
\big|\log\big(\sqrt{-\frac{12}{\Lambda}}\bar{\varrho}_{n}(u,v)\big)\big|\le3M(1+10e^{-\frac{1}{2}C_{0}^{\frac{3}{2}}M}),\label{eq:BoundVarRhoBar}
\end{equation}
\begin{equation}
\big|\log\big(\sqrt{-\frac{12}{\Lambda}}\varrho_{n}(u,v)\big)-\log\big(\sqrt{-\frac{12}{\Lambda}}\varrho_{n-1}(u,v)\big)\big|\le C_{0}Me^{-C_{0}^{\frac{3}{2}}M}\big(\mathfrak{D}_{n-1}+\mathfrak{D}_{n-2}\big),\label{eq:BoundDifferenceVarRho}
\end{equation}
and 
\begin{equation}
\big|\log\big(\sqrt{-\frac{12}{\Lambda}}\bar{\varrho}_{n}(u,v)\big)-\log\big(\sqrt{-\frac{12}{\Lambda}}\bar{\varrho}_{n-1}(u,v)\big)\big|\le C_{0}Me^{-C_{0}^{\frac{3}{2}}M}\big(\mathfrak{D}_{n-1}+\mathfrak{D}_{n-2}\big).\label{eq:BoundDifferenceVarRhoBar}
\end{equation}

Integrating (\ref{eq:VarRhoBar}) in $u$ and using (\ref{eq:U0UpperBound}),
(\ref{eq:InitialDataSequenceInduct}), (\ref{eq:BoundVarRhoBar})
and (\ref{eq:BoundDifferenceVarRhoBar}), we readily obtain that,
for any $(u,v)\in\mathcal{W}\cap\{v\le v_{2}\}$: 
\begin{equation}
|\text{\textgreek{r}}_{n}(u,v)-\text{\textgreek{r}}_{n}(0,v)|\le4Me^{-C_{0}^{\frac{3}{2}}M}\label{eq:FirstBoundRhoNearRegion}
\end{equation}
and 
\begin{equation}
|\text{\textgreek{r}}_{n}(u,v)-\text{\textgreek{r}}_{n-1}(u,v)|\le C_{0}Me^{-C_{0}^{\frac{3}{2}}M}\big(\mathfrak{D}_{n-1}+\mathfrak{D}_{n-2}\big).\label{eq:DifferenceRhoNearRegion}
\end{equation}
Integrating (\ref{eq:Varrho}) in $v$ in the region $\mathcal{W}\cap\{v_{2}\le v\le v_{2}+u\}$
and using (\ref{eq:U0UpperBound}), (\ref{eq:FirstBoundRhoNearRegion})
and (\ref{eq:DifferenceRhoNearRegion}) (for $v=v_{2}$), we readily
obtain for any $(u,v)\in\mathcal{W}\cap\{v_{2}\le v\le v_{2}+u\}$:
\begin{equation}
|\text{\textgreek{r}}_{n}(u,v)-\text{\textgreek{r}}_{n}(0,v_{2})|\le4Me^{-C_{0}^{\frac{3}{2}}M}\label{eq:FirstBoundRhoAway}
\end{equation}
and 
\begin{equation}
|\text{\textgreek{r}}_{n}(u,v)-\text{\textgreek{r}}_{n-1}(u,v)|\le C_{0}Me^{-C_{0}^{\frac{3}{2}}M}\big(\mathfrak{D}_{n-1}+\mathfrak{D}_{n-2}\big).\label{eq:DifferenceRhoAway}
\end{equation}

By integrating equation (\ref{eq:DerivativeInUtildeMInduction}) in
$u$ and using (\ref{eq:InitialDataSequenceInduct}), (\ref{eq:ConservationTildeMAxisInduct}),
(\ref{eq:EqualityForEnergyMomentum}), (\ref{eq:UpperBoundInitialData}),
(\ref{eq:BoundForInductionPreviousStep}), (\ref{eq:UpperBoundK_nEverywhere})
and (\ref{eq:DifferenceKappaEverywhere}), we readily obtain that
\begin{equation}
\sqrt{-\Lambda}|\tilde{m}_{n}|\le M(1+10e^{-C_{0}^{\frac{3}{2}}M})\label{eq:BoundInMass}
\end{equation}
and 
\begin{equation}
\sqrt{-\Lambda}|\tilde{m}_{n}-\tilde{m}_{n-1}|\le C_{0}Me^{-C_{0}^{\frac{3}{2}}M}\big(\mathfrak{D}_{n-1}+\mathfrak{D}_{n-2}\big).\label{eq:DifferenceInMass}
\end{equation}
From (\ref{eq:InitialDataSequenceInduct}), (\ref{eq:UpperBoundInitialData}),
(\ref{eq:FirstBoundRhoNearRegion}), (\ref{eq:FirstBoundRhoAway})
and (\ref{eq:BoundInMass}), we also obtain for all $(u,v)\in\mathcal{W}$:
\begin{equation}
\Big|\log\Big(\frac{1-\frac{2\sqrt{-\Lambda}\tilde{m}_{n-1}}{\tan\text{\textgreek{r}}_{n}}+\frac{1}{3}\tan^{2}\text{\textgreek{r}}_{n}}{1+\frac{1}{3}\tan^{2}\text{\textgreek{r}}_{n}}\Big)\Big|\le M(1+10e^{-C_{0}^{\frac{3}{2}}M})\label{eq:BoundNonTrappedInductionStep}
\end{equation}
Finally, from (\ref{eq:EqualityEnergyMomentum}), (\ref{eq:UpperBoundInitialData}),
(\ref{eq:SmallnessInL1OfT}), (\ref{eq:RhoBoundaryInduct}), (\ref{eq:BoundVarRho})
and (\ref{eq:BoundVarRhoBar}) we can readily estimate 
\begin{equation}
\sup_{0\le\bar{u}\le u_{0}}\int_{\mathcal{W}\cap\{u=\bar{u}\}}\sqrt{-\Lambda}\frac{\bar{\text{\textgreek{t}}}_{n}}{\text{\textgreek{r}}_{n}}\, dv+\sup_{0\le\bar{u}\le u_{0}}\sup_{v\in[v_{1}+\bar{u},v_{2}+\bar{u}]}\int_{\max\{v-u_{0},v_{1}+\bar{u}\}}^{\min\{v+u_{0},v_{2}+\bar{u}\}}\frac{\bar{\text{\textgreek{t}}}_{n}}{|\bar{v}-v|+r_{0}}\, d\bar{v}\le10M(1+10e^{-C_{0}^{\frac{3}{2}}M}).\label{eq:UpperBoundTauBarIntegral}
\end{equation}

The bound (\ref{eq:BoundForInductionNextStep}) readily follows from
(\ref{eq:BoundKappaBarEveywhereInW}), (\ref{eq:UpperBoundK_nEverywhere}),
(\ref{eq:BoundVarRho}), (\ref{eq:BoundVarRhoBar}), (\ref{eq:BoundInMass}),
(\ref{eq:BoundNonTrappedInductionStep}) and (\ref{eq:UpperBoundTauBarIntegral}),
while the bound (\ref{eq:DifferenceRate}) follows from (\ref{eq:EqualityEnergyMomentum}),
(\ref{eq:BoundKappaBarEveywhereInWDifference}), (\ref{eq:DifferenceKappaEverywhere}),
(\ref{eq:BoundDifferenceVarRho}), (\ref{eq:BoundDifferenceVarRhoBar}),
(\ref{eq:DifferenceRhoNearRegion}), (\ref{eq:DifferenceRhoAway})
and (\ref{eq:DifferenceInMass}). 

The estimate (\ref{eq:SmallnessNextStep}) follows from (\ref{eq:SmallnessInL1OfT}),
(\ref{eq:U0UpperBound}), (\ref{eq:EqualityEnergyMomentum}), (\ref{eq:RhoBoundaryInduct}),
(\ref{eq:BoundVarRho}) and (\ref{eq:BoundVarRhoBar}). Finally, the
bound (\ref{eq:UpperBoundForTauTauBarNextStep}) follows from (\ref{eq:BoundForEnergyMomentum}),
while the lower bound (\ref{eq:LowerBoundRhoNextStep}) follows immediately
from the fact that $\varrho_{n}>0$ (in view of (\ref{eq:BoundVarRhoBar}))
and the boundary condition (\ref{eq:RhoBoundaryInduct}). 

\medskip{}

\noindent \emph{Step 2. }The second step of the proof will consist
of showing that the smooth solution constructed in the previous step
is actually unique. Let $(\text{\textgreek{r}}_{1},\text{\textgreek{k}}_{1},\bar{\text{\textgreek{k}}}_{1},\text{\textgreek{t}}_{1},\bar{\text{\textgreek{t}}}_{1},\tilde{m}_{1})$
and $(\text{\textgreek{r}}_{2},\text{\textgreek{k}}_{2},\bar{\text{\textgreek{k}}}_{2},\text{\textgreek{t}}_{2},\bar{\text{\textgreek{t}}}_{2},\tilde{m}_{2})$
be two $C^{0}$ solutions of the system (\ref{eq:DerivativeInUDirectionKappaRenormalised})--(\ref{eq:ConservationTau})
on $\mathcal{W}$ satisfying (\ref{eq:InitialDataRecastRight})--(\ref{eq:ConservationTildeMAxis}),
such that $(\text{\textgreek{r}}_{1},\text{\textgreek{k}}_{1},\bar{\text{\textgreek{k}}}_{1},\text{\textgreek{t}}_{1},\bar{\text{\textgreek{t}}}_{1},\tilde{m}_{1})$
is the solution constructed in the previous step. In particular, $(\text{\textgreek{r}}_{1},\text{\textgreek{k}}_{1},\bar{\text{\textgreek{k}}}_{1},\text{\textgreek{t}}_{1},\bar{\text{\textgreek{t}}}_{1},\tilde{m}_{1})$
satisfies the bounds (\ref{eq:BoundToShowInitialValueFirstStep})--(\ref{eq:UpperBoundForTauTauBar}),
and the two solutions satisfy the same initial data, i.\,e.
\begin{equation}
(\text{\textgreek{r}}_{1},\text{\textgreek{k}}_{1},\tilde{m}_{1},\bar{\text{\textgreek{t}}}_{1})|_{\{0\}\times[v_{1},v_{2}]}=(\text{\textgreek{r}}_{2},\text{\textgreek{k}}_{2},\tilde{m}_{2},\bar{\text{\textgreek{t}}}_{2})|_{\{0\}\times[v_{1},v_{2}]}.\label{eq:SameInitialData}
\end{equation}

Let $\mathscr{B}$ be the set of all subsets $\mathcal{B}$ of $\mathcal{W}$
satisfying the following properties: 

\begin{enumerate}

\item The closure $clos(\mathcal{V})$ of $\mathcal{V}$ (in the
ambient topology of $\mathbb{R}^{2}$) contains $\{0\}\times[v_{1},v_{2}]$,

\item For any $(u_{*},v_{*})\in\mathcal{V}$, the line segments $\{u=u_{*}\}\cap\{v_{1}+u_{*}\le v\le v_{2}+u_{*}\}$
and $\{0\le u\le u_{*}\}\cap\{v=v_{*}\}$ are contained in $\mathcal{V}$,

\end{enumerate}

Let $\mathcal{W}_{0}\in\mathscr{B}$ be a subset of $\mathcal{W}$,
which is closed in the induced topology of $\mathcal{W}\subset\mathbb{R}^{2}$,
such that 
\begin{equation}
(\text{\textgreek{r}}_{1},\text{\textgreek{k}}_{1},\bar{\text{\textgreek{k}}}_{1},\text{\textgreek{t}}_{1},\bar{\text{\textgreek{t}}}_{1},\tilde{m}_{1})|_{\mathcal{W}_{0}}=(\text{\textgreek{r}}_{2},\text{\textgreek{k}}_{2},\bar{\text{\textgreek{k}}}_{2},\text{\textgreek{t}}_{2},\bar{\text{\textgreek{t}}}_{2},\tilde{m}_{2})|_{\mathcal{W}_{0}}.\label{eq:EqualitySolutions}
\end{equation}
Since (\ref{eq:EqualitySolutions}) implies automatically that $(\text{\textgreek{r}}_{2},\text{\textgreek{k}}_{2},\bar{\text{\textgreek{k}}}_{2},\text{\textgreek{t}}_{2},\bar{\text{\textgreek{t}}}_{2},\tilde{m}_{2})$
satisfies the bounds (\ref{eq:BoundToShowInitialValueFirstStep})--(\ref{eq:UpperBoundForTauTauBar})
on $\mathcal{W}_{0}$, we can always find an set $\mathcal{V}\supseteq\mathcal{W}_{0}$
with $\mathcal{V}\in\mathscr{B}$, such that $\mathcal{V}$ is open
in the induced topology of $\mathcal{W}\subset\mathbb{R}^{2}$ and
$(\text{\textgreek{r}}_{2},\text{\textgreek{k}}_{2},\bar{\text{\textgreek{k}}}_{2},\text{\textgreek{t}}_{2},\bar{\text{\textgreek{t}}}_{2},\tilde{m}_{2})$
satisfies (\ref{eq:BoundToShowInitialValueFirstStep})--(\ref{eq:UpperBoundForTauTauBar})
on $\mathcal{V}$. Therefore, it suffices to show that (\ref{eq:EqualitySolutions})
holds on any such set $\mathcal{V}$, since this will imply that $\mathcal{W}_{0}\subseteq\mathcal{V}$
and, therefore, $\mathcal{W}_{0}=\mathcal{W}$.

Let $\mathcal{V}\in\mathscr{B}$ such that $(\text{\textgreek{r}}_{2},\text{\textgreek{k}}_{2},\bar{\text{\textgreek{k}}}_{2},\text{\textgreek{t}}_{2},\bar{\text{\textgreek{t}}}_{2},\tilde{m}_{2})$
satisfies the bounds (\ref{eq:BoundToShowInitialValueFirstStep})--(\ref{eq:UpperBoundForTauTauBar}).
Subtracting equations (\ref{eq:DerivativeInUDirectionKappaRenormalised})--(\ref{eq:ConservationTau})
for $(\text{\textgreek{r}}_{2},\text{\textgreek{k}}_{2},\bar{\text{\textgreek{k}}}_{2},\text{\textgreek{t}}_{2},\bar{\text{\textgreek{t}}}_{2},\tilde{m}_{2})$
from the same equations for $(\text{\textgreek{r}}_{1},\text{\textgreek{k}}_{1},\bar{\text{\textgreek{k}}}_{1},\text{\textgreek{t}}_{1},\bar{\text{\textgreek{t}}}_{1},\tilde{m}_{1})$,
we obtain: 
\begin{align}
\partial_{u}\big(\log(\text{\textgreek{k}}_{2})-\log(\text{\textgreek{k}}_{1})\big)= & -\sqrt{-\frac{\Lambda}{3}}\Big(F_{1}\big(\tan\text{\textgreek{r}}_{2};\sqrt{-\frac{\Lambda}{3}}\tilde{m}_{2}\big)\bar{\text{\textgreek{k}}}_{2}^{-1}\text{\textgreek{t}}_{2}-F_{1}\big(\tan\text{\textgreek{r}}_{1};\sqrt{-\frac{\Lambda}{3}}\tilde{m}_{1}\big)\bar{\text{\textgreek{k}}}_{1}^{-1}\text{\textgreek{t}}_{1}\Big),\label{eq:DerivativeInUDirectionKappaRenormalised-1}\\
\partial_{v}\big(\log(\bar{\text{\textgreek{k}}}_{2})-\log(\bar{\text{\textgreek{k}}}_{1})\big)= & \sqrt{-\frac{\Lambda}{3}}\Big(F_{1}\big(\tan\text{\textgreek{r}}_{2};\sqrt{-\frac{\Lambda}{3}}\tilde{m}_{2}\big)\bar{\text{\textgreek{k}}}_{2}^{-1}\text{\textgreek{t}}_{2}-F_{1}\big(\tan\text{\textgreek{r}}_{1};\sqrt{-\frac{\Lambda}{3}}\tilde{m}_{1}\big)\bar{\text{\textgreek{k}}}_{1}^{-1}\text{\textgreek{t}}_{1}\Big),\label{eq:DerivativeInVDirectionKappaBarRenormalised-1}\\
\partial_{u}\partial_{v}\big(\text{\textgreek{r}}_{2}-\text{\textgreek{r}}_{1}\big)= & (-\frac{\Lambda}{3})\Big(F_{2}\big(\tan\text{\textgreek{r}}_{2};\sqrt{-\frac{\Lambda}{3}}\tilde{m}_{2}\big)\text{\textgreek{k}}_{2}\bar{\text{\textgreek{k}}}_{2}-F_{2}\big(\tan\text{\textgreek{r}}_{1};\sqrt{-\frac{\Lambda}{3}}\tilde{m}_{1}\big)\text{\textgreek{k}}_{1}\bar{\text{\textgreek{k}}}_{1}\Big),\label{eq:EquationRForRhoRenormalised-1}\\
\partial_{u}\big(\tilde{m}_{2}-\tilde{m}_{1}\big)= & -4\pi\big(\bar{\text{\textgreek{k}}}_{2}^{-1}\text{\textgreek{t}}_{2}-\bar{\text{\textgreek{k}}}_{1}^{-1}\text{\textgreek{t}}_{1}\big),\label{eq:DerivativeInUtildeMRenormalised-1}\\
\partial_{u}\big(\bar{\text{\textgreek{t}}}_{2}-\bar{\text{\textgreek{t}}}_{1}\big)= & 0,\label{eq:ConservationTauBar-1}\\
\partial_{v}\big(\text{\textgreek{t}}_{2}-\text{\textgreek{t}}_{1}\big)= & 0.\label{eq:ConservationTau-1}
\end{align}
 Integrating equations (\ref{eq:DerivativeInUDirectionKappaRenormalised-1}),
(\ref{eq:EquationRForRhoRenormalised-1}), (\ref{eq:DerivativeInUtildeMRenormalised-1})
and (\ref{eq:ConservationTauBar-1}) in $u$ and equations (\ref{eq:DerivativeInVDirectionKappaBarRenormalised-1}),
(\ref{eq:EquationRForRhoRenormalised-1}) and (\ref{eq:ConservationTau-1})
in $v$, using also the initial condition (\ref{eq:SameInitialData}),
the boundary conditions (\ref{eq:RhoBoundary})--(\ref{eq:ConservationTildeMAxis})
(noticing that (\ref{eq:RhoBoundary}) implies that $\partial_{v}\text{\textgreek{r}}_{i}=-\partial_{u}\text{\textgreek{r}}_{i}$
on $\text{\textgreek{g}}_{0;u_{0}}\cap clos(\mathcal{V})$ and $\mathcal{I}_{u_{0}}\cap clos(\mathcal{V})$,
$i=1,2$) and the bounds (\ref{eq:BoundToShowInitialValueFirstStep})--(\ref{eq:UpperBoundForTauTauBar})
and (\ref{eq:BoundednessF_1,F_2,F_3}), we obtain for any $(u,v)\in\mathcal{V}$
\begin{align}
\big|\log\big(\text{\textgreek{k}}_{2}(u,v)\big)-\log\big(\text{\textgreek{k}}_{1}(u,v)\big)\big| & \le C_{M,\text{\textgreek{r}}_{0},T_{*}}\sqrt{-\Lambda}\Big\{\int_{\max\{0,v-v_{2}\}}^{u}\mathfrak{D}(\bar{u},v)\, d\bar{u}+\text{\textgreek{q}}_{v>v_{2}}(v)\int_{v-v_{2}+v_{1}}^{v}\mathfrak{D}\big(v-v_{2},\bar{v}\big)\, d\bar{v}+\label{eq:FirstDifferenceForUniqueness}\\
 & \hphantom{\le C_{M,\text{\textgreek{r}}_{0}}\sqrt{-\Lambda}\Bigg\{++}++\text{\textgreek{q}}_{v>v_{2}}(v)\int_{0}^{v-v_{2}}\mathfrak{D}(\bar{u},v_{1}+v-v_{2})\, d\bar{u}\Big\},\nonumber \\
\big|\log\big(\bar{\text{\textgreek{k}}}_{2}(u,v)\big)-\log\big(\bar{\text{\textgreek{k}}}_{1}(u,v)\big)\big| & \le C_{M,\text{\textgreek{r}}_{0},T_{*}}\sqrt{-\Lambda}\Big\{\int_{v_{1}+u}^{v}\mathfrak{D}(u,\bar{v})\, d\bar{v}+\int_{0}^{u}\mathfrak{D}(\bar{u},v_{1}+u)\, d\bar{u}\Big\},\\
\big|\partial_{v}\text{\textgreek{r}}_{2}(u,v)-\partial_{v}\text{\textgreek{r}}_{1}(u,v)\big| & \le C_{M,\text{\textgreek{r}}_{0},T_{*}}(-\Lambda)\Big\{\int_{\max\{0,v-v_{2}\}}^{u}\mathfrak{D}(\bar{u},v)\, d\bar{u}+\text{\textgreek{q}}_{v>v_{2}}(v)\int_{v-v_{2}+v_{1}}^{v}\mathfrak{D}\big(v-v_{2},\bar{v}\big)\, d\bar{v}+\le\\
 & \hphantom{\le C_{M,\text{\textgreek{r}}_{0},T_{*}}\sqrt{-\Lambda}\Bigg\{++}++\text{\textgreek{q}}_{v>v_{2}}(v)\int_{0}^{v-v_{2}}\mathfrak{D}(\bar{u},v_{1}+v-v_{2})\, d\bar{u}\Big\},\nonumber \\
\big|\partial_{u}\text{\textgreek{r}}_{2}(u,v)-\partial_{u}\text{\textgreek{r}}_{1}(u,v)\big| & \le C_{M,\text{\textgreek{r}}_{0},T_{*}}(-\Lambda)\Big\{\int_{v_{1}+u}^{v}\mathfrak{D}(u,\bar{v})\, d\bar{v}+\int_{0}^{u}\mathfrak{D}(\bar{u},v_{1}+u)\, d\bar{u}\Big\},\\
\big|\tilde{m}_{2}(u,v)-\tilde{m}_{1}(u,v)\big| & \le C_{M,\text{\textgreek{r}}_{0},T_{*}}\int_{\max\{0,v-v_{2}\}}^{u}\mathfrak{D}(\bar{u},v)\, d\bar{u},\\
\big|\bar{\text{\textgreek{t}}}_{2}(u,v)-\bar{\text{\textgreek{t}}}_{1}(u,v)\big| & =0,\\
\big|\text{\textgreek{t}}_{2}(u,v)-\text{\textgreek{t}}_{1}(u,v)\big| & =0,\label{eq:LastDifferenceForUniqueness}
\end{align}
where 
\begin{align}
\mathfrak{D}(u,v)\doteq\big|\log & \big(\text{\textgreek{k}}_{2}(u,v)\big)-\log\big(\text{\textgreek{k}}_{1}(u,v)\big)\big|+\big|\log\big(\bar{\text{\textgreek{k}}}_{2}(u,v)\big)-\log\big(\bar{\text{\textgreek{k}}}_{1}(u,v)\big)\big|+\\
 & +\big|\text{\textgreek{r}}_{2}(u,v)-\text{\textgreek{r}}_{1}(u,v)\big|+\sqrt{-\Lambda}\big|\tilde{m}_{2}(u,v)-\tilde{m}_{1}(u,v)\big|+\nonumber \\
 & +\big|\bar{\text{\textgreek{t}}}_{2}(u,v)-\bar{\text{\textgreek{t}}}_{1}(u,v)\big|+\big|\text{\textgreek{t}}_{2}(u,v)-\text{\textgreek{t}}_{1}(u,v)\big|,\nonumber 
\end{align}
\begin{equation}
\text{\textgreek{q}}_{v>v_{2}}(v)=\begin{cases}
1, & \mbox{ if }v>v_{2},\\
0, & \mbox{ if }v\le v_{2}
\end{cases}
\end{equation}
and $C_{M,\text{\textgreek{r}}_{0},T_{*}}$ depends on $M,\text{\textgreek{r}}_{0}$
and 
\[
T_{*}\doteq\sup_{v_{1}<\bar{v}<v_{2}}r_{/}^{2}(T_{vv})_{/}(\bar{v}).
\]
 Inequalities (\ref{eq:FirstDifferenceForUniqueness})--(\ref{eq:LastDifferenceForUniqueness})
now readily yield (\ref{eq:EqualitySolutions}). 
\end{proof}
Our next result is a well-posedness result for the initial data introduced
by Definition \ref{def:TypeIII}.
\begin{prop}
\label{Prop:LocalExistenceTypeIII}Let $C_{0}\gg1$ be a (large) absolute
constant. For any $u_{1}<u_{2}$, $v_{1}<v_{2}$ and $r_{0}>0$, let
$r_{\backslash}:[u_{1},u_{2}]\rightarrow[r_{0},+\infty)$, $\text{\textgreek{W}}_{\backslash}:[u_{1},u_{2}]\rightarrow(0,+\infty)$,
$\bar{f}_{in\backslash},\bar{f}_{out\backslash}:[u_{1},u_{2}]\times(0,+\infty)\rightarrow[0,+\infty)$,
$r_{/}:[v_{1},v_{2}]\rightarrow(r_{0},+\infty)$, $\text{\textgreek{W}}_{/}:[v_{1},v_{2}]\rightarrow(0,+\infty)$
and $\bar{f}_{in/},\bar{f}_{out/}:[v_{1},v_{2}]\times(0,+\infty)\rightarrow[0,+\infty)$
be smooth functions, so that $(r_{\backslash},\text{\textgreek{W}}_{\backslash}^{2},\bar{f}_{in\backslash},\bar{f}_{out\backslash})$
and $(r_{/},\text{\textgreek{W}}_{/}^{2},\bar{f}_{in/},\bar{f}_{out/})$
consist a boundary-double characteristic initial data set for the
system (\ref{eq:RequationFinal})--(\ref{eq:OutgoingVlasovFinal})
satisfying the reflecting boundary condition at $r=r_{0}$, according
to Definition \ref{def:TypeIII}. We will assume, in addition, that
$r_{\backslash}$ satisfies for all $u\in[u_{1},u_{2}]$ 
\begin{equation}
\partial_{u}r_{\backslash}(u)<0\label{eq:NegativeDerivativeInitially-1}
\end{equation}
and $r_{/}$ satisfies for all $v\in[v_{1},v_{2}]$
\begin{equation}
\partial_{v}r_{/}(v)>0.\label{eq:PositiveD_vDerivativeRInitially}
\end{equation}
Note that (\ref{eq:RelationHawkingMass}), (\ref{eq:PositiveDerivativeD_vonIngoing})
and (\ref{eq:NegativeDerivativeInitially-1}) imply that 
\begin{equation}
1-\frac{2m_{\backslash}}{r_{\backslash}}>0,
\end{equation}
while (\ref{eq:RelationHawkingMass}), (\ref{eq:ConstraintVDef}),
(\ref{eq:PositiveD_vDerivativeRInitially}) and (\ref{eq:NegativeDerivativeInitially-1})
imply that 
\begin{equation}
1-\frac{2m_{/}}{r_{/}}>0,
\end{equation}
where 
\begin{gather}
m_{\backslash}(v)=\frac{r_{\backslash}}{2}\big(1+4\text{\textgreek{W}}_{\backslash}^{-2}\partial_{u}r_{\backslash}(\partial_{v}r)_{\backslash}\big)(v),\\
m_{/}(v)=\frac{r_{/}}{2}\big(1+4\text{\textgreek{W}}_{/}^{-2}(\partial_{u}r)_{/}\partial_{v}r_{/}\big)(v)
\end{gather}
and $(\partial_{v}r)_{\backslash}(\partial_{u}r)_{/},$ are defined
according to the formulas (\ref{eq:TransversalDerivativeV}) and (\ref{eq:TransversalDerivativeU}),
respectively. Let us also set
\begin{align}
M\doteq\max_{u\in[u_{1},u_{2}]} & \Bigg\{\Big|\log\big(\frac{\text{\textgreek{W}}_{\backslash}^{2}}{1-\frac{1}{3}\Lambda r_{\backslash}^{2}}\big)\Big|+\Big|\log\Big(\frac{-2\partial_{u}r_{\backslash}}{1-\frac{2m_{\backslash}}{r_{\backslash}}}\Big)\Big|+\Big|\log\Big(\frac{1-\frac{2m_{\backslash}}{r_{\backslash}}}{1-\frac{1}{3}\Lambda r_{\backslash}^{2}}\Big)\Big|\Bigg\}(u)+\int_{u_{1}}^{u_{2}}r_{\backslash}(T_{uu})_{\backslash}\, d\bar{u}+\label{eq:UpperBoundInitialData-1}\\
 & +\max_{v\in[v_{1},v_{2}]}\Bigg\{\Big|\log\big(\frac{\text{\textgreek{W}}_{/}^{2}}{1-\frac{1}{3}\Lambda r_{/}^{2}}\big)\Big|+\Big|\log\Big(\frac{2\partial_{v}r_{/}}{1-\frac{2m_{/}}{r_{/}}}\Big)\Big|+\Big|\log\Big(\frac{1-\frac{2m_{/}}{r_{/}}}{1-\frac{1}{3}\Lambda r_{/}^{2}}\Big)\Big|+\sqrt{-\Lambda}|\tilde{m}_{/}|\Bigg\}(v)\nonumber 
\end{align}
and, for any $0<\text{\textgreek{d}}<1$: 
\begin{equation}
u_{in}(\text{\textgreek{d}})\doteq\sup\Bigg\{0\le u_{*}\le u_{2}-u_{1}:\mbox{ }\sup_{u\in[u_{1},u_{2}]}\int_{\max\{u-u_{*},u_{1}\}}^{\min\{u+u_{*},u_{2}\}}\frac{r_{\backslash}^{2}(T_{uu})_{\backslash}(\bar{u})}{|\bar{u}-u|+r_{0}}\, d\bar{u}<\text{\textgreek{d}}\Bigg\},\label{eq:SmallnessInL1OfT-1}
\end{equation}
where 
\begin{equation}
m_{/}(v)=\frac{r_{/}}{2}\big(1+4\text{\textgreek{W}}_{/}^{-2}(\partial_{u}r)_{/}\partial_{v}r_{/}\big)(v),
\end{equation}
 
\begin{equation}
\tilde{m}_{/}(v)=m(v)-\frac{1}{6}\Lambda r_{/}^{3}(v)
\end{equation}
and $(\partial_{u}r)_{/}$ is defined according to (\ref{eq:TransversalDerivativeU-1}).
Then, provided $v_{2}-v_{1}$ is sufficiently small so that 
\begin{equation}
v_{2}-v_{1}<\frac{u_{in}(2e^{-C_{0}^{2}((-\Lambda)(u_{2}-u_{1})^{2}+1)M}M)}{e^{C_{0}^{2}((-\Lambda)(u_{2}-u_{1})^{2}+1)}}\label{eq:U0UpperBound-1}
\end{equation}
and 
\begin{equation}
\int_{v_{1}}^{v_{2}}r_{/}(T_{vv})_{/}\, d\bar{v}<2e^{-C_{0}^{2}((-\Lambda)(u_{2}-u_{1})^{2}+1)M}M,
\end{equation}
the following holds: Setting 
\begin{equation}
\mathcal{V}=\{u_{1}<u<v-v_{1}+u_{2}\}\cap\{v_{1}<v<v_{2}\}
\end{equation}
and 
\begin{equation}
\text{\textgreek{g}}_{0;\mathcal{V}}=\{u-u_{2}=v-v_{1}\}\cap\{v_{1}<v<v_{2}\},
\end{equation}
there exist smooth functions $r:\mathcal{V}\cup\text{\textgreek{g}}_{0;\mathcal{V}}\rightarrow(r_{0},+\infty)$,
$\text{\textgreek{W}}:\mathcal{V}\cup\text{\textgreek{g}}_{0;\mathcal{V}}\rightarrow(0,+\infty)$
and $\bar{f}_{in},\bar{f}_{out}:\mathcal{V}\cup\text{\textgreek{g}}_{0;\mathcal{V}}\times(0,+\infty)\rightarrow[0,+\infty)$
solving equations (\ref{eq:RequationFinal})--(\ref{eq:OutgoingVlasovFinal})
on $\mathcal{V}$ (with $T_{uu},T_{vv}$ given by the formulas (\ref{eq:T_uuComponent}),
(\ref{eq:T_vvComponent})), such that:

\begin{enumerate}

\item The functions $r,\text{\textgreek{W}}^{2},\bar{f}_{in},\bar{f}_{out}$
satisfy the given initial conditions on $[u_{1},u_{2}]\times\{v_{1}\}$
and $\{u_{1}\}\times[v_{1},v_{2}]$ , i.\,e.:
\begin{equation}
(r,\text{\textgreek{W}}^{2},\bar{f}_{in},\bar{f}_{out})|_{[u_{1},u_{2}]\times\{v_{1}\}}=(r_{\backslash},\text{\textgreek{W}}_{\backslash}^{2},\bar{f}_{in\backslash},\bar{f}_{out\backslash})\label{eq:InitialDataLeft-1-1}
\end{equation}
and 
\begin{equation}
(r,\text{\textgreek{W}}^{2},\bar{f}_{in},\bar{f}_{out})|_{\{u_{1}\}\times[v_{1},v_{2}]}=(r_{/},\text{\textgreek{W}}_{/}^{2},\bar{f}_{in/},\bar{f}_{out/}).\label{eq:InitialDataRight-1-1}
\end{equation}

\item The functions $(r,\bar{f}_{in},\bar{f}_{out})$ satisfy on
$\text{\textgreek{g}}_{0;\mathcal{V}}$ the boundary conditions 
\begin{equation}
r|_{\text{\textgreek{g}}_{0;\mathcal{V}}}=r_{0}\label{eq:MirrorLocalExistence-1}
\end{equation}
 and 
\begin{equation}
\bar{f}_{out}\big(u_{*},v_{*};\, p\big)=\bar{f}_{in}\big(u_{*},v_{*};\, p\big)\label{eq:ReflectionMirrorLocalExistence-1}
\end{equation}
for all $(u_{*},v_{*})\in\text{\textgreek{g}}_{0;\mathcal{V}}$ and
$p>0$, as well as the reflecting gauge condition 
\begin{equation}
\partial_{u}r|_{\text{\textgreek{g}}_{0;\mathcal{V}}}=-\partial_{v}r|_{\text{\textgreek{g}}_{0;\mathcal{V}}}.\label{eq:GaugeMirrorLocalExistence-1}
\end{equation}

\item The function $r$ satisfies 
\begin{equation}
\sup_{\mathcal{V}}\partial_{u}r<0.\label{eq:SignConditionDuR-1-1-2-1}
\end{equation}

\item The following estimates hold on $\mathcal{V}$: 
\begin{align}
\sup_{\mathcal{V}}\Bigg\{\Big|\log\big(\frac{\text{\textgreek{W}}^{2}}{1-\frac{1}{3}\Lambda r^{2}}\big)\Big|+\Big|\log\Big(\frac{2\partial_{v}r}{1-\frac{2m}{r}}\Big)\Big|+\Big|\log\Big(\frac{1-\frac{2m}{r}}{1-\frac{1}{3}\Lambda r^{2}}\Big)\Big|+\sqrt{-\Lambda}|\tilde{m}|\Bigg\}+\label{eq:BoundToShowInitialValueForContinuation-1-2}\\
+\sup_{\bar{v}}\int_{\{v=\bar{v}\}\cap\mathcal{V}}rT_{uu}\, du & \le C_{0}M\nonumber 
\end{align}
and 
\begin{equation}
\sup_{\bar{u}}\int_{\{u=\bar{u}\}\cap\mathcal{V}}rT_{vv}\, dv\le e^{-C_{0}M}M.\label{eq:BoundSmallnessForTypeIII}
\end{equation}

\end{enumerate}
\end{prop}
\begin{figure}[h] 
\centering 
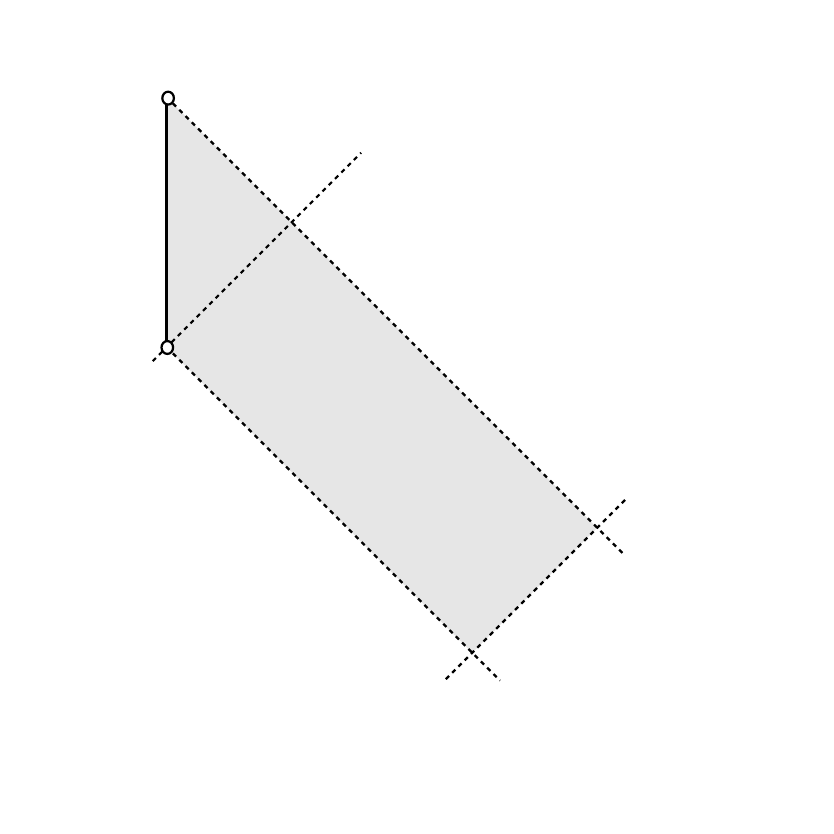 
\caption{Schematic depiction of the domain $\mathcal{V}$ in the statement of Proposition \ref{Prop:LocalExistenceTypeIII}.}
\end{figure}
\begin{proof}
As in the proof of Proposition \ref{Prop:LocalExistenceTypeII}, by
introducing the new variables (\ref{eq:RhoVariable})--(\ref{eq:TVariable}),
the proof of Proposition \ref{Prop:LocalExistenceTypeIII} will follow
by establishing the existence and uniqueness of a smooth solution
to the system (\ref{eq:DerivativeInUDirectionKappaRenormalised})--(\ref{eq:ConservationTau})
on $\mathcal{V}\cup\text{\textgreek{g}}_{0;\mathcal{V}}$ satisfying
the initial conditions 
\begin{equation}
(\text{\textgreek{r}},\text{\textgreek{k}},\tilde{m},\bar{\text{\textgreek{t}}})|_{\{u_{1}\}\times[v_{1},v_{2}]}=\Big(\tan^{-1}\big(\sqrt{-\frac{\Lambda}{3}}r_{/}\big),\frac{2\partial_{v}r_{/}}{1-\frac{2m_{/}}{r_{/}}},\tilde{m}_{/},r_{/}^{2}(T_{vv})_{/}\Big)\label{eq:InitialDataRecastRight-2}
\end{equation}
and 
\begin{equation}
(\text{\textgreek{r}},\bar{\text{\textgreek{k}}},\text{\textgreek{t}})|_{[u_{1},u_{2}]\times\{v_{1}\}}=\Big(\tan^{-1}\big(\sqrt{-\frac{\Lambda}{3}}r_{\backslash}\big),\frac{-2\partial_{u}r_{\backslash}}{1-\frac{2m_{\backslash}}{r_{\backslash}}},r_{\backslash}^{2}(T_{uu})_{\backslash}\Big),
\end{equation}
as well as the boundary conditions
\begin{equation}
\text{\textgreek{r}}|_{\text{\textgreek{g}}_{0;\mathcal{V}}}=\text{\textgreek{r}}_{0},\label{eq:RhoBoundary-1}
\end{equation}
 
\begin{equation}
\frac{\text{\textgreek{k}}}{\bar{\text{\textgreek{k}}}}\Big|_{\text{\textgreek{g}}_{0;\mathcal{V}}}=1\label{eq:KappaBoundary-1}
\end{equation}
 and
\begin{equation}
\frac{\text{\textgreek{t}}}{\bar{\text{\textgreek{t}}}}\Big|_{\text{\textgreek{g}}_{0;\mathcal{V}}}=1,\mbox{ }\label{eq:TauBoundary-1}
\end{equation}
where $\text{\textgreek{r}}_{0}$ is defined by (\ref{eq:Rho_0}),
such that, moreover, the estimates (\ref{eq:BoundToShowInitialValueForContinuation-1-2})--(\ref{eq:BoundSmallnessForTypeIII})
are satisfied. The proof follows in the same way as the proof of Proposition
\ref{Prop:LocalExistenceTypeII}, and hence the details will be omitted.\end{proof}
\begin{prop}
\label{prop:LocalExistenceTypeI}For any $u_{1}<u_{2}$, $v_{1}<v_{2}$
and $0<r_{1}<R_{1}<+\infty$, let $r_{\backslash}:[u_{1},u_{2}]\rightarrow[r_{1},R_{1}]$,
$\text{\textgreek{W}}_{\backslash}:[u_{1},u_{2}]\rightarrow(0,+\infty)$,
$\bar{f}_{in\backslash},\bar{f}_{out\backslash}:[u_{1},u_{2}]\times(0,+\infty)\rightarrow[0,+\infty)$,
$r_{/}:[v_{1},v_{2}]\rightarrow[r_{1},R_{1}]$, $\text{\textgreek{W}}_{/}:[v_{1},v_{2}]\rightarrow(0,+\infty)$
and $\bar{f}_{in/},\bar{f}_{out/}:[v_{1},v_{2}]\times(0,+\infty)\rightarrow[0,+\infty)$
be smooth functions consisting a characteristic initial data set for
the system (\ref{eq:RequationFinal})--(\ref{eq:OutgoingVlasovFinal}),
according to Definition \ref{def:TypeI}, satisfying, in addition,
for all $u\in[u_{1},u_{2}]$, the sign condition 
\begin{equation}
\partial_{u}r_{\backslash}(u)<0.\label{eq:NegativeDerivativeInitially}
\end{equation}
 Let us also set 
\begin{align}
M\doteq\max_{u\in[u_{1},u_{2}]} & \Bigg\{\Big|\log(-\partial_{u}r_{\backslash})\Big|+\Big|\log(\text{\textgreek{W}}_{\backslash}^{2})\Big|+\sqrt{-\Lambda}|\tilde{m}_{\backslash}|\Bigg\}(u)+\int_{u_{1}}^{u_{2}}r_{\backslash}(T_{uu})_{\backslash}\, d\bar{u}+\label{eq:UpperBoundInitialDataCharacteristi}\\
 & +\max_{v\in[v_{1},v_{2}]}\Bigg\{\Big|\log\big(-(\partial_{u}r)_{/}\big)\Big|+\Big|\log(\text{\textgreek{W}}_{/}^{2})\Big|+\sqrt{-\Lambda}|\tilde{m}_{/}|\Bigg\}(v)\nonumber 
\end{align}
 where 
\begin{align}
\tilde{m}_{\backslash}(u) & =\frac{r_{\backslash}}{2}\big(1+4\text{\textgreek{W}}_{\backslash}^{-2}\partial_{u}r_{\backslash}(\partial_{v}r)_{\backslash}\big)(u)-\frac{1}{6}\Lambda r_{\backslash}^{3}(u),\label{eq:DefinitionHawkingMassCharacteristic}\\
\tilde{m}_{/}(v) & =\frac{r_{/}}{2}\big(1+4\text{\textgreek{W}}_{/}^{-2}(\partial_{u}r)_{/}\partial_{v}r_{/}\big)(v)-\frac{1}{6}\Lambda r_{/}^{3}(v),\nonumber 
\end{align}
and the transversal derivatives $(\partial_{v}r)_{\backslash},(\partial_{u}r)_{/}$
are computed according to (\ref{eq:TransversalDerivativeV})--(\ref{eq:TransversalDerivativeU}),
i.\,e.: 
\begin{align}
(\partial_{v}r)_{\backslash}(u) & =\partial_{v}r_{/}(v_{1})-\frac{1}{4r_{\backslash}(u)}\int_{u_{1}}^{u}\big(1-\Lambda r_{\backslash}^{2}(\bar{u})\big)\text{\textgreek{W}}_{\backslash}^{2}(\bar{u})\, d\bar{u},\label{eq:TransversalDerivativesCharacteristic}\\
(\partial_{u}r)_{/}(v) & =\partial_{u}r_{\backslash}(u_{1})-\frac{1}{4r_{/}(v)}\int_{v_{1}}^{v}\big(1-\Lambda r_{/}^{2}(\bar{v})\big)\text{\textgreek{W}}_{/}^{2}(\bar{v})\, d\bar{v}.\nonumber 
\end{align}
We will also define for any $0<\text{\textgreek{d}}<1$: 
\begin{equation}
u_{in}(\text{\textgreek{d}})\doteq\sup\Bigg\{0\le u_{*}\le u_{2}-u_{1}:\mbox{ }\sup_{u\in[u_{1},u_{2}]}\int_{\max\{u-u_{*},u_{1}\}}^{\min\{u+u_{*},u_{2}\}}\frac{r_{\backslash}^{2}(T_{uu})_{\backslash}(\bar{u})}{|\bar{u}-u|+r_{0}}\, d\bar{u}<\text{\textgreek{d}}\Bigg\}.\label{eq:SmallnessInL1OfT-1-1}
\end{equation}
Then, provided 
\begin{equation}
v_{2}-v_{1}<\frac{u_{in}(2e^{-C_{r_{1}R_{1}}^{2}((-\Lambda)(u_{2}-u_{1})^{2}+1)M}M)}{e^{C_{r_{1}R_{1}}^{2}((-\Lambda)(u_{2}-u_{1})^{2}+1)}}\label{eq:U0UpperBound-1-1}
\end{equation}
and 
\begin{equation}
\int_{v_{1}}^{v_{2}}r_{/}(T_{vv})_{/}\, d\bar{v}<2e^{-C_{r_{1}R_{1}}^{2}((-\Lambda)(u_{2}-u_{1})^{2}+1)M}M,
\end{equation}
where $C_{r_{1}R_{1}}\gg1$ is a constant depending only on $r_{1}$
and $R_{1}$, there exist unique smooth functions $r:[u_{1},u_{2}]\times[v_{1},v_{2}]\rightarrow[\frac{1}{2}r_{1},R_{1}]$,
$\text{\textgreek{W}}:[u_{1},u_{2}]\times[v_{1},v_{2}]\rightarrow(0,+\infty)$
and $\bar{f}_{in},\bar{f}_{out}:[u_{1},u_{2}]\times[v_{1},v_{2}]\times(0,+\infty)\rightarrow[0,+\infty)$
solving equations (\ref{eq:RequationFinal})--(\ref{eq:OutgoingVlasovFinal})
(with $T_{uu},T_{vv}$ given by the formulas (\ref{eq:T_uuComponent}),
(\ref{eq:T_vvComponent})), such that:

\begin{enumerate}

\item The functions $r,\text{\textgreek{W}}^{2},\bar{f}_{in},\bar{f}_{out}$
satisfy the initial conditions on $[u_{1},u_{2}]\times\{v_{1}\}$and
$\{u_{1}\}\times[v_{1},v_{2}]$, i.\,e.: 
\begin{equation}
(r,\text{\textgreek{W}}^{2},\bar{f}_{in},\bar{f}_{out})|_{[u_{1},u_{2}]\times\{v_{1}\}}=(r_{\backslash},\text{\textgreek{W}}_{\backslash}^{2},\bar{f}_{in\backslash},\bar{f}_{out\backslash})\label{eq:InitialDataLeft}
\end{equation}
and 
\begin{equation}
(r,\text{\textgreek{W}}^{2},\bar{f}_{in},\bar{f}_{out})|_{\{u_{1}\}\times[v_{1},v_{2}]}=(r_{/},\text{\textgreek{W}}_{/}^{2},\bar{f}_{in/},\bar{f}_{out/}).\label{eq:InitialDataRight}
\end{equation}

\item The function $r$ satisfies 
\[
\sup_{[u_{1},u_{2}]\times[v_{1},v_{2}]}\partial_{u}r<0.
\]

\item The following estimates hold on $[u_{1},u_{2}]\times[v_{1},v_{2}]$:
\begin{equation}
\sup_{[u_{1},u_{2}]\times[v_{1},v_{2}]}\Bigg\{\log\big(1+|\partial_{v}r|\big)+\Big|\log(\text{\textgreek{W}}^{2})\Big|+\Big|\log(-\partial_{u}r)\Big|+\sqrt{-\Lambda}|\tilde{m}|\Bigg\}+\sup_{\bar{v}\in[v_{1},v_{2}]}\int_{u_{1}}^{u_{2}}rT_{uu}(u,\bar{v})\, du<C_{r_{1}R_{1}}M
\end{equation}
and 
\begin{equation}
\sup_{\bar{u}\in[u_{1},u_{2}]}\int_{v_{1}}^{v_{2}}rT_{vv}(\bar{u},v)\, dv\le e^{-C_{r_{1}R_{1}}M}M.\label{eq:BoundSmallnessForTypeIII-1}
\end{equation}

\end{enumerate}\end{prop}
\begin{proof}
Using the expression 
\begin{equation}
\frac{\partial_{v}r_{/}}{1-\frac{2\tilde{m}_{/}}{r_{/}}-\frac{1}{3}\Lambda r_{/}^{2}}=\frac{1}{4}\frac{\text{\textgreek{W}}_{/}^{2}}{-(\partial_{u}r)_{/}},\label{eq:D_vrFromK}
\end{equation}
the relation (\ref{eq:UpperBoundInitialDataCharacteristi}) implies
that 
\begin{equation}
\max_{v\in[v_{1},v_{2}]}\log\big(1+|\partial_{v}r_{/}|\big)\le4M.\label{eq:BoundKInitially}
\end{equation}

Equations (\ref{eq:RequationFinal})--(\ref{eq:OutgoingVlasovFinal})
for $(r,\text{\textgreek{W}}^{2},\bar{f}_{in},\bar{f}_{out})$ on
$[u_{1},u_{2}]\times[v_{1},v_{2}]$ (together with the initial constraint
(\ref{eq:ConstraintVDef})) are equivalent to the system (\ref{eq:RequationFinal-2})--(\ref{eq:OutgoingConservationClosed})
for $(r,\text{\textgreek{W}}^{2},T_{uu},T_{vv})$. In particular,
if $(r,\text{\textgreek{W}}^{2},T_{uu},T_{vv})$ is a solution to
the system (\ref{eq:RequationFinal-2})--(\ref{eq:OutgoingConservationClosed})
on $[u_{1},u_{2}]\times[v_{1},v_{2}]$ satisfying the initial conditions
\begin{equation}
(r,\text{\textgreek{W}}^{2},T_{uu})|_{[u_{1},u_{2}]\times\{v_{1}\}}=(r_{\backslash},\text{\textgreek{W}}_{\backslash}^{2},(T_{uu})_{\backslash})\label{eq:InitialLeft}
\end{equation}
and 
\begin{equation}
(r,\text{\textgreek{W}}^{2},T_{vv})|_{\{u_{1}\}\times[v_{1},v_{2}]}=(r_{/},\text{\textgreek{W}}_{/}^{2},(T_{vv})_{/})\label{eq:InitialRight-1}
\end{equation}
(which are assumed to satisfy the constraints (\ref{eq:ConstraintUDef})--(\ref{eq:ConstraintVDef})),
then, by solving equations (\ref{eq:IngoingVlasovFinal}) and (\ref{eq:OutgoingVlasovFinal})
for $\bar{f}_{in},\bar{f}_{out}$ on $[u_{1},u_{2}]\times[v_{1},v_{2}]$
with initial data 
\begin{equation}
(\bar{f}_{in},\bar{f}_{out})|_{[u_{1},u_{2}]\times\{v_{1}\}}=(\bar{f}_{in\backslash},\bar{f}_{out\backslash})
\end{equation}
and 
\begin{equation}
(\bar{f}_{in},\bar{f}_{out})|_{\{u_{1}\}\times[v_{1},v_{2}]}=(\bar{f}_{in/},\bar{f}_{out/}),
\end{equation}
one obtains a solution $(r,\text{\textgreek{W}}^{2},\bar{f}_{in},\bar{f}_{out})$
of the system (\ref{eq:RequationFinal})--(\ref{eq:OutgoingVlasovFinal})
satisfying (\ref{eq:InitialDataLeft}) and (\ref{eq:InitialDataRight}).
Therefore, it suffices to establish the existence and uniqueness of
a smooth solution $(r,\text{\textgreek{W}}^{2},T_{uu},T_{vv})$ to
the system (\ref{eq:RequationFinal-2})--(\ref{eq:OutgoingConservationClosed})
on $[u_{1},u_{2}]\times[v_{1},v_{2}]$ satisfying (\ref{eq:InitialLeft})
and (\ref{eq:InitialRight-1}).

Introducng the new set of variables 
\begin{gather}
\text{\textgreek{l}}=\partial_{v}r,\label{eq:OmegaVariable}\\
\text{\textgreek{k}}=\frac{1}{4}\frac{\text{\textgreek{W}}^{2}}{-\partial_{u}r},\label{eq:KappaVariable-1}\\
\bar{\text{\textgreek{t}}}=r^{2}T_{vv},\label{eq:Tbar-1}\\
\text{\textgreek{t}}=r^{2}T_{uu},\label{eq:T-1}
\end{gather}
the system (\ref{eq:RequationFinal-2})--(\ref{eq:OutgoingConservationClosed})
yields: 
\begin{align}
\partial_{u}\text{\textgreek{l}}= & -\frac{1}{2}\frac{\tilde{m}-\frac{1}{3}\Lambda r^{3}}{r^{2}}\text{\textgreek{W}}^{2},\label{eq:EquationLambdaRenormalisedOO}\\
\partial_{u}\partial_{v}\log(\text{\textgreek{W}}^{2})= & \frac{\text{\textgreek{W}}^{2}}{2r^{2}}\big(1+\text{\textgreek{k}}^{-1}\text{\textgreek{l}}\big),\label{eq:EquationOmega-1-1}\\
\partial_{u}\text{\textgreek{k}}= & -4\pi r^{-1}\text{\textgreek{t}}\text{\textgreek{W}}^{-2},\label{eq:ConstraintU-1-1}\\
\partial_{u}\bar{\text{\textgreek{t}}}= & 0,\label{eq:ConservationT_vv-1-1}\\
\partial_{v}\text{\textgreek{t}}= & 0,\label{eq:ConservationT_uu-1-1}\\
\partial_{u}\tilde{m}= & -8\pi\text{\textgreek{W}}^{-2}\text{\textgreek{l}}\text{\textgreek{t}},\label{eq:NewMTildeRenormalised}\\
\partial_{v}r= & \text{\textgreek{l}}.\label{eq:EquationForROO}
\end{align}
while the initial conditions (\ref{eq:InitialLeft})--(\ref{eq:InitialRight-1})
give rise to the following initial conditions for $(\text{\textgreek{l}},\text{\textgreek{W}}^{2},\text{\textgreek{k}},\bar{\text{\textgreek{t}}},\text{\textgreek{t}},\tilde{m},r)$:
\begin{equation}
(r,\text{\textgreek{W}}^{2},\text{\textgreek{t}})|_{[u_{1},u_{2}]\times\{v_{1}\}}=\Big(r_{\backslash},\text{\textgreek{W}}_{\backslash}^{2},r_{\backslash}^{2}(T_{uu})_{\backslash}\Big)\label{eq:InitialDataRecastLeft}
\end{equation}
and 
\begin{equation}
(\text{\textgreek{l}},\text{\textgreek{W}}^{2},\text{\textgreek{k}},\bar{\text{\textgreek{t}}},\tilde{m})|_{\{u_{1}\}\times[v_{1},v_{2}]}=\Big(\partial_{v}r_{/},\text{\textgreek{W}}_{/}^{2},\frac{-(\partial_{u}r)_{/}}{\text{\textgreek{W}}_{/}^{2}},r_{/}^{2}(T_{vv})_{/},\tilde{m}_{/}\Big).\label{eq:InitialDataRecastRight-1}
\end{equation}

The proof of the fact that the system (\ref{eq:EquationLambdaRenormalisedOO})--(\ref{eq:EquationForROO})
admits a unique smooth solution on $[u_{1},u_{2}]\times[v_{1},v_{2}]$
satifying the initial conditions (\ref{eq:InitialDataRecastLeft})--(\ref{eq:InitialDataRecastRight-1}),
such that, moreover, 
\begin{equation}
\sup_{[u_{1},u_{2}]\times[v_{1},v_{2}]}\Bigg\{\log\big(1+|\text{\textgreek{l}}|\big)+\Big|\log(\text{\textgreek{W}}^{2})\Big|+\Big|\log(\text{\textgreek{k}})\Big|+\sqrt{-\Lambda}|\tilde{m}|\Bigg\}+\sup_{\bar{v}\in[v_{1},v_{2}]}\int_{u_{1}}^{u_{2}}\frac{\text{\textgreek{t}}}{r}(u,\bar{v})\, du<C_{r_{1}R_{1}}M
\end{equation}
and 
\begin{equation}
\sup_{\bar{u}\in[u_{1},u_{2}]}\int_{v_{1}}^{v_{2}}\frac{\bar{\text{\textgreek{t}}}}{r}(\bar{u},v)\, dv\le e^{-C_{r_{1}R_{1}}M}M\label{eq:BoundSmallnessForTypeIII-1-1}
\end{equation}
hold, follows by similar arguments as in the proof of Proposition
\ref{Prop:LocalExistenceTypeII}. We will omit the details.
\end{proof}

\subsection{\label{sub:Continuation-criteria}Continuation criteria}

In this Section, we will establish two continuation criteria, i.\,e.~conditions
under which smooth solutions $(r,\text{\textgreek{W}}^{2},\bar{f}_{in},\bar{f}_{out})$
of (\ref{eq:RequationFinal})--(\ref{eq:OutgoingVlasovFinal}) on
open regions $\mathcal{U}\subset\mathbb{R}^{2}$ admit a smooth extension
across $\partial\mathcal{U}$.

The following lemma is a continuation criterion away from $r=0$ and
$r=+\infty$ (cf.~the extension principle for the spherically symmetric
Einstein--massive Vlasov system, not reduced to the radial case, in
\cite{DafermosRendall} and the generalised extension principle for
strongly tame matter models in \cite{Kommemi2013}).

\begin{figure}[h] 
\centering 
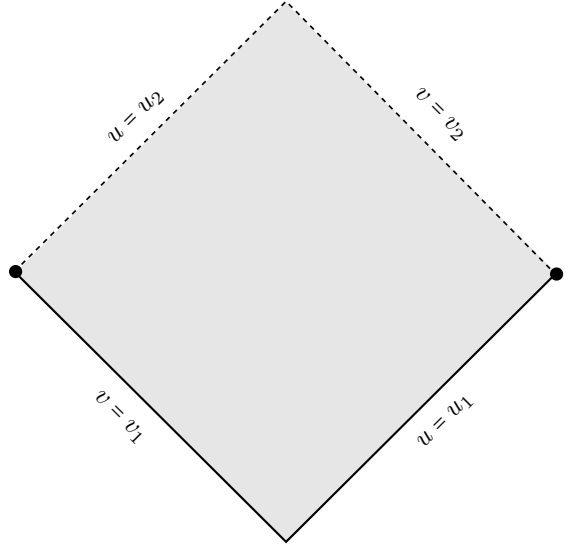 
\caption{Schematic depiction of the domain of definition $[u_1,u_2)\times [v_1,v_2)$ of $(r,\Omega^2,\bar{f}_{in},\bar{f}_{out})$ in Lemma \ref{lem:cotinuationCriterion}. The functions $(r,\Omega^2,\bar{f}_{in},\bar{f}_{out})$ are assumed to extend smoothly on $\{u_2\}\times\{v_1\}$ and $\{u_1\}\times\{v_2\}$ and,  under the conditions \eqref{eq:BoundedBelowRContinuation} and \eqref{eq:BoundedAboveRContinuation}, it is shown that they extend smoothly on $\{u_2\}\times [v_1,v_2]$ and $[u_1,u_2]\times \{v_2\}$.}
\end{figure}
\begin{lem}
\label{lem:cotinuationCriterion}For any $u_{1}<u_{2}$ and $v_{1}<v_{2}$,
let $(r,\text{\textgreek{W}}^{2},\bar{f}_{in},\bar{f}_{out})$ be
a smooth solution of the system (\ref{eq:RequationFinal})--(\ref{eq:OutgoingVlasovFinal})
on $[u_{1},u_{2})\times[v_{1},v_{2})$, such that $(r,\log(\text{\textgreek{W}}^{2}))|_{[u_{1},u_{2})\times\{v_{1}\}}$
and $(r,\log(\text{\textgreek{W}}^{2}))|_{\{u_{1}\}\times[v_{1},v_{2})}$
extend smoothly on $\{u_{2}\}\times\{v_{1}\}$ and $\{u_{1}\}\times\{v_{2}\}$,
respectively, and $(\bar{f}_{in},\bar{f}_{out})|_{[u_{1},u_{2})\times\{v_{1}\}\times(0,+\infty)}$
and $(\bar{f}_{in},\bar{f}_{out})|_{\{u_{1}\}\times[v_{1},v_{2})\times(0,+\infty)}$
extend smoothly on $\{u_{2}\}\times\{v_{1}\}\times(0,+\infty)$ and
$\{u_{1}\}\times\{v_{2}\}\times(0,+\infty)$, respectively. Assume,
moreover, that 
\begin{equation}
\inf_{[u_{1},u_{2})\times[v_{1},v_{2})}r>0\label{eq:BoundedBelowRContinuation}
\end{equation}
and
\begin{equation}
\sup_{[u_{1},u_{2})\times[v_{1},v_{2})}r<+\infty.\label{eq:BoundedAboveRContinuation}
\end{equation}
Then, $(r,\log(\text{\textgreek{W}}^{2}))$ extend smoothly on the
whole of $[u_{1},u_{2}]\times[v_{1},v_{2}]$ and $(\bar{f}_{in},\bar{f}_{out})$
extend smoothly on the whole of $[u_{1},u_{2}]\times[v_{1},v_{2}]\times(0,+\infty)$.\end{lem}
\begin{proof}
It suffices to show that $(r,\log(\text{\textgreek{W}}^{2}))$ extend
smoothly on the whole of $[u_{1},u_{2}]\times[v_{1},v_{2}]$, since
then the smooth extension of $(\bar{f}_{in},\bar{f}_{out})$ will
readily follow by integrating equations (\ref{eq:IngoingVlasovFinal})--(\ref{eq:OutgoingVlasovFinal}).
In fact, we will only show that $(r,\log(\text{\textgreek{W}}^{2}))$
extend continuously on the whole of $[u_{1},u_{2}]\times[v_{1},v_{2}]$.
Assuming that $(r,\log(\text{\textgreek{W}}^{2}))\in C^{0}\big([u_{1},u_{2}]\times[v_{1},v_{2}]\big)$,
equation (\ref{eq:RequationFinal}), combined with the fact that $r|_{[u_{1},u_{2})\times\{v_{1}\}}$
and $r|_{\{u_{1}\}\times[v_{1},v_{2})}$ extend smoothly on $\{u_{2}\}\times\{v_{1}\}$
and $\{u_{1}\}\times\{v_{2}\}$, implies that $\partial_{v}r$, $\partial_{u}r$
are continuous on $[u_{1},u_{2}]\times[v_{1},v_{2}]$. Similarly,
equation (\ref{eq:OmegaEquationFinal}) yields, in turn, that $\partial_{u}\log(\text{\textgreek{W}}^{2}),\partial_{v}\log(\text{\textgreek{W}}^{2})$
are continuous on $[u_{1},u_{2}]\times[v_{1},v_{2}]$. Commuting (\ref{eq:RequationFinal})--(\ref{eq:OmegaEquationFinal})
successively with $\partial_{u},\partial_{v}$ and treating the commuted
equations as linear equations in the highest order terms, the smoothness
of $(r,\log(\text{\textgreek{W}}^{2}))$ then follows readily.

In view of (\ref{eq:DefinitionHawkingMass}), (\ref{eq:RenormalisedHawkingMass}),
(\ref{eq:T_uuComponent}), (\ref{eq:T_vvComponent}), the system (\ref{eq:RequationFinal})--(\ref{eq:OutgoingVlasovFinal})
yields 
\begin{align}
\partial_{u}\partial_{v}(r^{2})= & -\frac{1}{2}(1-\Lambda r^{2})\text{\textgreek{W}}^{2},\label{eq:RequationFinal-1}\\
\partial_{u}\partial_{v}\log(\text{\textgreek{W}}^{2})= & (\tilde{m}r^{-3}+\frac{1}{6}\Lambda)\text{\textgreek{W}}^{2},\label{eq:OmegaEquationFinal-1}\\
\partial_{u}\partial_{v}\tilde{m}= & 32\pi^{2}r^{-1}\text{\textgreek{W}}^{-2}(r^{2}T_{vv})(r^{2}T_{uu})\label{eq:MassWaveFinal-1}\\
\partial_{v}(\text{\textgreek{W}}^{-2}\partial_{v}r)= & -4\pi rT_{vv}\text{\textgreek{W}}^{-2},\label{eq:Constraint-1-1}\\
\partial_{u}(\text{\textgreek{W}}^{-2}\partial_{u}r)= & -4\pi rT_{uu}\text{\textgreek{W}}^{-2},\label{eq:Constraint-2-1}\\
\partial_{v}(r^{2}T_{uu})= & 0,\label{eq:IngoingVlasovFinal-1}\\
\partial_{u}(r^{2}T_{vv})= & 0.\label{eq:OutgoingVlasovFinal-1}
\end{align}
Integrating equations (\ref{eq:RequationFinal-1})--(\ref{eq:MassWaveFinal-1})
we obtain for any $(u,v)\in[u_{1},u_{1})\times[v_{1},v_{2})$: 
\begin{align}
r^{2}(u,v)= & r^{2}(u_{1},v)+r^{2}(u,v_{1})-r^{2}(u_{1},v_{1})-\frac{1}{2}\int_{u_{1}}^{u}\int_{v_{1}}^{v}(1-\Lambda r^{2})\text{\textgreek{W}}^{2}\, d\bar{v}d\bar{u},\label{eq:RequationIntegrated}\\
\log\big(\text{\textgreek{W}}^{2}(u,v)\big)= & \log\big(\text{\textgreek{W}}^{2}(u_{1},v)\big)+\log\big(\text{\textgreek{W}}^{2}(u,v_{1})\big)-\log\big(\text{\textgreek{W}}^{2}(u_{1},v_{1})\big)+\int_{u_{1}}^{u}\int_{v_{1}}^{v}(\tilde{m}r^{-3}+\frac{1}{6}\Lambda)\text{\textgreek{W}}^{2}\, d\bar{v}d\bar{u},\label{eq:OmegaEquationIntegrated}\\
\tilde{m}(u,v)= & \tilde{m}(u_{1},v)+\tilde{m}(u,v_{1})-\tilde{m}(u_{1},v_{1})+32\pi^{2}\int_{u_{1}}^{u}\int_{v_{1}}^{v}r^{-1}\text{\textgreek{W}}^{-2}(r^{2}T_{vv})(r^{2}T_{uu})\, d\bar{v}d\bar{u}.\label{eq:MassIntegrated}
\end{align}

Since $(r,\log(\text{\textgreek{W}}^{2}))|_{[u_{1},u_{2})\times\{v_{1}\}}$
and $(r,\log(\text{\textgreek{W}}^{2}))|_{\{u_{1}\}\times[v_{1},v_{2})}$
extend smoothly on $\{u_{2}\}\times\{v_{1}\}$ and $\{u_{1}\}\times\{v_{2}\}$,
the functions $r(u_{1},v)$, $\log\big(\text{\textgreek{W}}^{2}(u_{1},v)\big)$
and $\tilde{m}(u_{1},v)$ extend continuously to $v=v_{2}$, while
the functions $r(u,v_{1})$, $\log\big(\text{\textgreek{W}}^{2}(u,v_{1})\big)$
and $\tilde{m}(u,v_{1})$ extend continuously to $u=u_{2}$. Therefore,
in view of (\ref{eq:BoundedBelowRContinuation})--(\ref{eq:BoundedAboveRContinuation})
and (\ref{eq:RequationIntegrated})--(\ref{eq:MassIntegrated}), the
continuous extension of $r,\log(\text{\textgreek{W}}^{2}),\tilde{m}$
on the whole of $[u_{1},u_{2}]\times[v_{1},v_{2}]$ will follow if
we establish 
\begin{equation}
\sup_{[u_{1},u_{2})\times[v_{1},v_{2})}\int_{u_{1}}^{u}\int_{v_{1}}^{v}\text{\textgreek{W}}^{2}\, d\bar{v}d\bar{u}<+\infty,\label{eq:IntegrabilityOmega-1}
\end{equation}
\begin{equation}
\sup_{[u_{1},u_{2})\times[v_{1},v_{2})}|\tilde{m}(u,v)|<+\infty\label{eq:UpperBoundMass}
\end{equation}
and 
\begin{equation}
\sup_{[u_{1},u_{2}]\times[v_{1},v_{2}]}\big(T_{uu}+T_{vv}\big)<+\infty.\label{eq:UpperBoundEnergyMomentum}
\end{equation}

Since $(r,\log(\text{\textgreek{W}}^{2}))|_{[u_{1},u_{2})\times\{v_{1}\}}$
and $(r,\log(\text{\textgreek{W}}^{2}))|_{\{u_{1}\}\times[v_{1},v_{2})}$
extend smoothly on $\{u_{2}\}\times\{v_{1}\}$ and $\{u_{1}\}\times\{v_{2}\}$,
from (\ref{eq:BoundedBelowRContinuation}), (\ref{eq:Constraint-1-1})
and (\ref{eq:Constraint-2-1}) we infer that 
\begin{equation}
\sup_{[u_{1},u_{2})\times\{v_{1}\}}T_{uu}+\sup_{\{u_{1}\}\times[v_{1},v_{2})}T_{vv}<+\infty,\label{eq:UpperBoundEnergyMomentum-1}
\end{equation}
Integrating (\ref{eq:IngoingVlasovFinal-1})--(\ref{eq:OutgoingVlasovFinal-1})
and using (\ref{eq:UpperBoundEnergyMomentum-1}) and (\ref{eq:BoundedAboveRContinuation}),
we readily infer that $r^{2}T_{uu}$ and $r^{2}T_{vv}$ extend continuously
on the whole of $[u_{1},u_{2}]\times[v_{1},v_{2}]$ and satisfy (\ref{eq:UpperBoundEnergyMomentum}).
In view of (\ref{eq:BoundedAboveRContinuation}), equation (\ref{eq:RequationIntegrated})
yields 
\begin{equation}
\sup_{[u_{1},u_{2})\times[v_{1},v_{2})}\int_{u_{1}}^{u}\int_{v_{1}}^{v}(1-\Lambda r^{2})\text{\textgreek{W}}^{2}\, d\bar{v}d\bar{u}<+\infty\label{eq:IntegrabilityOmega}
\end{equation}
and, therefore, (\ref{eq:IntegrabilityOmega-1}). Thus, it only remains
to establish (\ref{eq:UpperBoundMass})

In view of the fact that $T_{uu},T_{vv}\ge0$ and $\tilde{m}$ extends
smoothly on $\{u_{2}\}\times\{v_{1}\}$ and $\{u_{1}\}\times\{v_{2}\}$,
equation (\ref{eq:MassIntegrated}) implies that 
\begin{equation}
\inf_{[u_{1},u_{2})\times[v_{1},v_{2})}\tilde{m}>-\infty.\label{eq:LowerBoundTildeM}
\end{equation}
Equation (\ref{eq:OmegaEquationIntegrated}) then yields that 
\begin{align}
\inf_{[u_{1},u_{2})\times[v_{1},v_{2})}\big(\log\big(\text{\textgreek{W}}^{2}(u,v)\big)\big)\ge & \inf_{[u_{1},u_{2})\times[v_{1},v_{2})}\big(\log\big(\text{\textgreek{W}}^{2}(u_{1},v)\big)+\log\big(\text{\textgreek{W}}^{2}(u,v_{1})\big)-\log\big(\text{\textgreek{W}}^{2}(u_{1},v_{1})\big)\big)+\label{eq:LowerBoundLogOmega}\\
 & \hphantom{+++}+\Big(\big(\inf_{[u_{1},u_{2})\times[v_{1},v_{2})}\tilde{m}\big)\big(\inf_{[u_{1},u_{2})\times[v_{1},v_{2})}r\big)^{-3}+\frac{1}{6}\Lambda\Big)\sup_{[u_{1},u_{2})\times[v_{1},v_{2})}\int_{u_{1}}^{u}\int_{v_{1}}^{v}\text{\textgreek{W}}^{2}\, d\bar{v}d\bar{u}>\nonumber \\
 & >-\infty\nonumber 
\end{align}
in view of (\ref{eq:BoundedBelowRContinuation}), (\ref{eq:IntegrabilityOmega-1})
and (\ref{eq:LowerBoundTildeM}). Returning to equation (\ref{eq:MassIntegrated})
and considering the supremum of the right hand side, we infer: 
\begin{align}
\sup_{[u_{1},u_{2})\times[v_{1},v_{2})}\tilde{m}(u,v)\le & \sup_{[u_{1},u_{2})\times[v_{1},v_{2})}\big(\tilde{m}(u_{1},v)+\tilde{m}(u,v_{1})-\tilde{m}(u_{1},v_{1})\big)+\label{eq:ForUpperBoundMTilde}\\
 & \hphantom{+++}+32\pi^{2}\sup_{[u_{1},u_{2})\times[v_{1},v_{2})}\big(r^{2}T_{uu}+r^{2}T_{vv}\big)^{2}\big(\inf_{[u_{1},u_{2})\times[v_{1},v_{2})}r\big)^{-1}e^{-\inf_{[u_{1},u_{2})\times[v_{1},v_{2})}\log\text{\textgreek{W}}^{2}}|u_{2}-u_{1}||v_{2}-v_{1}|.\nonumber 
\end{align}
Thus, (\ref{eq:BoundedBelowRContinuation}), (\ref{eq:BoundedAboveRContinuation}),
(\ref{eq:UpperBoundEnergyMomentum}), (\ref{eq:LowerBoundLogOmega})
and (\ref{eq:ForUpperBoundMTilde}) imply that 
\begin{equation}
\sup_{[u_{1},u_{2})\times[v_{1},v_{2})}\tilde{m}(u,v)<+\infty.\label{eq:UpperBoundMTilde}
\end{equation}
Therefore, (\ref{eq:UpperBoundMass}) follows from (\ref{eq:LowerBoundTildeM})
and (\ref{eq:UpperBoundMTilde}). 
\end{proof}
The following lemma is a continuation criterion on the mirror $\text{\textgreek{g}}_{0}$.

\begin{figure}[h] 
\centering 
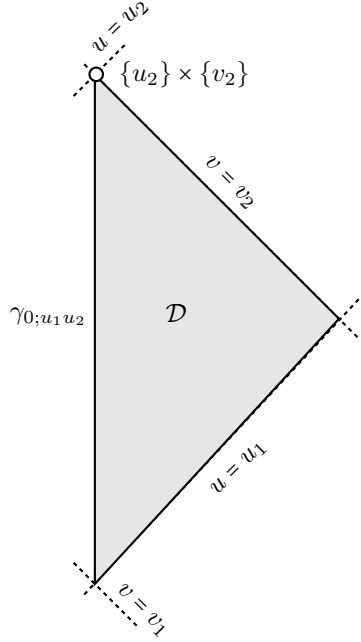 
\caption{Schematic depiction of the domain $\mathcal{D}$  in Lemma \ref{lem:cotinuationCriterionMirror}. The functions $(r,\Omega^2,\bar{f}_{in},\bar{f}_{out})$ are assumed to be smooth on $\mathcal{D}\backslash \{u_2\}\times \{v_2\}$ and, under the condition \eqref{eq:NonNullMirrorContinuation}, it is shown that they extend smoothly on $\{u_2\}\times \{v_2\}$.}
\end{figure}
\begin{lem}
\label{lem:cotinuationCriterionMirror}For any $u_{1}<u_{2}$ and
$v_{1}<v_{2}$ such that $u_{2}-u_{1}=v_{2}-v_{1}$, let 
\begin{equation}
\text{\textgreek{g}}_{0;u_{1}u_{2}}\doteq\{u_{1}\le u<u_{2}\}\cap\{v=u\}
\end{equation}
and 
\begin{equation}
\mathcal{D}\doteq\{u_{1}\le u\le u_{2}\}\cap\{v_{1}\le v\le u\}.
\end{equation}
 For any $r_{0}>0$, let $(r,\text{\textgreek{W}}^{2},\bar{f}_{in},\bar{f}_{out})$
be a smooth solution of the system (\ref{eq:RequationFinal})--(\ref{eq:OutgoingVlasovFinal})
on $\mathcal{D}\backslash\{u_{2}\}\times\{v_{2}\}$, such that 
\begin{equation}
r|_{\text{\textgreek{g}}_{0;u_{1}u_{2}}}=r_{0},\label{eq:ConditionRMirrorContinuation}
\end{equation}
\begin{equation}
\partial_{u}r|_{\text{\textgreek{g}}_{0;u_{1}u_{2}}}=-\partial_{v}r|_{\text{\textgreek{g}}_{0;u_{1}u_{2}}}\label{eq:gaugeconditionMirrorContinuation}
\end{equation}
and $\bar{f}_{in},\bar{f}_{out}$ satisfy the reflecting boundary
condition (\ref{eq:LeftBoundaryCondition}) on $\text{\textgreek{g}}_{0;u_{1}u_{2}}$
with $w=r-r_{0}$. Then, provided 
\begin{equation}
\inf_{\text{\textgreek{g}}_{0;u_{1}u_{2}}}\partial_{v}r>0,\label{eq:NonNullMirrorContinuation}
\end{equation}
the pair $(r,\log(\text{\textgreek{W}}^{2}))$ extends smoothly on
the whole of $\mathcal{D}$ and the pair $(\bar{f}_{in},\bar{f}_{out})$
extends smoothly on the whole of $\mathcal{D}\times(0,+\infty)$.\end{lem}
\begin{proof}
It suffices to show that 
\begin{equation}
M\doteq\sup_{\mathcal{D}\backslash\{u_{2}\}\times\{v_{2}\}}\Bigg\{\big|\log\text{\textgreek{W}}^{2}\big|+\big|\log(\frac{r}{r_{0}})|+\Big|\log(\partial_{v}r)\Big|+\Big|\log(-\partial_{u}r)\Big|+\Big|\log\Big(1-\frac{2m}{r}\Big)\Big|+\sqrt{-\Lambda}|\tilde{m}|+r^{2}T_{vv}+r^{2}T_{uu}\Bigg\}<+\infty.\label{eq:BoundToShowContinuation}
\end{equation}
Provided (\ref{eq:BoundToShowContinuation}) has been established,
in view of the fact that $(r,\text{\textgreek{W}}^{2},\bar{f}_{in},\bar{f}_{out})$
are smooth on $\mathcal{D}\backslash\{u_{2}\}\times\{v_{2}\}$, choosing
some $\text{\textgreek{d}}>0$ small enough in tems of $M$, $r_{0}$
and $v_{2}-v_{1}$, we can smoothly extend the boundary-double characteristic
initial data set 
\[
(r_{/},\text{\textgreek{W}}_{/}^{2},\bar{f}_{in/},\bar{f}_{out/})^{\prime}\doteq(r,\text{\textgreek{W}}^{2},\bar{f}_{in},\bar{f}_{out})|_{u=u_{1}}
\]
and 
\[
(r_{\backslash},\text{\textgreek{W}}_{\backslash}^{2},\bar{f}_{in\backslash},\bar{f}_{out\backslash})^{\prime}\doteq(r,\text{\textgreek{W}}^{2},\bar{f}_{in},\bar{f}_{out})|_{v=v_{2}-\text{\textgreek{d}}}
\]
induced by $(r,\text{\textgreek{W}}^{2},\bar{f}_{in},\bar{f}_{out})$
on $\big([u_{1},u_{2}-\text{\textgreek{d}}]\times\{v_{2}-\text{\textgreek{d}}\}\big)\cup\big(\{u_{1}\}\times[v_{2}-\text{\textgreek{d}},v_{2}]\big)$
(see Definition \ref{def:TypeIII}), to a boundary-double characteristic
initial data set $(r_{/},\text{\textgreek{W}}_{/}^{2},\bar{f}_{in/},\bar{f}_{out/})$,
$(r_{\backslash},\text{\textgreek{W}}_{\backslash}^{2},\bar{f}_{in\backslash},\bar{f}_{out\backslash})$
 on $\big([u_{1},u_{2}-\text{\textgreek{d}}]\times\{v_{2}-\text{\textgreek{d}}\}\big)\cup\big(\{u_{1}\}\times[v_{2}-\text{\textgreek{d}},v_{2}+\text{\textgreek{d}}]\big)$,
satisfying 
\begin{align}
\max_{u\in[u_{1},u_{2}-\text{\textgreek{d}}]} & \Bigg\{\Big|\log\big(\frac{\text{\textgreek{W}}_{\backslash}^{2}}{1-\frac{1}{3}\Lambda r_{\backslash}^{2}}\big)\Big|+\Big|\log\Big(\frac{-2\partial_{u}r_{\backslash}}{1-\frac{2m_{\backslash}}{r_{\backslash}}}\Big)\Big|+\Big|\log\Big(\frac{1-\frac{2m_{\backslash}}{r_{\backslash}}}{1-\frac{1}{3}\Lambda r_{\backslash}^{2}}\Big)\Big|\Bigg\}(u)+\int_{u_{1}}^{u_{2}-\text{\textgreek{d}}}r_{\backslash}(T_{uu})_{\backslash}\, d\bar{u}+\label{eq:UpperBoundInitialData-1-1}\\
 & +\max_{v\in[v_{2}-\text{\textgreek{d}},v_{2}+\text{\textgreek{d}}]}\Bigg\{\Big|\log\big(\frac{\text{\textgreek{W}}_{/}^{2}}{1-\frac{1}{3}\Lambda r_{/}^{2}}\big)\Big|+\Big|\log\Big(\frac{2\partial_{v}r_{/}}{1-\frac{2m_{/}}{r_{/}}}\Big)\Big|+\Big|\log\Big(\frac{1-\frac{2m_{/}}{r_{/}}}{1-\frac{1}{3}\Lambda r_{/}^{2}}\Big)\Big|+\sqrt{-\Lambda}|\tilde{m}_{/}|\Bigg\}(v)\le2M,\nonumber 
\end{align}
\begin{equation}
(v_{2}+\text{\textgreek{d}})-(v_{2}-\text{\textgreek{d}})<\frac{u_{in}(2e^{-C_{0}^{2}((-\Lambda)(u_{2}-u_{1})^{2}+1)M}M)}{e^{C_{0}^{2}((-\Lambda)(u_{2}-u_{1})^{2}+1)}}\label{eq:U0UpperBound-1-2}
\end{equation}
and 
\begin{equation}
\int_{v_{2}-\text{\textgreek{d}}}^{v_{2}+\text{\textgreek{d}}}r_{/}(T_{vv})_{/}\, d\bar{v}<2e^{-C_{0}^{2}((-\Lambda)(u_{2}-u_{1})^{2}+1)M}M.
\end{equation}
Therefore, by applying Proposition \ref{Prop:LocalExistenceTypeIII}
for the boundary-double characteristic initial data set $(r_{/},\text{\textgreek{W}}_{/}^{2},\bar{f}_{in/},\bar{f}_{out/})$,
$(r_{\backslash},\text{\textgreek{W}}_{\backslash}^{2},\bar{f}_{in\backslash},\bar{f}_{out\backslash})$
 on $\big([u_{1},u_{2}-\text{\textgreek{d}}]\times\{v_{2}-\text{\textgreek{d}}\}\big)\cup\big(\{u_{1}\}\times[v_{2}-\text{\textgreek{d}},v_{2}+\text{\textgreek{d}}]\big)$,
we readily infer that this initial data set admits a smooth development
on $\{u_{1}<u<u_{2}-\text{\textgreek{d}}\}\cap\{u\le v\}\cap\{v_{2}-\text{\textgreek{d}}<v<v_{2}+\text{\textgreek{d}}\}$,
which coincides with $(r,\text{\textgreek{W}}^{2},\bar{f}_{in},\bar{f}_{out})$
on $\mathcal{D}\backslash\{u_{2}\}\times\{v_{2}\}$. This fact then
implies the statement of Lemma \ref{lem:cotinuationCriterionMirror}.

Let us set 
\begin{equation}
\text{\textgreek{l}}_{0}=\inf_{\text{\textgreek{g}}_{0;u_{1}u_{2}}}\partial_{v}r.\label{eq:LowerBoundDvRAxis}
\end{equation}
Note that $\text{\textgreek{l}}_{0}>0$, in view of (\ref{eq:NonNullMirrorContinuation}). 

By integrating equation (\ref{eq:RequationFinal}) in $v$ and using
(\ref{eq:ConditionRMirrorContinuation}), (\ref{eq:gaugeconditionMirrorContinuation})
and (\ref{eq:LowerBoundDvRAxis}), we readily obtain that 
\begin{equation}
r\partial_{u}r\le-r_{0}\text{\textgreek{l}}_{0}<0\label{eq:NegativeDuDerivativeContinuation}
\end{equation}
on $\mathcal{D}\backslash\{u_{2}\}\times\{v_{2}\}$. Therefore, (\ref{eq:NegativeDuDerivativeContinuation})
and (\ref{eq:ConditionRMirrorContinuation}) imply that 
\begin{equation}
r_{0}\le r\le\max_{\{u_{1}\}\times[v_{1},v_{2}]}r\doteq r_{+}<+\infty\label{eq:BoundsForRContinuation}
\end{equation}
on $\mathcal{D}\backslash\{u_{2}\}\times\{v_{2}\}$. Furthermore,
integrating (\ref{eq:RequationFinal}) in $u$ and using (\ref{eq:ConditionRMirrorContinuation})
and (\ref{eq:LowerBoundDvRAxis}), we obtain for all points in $\mathcal{D}\backslash\{u_{2}\}\times\{v_{2}\}$:
\begin{equation}
0<r_{0}\text{\textgreek{l}}_{0}\le r\partial_{v}r\le\max_{\{u_{1}\}\times[v_{1},v_{2}]}r\partial_{v}r.\label{eq:LowerBoundDvRContinuation}
\end{equation}

In view of the (\ref{eq:RelationHawkingMass}), the fact that $\text{\textgreek{W}}^{2}>0$
on $\mathcal{D}\backslash\{u_{2}\}\times\{v_{2}\}$ combined with
(\ref{eq:NegativeDuDerivativeContinuation}), (\ref{eq:BoundsForRContinuation})
and (\ref{eq:LowerBoundDvRContinuation}), implies that 
\begin{equation}
1-\frac{2m}{r}>0\label{eq:NonTrappingContinuation}
\end{equation}
everywhere on $\mathcal{D}\backslash\{u_{2}\}\times\{v_{2}\}$. Setting
\begin{equation}
\text{\textgreek{d}}_{0}\doteq\min_{\{u_{1}\}\times[v_{1},v_{2}]}\big(1-\frac{2m}{r}\big),\label{eq:LowerBoundTrappingContinuation}
\end{equation}
inequality (\ref{eq:NonTrappingContinuation}) impies that $\text{\textgreek{d}}_{0}>0$.
In view of (\ref{eq:NegativeDuDerivativeContinuation}) and the fact
that $T_{uu}\ge0$, the relation (\ref{eq:DerivativeInVDirectionKappaBar})
(which is well defined on $\mathcal{D}\backslash\{u_{2}\}\times\{v_{2}\}$
in view of (\ref{eq:LowerBoundDvRContinuation}) and (\ref{eq:NonTrappingContinuation}))
yields 
\begin{equation}
\partial_{u}\log\big(\frac{\partial_{v}r}{1-\frac{2m}{r}}\big)\le0.\label{eq:NegativeDerivatineKapaContinuation}
\end{equation}
Integrating (\ref{eq:NegativeDerivatineKapaContinuation}) in $u$
starting from $u=u_{1}$, we obtain for all points in $\mathcal{D}\backslash\{u_{2}\}\times\{v_{2}\}$:
\begin{equation}
\frac{\partial_{v}r}{1-\frac{2m}{r}}\le\max_{\{u_{1}\}\times[v_{1},v_{2}]}\frac{\partial_{v}r}{1-\frac{2m}{r}}\doteq\text{\textgreek{k}}_{0}<+\infty.\label{eq:UpperBoundKappaContinuaton}
\end{equation}
The bounds (\ref{eq:BoundsForRContinuation}), (\ref{eq:LowerBoundDvRContinuation}),
(\ref{eq:LowerBoundTrappingContinuation}) and (\ref{eq:NegativeDerivatineKapaContinuation})
yield for all points in $\mathcal{D}\backslash\{u_{2}\}\times\{v_{2}\}$:
\begin{equation}
1-\frac{2m}{r}\ge\text{\textgreek{k}}_{0}^{-1}\text{\textgreek{l}}_{0}r_{0}r_{+}^{-1}>0.\label{eq:LowerBoundTrapping}
\end{equation}

Integrating (\ref{eq:ConservationT_vv}) in $u$, we obtain for all
points in $\mathcal{D}\backslash\{u_{2}\}\times\{v_{2}\}$: 
\begin{equation}
r^{2}T_{vv}\le\max_{\{u_{1}\}\times[v_{1},v_{2}]}(r^{2}T_{vv})<+\infty.\label{eq:UpperBoundTvvContinuation}
\end{equation}
Since $\bar{f}_{in},\bar{f}_{out}$ satisfy the boundary condition
(\ref{eq:LeftBoundaryCondition}) on $\text{\textgreek{g}}_{0;u_{1}u_{2}}$
with $w=r-r_{0}$, $T_{uu}|_{\text{\textgreek{g}}_{0;u_{1}u_{2}}}$
and $T_{vv}|_{\text{\textgreek{g}}_{0;u_{1}u_{2}}}$ are related by
(\ref{eq:LeftBoundaryConditionT}), i.\,e. (in view of (\ref{eq:ConditionRMirrorContinuation})
and (\ref{eq:gaugeconditionMirrorContinuation})): 
\begin{equation}
\frac{T_{uu}}{T_{vv}}\Big|_{\text{\textgreek{g}}_{0;u_{1}u_{2}}}=1.\label{eq:ReflectionTContinuation}
\end{equation}
Therefore, integrating (\ref{eq:ConservationT_uu}) in $v$ and using
(\ref{eq:ReflectionTContinuation}), (\ref{eq:NonNullMirrorContinuation})
and (\ref{eq:UpperBoundTvvContinuation}), we infer that, for all
points in $\mathcal{D}\backslash\{u_{2}\}\times\{v_{2}\}$: 
\begin{equation}
r^{2}T_{uu}\le\max_{\{u_{1}\}\times[v_{1},v_{2}]}(r^{2}T_{vv})<+\infty.\label{eq:UpperBoundTuuContinuation}
\end{equation}
Integrating (\ref{eq:DerivativeInVDirectionKappaBar}) in $u$ starting
from $u=u_{1}$ and using (\ref{eq:NegativeDuDerivativeContinuation}),
(\ref{eq:UpperBoundTuuContinuation}) and (\ref{eq:BoundsForRContinuation}),
we obtain: 
\begin{equation}
\sup_{\mathcal{D}\backslash\{u_{2}\}\times\{v_{2}\}}\Big|\log\Big(\frac{\partial_{v}r}{1-\frac{2m}{r}}\Big)\Big|<+\infty.\label{eq:UpperBoundKappaContinuation}
\end{equation}

The relation (\ref{eq:DerivativeInVDirectionKappaBar}) is well defined
on $\mathcal{D}\backslash\{u_{2}\}\times\{v_{2}\}$, in view of (\ref{eq:NegativeDuDerivativeContinuation})
and (\ref{eq:NonTrappingContinuation}). Integrating (\ref{eq:DerivativeInVDirectionKappaBar})
in $v$ starting from $\text{\textgreek{g}}_{0;u_{1}u_{2}}$ and using
(\ref{eq:gaugeconditionMirrorContinuation}), (\ref{eq:BoundsForRContinuation}),
(\ref{eq:LowerBoundDvRContinuation}), (\ref{eq:UpperBoundKappaContinuation})
and (\ref{eq:UpperBoundTvvContinuation}), we infer that 
\begin{equation}
\sup_{\mathcal{D}\backslash\{u_{2}\}\times\{v_{2}\}}\Big|\log\Big(\frac{-\partial_{u}r}{1-\frac{2m}{r}}\Big)\Big|<\infty.\label{eq:BoundKappaBarContinuatio}
\end{equation}
Thus, (\ref{eq:RelationHawkingMass}), (\ref{eq:LowerBoundDvRContinuation})
and (\ref{eq:BoundKappaBarContinuatio}) yield: 
\begin{equation}
\sup_{\mathcal{D}\backslash\{u_{2}\}\times\{v_{2}\}}\Big|\log(\text{\textgreek{W}}^{2})\Big|<+\infty.\label{eq:BoundOmegaContinuation}
\end{equation}

Finally, by integrating equation 
\begin{equation}
\partial_{u}\tilde{m}=-2\pi\big(1-\frac{2\tilde{m}}{r}-\frac{1}{3}\Lambda r^{2}\big)\cdot\frac{r^{2}T_{uu}}{-\partial_{u}r}
\end{equation}
in $u$ starting from $u=u_{1}$ and using (\ref{eq:NegativeDuDerivativeContinuation}),
(\ref{eq:BoundsForRContinuation}) and (\ref{eq:UpperBoundTuuContinuation}),
we infer that: 
\begin{equation}
\sup_{\mathcal{D}\backslash\{u_{2}\}\times\{v_{2}\}}|\tilde{m}|<+\infty.\label{eq:AbsoluteBoundRenormalisdMassContinuation}
\end{equation}
The bound (\ref{eq:BoundToShowContinuation}) now readily follows
from (\ref{eq:BoundsForRContinuation}), (\ref{eq:LowerBoundDvRContinuation}),
(\ref{eq:NonTrappingContinuation}), (\ref{eq:UpperBoundTvvContinuation}),
(\ref{eq:UpperBoundTuuContinuation}), (\ref{eq:BoundKappaBarContinuatio}),
(\ref{eq:BoundOmegaContinuation}) and (\ref{eq:AbsoluteBoundRenormalisdMassContinuation}).
\end{proof}

\subsection{\label{sub:ProofOfProp}Proof of Theorem \ref{thm:maximalExtension}}

The construction of the maximal future development $(\mathcal{U};r,\text{\textgreek{W}}^{2},\bar{f}_{in},\bar{f}_{out})$
will be performed in two steps: In the first step, we will construct
the domain of outer communications $(J^{-}(\mathcal{I})\cap\mathcal{U};r,\text{\textgreek{W}}^{2},\bar{f}_{in},\bar{f}_{out})$
of the maximal future development, using Proposition \ref{Prop:LocalExistenceTypeII}
as the main tool. In the second step, we will construct the rest of
the maximal future development, i.\,e.~$(\mathcal{U}\backslash J^{-}(\mathcal{I});r,\text{\textgreek{W}}^{2},\bar{f}_{in},\bar{f}_{out})$,
using Propositions \ref{Prop:LocalExistenceTypeIII} and \ref{prop:LocalExistenceTypeI}
as the main tool. Notice that $\mathcal{U}\backslash J^{-}(\mathcal{I})$
can possibly be empty; in the proof, we will actually consider the
case $\mathcal{U}\backslash J^{-}(\mathcal{I})=\emptyset$ separately. 

The uniqueness and maximality of $(\mathcal{U};r,\text{\textgreek{W}}^{2},\bar{f}_{in},\bar{f}_{out})$
will follow readily from our construction, in conjunction with the
uniqueness statements of Propositions \ref{Prop:LocalExistenceTypeII},
\ref{Prop:LocalExistenceTypeIII} and \ref{prop:LocalExistenceTypeI}.
The properties 1--6 of $(\mathcal{U};r,\text{\textgreek{W}}^{2},\bar{f}_{in},\bar{f}_{out})$
stated in Theorem \ref{thm:maximalExtension} will also be established
during the construction of $(\mathcal{U};r,\text{\textgreek{W}}^{2},\bar{f}_{in},\bar{f}_{out})$.

In order to better keep track of the notations introduced throughout
the proof, the reader is advised to refer to Figure \ref{fig:pieces}.

\begin{figure}[h!] 
\centering 
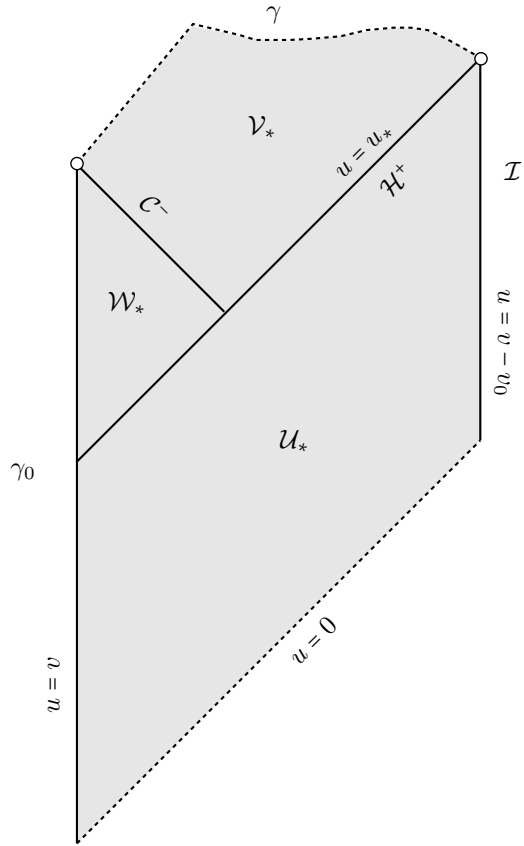 
\caption{The construction of the maximal development $(\mathcal{U};r,\Omega^2,\bar{f}_{in},\bar{f}_{out}$ of $(r_{\slash},\Omega^2_{\slash},\bar{f}_{in\slash},\bar{f}_{out\slash})$ will proceed in two steps: In the first step, we will construct the domain of outer communications $\mathcal{U}\backslash \mathcal{I}$, corresponding to the region $\mathcal{U}_{*}=\{0<u<u_{*}\}$ depicted above. In the case when $u_{*}<+\infty$ (which is the case depicted), the solution will have a non-empty future event horizon $\mathcal{H}^{+}=\{u=u_{*}\}$, and the second step of the construction will consist of constructing the part of the solution lying to the future of $\mathcal{H}^{+}$. In the figure, this corresponds to the domains $\mathcal{W}_{*}$ and $\mathcal{V}_{*}$. The construction of $\mathcal{U}$ will require the use of Proposition \ref{Prop:LocalExistenceTypeII} in the region $\mathcal{U}_{*}$, Proposition \ref{Prop:LocalExistenceTypeIII} in the region $\mathcal{W}_{*}$ and Proposition \ref{Prop:LocalExistenceTypeI} in the region $\mathcal{V}_{*}$.}
\label{fig:pieces}
\end{figure}

\subsubsection*{Step 1: Construction of $J^{-}(\mathcal{I})\cap\mathcal{U}$}

We will first construct the domain of outer communications $(J^{-}(\mathcal{I})\cap\mathcal{U};r,\text{\textgreek{W}}^{2},\bar{f}_{in},\bar{f}_{out})$
of the maximal future development $(\mathcal{U};r,\text{\textgreek{W}}^{2},\bar{f}_{in},\bar{f}_{out})$
of $(r_{/},\text{\textgreek{W}}_{/}^{2},\bar{f}_{in/},\bar{f}_{out/})$.

Let $\mathscr{U}_{\mathcal{I}}$ be the set of all developments $\mathscr{D}=(\mathcal{U}_{\mathscr{D}};r,\text{\textgreek{W}}^{2},\bar{f}_{in},\bar{f}_{out})$
of $(r_{/},\text{\textgreek{W}}_{/}^{2},\bar{f}_{in/},\bar{f}_{out/})$
(according to Definition \ref{def:Development}) such that, for some
$0<u_{\mathscr{D}}<+\infty$: 
\begin{equation}
\mathcal{U}_{\mathscr{D}}=\big\{0<u<u_{\mathscr{D}}\big\}\cap\big\{ u<v<u+v_{0}\big\}.\label{eq:DomainOfExtension}
\end{equation}
 In view of Proposition \ref{Prop:LocalExistenceTypeII}, $\mathscr{U}_{\mathcal{I}}\neq\emptyset$.
Furthermore, in view of the remark below Definition \ref{def:Development}
and the form (\ref{eq:DomainOfExtension}) of the domain of the developments
belonging to $\mathscr{U}_{\mathcal{I}}$, any two developments $\mathscr{D}_{1},\mathscr{D}_{2}\in\mathscr{U}_{\mathcal{I}}$
will satisfy $\mathscr{D}_{1}\subseteq\mathscr{D}_{2}$ or $\mathscr{D}_{2}\subseteq\mathscr{D}_{1}$.
Therefore, there exists a unique $0<u_{*}\le+\infty$ and a unique
development $\mathscr{D}_{*}=(\mathcal{U}_{*};r,\text{\textgreek{W}}^{2},\bar{f}_{in},\bar{f}_{out})$
of $(r_{/},\text{\textgreek{W}}_{/}^{2},\bar{f}_{in/},\bar{f}_{out/})$,
where 
\begin{equation}
\mathcal{U}_{*}=\big\{0<u<u_{*}\big\}\cap\big\{ u<v<u+v_{0}\big\},\label{eq:DomainOfExtension-1}
\end{equation}
such that any $\mathscr{D}\in\mathscr{U}_{\mathcal{I}}$ satisfies
$\mathscr{D}\subseteq\mathscr{D}_{*}$. Let us set 
\begin{equation}
\mathcal{I}=\{u=v-v_{0}\}\cap\{0\le u<u_{*}\}
\end{equation}
and 
\begin{equation}
\text{\textgreek{g}}_{0}^{*}=\{u=v\}\cap\{0\le u<u_{*}\}.
\end{equation}

We will now establish that $\mathscr{D}_{*}$ has the property that,
for any $u^{\prime}<u_{*}$:

\begin{equation}
\sup_{\mathcal{U}_{*}\cap\{u\le u^{\prime}\}}\max\Bigg\{\Big|\log\big(\frac{\text{\textgreek{W}}^{2}}{1-\frac{1}{3}\Lambda r^{2}}\big)\Big|,\Big|\log\Big(\frac{2\partial_{v}r}{1-\frac{2m}{r}}\Big)\Big|,\Big|\log\Big(\frac{1-\frac{2m}{r}}{1-\frac{1}{3}\Lambda r^{2}}\Big)\Big|\Big|,\sqrt{-\Lambda}|\tilde{m}|,r^{2}T_{vv}\Bigg\}<+\infty\label{eq:BoundednessOfQuantitiesOutsideTheEventHorizon}
\end{equation}
and 

\begin{itemize}

\item{$u_{*}=+\infty$ or}

\item{$u_{*}<+\infty$ and 
\begin{equation}
\sup_{\mathcal{U}_{*}}\max\Bigg\{\Big|\log\big(\frac{\text{\textgreek{W}}^{2}}{1-\frac{1}{3}\Lambda r^{2}}\big)\Big|,\Big|\log\Big(\frac{2\partial_{v}r}{1-\frac{2m}{r}}\Big)\Big|,\Big|\log\Big(\frac{1-\frac{2m}{r}}{1-\frac{1}{3}\Lambda r^{2}}\Big)\Big|\Big|,\sqrt{-\Lambda}|\tilde{m}|,r^{2}T_{vv}\Bigg\}=+\infty.\label{eq:InfinityForMaximalDomainOfOuterCommunications}
\end{equation}
}

\end{itemize}

\noindent This can be inferred as follows: Let $\mathscr{D}=(\mathcal{U}_{\mathscr{D}};r,\text{\textgreek{W}}^{2},\bar{f}_{in},\bar{f}_{out})$
be any development of $(r_{/},\text{\textgreek{W}}_{/}^{2},\bar{f}_{in/},\bar{f}_{out/})$
belonging to the set $\mathscr{U}_{\mathcal{I}}$. Provided that 
\begin{equation}
M\doteq\sup_{\mathcal{U}_{\mathscr{D}}}\max\Bigg\{\Big|\log\big(\frac{\text{\textgreek{W}}^{2}}{1-\frac{1}{3}\Lambda r^{2}}\big)\Big|,\Big|\log\Big(\frac{2\partial_{v}r}{1-\frac{2m}{r}}\Big)\Big|,\Big|\log\Big(\frac{1-\frac{2m}{r}}{1-\frac{1}{3}\Lambda r^{2}}\Big)\Big|\Big|,\sqrt{-\Lambda}|\tilde{m}|,r^{2}T_{vv}\Bigg\}<+\infty,\label{eq:ForContinuingInitialDataInDomain}
\end{equation}
by applying Proposition \ref{Prop:LocalExistenceTypeII} for the initial
data induced by $(r,\text{\textgreek{W}}^{2},\bar{f}_{in},\bar{f}_{out})$
on $\{u=u_{\mathscr{D}}-u^{\prime}\}\times\{u\le v\le v_{0}+u\}$
for some $u^{\prime}$ small enough in terms of $r_{0},v_{0}$ and
$M$, we infer that there exists some $\mathscr{D}^{\prime}\in\mathscr{U}_{\mathcal{I}}$
strictly extending $\mathscr{D}$, i.\,e.~$\mathscr{D}\subseteq\mathscr{D}^{\prime}$
and $\mathscr{D}\neq\mathscr{D}^{\prime}$. Therefore, in view of
the inextendibility of $\mathscr{D}_{*}$ in $\mathscr{U}_{\mathcal{I}}$,
either $u_{*}=+\infty$, or (\ref{eq:InfinityForMaximalDomainOfOuterCommunications})
holds. Moreover, it can be readily verified that (\ref{eq:ForContinuingInitialDataInDomain})
always holds if $\mathscr{D}$ has a strict extension $\mathscr{D}^{\prime}$
in $\mathscr{U}_{\mathcal{I}}$. Therefore, (\ref{eq:BoundednessOfQuantitiesOutsideTheEventHorizon})
holds.

\paragraph*{Some basic estimates on $\mathcal{I}$ and $\mathcal{U}_{*}$.}

It follows readily from the proof of Proposition \ref{Prop:LocalExistenceTypeII}
that the quantities (\ref{eq:RhoVariable})--(\ref{eq:TVariable})
and $\tilde{m}$ extend smoothly on $\mathcal{I}$. The relations
(\ref{eq:DerivativeTildeUMass})--(\ref{eq:DerivativeTildeVMass})
and the conditions (\ref{eq:MirrorRMaximal}), (\ref{eq:GaugeMirrorMaximal}),
(\ref{eq:InfinityRMaximal}) and (\ref{eq:GaugeInfinityMaximal})
imply that $\tilde{m}$ is constant on $\text{\textgreek{g}}_{0}^{*}$
and $\mathcal{I}$, i.\,e. 
\begin{equation}
\tilde{m}|_{\text{\textgreek{g}}_{0}^{*}}=\tilde{m}_{/}(0)\mbox{ and }\tilde{m}|_{\mathcal{I}}=\lim_{v\rightarrow v_{0}^{-}}\tilde{m}_{/}(v).\label{eq:ConservationTildeM}
\end{equation}

The relations (\ref{eq:DerivativeTildeUMass})--(\ref{eq:DerivativeTildeVMass})
and the bound (\ref{eq:BoundednessOfQuantitiesOutsideTheEventHorizon})
imply that 
\begin{equation}
\partial_{u}\tilde{m}\le0\mbox{ and }\partial_{v}\tilde{m}\ge0\mbox{ on }\mathcal{U}_{*}
\end{equation}
 and, hence 
\begin{equation}
\tilde{m}_{/}(0)\le\tilde{m}\le\lim_{v\rightarrow v_{0}^{-}}\tilde{m}_{/}(v)\mbox{ on }\mathcal{U}_{*}.\label{eq:MassBounds}
\end{equation}
Moreover, the relations (\ref{eq:RelationHawkingMass}) and (\ref{eq:RenormalisedHawkingMass})
imply, in view of the fact that (\ref{eq:KappaVariable})--(\ref{eq:KappaBarVariable})
extend smoothly on $\mathcal{I}$, that the quantity 
\begin{equation}
\text{\textgreek{w}}\doteq\sqrt{\frac{\text{\textgreek{W}}^{2}}{1-\frac{1}{3}\Lambda r^{2}}}
\end{equation}
 extends smoothly on $\mathcal{I}$, with 
\begin{equation}
\text{\textgreek{w}}^{2}=4\Big(\frac{\partial_{v}r}{1-\frac{2m}{r}}\Big)^{2}\Big|_{\mathcal{I}}\label{eq:RelationConformalFactorInfinity}
\end{equation}
in view of (\ref{eq:RelationHawkingMass}), (\ref{eq:RenormalisedHawkingMass})
and (\ref{eq:GaugeInfinityMaximal}).

\subsubsection*{End of the proof in the case $u_{_{*}}=+\infty$}

In the case $u_{*}=+\infty$, we will set 
\begin{equation}
\mathcal{U}=\mathcal{U}_{*}.
\end{equation}
Note that, in this case, we necessarily have 
\begin{equation}
\mathcal{U}\backslash J(\mathcal{I})=\emptyset
\end{equation}
and, thus, 
\begin{equation}
\mathcal{H}^{+}=\emptyset
\end{equation}
 and 
\begin{equation}
\text{\textgreek{g}}_{0}=\text{\textgreek{g}}_{0}^{*}.
\end{equation}

In order to complete the proof of Theorem \ref{thm:maximalExtension}
in this case, it remains to establish (\ref{eq:NegativeDerivativeRMaximal})--(\ref{eq:D_vRPositiveMaximal}),
(\ref{eq:DerivativeTildeUMass})--(\ref{eq:DerivativeTildeVMass})
and (\ref{eq:CompleConformalInfinity}).

The bounds (\ref{eq:NegativeDerivativeRMaximal})--(\ref{eq:D_vRPositiveMaximal})
follow readily from (\ref{eq:BoundednessOfQuantitiesOutsideTheEventHorizon}),
while (\ref{eq:DerivativeTildeUMass})--(\ref{eq:DerivativeTildeVMass})
follow immediately from (\ref{eq:ConservationTildeM}) and the fact
that $\text{\textgreek{g}}_{0}=\text{\textgreek{g}}_{0}^{*}$.

The proof of (\ref{eq:CompleConformalInfinity}) will follow by showing
that 
\begin{equation}
\int_{\mathcal{I}}\text{\textgreek{w}}\, du=\lim_{u\rightarrow u_{*}}\int_{0}^{u}\text{\textgreek{w}}(\bar{u},\bar{u}+v_{0})\, d\bar{u}=+\infty.\label{eq:InfiniteConformalLengthToShow}
\end{equation}
In view of (\ref{eq:RelationConformalFactorInfinity}), it suffices
to show that, for any $u\ge v_{0}$: 
\begin{equation}
\int_{u-v_{0}}^{u}\Big(\frac{\partial_{v}r}{1-\frac{2m}{r}}\Big)\Big|_{\mathcal{I}}(\bar{u},\bar{u}+v_{0})\, d\bar{u}\ge c_{1}>0\label{eq:InequalityForDerivativeRAtInfinity}
\end{equation}
 for some absolute constant $c_{1}$ depending only on the initial
data $(r_{/},\text{\textgreek{W}}_{/}^{2},\bar{f}_{in/},\bar{f}_{out/})$. 

The lower bound (\ref{eq:InequalityForDerivativeRAtInfinity}) is
deduced as follows: Inequality (\ref{eq:DerivativeInUDirectionKappa})
and the bound (\ref{eq:BoundednessOfQuantitiesOutsideTheEventHorizon})
imply that 
\begin{equation}
\partial_{u}\Big(\frac{\partial_{v}r}{1-\frac{2m}{r}}\Big)\le0\label{eq:InequalityDerivativeKappa}
\end{equation}
on $\mathcal{U}_{*}$ and, thus, for any $u\ge v_{0}$ and $u\le v\le u+v_{0}$:
\begin{equation}
\Big(\frac{\partial_{v}r}{1-\frac{2m}{r}}\Big)\Big|_{\mathcal{I}}(v-v_{0},v)\ge\frac{\partial_{v}r}{1-\frac{2m}{r}}(u,v).\label{eq:InequalityKappa}
\end{equation}
 For any $u\ge0$, we compute: 
\begin{align}
\sqrt{-\frac{\Lambda}{3}}\int_{u}^{u+v_{0}}\frac{\partial_{v}r}{1-\frac{1}{3}\Lambda r^{2}}(u,v)\, dv & =\tan^{-1}(\sqrt{-\frac{\Lambda}{3}}r)\Big|_{\mathcal{I}}-\tan^{-1}(\sqrt{-\frac{\Lambda}{3}}r)\Big|_{\text{\textgreek{g}}_{0}}\label{eq:LowerBoundRDifferenceeMaximal}\\
 & =\frac{\pi}{2}-\tan^{-1}(\sqrt{-\frac{\Lambda}{3}}r_{0}).\nonumber 
\end{align}
Therefore, (\ref{eq:InequalityKappa}), (\ref{eq:LowerBoundRDifferenceeMaximal})
and (\ref{eq:MassBounds}) readily yield that, for any $u\ge0$ 
\begin{align}
\int_{u-v_{0}}^{u}\Big(\frac{\partial_{v}r}{1-\frac{2m}{r}}\Big)\Big|_{\mathcal{I}}(\bar{u},\bar{u}+v_{0})\, d\bar{u} & =\int_{u}^{u+v_{0}}\Big(\frac{\partial_{v}r}{1-\frac{2m}{r}}\Big)\Big|_{\mathcal{I}}(v-v_{0},v)\, dv\ge\\
 & \ge\int_{u}^{u+v_{0}}\Big(\frac{\partial_{v}r}{1-\frac{2m}{r}}\Big)(u,v)\, dv\ge\nonumber \\
 & \ge\sqrt{-\frac{3}{\Lambda}}\Big(1+\frac{2\max\{0,-\tilde{m}_{/}(0)\}}{r_{0}}\Big)^{-1}\cdot\Big(\frac{\pi}{2}-\tan^{-1}(\sqrt{-\frac{\Lambda}{3}}r_{0})\Big)\nonumber 
\end{align}
and, thus, (\ref{eq:InequalityForDerivativeRAtInfinity}) holds.

\subsubsection*{The case $u_{*}<+\infty$}

For the rest of the proof, we will assume without loss of generality
that 
\begin{equation}
u_{*}<+\infty.\label{eq:U_*<infty}
\end{equation}
We will set 
\begin{equation}
\mathcal{H}^{+}\doteq\{u=u_{*}\}\cap\{u_{*}\le v<u_{*}+v\}.
\end{equation}
We will show that $(r,\text{\textgreek{W}}^{2},\bar{f}_{in},\bar{f}_{out})$
extend smoothly beyond $\mathcal{H}^{+}$ and, therefore, in the case
(\ref{eq:U_*<infty}), we have 
\begin{equation}
\mathcal{U}\backslash J^{-}(\mathcal{I})\neq\emptyset\label{eq:NonEmptyBlackHole}
\end{equation}
 and $\mathcal{H}^{+}$ wll actually be the future evnt horizon defined
by (\ref{eq:DefinitionHorizon}).

Notice that, equations (\ref{eq:ConservationT_vv})--(\ref{eq:ConservationT_uu})
and the reflecting boundary conditions 
\begin{equation}
\frac{r^{2}T_{uu}}{r^{2}T_{vv}}\Big|_{\text{\textgreek{g}}_{0}^{*}}=1\mbox{ and }\frac{r^{2}T_{uu}}{r^{2}T_{vv}}\Big|_{\mathcal{I}}=1\label{eq:ReflectiveConditionsTMaximal}
\end{equation}
 imply that 
\begin{equation}
r^{2}T_{uu},r^{2}T_{vv}\le\sup_{[0,v_{0})}r_{/}^{2}(T_{vv})_{/}\label{eq:BoundEnergyMomentumTensorMaximal}
\end{equation}
and, thus, in view of (\ref{eq:RelationHawkingMass}), (\ref{eq:BoundednessOfQuantitiesOutsideTheEventHorizon}),
(\ref{eq:MassBounds}) and (\ref{eq:BoundEnergyMomentumTensorMaximal}),
the condition (\ref{eq:InfinityForMaximalDomainOfOuterCommunications})
is equivalent to 
\begin{equation}
\limsup_{\bar{u}\rightarrow u_{*}}\Bigg\{\sup_{u=\bar{u}}\Bigg(\Big|\log\Big(\frac{\partial_{v}r}{1-\frac{2m}{r}}\Big)\Big|+\Big|\log\Big(\frac{-\partial_{u}r}{1-\frac{2m}{r}}\Big)\Big|+\Big(1-\frac{2m}{r}\Big)^{-1}\Bigg)\Bigg\}=+\infty.\label{eq:InfinityForMaximalDomainOfOuterCommunicationsAlt}
\end{equation}

\paragraph*{Smooth extension across $\mathcal{H}^{+}$.}

We will now show that $(r,\log\text{\textgreek{W}}^{2})$ extend smoothly
on the whole of $\mathcal{H}^{+}$, and $(\bar{f}_{in},\bar{f}_{out})$
extends smoothly on the whole of $\mathcal{H}^{+}\times(0,+\infty)$
and, thus, (\ref{eq:NonEmptyBlackHole}) holds. 

Assuming, first, that $(r,\log\text{\textgreek{W}}^{2})$ extend smoothly
on $\{u_{*}\}\times\{u_{*}\}$ and $(\bar{f}_{in},\bar{f}_{out})$
extend smoothly on $\{u_{*}\}\times\{u_{*}\}\times(0,+\infty)$, we
can readily deduce that $(r,\log\text{\textgreek{W}}^{2})$ extends
smoothly on the whole of $\mathcal{H}^{+}$, and $(\bar{f}_{in},\bar{f}_{out})$
extends smoothly on the whole of $\mathcal{H}^{+}\times(0,+\infty)$,
by applying Lemma \ref{lem:cotinuationCriterion} for $u_{1}=u_{*}-\text{\textgreek{d}}$,
$u_{2}=u_{*}$, $v_{1}=u_{*}$, $v_{2}=u_{*}+v_{0}-2\text{\textgreek{d}}$
for any $0<\text{\textgreek{d}}\ll1$. 

Thus, it remains to show that $(r,\log\text{\textgreek{W}}^{2})$
extend smoothly on the point $\{u_{*}\}\times\{u_{*}\}$, and $(\bar{f}_{in},\bar{f}_{out})$
extend smoothly on $\{u_{*}\}\times\{u_{*}\}\times(0,+\infty)$. Provided
\begin{equation}
\limsup_{u\rightarrow u_{*}^{-}}\partial_{v}r(u,u)>0,\label{eq:LowerBoundDvROnAxis}
\end{equation}
the smooth extension of $(r,\log\text{\textgreek{W}}^{2})$ on $\{u_{*}\}\times\{u_{*}\}$
and $(\bar{f}_{in},\bar{f}_{out})$ on $\{u_{*}\}\times\{u_{*}\}\times(0,+\infty)$
follows readily from the continuation criterion of Lemma \ref{lem:cotinuationCriterionMirror}
with $u_{1}=v_{1}=0$ and $u_{2}=v_{2}=v_{0}$. Hence, it suffices
to establish (\ref{eq:LowerBoundDvROnAxis}). 

\medskip{}

\noindent \emph{Proof of (\ref{eq:LowerBoundDvROnAxis}).} Assume,
for the sake of contradiction, that (\ref{eq:LowerBoundDvROnAxis})
is false, i.\,e.
\begin{equation}
\limsup_{u\rightarrow u_{*}^{-}}\partial_{v}r(u,u)=0.\label{eq:ContradictionForTrappingAxis}
\end{equation}

In view of (\ref{eq:ConstrainVFinal}) and (\ref{eq:BoundednessOfQuantitiesOutsideTheEventHorizon}),
we can bound 
\begin{equation}
\partial_{v}r>0\label{eq:PositiveDerivativeU_*}
\end{equation}
and 
\begin{equation}
\partial_{v}(\text{\textgreek{W}}^{-2}\partial_{v}r)=-4\pi rT_{vv}\text{\textgreek{W}}^{-2}\le0\label{eq:ConcaveR}
\end{equation}
everywhere on $\mathcal{U}_{*}.$ Thus, (\ref{eq:MirrorRMaximal}),
(\ref{eq:ContradictionForTrappingAxis}), (\ref{eq:PositiveDerivativeU_*})
and (\ref{eq:ConcaveR}) imply that $r$ extends continuously on $\mathcal{H}^{+}$
so that 
\begin{equation}
r|_{\mathcal{H}^{+}}=r_{0}
\end{equation}
 and, for any $u_{*}<v<u_{*}+u_{0}$: 
\begin{equation}
\lim_{u\rightarrow u_{*}}T_{vv}(u,v)=0.\label{eq:ZeroFluxFromTheHorizon}
\end{equation}

Equations (\ref{eq:ConservationT_vv})--(\ref{eq:ConservationT_uu})
and the reflecting boundary conditions (\ref{eq:ReflectiveConditionsTMaximal})
imply, in view of (\ref{eq:ZeroFluxFromTheHorizon}), that 
\begin{equation}
T_{uu}(u,v)=T_{vv}(u,v)=0\label{eq:VanishingEnergyMomentumTensor}
\end{equation}
for all points $(u,v)\in\mathcal{U}_{*}\backslash\text{\textgreek{g}}_{\vdash}^{-}(u_{*},u_{*})$,
where $\text{\textgreek{g}}_{\vdash}^{-}(u_{*},u_{*})$ is the past
directed null geodesic emanating from $(u_{*},u_{*})$, reflected
successively on $\mathcal{I}$ and $\text{\textgreek{g}}_{0}^{*}$,
i.\,e.~ 
\[
\text{\textgreek{g}}_{\vdash}^{-}(u_{*},u_{*})=\cup_{n\in\mathbb{N}}\text{\textgreek{g}}_{\vdash}^{-(n)}(u_{*},u_{*})
\]
with 
\begin{equation}
\begin{cases}
\text{\textgreek{g}}_{\vdash}^{-(2n)}(u_{*},u_{*})=\{v=u_{*}-nv_{0}\}\cap\{u_{*}-(n+1)v_{0}\le u\le u_{*}-nv_{0}\}, & n\ge0,\\
\text{\textgreek{g}}_{\vdash}^{-(2n+1)}(u_{*},u_{*})=\{u=u_{*}-(n+1)v_{0}\}\cap\{u_{*}-(n+1)v_{0}\le v\le u_{*}-nv_{0}\}, & n\ge0.
\end{cases}
\end{equation}
Since $T_{uu},T_{vv}$ are smooth on $\mathcal{U}_{*}$, we infer
that (\ref{eq:VanishingEnergyMomentumTensor}) holds on the whole
of $\mathcal{U}_{*}$. Thus, $\tilde{m}$ is constant on $\mathcal{U}_{*}$.
Since $\partial_{v}r|_{\mathcal{S}_{v_{0}}}>0$, we infer that 
\begin{equation}
1-\frac{2\overline{M}}{r_{0}}-\frac{1}{3}\Lambda r_{0}^{2}>0.\label{eq:SchwarzschildOutsideHorizon}
\end{equation}

Equations (\ref{eq:DerivativeInUDirectionKappa}) and (\ref{eq:DerivativeInVDirectionKappaBar}),
combined with (\ref{eq:GaugeMirrorMaximal}) and (\ref{eq:GaugeInfinityMaximal}),
imply in this case that 
\begin{equation}
\Big|\log\Big(\frac{\partial_{v}r}{1-\frac{2\overline{M}}{r}-\frac{1}{3}\Lambda r^{2}}\Big)\Big|,\Big|\log\Big(\frac{-\partial_{u}r}{1-\frac{2\overline{M}}{r}-\frac{1}{3}\Lambda r^{2}}\Big)\Big|\le\sup_{\{0\}\times[0,v_{0})}\Big|\log\Big(\frac{\partial_{v}r}{1-\frac{2\overline{M}}{r}-\frac{1}{3}\Lambda r^{2}}\Big)\Big|\label{eq:BoundednessKappa&KappaBarSchwarzschild}
\end{equation}
everywhere on $\mathcal{U}_{*}$. Therefore, (\ref{eq:SchwarzschildOutsideHorizon})
and (\ref{eq:BoundednessKappa&KappaBarSchwarzschild}) imply that
(\ref{eq:InfinityForMaximalDomainOfOuterCommunicationsAlt}) is false,
which is a contradiction. 

\medskip{}

\noindent \emph{Useful bounds on $\mathcal{H}^{+}$.} Notice that,
in view of (\ref{eq:BoundednessOfQuantitiesOutsideTheEventHorizon}):
\begin{equation}
\partial_{v}r|_{\mathcal{H}^{+}}\ge0.\label{eq:IncreasingRHorizon}
\end{equation}
In addition, (\ref{eq:ConcaveR}) implies that, for any $\bar{v}\in[u_{*},u_{*}+v_{0}]$:
\begin{equation}
\partial_{v}r(u_{*},\bar{v})=0\mbox{ }\Rightarrow\mbox{ }\partial_{v}r(u_{*},v)=0\mbox{ for all}v\ge\bar{v}.\label{eq:D_vRCan'tvanish}
\end{equation}
Integrating (\ref{eq:RequationFinal}) in $v$ along $u=\bar{u}$,
$\bar{u}\le u_{*}$ and using (\ref{eq:GaugeMirrorMaximal}), we also
infer that 
\begin{equation}
\inf_{\mathcal{U}_{*}\cup\mathcal{H}^{+}}(-\partial_{u}r)>0.\label{eq:LowerBoundDuRHorizon}
\end{equation}

\paragraph*{Proof of (\ref{eq:UpperBoundRHorizon}) and (\ref{eq:TrappingAsymptoticallyHorizon}).}

In order to establish (\ref{eq:UpperBoundRHorizon}), we will first
establish 
\begin{equation}
\sup_{\mathcal{H}^{+}}r\le r_{S}\label{eq:UpperBoundRHorizonGeneral}
\end{equation}
and (\ref{eq:TrappingAsymptoticallyHorizon}), and then show that
\begin{equation}
\sup_{\mathcal{H}^{+}}r\ge r_{S}.\label{eq:LowerBoundRHorizonSpecialForMatterField}
\end{equation}

\medskip{}

\noindent \emph{Proof of (\ref{eq:UpperBoundRHorizon}) and (\ref{eq:TrappingAsymptoticallyHorizon}).}
The upper bound (\ref{eq:UpperBoundRHorizon}) follows by a contradiction
argument: Assuming that (\ref{eq:UpperBoundRHorizonGeneral}) is false,
in view of (\ref{eq:IncreasingRHorizon}) and (\ref{eq:LowerBoundDuRHorizon})
we infer that there exists some (possibly small) $\text{\textgreek{d}}>0$
such that 
\begin{equation}
\inf_{\mathcal{V}_{\text{\textgreek{d}}}}r\ge(1+\text{\textgreek{d}})r_{S},\label{eq:BigRAwayOnTheHorizon}
\end{equation}
where 
\begin{equation}
\mathcal{V}_{\text{\textgreek{d}}}\doteq\{u_{*}-\text{\textgreek{d}}\le u\le u_{*}\}\cap\{u_{*}+v_{0}-\text{\textgreek{d}}\le v<u+v_{0}\}.
\end{equation}
In view of (\ref{eq:MassBounds}) and (\ref{eq:DefinitionRs}), the
lower bound (\ref{eq:BigRAwayOnTheHorizon}) implies that 
\begin{equation}
C_{tr}\doteq\sup_{\mathcal{V}_{\text{\textgreek{d}}}}\Big\{\big(1-\frac{2m}{r}\big)^{-1}\Big\}<+\infty.\label{eq:NonTrappingAwayFromHorizon}
\end{equation}

By integrating (\ref{eq:DerivativeInUDirectionKappa}) in $u$ and
(\ref{eq:DerivativeInVDirectionKappaBar}) in $v$ and using condition
(\ref{eq:GaugeInfinityMaximal}) on $\mathcal{I}$, we infer that,
for any $u\in[u_{*}-\text{\textgreek{d}},u_{*})$: 
\begin{equation}
\log\Big(\frac{\partial_{v}r}{1-\frac{2m}{r}}\Big)(u_{*},u+v_{0})=\log\Big(\frac{-\partial_{u}r}{1-\frac{2m}{r}}\Big)(u,u_{*}+v_{0}-\text{\textgreek{d}})+4\pi\int_{u_{*}+v_{0}-\text{\textgreek{d}}}^{u+v_{0}}r\frac{T_{vv}}{\partial_{v}r}(u,\bar{v})\, d\bar{v}-4\pi\int_{u}^{u_{*}}r\frac{T_{uu}}{-\partial_{u}r}(\bar{u},u+v_{0})\, d\bar{u}.\label{eq:FromIntegrationDvDur}
\end{equation}
In view of (\ref{eq:DerivativeTildeUMass}), (\ref{eq:DerivativeTildeVMass}),
(\ref{eq:PositiveDerivativeU_*}) and (\ref{eq:NonTrappingAwayFromHorizon}),
we can estimate: 
\begin{align}
\Big|4\pi\int_{u_{*}+v_{0}-\text{\textgreek{d}}}^{u+v_{0}}r\frac{T_{vv}}{\partial_{v}r}(u,\bar{v})\, d\bar{v}-4\pi & \int_{u}^{u_{*}}r\frac{T_{uu}}{-\partial_{u}r}(\bar{u},u+v_{0})\, d\bar{u}\Big|\le\label{eq:BoundFrmTheMassAwayFromTrapping}\\
 & \le2C_{tr}r_{0}^{-1}\big(\tilde{m}(u,u+v_{0})-\tilde{m}(u,u_{*}+v_{0}-\text{\textgreek{d}})+\tilde{m}(u,u+v_{0})-\tilde{m}(u_{*},u+v_{0})\big).\nonumber 
\end{align}
Thus, in view of (\ref{eq:MassBounds}), (\ref{eq:LowerBoundDuRHorizon}),
(\ref{eq:NonTrappingAwayFromHorizon}) and (\ref{eq:BoundFrmTheMassAwayFromTrapping}),
from (\ref{eq:FromIntegrationDvDur}) we infer that 
\begin{equation}
\sup_{\mathcal{V}_{\text{\textgreek{d}}}}\Big|\log\Big(\frac{\partial_{v}r}{1-\frac{2m}{r}}\Big)\Big|<+\infty.\label{eq:UpperBoundKappaAwayFromHorizon}
\end{equation}

By integrating equation (\ref{eq:DerivativeInVDirectionKappaBar})
in $v$ starting from the point $(u,u_{*}+v_{0}-\text{\textgreek{d}})$
for any $u\in[u_{*}-\text{\textgreek{d}},u_{*})$ and using (\ref{eq:LowerBoundDuRHorizon}),
(\ref{eq:BoundEnergyMomentumTensorMaximal}), (\ref{eq:NonTrappingAwayFromHorizon})
and (\ref{eq:UpperBoundKappaAwayFromHorizon}), we also infer that
\begin{equation}
\sup_{\mathcal{V}_{\text{\textgreek{d}}}}\Big|\log\Big(-\frac{\partial_{u}r}{1-\frac{2m}{r}}\Big)\Big|<+\infty.\label{eq:UpperBoundKappaBarAwayFromHorizon}
\end{equation}
The bounds (\ref{eq:NonTrappingAwayFromHorizon}), (\ref{eq:UpperBoundKappaAwayFromHorizon})
and (\ref{eq:UpperBoundKappaBarAwayFromHorizon}) combine into: 
\begin{equation}
\sup_{\mathcal{V}_{\text{\textgreek{d}}}}\Bigg(\Big|\log\Big(\frac{\partial_{v}r}{1-\frac{2m}{r}}\Big)\Big|+\Big|\log\Big(-\frac{\partial_{u}r}{1-\frac{2m}{r}}\Big)\Big|+\Big(1-\frac{2m}{r}\Big)^{-1}\Bigg)<+\infty.\label{eq:BoundVdelta}
\end{equation}

In view of (\ref{eq:BoundednessOfQuantitiesOutsideTheEventHorizon}),
(\ref{eq:D_vRCan'tvanish}) and (\ref{eq:BoundVdelta}), we infer
that 
\begin{equation}
\sup_{\{u_{*}-\text{\textgreek{d}}\le u\le u_{*}\}\cap\{u\le v\le u_{*}+v_{0}-\text{\textgreek{d}}\}}\big|\log(\partial_{v}r)\big|<+\infty.\label{eq:UpperBoundDvRAwayTInf}
\end{equation}
In view of (\ref{eq:RelationHawkingMass}), (\ref{eq:LowerBoundDuRHorizon}),
(\ref{eq:UpperBoundDvRAwayTInf}) and the fact that $\log(\text{\textgreek{W}}^{2})$
extends continuously on $\mathcal{H}^{+}$, we also infer that 
\begin{equation}
\sup_{\{u_{*}-\text{\textgreek{d}}\le u\le u_{*}\}\cap\{u\le v\le u_{*}+v_{0}-\text{\textgreek{d}}\}}\Big(1-\frac{2m}{r}\Big)^{-1}<+\infty.\label{eq:UpperBoundTrappingAwayTInf}
\end{equation}
Therefore, (\ref{eq:LowerBoundDuRHorizon}), (\ref{eq:UpperBoundDvRAwayTInf})
and (\ref{eq:UpperBoundTrappingAwayTInf}) imply that 
\begin{equation}
\sup_{\{u_{*}-\text{\textgreek{d}}\le u\le u_{*}\}\cap\{u\le v\le u_{*}+v_{0}-\text{\textgreek{d}}\}}\Bigg(\Big|\log\Big(\frac{\partial_{v}r}{1-\frac{2m}{r}}\Big)\Big|+\Big|\log\Big(-\frac{\partial_{u}r}{1-\frac{2m}{r}}\Big)\Big|+\Big(1-\frac{2m}{r}\Big)^{-1}\Bigg)<+\infty.\label{eq:BoundHorizonAwayTinf}
\end{equation}
Combining (\ref{eq:BoundednessOfQuantitiesOutsideTheEventHorizon}),
(\ref{eq:BoundVdelta}) and (\ref{eq:BoundHorizonAwayTinf}), we thus
obtain 
\begin{equation}
\sup_{\mathcal{U}_{*}}\Bigg(\Big|\log\Big(\frac{\partial_{v}r}{1-\frac{2m}{r}}\Big)\Big|+\Big|\log\Big(-\frac{\partial_{u}r}{1-\frac{2m}{r}}\Big)\Big|+\Big(1-\frac{2m}{r}\Big)^{-1}\Bigg)<+\infty.
\end{equation}
This is a contradiction, in view of (\ref{eq:InfinityForMaximalDomainOfOuterCommunicationsAlt}).
Thus, (\ref{eq:UpperBoundRHorizon}) holds.

The relation (\ref{eq:TrappingAsymptoticallyHorizon}) also follows
by a similar argument: Assuming that (\ref{eq:TrappingAsymptoticallyHorizon})
is false, i.\,e.
\begin{equation}
\inf_{\mathcal{H}^{+}}\Big(1-\frac{2m}{r}\Big)>0,\label{eq:LowerBoundTrappingContradiction}
\end{equation}
the relation (\ref{eq:RelationHawkingMass}) and the fact that $\log(\text{\textgreek{W}}^{2})$
extends continuously on $\mathcal{H}^{+}$ implies that 
\begin{equation}
\partial_{v}r|_{\mathcal{H}^{+}}>0.\label{eq:PositiveDerivativeHorizon}
\end{equation}
Thus, the inequality 
\begin{equation}
\partial_{u}\log\Big(\frac{\partial_{v}r}{1-\frac{2m}{r}}\Big)\le0
\end{equation}
(following from (\ref{eq:DerivativeInUDirectionKappa}) and (\ref{eq:LowerBoundDuRHorizon})),
combined with (\ref{eq:PositiveDerivativeU_*}) (on $\mathcal{U}_{*}$),
(\ref{eq:LowerBoundTrappingContradiction}) and (\ref{eq:PositiveDerivativeHorizon}),
implies (\ref{eq:NonTrappingAwayFromHorizon}). Therefore, repeating
the same arguments as for the proof of (\ref{eq:UpperBoundRHorizonGeneral}),
we reach a contradiction.

\medskip{}

\noindent \emph{Proof of (\ref{eq:LowerBoundRHorizonSpecialForMatterField}).}
In view of (\ref{eq:UpperBoundRHorizonGeneral}) and (\ref{eq:TrappingAsymptoticallyHorizon}),
it suffices to show that 
\begin{equation}
\lim_{\bar{v}\rightarrow u_{*}+v_{0}}\tilde{m}|_{\mathcal{H}^{+}\cap\{v=\bar{v}\}}=\lim_{v\rightarrow v_{0}}\tilde{m}_{/}(v).\label{eq:LimitOfRenormalisedMassOnHorizon}
\end{equation}

Integrating (\ref{eq:DerivativeTildeUMass}) in $u$, we calculate
for any $u_{*}+v_{0}-\text{\textgreek{d}}\le\bar{v}<u_{*}+v_{0}$:
\begin{equation}
\tilde{m}|_{\mathcal{H}^{+}\cap\{v=\bar{v}\}}=\tilde{m}|_{\mathcal{I}\cap\{v=\bar{v}\}}+2\pi\int_{\bar{v}-v_{0}}^{u_{*}}\frac{1-\frac{2m}{r}}{-\partial_{u}r}rT_{uu}(u,\bar{v})\, du.\label{eq:RelationForMassDifferenceHorizonInfinity}
\end{equation}
The relation 
\begin{equation}
\partial_{v}\Big(\frac{1-\frac{2m}{r}}{-\partial_{u}r}\Big)\ge0\label{eq:DecreasingKappaBar}
\end{equation}
(following from (\ref{eq:DerivativeInVDirectionKappaBar}) and (\ref{eq:PositiveDerivativeU_*})),
combined with the bound (\ref{eq:BoundEnergyMomentumTensorMaximal})
(and the fact that $r\ge r_{0}$ on $\mathcal{U}_{*}$) implies that
\begin{equation}
\sup_{\{u_{*}-\text{\textgreek{d}}\le u\le u_{*}\}\cap\{u_{*}+v_{0}-\text{\textgreek{d}}\le v\le u+v_{0}\}}\frac{1-\frac{2m}{r}}{-\partial_{u}r}rT_{uu}<+\infty.\label{eq:UpperBoundIntegrand}
\end{equation}
Thus, in view of (\ref{eq:ConservationTildeM}) and (\ref{eq:UpperBoundIntegrand}),
the relation (\ref{eq:LimitOfRenormalisedMassOnHorizon}) is obtained
by taking the limit $\bar{v}\rightarrow u_{*}+v_{0}$ in (\ref{eq:RelationForMassDifferenceHorizonInfinity}).

\paragraph*{Proof of (\ref{eq:CompleConformalInfinity})}

In view of (\ref{eq:RelationHawkingMass}), (\ref{eq:MassBounds})
and (\ref{eq:GaugeInfinityMaximal}), (\ref{eq:CompleConformalInfinity})
will follow from 
\begin{equation}
\int_{\mathcal{I}}\frac{-\partial_{u}r}{1-\frac{2m}{r}}\Big|_{\mathcal{I}}\, du=+\infty.\label{eq:ForInfiniteLengthInfinity}
\end{equation}
Assume, for the sake of contradiction, that (\ref{eq:ForInfiniteLengthInfinity})
is false. Then 
\begin{equation}
\lim_{\bar{u}\rightarrow u_{*}}\int_{\bar{u}}^{u_{*}}\frac{-\partial_{u}r}{1-\frac{2m}{r}}(u,u+v_{0})\, du=\lim_{\bar{u}\rightarrow u_{*}}\int_{\mathcal{I}\cap\{u\ge\bar{u}\}}\frac{-\partial_{u}r}{1-\frac{2m}{r}}\Big|_{\mathcal{I}}\, du=0.\label{eq:ZeroForContradictionInfinity}
\end{equation}
The relation (\ref{eq:DecreasingKappaBar}) implies, in view of (\ref{eq:ZeroForContradictionInfinity})
that:
\begin{equation}
\lim_{\bar{u}\rightarrow u_{*}}\int_{\bar{u}}^{u_{*}}\frac{-\partial_{u}r}{1-\frac{2m}{r}}(u,\bar{u}+v_{0})\, du=0.\label{eq:LengthRayHorizonInfinity}
\end{equation}
Thus, (\ref{eq:MassBounds}) and (\ref{eq:LengthRayHorizonInfinity})
imply that 
\begin{align}
0 & =\lim_{\bar{u}\rightarrow u_{*}}\int_{\bar{u}}^{u_{*}}\frac{-\partial_{u}r}{1-\frac{1}{3}\Lambda r^{2}}(u,\bar{u}+v_{0})\, du=\\
 & =\sqrt{-\frac{3}{\Lambda}}\lim_{\bar{u}\rightarrow u_{*}}\Big(\tan^{-1}\big(\sqrt{-\frac{\Lambda}{3}}r\big)\Big|_{\mathcal{I}\cap\{v=\bar{u}+v_{0}\}}-\tan^{-1}\big(\sqrt{-\frac{\Lambda}{3}}r\big)\Big|_{\mathcal{H}^{+}\cap\{v=\bar{u}+v_{0}\}}\Big)=\nonumber \\
 & =\sqrt{-\frac{3}{\Lambda}}\Big(\frac{\pi}{2}-\lim_{\bar{u}\rightarrow u_{*}}\tan^{-1}\big(\sqrt{-\frac{\Lambda}{3}}r\big)\Big|_{\mathcal{H}^{+}\cap\{v=\bar{u}+v_{0}\}}\Big),\nonumber 
\end{align}
which is a contradiction in view of (\ref{eq:UpperBoundRHorizonGeneral}).

\paragraph*{Proof of (\ref{eq:InfiniteLengthHorizon}).}

In view of (\ref{eq:RelationHawkingMass}), (\ref{eq:LowerBoundDuRHorizon}),
(\ref{eq:UpperBoundRHorizonGeneral}) and the fact that 
\begin{equation}
\partial_{v}(-r\partial_{u}r)\ge0\label{eq:IncreasingDuR}
\end{equation}
(following from (\ref{eq:RequationFinal})), in order to establish
(\ref{eq:InfiniteLengthHorizon}) it suffices to show that 
\begin{equation}
\int_{\mathcal{H}^{+}}\frac{\text{\textgreek{W}}^{2}}{-\partial_{u}r}\, dv=+\infty.\label{eq:BoundToShowHorizon}
\end{equation}

In view of (\ref{eq:LowerBoundDuRHorizon}), (\ref{eq:UpperBoundRHorizonGeneral}),
(\ref{eq:IncreasingDuR}) and (\ref{eq:BoundEnergyMomentumTensorMaximal}),
we can bound 
\begin{equation}
C_{T}\doteq\sup_{\mathcal{U}_{*}}\Big(\frac{rT_{uu}}{-\partial_{u}r}\Big)<+\infty.\label{eq:BoundForChangeDvR}
\end{equation}
Integrating equation 
\begin{equation}
\partial_{u}\log\big(\frac{\text{\textgreek{W}}^{2}}{-\partial_{u}r}\big)=-4\pi r\frac{T_{uu}}{-\partial_{u}r}\label{eq:RephrasedConstraint}
\end{equation}
 in $u$ and using (\ref{eq:BoundForChangeDvR}),%
\footnote{Equation (\ref{eq:RephrasedConstraint}) is readily obtained from
(\ref{eq:ConstraintUFinal}).%
} we obtain for any $v_{0}\le\bar{v}<u_{*}+v_{0}$: 
\begin{equation}
\Bigg|\log\Big(\frac{\text{\textgreek{W}}^{2}}{-\partial_{u}r}\Big)\Big|_{\mathcal{H}^{+}\cap\{v=\bar{v}\}}-\log\Big(\frac{\partial_{v}r}{1-\frac{2m}{r}}\Big)\Big|_{\mathcal{I}\cap\{v=\bar{v}\}}\Bigg|\le C_{T}(u_{*}+v_{0}-\bar{v}).\label{eq:DifferenceLenghtHorizonInfinity}
\end{equation}
Integrating (\ref{eq:DifferenceLenghtHorizonInfinity}) in $\bar{v}\in[v_{0},u_{*}+v_{0})$
and using (\ref{eq:RelationHawkingMass}) on $\mathcal{I}$, (\ref{eq:ForInfiniteLengthInfinity})
and (\ref{eq:GaugeInfinityMaximal}), we thus infer (\ref{eq:BoundToShowHorizon}).

\subsubsection*{Step 2: Construction of $\mathcal{U}\backslash J^{-}(\mathcal{U})$
in the case $u_{*}<+\infty$}

By applying Proposition \ref{Prop:LocalExistenceTypeIII} for the
initial data induced by $(r,\text{\textgreek{W}}^{2},\bar{f}_{in},\bar{f}_{out})$
on $[u_{*}-\text{\textgreek{d}}_{1},u_{*}]\times\{u_{*}\}\cup\{u_{*}-\text{\textgreek{d}}_{1}\}\times[u_{*},u_{*}+\text{\textgreek{d}}_{2}]$
for some $0<\text{\textgreek{d}}_{2}\ll\text{\textgreek{d}}_{1}\ll1$,
we infer that $(r,\text{\textgreek{W}}^{2},\bar{f}_{in},\bar{f}_{out})$
extends smoothly as a solution to (\ref{eq:RequationFinal})--(\ref{eq:OutgoingVlasovFinal})
(satisfying (\ref{eq:MirrorRMaximal})--(\ref{eq:GaugeMirrorMaximal})
on $u=v$) on 
\begin{equation}
\mathcal{W}_{\text{\textgreek{d}}_{2}}\doteq\{u_{*}\le u\le v\}\cap\{v\le u_{*}+\text{\textgreek{d}}_{2}\}.\label{eq:FirstExtension}
\end{equation}
Repeating the same procedure as for the construction of $\mathcal{U}_{*}$
with the use of Proposition \ref{Prop:LocalExistenceTypeII} in the
previous step, by using Proposition \ref{Prop:LocalExistenceTypeIII}
for the initial data induced by $(r,\text{\textgreek{W}}^{2},\bar{f}_{in},\bar{f}_{out})$
on sets of the form $[u_{*},u_{*}+\bar{u}]\times\{u_{*}+\bar{u}\}\cup\{u_{*}\}\times[u_{*}+\bar{u},u_{*}+\bar{u}+\text{\textgreek{d}}]$
(starting from $\bar{u}=\text{\textgreek{d}}_{2}$), we infer that
there exists a $v_{*}\le u_{*}+v_{0}$ with the following properties:

\begin{enumerate}

\item Setting 
\begin{equation}
\mathcal{W}_{*}\doteq\{u_{*}\le u<v\}\cap\{v<v_{*}\},\label{eq:DefinitionW_*}
\end{equation}
\begin{equation}
\mathcal{C}^{-}\doteq\{v=v_{*}\}\cap\{u\le u_{*}<v_{*}\}\label{eq:PastOfTerminalMirror}
\end{equation}
and 
\begin{equation}
\text{\textgreek{g}}_{0}\doteq\{u=v\}\cap\{0\le u<v_{*}\},\label{eq:DefinitionGamma0Full}
\end{equation}
the functions $(r,\text{\textgreek{W}}^{2},\bar{f}_{in},\bar{f}_{out})$
extend smoothly on $\mathcal{W}_{*}\cup\text{\textgreek{g}}_{0}$,
so that $(\mathcal{U}_{*}\cup\mathcal{W}_{*};r,\text{\textgreek{W}}^{2},\bar{f}_{in},\bar{f}_{out})$
is a development of the initial data $(r_{/},\text{\textgreek{W}}_{/}^{2},\bar{f}_{in/},\bar{f}_{out/})$
according to Definition \ref{def:Development}.

\item For any $\bar{v}<v_{*}$, we can bound 
\begin{equation}
\sup_{\mathcal{W}_{*}\cap\{v\le\bar{v}\}}\Bigg\{\Big|\log\big(\frac{\text{\textgreek{W}}^{2}}{1-\frac{1}{3}\Lambda r^{2}}\big)\Big|+\Big|\log\Big(\frac{2\partial_{v}r}{1-\frac{2m}{r}}\Big)\Big|+\Big|\log\Big(\frac{1-\frac{2m}{r}}{1-\frac{1}{3}\Lambda r^{2}}\Big)\Big|+\sqrt{-\Lambda}|\tilde{m}|+r^{2}T_{uu}+r^{2}T_{vv}\Bigg\}<+\infty\label{eq:BoundednessNearAxis}
\end{equation}
and, in the case $v_{*}<u_{*}+v_{0}$: 
\begin{equation}
\sup_{\mathcal{W}_{*}}\Bigg\{\Big|\log\big(\frac{\text{\textgreek{W}}^{2}}{1-\frac{1}{3}\Lambda r^{2}}\big)\Big|+\Big|\log\Big(\frac{2\partial_{v}r}{1-\frac{2m}{r}}\Big)\Big|+\Big|\log\Big(\frac{1-\frac{2m}{r}}{1-\frac{1}{3}\Lambda r^{2}}\Big)\Big|+\sqrt{-\Lambda}|\tilde{m}|+r^{2}T_{uu}+r^{2}T_{vv}\Bigg\}=+\infty.\label{eq:InfinityNearAxis}
\end{equation}

\end{enumerate}

We will now proceed to define the domain $\mathcal{U}$ of the maximal
future development of $(r_{/},\text{\textgreek{W}}_{/}^{2},\bar{f}_{in/},\bar{f}_{out/})$.
We have to distinguish between two cases: The case $v_{*}=u_{*}+v_{0}$,
and the case $v_{*}<u_{*}+v_{0}$. Let us remark already that, later
in the proof, we will establish that, necessarily, $v_{*}<u_{*}+v_{0}$,
and thus the former case can not be actually realised.

\medskip{}

\noindent \emph{The case $v_{*}=u_{*}+v_{0}$.} In this case, we set
\begin{equation}
\mathcal{U}=\mathcal{U}_{*}\cup\mathcal{W}_{*}.
\end{equation}

\noindent \medskip{}

\noindent \emph{The case $v_{*}<u_{*}+v_{0}$.} In this case, an application
of the continuation criterion of Lemma \ref{lem:cotinuationCriterion}
on the domains 
\begin{equation}
\mathcal{Y}_{\text{\textgreek{d}}}=[u_{*},v_{*}-2\text{\textgreek{d}}]\times[v_{*}-\text{\textgreek{d}},v_{*})
\end{equation}
for any $0<\text{\textgreek{d}}\ll1$ implies that $(r,\log\text{\textgreek{W}}^{2})$
extend smoothly on $\{v=v_{*}\}\cap\{u_{*}\le u<v_{*}\}$ and $(\bar{f}_{in},\bar{f}_{out})$
extend smoothly on $\{v=v_{*}\}\cap\{u_{*}\le u<v_{*}\}\times(0,+\infty)$.
Thus, considering the initial data induced by $(r,\text{\textgreek{W}}^{2},\bar{f}_{in},\bar{f}_{out})$
on $[u_{*},v_{*})\times\{v_{*}\}\cup\{u_{*}\}\times[v_{*},u_{*}+v_{0})$
and using Proposition \ref{prop:LocalExistenceTypeI} we infer (again
by repeating a similar procedure to the construction of $\mathcal{U}_{*}$
with the use of Proposition \ref{Prop:LocalExistenceTypeII}) that
there exists an open set 
\begin{equation}
\mathcal{V}^{\prime}\subseteq(u_{*},v_{*})\times(v_{*},u_{*}+v_{0})\label{eq:DefinitionFromCharacteristicInitialDataMaxExtension}
\end{equation}
 which is globally hyperbolic (with respect to the reference metric
(\ref{eq:ComparisonUVMetric})), such that $(r,\text{\textgreek{W}}^{2},\bar{f}_{in},\bar{f}_{out})$
extend smoothly on $\mathcal{V}^{\prime}$ as solutions to the system
(\ref{eq:RequationFinal})--(\ref{eq:OutgoingVlasovFinal}).

Let us set 
\begin{equation}
\mathcal{V}_{*}\doteq\mathcal{V}^{\prime}\cap\{r>r_{0}\}.\label{eq:R0MaxExtension}
\end{equation}
In view of (\ref{eq:ConstrainVFinal}) and the fact that $T_{vv}\ge0$,
any $\{u=const\}$ line in the region 
\[
\mathcal{U}^{\prime}=\mathcal{U}_{*}\cup\mathcal{H}^{+}\cup\mathcal{W}_{*}\cup\mathcal{C}^{-}\cup\mathcal{V}^{\prime}
\]
 can intersect the level set $\{r=r_{0}\}$ at most two times. Thus,
since $r|_{\text{\textgreek{g}}_{0}}=r_{0}$, we readily infer that
the boundary of $\mathcal{V}_{*}$ in $\mathcal{V}^{\prime}$ (which
consists of a subset of the level set $\{r=r_{0}\}$) is a smooth
achronal curve and $\mathcal{V}_{*}$ is globally hyperbolic with
respect to the reference metric (\ref{eq:ComparisonUVMetric}). Therefore,
the domain $\mathcal{U}_{*}\cup\mathcal{H}^{+}\cup\mathcal{W}_{*}\cup\mathcal{V}^{\prime}$
belongs to the set $\mathscr{U}_{v_{0}}$, introdced in Definition
\ref{def:DevelopmentSets}. In this case, we will set 
\begin{equation}
\mathcal{U}\doteq\mathcal{U}_{*}\cup\mathcal{H}^{+}\cup\mathcal{W}_{*}\cup\mathcal{C}^{-}\cup\mathcal{V}_{*}.
\end{equation}

\medskip{}

In both the case $v_{*}=u_{*}+v_{0}$ and the case $v_{*}<u_{*}+v_{0}$,
the boundary $\partial\mathcal{U}$ of $\mathcal{U}$ splits as (\ref{eq:BoundaryOfU}),
for a continuous achronal curve \textgreek{g} with a parametrization
$\text{\textgreek{g}}:(0,v_{0})\rightarrow\mathbb{R}^{2}$ of the
form $\text{\textgreek{g}}(v)=(u_{\text{\textgreek{g}}}(x),v_{\text{\textgreek{g}}}(x))$,
where 
\begin{equation}
u_{\text{\textgreek{g}}}(x)=\begin{cases}
v_{*}, & 0<x\le v_{1}\\
f_{1}(x), & v_{1}<x<v_{2}\\
u_{*}+v_{0}-x, & v_{2}\le x<v_{0}
\end{cases}\label{eq:ucoordinateofgamma}
\end{equation}
 and 
\begin{equation}
v_{\text{\textgreek{g}}}(x)=\begin{cases}
v_{*}+x, & 0<x\le v_{1}\\
f_{2}(x), & v_{1}<x<v_{2}\\
u_{*}+v_{0}, & v_{2}\le x<v_{0}
\end{cases}\label{eq:vcoordinateofgamma}
\end{equation}
for some $0\le v_{1}\le v_{2}\le v_{0}$, where 
\begin{equation}
u_{*}<f_{1}(x)<v_{*}\label{eq:Bound1F1}
\end{equation}
 and 
\begin{equation}
v_{*}<f_{2}(x)<u_{*}+v_{0}\label{eq:Bound1F2}
\end{equation}
 for all $x\in(v_{1},v_{2})$. Note that the properties (\ref{eq:Bound1F1}),
(\ref{eq:Bound1F2}) of $f_{1},f_{2}$ imply, by an application of
Lemma \ref{lem:cotinuationCriterion}, that $(r,\log\text{\textgreek{W}}^{2},\bar{f}_{in},\bar{f}_{out})$
extends smoothly across $\text{\textgreek{g}}\big((v_{1},v_{2})\big)$.

We will now proceed to show that 
\begin{equation}
v_{2}=v_{0}\label{eq:NoCauchyHorizon}
\end{equation}
 in (\ref{eq:ucoordinateofgamma})--(\ref{eq:vcoordinateofgamma}),
and that $r$ extends continuously on $\text{\textgreek{g}}$ with
\begin{equation}
r|_{\text{\textgreek{g}}}=r_{0}.\label{eq:RonGamma}
\end{equation}
Since in the case $v_{*}=u_{*}+v_{0}$ it is necessary that $v_{2}=0$,
(\ref{eq:NoCauchyHorizon}) will imply that $v_{*}<u_{*}+v_{0}$ always.
Furthermore, since $r$ extends smoothly across $\text{\textgreek{g}}\big((v_{1},v_{2})\big)$,
we will also infer from (\ref{eq:RonGamma}) that $\text{\textgreek{g}}\big((v_{1},v_{2})\big)$
is smooth. Note that,since $\text{\textgreek{g}}$ is continuous and
$\partial\mathcal{U}$ has the form (\ref{eq:BoundaryOfU}), the relation
(\ref{eq:NoCauchyHorizon}) also implies that, necessarily, $v_{1}<v_{2}=v_{0}$.

\medskip{}

\noindent \emph{Proof of (\ref{eq:NoCauchyHorizon}).} Integrating
equation (\ref{eq:RequationFinal}) in $v$ starting from $\text{\textgreek{g}}_{0}$,
we readily obtain that 
\begin{equation}
\sup_{\mathcal{U}}\partial_{u}r<0.\label{eq:NegativeDerivativeUEverywhere}
\end{equation}
Integrating (\ref{eq:ConservationT_vv}) in $u$ and (\ref{eq:ConservationT_uu})
in $v$ and using the boundary conditions (\ref{eq:ReflectiveConditionsTMaximal}),
we readily infer that (\ref{eq:BoundEnergyMomentumTensorMaximal})
holds on the whole of $\mathcal{U}$. In view of (\ref{eq:DerivativeTildeUMass}),
(\ref{eq:BoundEnergyMomentumTensorMaximal}), (\ref{eq:NegativeDerivativeUEverywhere}),
(\ref{eq:LimitOfRenormalisedMassOnHorizon}) and the fact that $\lim_{v\rightarrow v_{0}}\tilde{m}_{/}(v)>0$
in the case $\mathcal{H}^{+}\neq\emptyset$ (otherwise (\ref{eq:TrappingAsymptoticallyHorizon})
is violated), we infer that there exists some $0<\text{\textgreek{d}}\ll1$
so that 
\begin{equation}
\inf_{[u_{*},u_{*}+\text{\textgreek{d}}]\times[u_{*}+v_{0}-\text{\textgreek{d}},u_{*}+v_{0}]\cap\mathcal{U}}\tilde{m}>0.\label{eq:PositiveMassAway}
\end{equation}

Integrating equation 
\begin{equation}
\partial_{v}\log(-\partial_{u}r)=\frac{1}{2}\frac{\tilde{m}-\frac{1}{3}\Lambda r^{3}}{r^{2}}\frac{\text{\textgreek{W}}^{2}}{-\partial_{u}r}\label{eq:RephrasedDuR}
\end{equation}
(which is readily obtained from (\ref{eq:RequationFinal})) in $v$
starting from $v=u_{*}+v_{0}-\text{\textgreek{d}}$ and using (\ref{eq:PositiveMassAway}),
we obtain for any point $(u,v)\in[u_{*},u_{*}+\text{\textgreek{d}}]\times[u_{*}+v_{0}-\text{\textgreek{d}},u_{*}+v_{0}]\cap\mathcal{U}$:
\begin{equation}
\log(-\partial_{u}r)(u,v)\ge\log(-\partial_{u}r)(u,u_{*}+v_{0}-\text{\textgreek{d}})+c_{0}\int_{u_{*}+v_{0}-\text{\textgreek{d}}}^{v}\frac{\text{\textgreek{W}}^{2}}{-\partial_{u}r}(u,\bar{v})\, d\bar{v}.\label{eq:EquationForDerivativeU}
\end{equation}
for some $c_{0}>0$ depending on $r_{0}$ and (\ref{eq:PositiveMassAway}).
Integrating (\ref{eq:RephrasedConstraint}) in $u$ starting from
$u=u_{*}$ and using (\ref{eq:BoundEnergyMomentumTensorMaximal})
and (\ref{eq:NegativeDerivativeUEverywhere}), we can also estimate
\begin{equation}
\int_{u_{*}+v_{0}-\text{\textgreek{d}}}^{v}\frac{\text{\textgreek{W}}^{2}}{-\partial_{u}r}(u,\bar{v})\ge c_{1}\int_{u_{*}+v_{0}-\text{\textgreek{d}}}^{v}\frac{\text{\textgreek{W}}^{2}}{-\partial_{u}r}(u_{*},\bar{v}),\label{eq:EstimateLengthFromTheHorizon}
\end{equation}
for some $c_{1}>0$ depending on $r_{0}$, (\ref{eq:BoundEnergyMomentumTensorMaximal}),
(\ref{eq:NegativeDerivativeUEverywhere}) and (\ref{eq:PositiveMassAway}).
Thus, from (\ref{eq:EquationForDerivativeU}) and (\ref{eq:EstimateLengthFromTheHorizon})
we can bound for any $(u,v)\in[u_{*},u_{*}+\text{\textgreek{d}}]\times[u_{*}+v_{0}-\text{\textgreek{d}},u_{*}+v_{0}]\cap\mathcal{U}$:
\begin{align}
r(u_{*},v)-r(u,v) & =\int_{u_{*}}^{u}\big(-\partial_{u}r(\bar{u},v)\big)\, d\bar{u}\ge\label{eq:AlmostThereForDifferenceR}\\
 & \ge\exp\Big(c_{0}\sup_{\bar{u}\in[u_{*},u]}\int_{u_{*}+v_{0}-\text{\textgreek{d}}}^{v}\frac{\text{\textgreek{W}}^{2}}{-\partial_{u}r}(\bar{u},\bar{v})\, d\bar{v}\Big)\cdot\int_{u_{*}}^{u}\big(-\partial_{u}r(\bar{u},u_{*}+v_{0}-\text{\textgreek{d}})\big)\, d\bar{u}\ge\nonumber \\
 & \ge\exp\Big(c_{0}c_{1}\int_{u_{*}+v_{0}-\text{\textgreek{d}}}^{v}\frac{\text{\textgreek{W}}^{2}}{-\partial_{u}r}(u_{*},\bar{v})\, d\bar{v}\Big)\cdot\big(r(u_{*},u_{*}+v_{0}-\text{\textgreek{d}})-r(u,u_{*}+v_{0}-\text{\textgreek{d}})\big).\nonumber 
\end{align}

Assuming, for the sake of contradiction, that $v_{2}<v_{0}$ in (\ref{eq:ucoordinateofgamma})--(\ref{eq:vcoordinateofgamma}),
for any fixed $u_{0}\in(u_{*},u_{*}+v_{0}-v_{2})$ and any $v\in[u_{*}+v_{0}-\text{\textgreek{d}},u_{*}+v_{0})$,
the point $(u_{0},v)$ belongs to $[u_{*},u_{*}+\text{\textgreek{d}}]\times[u_{*}+v_{0}-\text{\textgreek{d}},u_{*}+v_{0}]\cap\mathcal{U}$.
Therefore, applying (\ref{eq:AlmostThereForDifferenceR}) for $(u,v)=(u_{0},v)$
and considering the limit $v\rightarrow u_{*}+v_{0}$, we obtain
\begin{equation}
\int_{u_{*}+v_{0}-\text{\textgreek{d}}}^{u_{*}+v_{0}}\frac{\text{\textgreek{W}}^{2}}{-\partial_{u}r}(u_{*},\bar{v})\, d\bar{v}\le C_{0}\limsup_{v\rightarrow u_{*}+v_{0}}\log\Big(\frac{r(u_{*},v)-r(u_{0},v)}{r(u_{*},u_{*}+v_{0}-\text{\textgreek{d}})-r(u_{0},u_{*}+v_{0}-\text{\textgreek{d}})}\Big),\label{eq:BoundForLengthFinal}
\end{equation}
where $C_{0}>0$ depends on $r_{0}$, (\ref{eq:BoundEnergyMomentumTensorMaximal}),
(\ref{eq:NegativeDerivativeUEverywhere}) and (\ref{eq:PositiveMassAway}).
In view of (\ref{eq:NegativeDerivativeUEverywhere}) and the fact
that $r_{0}\le r\le r_{S}$ on $\{u\ge u_{*}\}\cap\mathcal{U}$, the
right hand side of (\ref{eq:BoundForLengthFinal}) is finite, while
the left hand side is infinite in view of (\ref{eq:BoundToShowHorizon}),
which is a contradiction. Therefore, (\ref{eq:NoCauchyHorizon}) holds.

\medskip{}

\noindent \emph{Proof of (\ref{eq:RonGamma}).} In view of Lemma \ref{lem:cotinuationCriterionMirror},
(\ref{eq:InfinityNearAxis}) and the fact that $(r,\log\text{\textgreek{W}}^{2},\bar{f}_{in},\bar{f}_{out})$
are smooth on $\{u=u_{*}\}\cap\{u_{*}\le v\le v_{*}\}$, we infer
that, necessarily 
\begin{equation}
\lim_{v\rightarrow v^{*}}\partial_{v}r|_{\text{\textgreek{g}}_{0}\cap\{v=\bar{v}\}}=0.\label{eq:ZeroAtBreakingPoint}
\end{equation}
In view of the inequality (\ref{eq:ConcaveR}), the relation (\ref{eq:ZeroAtBreakingPoint})
implies that, for any $v\in[v_{*},v_{*}+v_{1}]$ (with $v_{1}$ as
in (\ref{eq:vcoordinateofgamma})) 
\begin{equation}
\lim_{u\rightarrow v_{*}}\partial_{v}r(u,v)\le0.\label{eq:BoundFromConcavity}
\end{equation}
Since $\mathcal{V}_{*}$ was defined by (\ref{eq:R0MaxExtension}),
we necessarily have for any point $p\in\text{\textgreek{g}}$: 
\begin{equation}
\liminf_{p^{\prime}\rightarrow p}r\ge r_{0}.\label{eq:LimInfGamma}
\end{equation}
Thus, since $r|_{\text{\textgreek{g}}_{0}}=r_{0}$, the relations
(\ref{eq:LimInfGamma}) and (\ref{eq:BoundFromConcavity}) imply that
$r$ extends continuously on $\overline{\text{\textgreek{g}}\big((0,v_{1})\big)},$with
\begin{equation}
r|_{clos\big(\text{\textgreek{g}}\big((0,v_{1})\big)\big)}=r_{0}.\label{eq:ROnGammaNull}
\end{equation}

Since 
\begin{equation}
\text{\textgreek{g}}\big((v_{1},v_{0})\big)\subset(u_{*},v_{*})\times(v_{*},u_{*}+v_{0}),
\end{equation}
an application of Lemma \ref{lem:cotinuationCriterion} yields that
$(r,\log\text{\textgreek{W}}^{2})$ extend smoothly across $\text{\textgreek{g}}$.
Since $\mathcal{V}_{*}$ was defined by (\ref{eq:R0MaxExtension}),
and $\mathcal{V}^{\prime}$ is the maximal globally hyperbolic development
of the\noun{ }characteristic initial data induced by $(r,\text{\textgreek{W}}^{2},\bar{f}_{in},\bar{f}_{out})$
on $[u_{*},v_{*})\times\{v_{*}\}\cup\{u_{*}\}\times[v_{*},u_{*}+v_{0})$,
we infer that $\text{\textgreek{g}}\big((v_{1},v_{0})\big)$ is the
bounary of $\mathcal{V}_{*}$ in $\mathcal{V}^{\prime}$ and, thus:
\begin{equation}
r|_{\text{\textgreek{g}}\big((v_{1},v_{0})\big)}=r_{0}.\label{eq:ROnGammaSpacelike}
\end{equation}
The relation (\ref{eq:RonGamma}) now follows from (\ref{eq:ROnGammaNull})
and (\ref{eq:ROnGammaSpacelike}).

\paragraph*{End of the proof. }

In order to finish the proof of Theorem \ref{thm:maximalExtension}
in the case when (\ref{eq:U_*<infty}) holds, it remains to establish
(\ref{eq:NegativeDerivativeRMaximal})--(\ref{eq:D_vRPositiveMaximal}),
(\ref{eq:ConstantMassMirror})--(\ref{eq:ConstantMassInfinity}),
(\ref{eq:MirrorExtendsBeyondHorizon}), as well as Property 6.

The bound (\ref{eq:NegativeDerivativeRMaximal}) has been already
established in (\ref{eq:NegativeDerivativeUEverywhere}). The bounds
(\ref{eq:NonTrappingMaximal})--(\ref{eq:D_vRPositiveMaximal}) follow
readily from (\ref{eq:BoundednessOfQuantitiesOutsideTheEventHorizon})(on
$J^{-}(\mathcal{I})$) and (\ref{eq:BoundednessNearAxis}) (on $J^{-}(\text{\textgreek{g}}_{0})$).

The conservation of $\tilde{m}$ on $\text{\textgreek{g}}_{0}$ and
$\mathcal{I}$, i.\,e.~(\ref{eq:ConstantMassMirror})--(\ref{eq:ConstantMassInfinity}),
follows readily fromhe relations (\ref{eq:DerivativeTildeUMass})--(\ref{eq:DerivativeTildeVMass})
and the conditions (\ref{eq:MirrorRMaximal}), (\ref{eq:GaugeMirrorMaximal}),
(\ref{eq:InfinityRMaximal}) and (\ref{eq:GaugeInfinityMaximal}).

The relation (\ref{eq:MirrorExtendsBeyondHorizon}) follows immediately
from the fact that $\text{\textgreek{g}}_{0}\cap\{u>u_{*}\}\neq\emptyset$
(see (\ref{eq:DefinitionGamma0Full})).

Finally, Property 6 is an immediate consequence of (\ref{eq:NoCauchyHorizon}).

\qed

\section{\label{sec:Proof-of-Cauchy}Cauchy stability in a rough norm, uniformly
in $r_{0}$}

In this Section, we will establish Theorem \ref{thm:CauchyStability}
and Corollary \ref{cor:CauchyStabilityOfAdS}.

\subsection{\label{sub:Proof-of-Cauchy-General}Proof of Theorem \ref{thm:CauchyStability}}

Let $C_{1}\gg1$ be a large, fixed constant. Using a standard continuity
argument, it suffices to show that, for any $0<u_{*}\le u_{0}$ such
that 
\[
\mathcal{W}_{u_{*}}\doteq\{0<u<u_{*}\}\cap\{u+v_{1}<v<u+v_{2}\}\subset\mathcal{U}_{2}
\]
and
\begin{align}
\sup_{\mathcal{W}_{u_{*}}}\Bigg\{\Big|\log\big(\frac{\text{\textgreek{W}}_{1}^{2}}{1-\frac{1}{3}\Lambda r_{1}^{2}}\big)-\log\big(\frac{\text{\textgreek{W}}_{2}^{2}}{1-\frac{1}{3}\Lambda r_{2}^{2}}\big)\Big|+\Big|\log\Big(\frac{2\partial_{v}r_{1}}{1-\frac{2m_{1}}{r_{1}}}\Big)-\log\Big(\frac{2\partial_{v}r_{2}}{1-\frac{2m_{2}}{r_{2}}}\Big)\Big|+\label{eq:UpperBoundNonTrappingForCauchyStabilityBootstrap}\\
+\Big|\log\Big(\frac{1-\frac{2m_{1}}{r_{1}}}{1-\frac{1}{3}\Lambda r_{1}^{2}}\Big)-\log\Big(\frac{1-\frac{2m_{2}}{r_{2}}}{1-\frac{1}{3}\Lambda r_{2}^{2}}\Big)\Big|+\sqrt{-\Lambda}|\tilde{m}_{1}-\tilde{m}_{2}|\Bigg\}+\nonumber \\
+\sup_{\bar{u}}\int_{\{u=\bar{u}\}\cap\mathcal{W}_{u_{*}}}\big|r_{1}(T_{vv})_{1}-r_{2}(T_{vv})_{2}\big|\, dv+\sup_{\bar{v}}\int_{\{v=\bar{v}\}\cap\mathcal{W}_{u_{*}}}\big|r_{1}(T_{uu})_{1}-r_{2}(T_{uu})_{2}\big|\, du & \le\nonumber \\
\le2\exp\big(\exp\big(C_{1} & (1+C_{0})\big)\frac{u_{0}}{v_{2}-v_{1}}\big)\text{\textgreek{d}},\nonumber 
\end{align}
the following improvement of (\ref{eq:UpperBoundNonTrappingForCauchyStabilityBootstrap})
actually holds: 
\begin{align}
\sup_{\mathcal{W}_{u_{*}}}\Bigg\{\Big|\log\big(\frac{\text{\textgreek{W}}_{1}^{2}}{1-\frac{1}{3}\Lambda r_{1}^{2}}\big)-\log\big(\frac{\text{\textgreek{W}}_{2}^{2}}{1-\frac{1}{3}\Lambda r_{2}^{2}}\big)\Big|+\Big|\log\Big(\frac{2\partial_{v}r_{1}}{1-\frac{2m_{1}}{r_{1}}}\Big)-\log\Big(\frac{2\partial_{v}r_{2}}{1-\frac{2m_{2}}{r_{2}}}\Big)\Big|+\label{eq:UpperBoundNonTrappingForCauchyStabilityToShow}\\
+\Big|\log\Big(\frac{1-\frac{2m_{1}}{r_{1}}}{1-\frac{1}{3}\Lambda r_{1}^{2}}\Big)-\log\Big(\frac{1-\frac{2m_{2}}{r_{2}}}{1-\frac{1}{3}\Lambda r_{2}^{2}}\Big)\Big|+\sqrt{-\Lambda}|\tilde{m}_{1}-\tilde{m}_{2}|\Bigg\}+\nonumber \\
+\sup_{\bar{u}}\int_{\{u=\bar{u}\}\cap\mathcal{W}_{u_{*}}}\big|r_{1}(T_{vv})_{1}-r_{2}(T_{vv})_{2}\big|\, dv+\sup_{\bar{v}}\int_{\{v=\bar{v}\}\cap\mathcal{W}_{u_{*}}}\big|r_{1}(T_{uu})_{1}-r_{2}(T_{uu})_{2}\big|\, du & \le\nonumber \\
\le\exp\big(\exp\big(C_{1} & (1+C_{0})\big)\cdot\frac{u_{0}}{v_{2}-v_{1}}\big)\text{\textgreek{d}}.\nonumber 
\end{align}

In order to establish (\ref{eq:UpperBoundNonTrappingForCauchyStabilityToShow}),
it suffices to show that, for any $u_{1}<u_{2}<u_{*}$ such that 
\begin{equation}
u_{2}-u_{1}\le v_{2}-v_{1},\label{eq:u_*small}
\end{equation}
setting 
\begin{equation}
\mathcal{W}_{u_{1};u_{2}}\doteq\{u_{1}<u<u_{2}\}\cap\{u<v<u+v_{0}\},
\end{equation}
 we can bound on $\mathcal{W}_{u_{1};u_{2}}$ 
\begin{align}
\sup_{\mathcal{W}_{u_{1};u_{2}}}\Bigg\{\Big|\log\big(\frac{\text{\textgreek{W}}_{1}^{2}}{1-\frac{1}{3}\Lambda r_{1}^{2}}\big)-\log\big(\frac{\text{\textgreek{W}}_{2}^{2}}{1-\frac{1}{3}\Lambda r_{2}^{2}}\big)\Big|+\Big|\log\Big(\frac{2\partial_{v}r_{1}}{1-\frac{2m_{1}}{r_{1}}}\Big)-\log\Big(\frac{2\partial_{v}r_{2}}{1-\frac{2m_{2}}{r_{2}}}\Big)\Big|+\label{eq:UpperBoundNonTrappingForCauchyStabilityToShowRed}\\
+\Big|\log\Big(\frac{1-\frac{2m_{1}}{r_{1}}}{1-\frac{1}{3}\Lambda r_{1}^{2}}\Big)-\log\Big(\frac{1-\frac{2m_{2}}{r_{2}}}{1-\frac{1}{3}\Lambda r_{2}^{2}}\Big)\Big|+\sqrt{-\Lambda}|\tilde{m}_{1}-\tilde{m}_{2}|\Bigg\}+\nonumber \\
+\sup_{\bar{u}}\int_{\{u=\bar{u}\}\cap\mathcal{W}_{u_{1};u_{2}}}\big|r_{1}(T_{vv})_{1}-r_{2}(T_{vv})_{2}\big|\, dv+\sup_{\bar{v}}\int_{\{v=\bar{v}\}\cap\mathcal{W}_{u_{1};u_{2}}}\big|r_{1}(T_{uu})_{1}-r_{2}(T_{uu})_{2}\big|\, du & \le\nonumber \\
\le\exp\big(\exp\big(C_{1} & (1+C_{0})\big)\big)\text{\textgreek{d}}_{u_{1}},\nonumber 
\end{align}
where 
\begin{align}
\text{\textgreek{d}}_{u_{1}}\doteq & \sup_{v\in(u_{1}+v_{1},u_{1}+v_{2})}\Bigg\{\Big|\log\big(\frac{\text{\textgreek{W}}_{1}^{2}}{1-\frac{1}{3}\Lambda r_{1}^{2}}\big)\big|_{(u_{1},v)}-\log\big(\frac{\text{\textgreek{W}}_{2}^{2}}{1-\frac{1}{3}\Lambda r_{2}^{2}}\big)\big|_{(u_{1},v)}\Big|+\Big|\log\Big(\frac{2\partial_{v}r_{1}}{1-\frac{2m_{1}}{r_{1}}}\Big)\big|_{(u_{1},v)}-\log\Big(\frac{2\partial_{v}r_{2}}{1-\frac{2m_{2}}{r_{2}}}\Big)\big|_{(u_{1},v)}\Big|+\label{eq:CauchyTransportedInitialDistance}\\
 & +\Big|\log\Big(\frac{1-\frac{2m_{_{1}}}{r_{1}}}{1-\frac{1}{3}\Lambda r_{1}^{2}}\Big)\big|_{(u_{1},v)}-\log\Big(\frac{1-\frac{2m_{_{2}}}{r_{2}}}{1-\frac{1}{3}\Lambda r_{2}^{2}}\Big)\big|_{(u_{1},v)}\Big|+\sqrt{-\Lambda}\Big|\tilde{m}_{1}\big|_{(u_{1},v)}-\tilde{m}_{2}\big|_{(u_{1},v)}\Big|\Bigg\}+\nonumber \\
 & +\sup_{v\in[u_{1}+v_{1},u_{1}+v_{2})}(-\Lambda)\int_{u_{1}+v_{1}}^{u_{1}+v_{2}}\Bigg|\frac{r_{1}^{2}(T_{vv})_{1}|_{(u_{1},\bar{v})}}{\big(|\text{\textgreek{r}}_{1}|_{(u_{1},\bar{v})}-\text{\textgreek{r}}_{1}|_{(u_{1},v)}|+\text{\textgreek{r}}_{1}|_{(u_{1},u_{1}+v_{1})}\big)\partial_{v}\text{\textgreek{r}}_{1}|_{(u_{1},\bar{v})}}-\nonumber \\
 & \hphantom{+\sup_{v\in[u_{1}+v_{1},u_{1}+v_{2})}(-\Lambda)\int_{u_{1}+v_{1}}^{u_{1}+v_{2}}\Bigg|}-\frac{r_{2}^{2}(T_{vv})_{2}|_{(u_{1},\bar{v})}}{\big(|\text{\textgreek{r}}_{2}|_{(u_{1},\bar{v})}-\text{\textgreek{r}}_{2}|_{(u_{1},v)}|+\text{\textgreek{r}}_{2}|_{(u_{1},u_{1}+v_{1})}\big)\partial_{v}\text{\textgreek{r}}_{2}|_{(u_{1},\bar{v})}}\Bigg|\, d\bar{v}\nonumber 
\end{align}
(with $\text{\textgreek{r}}_{i}$ defined in terms of $r_{i}$ by
(\ref{eq:RhoVariable})) measures the distance of the initial data
induced by $(r_{1},\text{\textgreek{W}}_{1}^{2},\bar{f}_{in1},\bar{f}_{out1})$
and $(r_{2},\text{\textgreek{W}}_{2}^{2},\bar{f}_{in2},\bar{f}_{out2})$
on $u=u_{1}$. Note that, in the general case when $u_{*}$ does not
satisfy (\ref{eq:u_*small}), the bound (\ref{eq:UpperBoundNonTrappingForCauchyStabilityToShow})
follows by applying (\ref{eq:UpperBoundNonTrappingForCauchyStabilityToShowRed})
successively on intervals of the form $\{u_{1}^{(n)}\le u\le u_{2}^{(n)}\}$,
where $u_{2}^{(n)}=n\cdot(v_{2}-v_{1})$ and $u_{1}^{(n)}=u_{2}^{(n-1)}$.
\begin{rem*}
Notice that the bound (\ref{eq:UpperBoundNonTrappingForCauchyStability})
yields the following bound for $\partial_{v}\tan^{-1}\Big(\sqrt{-\frac{\Lambda}{3}}r\Big)$:
\[
\Big|\log\Big(\sqrt{-\frac{3}{\Lambda}}\partial_{v}\tan^{-1}\Big(\sqrt{-\frac{\Lambda}{3}}r\Big)\Big)\Big|=\Big|\log\Big(\frac{\partial_{v}r}{1-\frac{1}{3}\Lambda r^{2}}\Big)\Big|\le C_{0}.
\]
Thus, in view of the fact that $\tan^{-1}\Big(\sqrt{-\frac{\Lambda}{3}}r\Big)\big|_{\text{\textgreek{g}}_{0}}=\tan^{-1}\Big(\sqrt{-\frac{\Lambda}{3}}r_{0}\Big)$
and $\tan^{-1}\Big(\sqrt{-\frac{\Lambda}{3}}r\Big)\big|_{+\infty}=\frac{\text{\textgreek{p}}}{2}$,
we readily infer that $|v_{2}-v_{1}|$ must necessarily satisfy the
bound: 
\begin{equation}
\Big|\log\big(\sqrt{-\Lambda}|v_{2}-v_{1}|\big)\Big|\le2C_{0}.\label{eq:UpperBoundWidth}
\end{equation}

\end{rem*}
Let us define the variables $\text{\textgreek{r}}_{i}$, $\text{\textgreek{k}}_{i}$,
$\bar{\text{\textgreek{k}}}_{i}$, $\text{\textgreek{t}}_{i}$ and
$\bar{\text{\textgreek{t}}}_{i}$, $i=1,2$, by (\ref{eq:RhoVariable})--(\ref{eq:TVariable}).
Recall that these variables satisfy equations (\ref{eq:DerivativeInUDirectionKappaRenormalised})--(\ref{eq:ConservationTau})
and the boundary conditions (\ref{eq:RhoBoundary})--(\ref{eq:ConservationTildeMAxis}).
In view of (\ref{eq:RelationHawkingMass}), the bounds (\ref{eq:UpperBoundNonTrappingForCauchyStability})
and (\ref{eq:UpperBoundNonTrappingForCauchyStabilityBootstrap}) yield
that: 
\begin{equation}
\max_{i=1,2}\sup_{\mathcal{W}_{u_{1};u_{2}}}\Big\{\big|\log(\text{\textgreek{k}}_{i})\big|+\big|\log(\bar{\text{\textgreek{k}}}_{i})\big|+\big|\log\big((-\Lambda)^{-\frac{1}{2}}\partial_{v}\text{\textgreek{r}}_{i}\big)\big|+\big|\log\big((-\Lambda)^{-\frac{1}{2}}\partial_{u}\text{\textgreek{r}}_{i}\big)\big|+\Big|\log\Big(\frac{1-\frac{2\sqrt{-\frac{\Lambda}{3}}\tilde{m}_{i}}{\tan\text{\textgreek{r}}_{i}}+\tan^{2}\text{\textgreek{r}}_{i}}{1+\tan^{2}\text{\textgreek{r}}_{i}}\Big)\Big|+\sqrt{-\Lambda}|\tilde{m}_{i}|\Bigg\}\le10C_{0}.\label{eq:UpperBoundNonTrappingUseful}
\end{equation}
Furthermore, using the fact that, for $i=1,2$, $\bar{\text{\textgreek{t}}}_{i}$
are functions of $v$ and $\text{\textgreek{t}}_{i}$ are functions
of $u$ (in view of equations (\ref{eq:ConservationTauBar}) and (\ref{eq:ConservationTau}),
respectively) from (\ref{eq:UpperBoundNonTrappingForCauchyStability}),
(\ref{eq:UpperBoundNonTrappingForCauchyStabilityBootstrap}) and (\ref{eq:UpperBoundNonTrappingUseful})
(using also (\ref{eq:UpperBoundWidth})) we obtain the following bound
for $\bar{\text{\textgreek{t}}}_{i},\text{\textgreek{t}}_{i}$ on
$\mathcal{W}_{u_{1};u_{2}}$: 
\begin{equation}
\sup_{\bar{u}\in(u_{1},u_{2})}\sup_{v\in(\bar{u}+v_{1},\bar{u}+v_{2})}\int_{\bar{u}+v_{1}}^{\bar{u}+v_{2}}\frac{\bar{\text{\textgreek{t}}}_{i}(\bar{u},v)}{|v-\bar{v}|+r_{0}}\, dv+\sup_{\bar{v}\in(v_{1},v_{2}+u_{2}-u_{1})}\sup_{\bar{u}\in(u_{1},u_{2})}\int_{u_{1}}^{u_{2}}\frac{\bar{\text{\textgreek{t}}}_{i}(u,\bar{v})}{|u-\bar{u}|+r_{0}}\, du\le e^{10(C_{0}+1)}.\label{eq:UpperBoundUsefulTauTauBar}
\end{equation}
Moreover, from (\ref{eq:CauchyTransportedInitialDistance}), (\ref{eq:UpperBoundNonTrappingUseful})
and (\ref{eq:UpperBoundUsefulTauTauBar}) we can estimate on $u=u_{1}$:
\begin{align}
\sup_{v\in(u_{1}+v_{1},u_{1}+v_{2})}\Bigg\{ & \Big|\text{\textgreek{k}}_{1}(u_{1},v)-\text{\textgreek{k}}_{2}(u_{1},v)\Big|+\Big|\text{\textgreek{k}}_{1}^{-1}(u_{1},v)-\text{\textgreek{k}}_{2}^{-1}(u_{1},v)\Big|+\label{eq:CauchyTransportedInitialDistanceUseful}\\
 & +\Big|\partial_{v}\text{\textgreek{r}}_{1}(u_{1},v)-\partial_{v}\text{\textgreek{r}}_{2}(u_{1},v)\Big|+\Big|(\partial_{v}\text{\textgreek{r}}_{1})^{-1}(u_{1},v)-(\partial_{v}\text{\textgreek{r}}_{2})^{-1}(u_{1},v)\Big|+\nonumber \\
 & +\Big|\Big(\frac{1-\frac{2\sqrt{-\frac{\Lambda}{3}}\tilde{m}_{1}}{\tan\text{\textgreek{r}}_{1}}+\tan^{2}\text{\textgreek{r}}_{1}}{1+\tan^{2}\text{\textgreek{r}}_{1}}\Big)(u_{1},v)-\Big(\frac{1-\frac{2\sqrt{-\frac{\Lambda}{3}}\tilde{m}_{2}}{\tan\text{\textgreek{r}}_{2}}+\tan^{2}\text{\textgreek{r}}_{2}}{1+\tan^{2}\text{\textgreek{r}}_{2}}\Big)(u_{1},v)\Big|+\nonumber \\
 & +\Big|\Big(\frac{1-\frac{2\sqrt{-\frac{\Lambda}{3}}\tilde{m}_{1}}{\tan\text{\textgreek{r}}_{1}}+\tan^{2}\text{\textgreek{r}}_{1}}{1+\tan^{2}\text{\textgreek{r}}_{1}}\Big)^{-1}(u_{1},v)-\Big(\frac{1-\frac{2\sqrt{-\frac{\Lambda}{3}}\tilde{m}_{2}}{\tan\text{\textgreek{r}}_{2}}+\tan^{2}\text{\textgreek{r}}_{2}}{1+\tan^{2}\text{\textgreek{r}}_{2}}\Big)^{-1}(u_{1},v)\Big|+\nonumber \\
 & +\sqrt{-\Lambda}\Big|\tilde{m}_{1}\big|_{(u_{1},v)}-\tilde{m}_{2}\big|_{(u_{1},v)}\Big|\Bigg\}+\sup_{v\in(u_{1}+v_{1},u_{1}+v_{2})}\int_{v_{1}}^{v_{2}}\frac{\big|\bar{\text{\textgreek{t}}}_{1}(u_{1},v)-\bar{\text{\textgreek{t}}}_{2}(u_{1},v)\big|}{|v-\bar{v}|+r_{0}}\, dv\le e^{100(C_{0}+1)}\text{\textgreek{d}}_{u_{1}}.\nonumber 
\end{align}
Thus, the proof of (\ref{eq:UpperBoundNonTrappingForCauchyStabilityToShowRed})
(and, therefore, the proof of Theorem \ref{thm:CauchyStability})
will follow (in view of (\ref{eq:UpperBoundNonTrappingUseful}) and
the boundary conditions (\ref{eq:RhoBoundary}) for $\text{\textgreek{r}}_{i}$)
by showing that 
\begin{align}
\sup_{\mathcal{W}_{u_{1};u_{2}}}\Bigg\{ & \Big|\log(\text{\textgreek{k}}_{1})-\log(\text{\textgreek{k}}_{2})\Big|+\Big|\log(\bar{\text{\textgreek{k}}}_{1})-\log(\bar{\text{\textgreek{k}}}_{2})\Big|+\Big|\log\big((-\Lambda)^{-\frac{1}{2}}\partial_{v}\text{\textgreek{r}}_{1}\big)-\log\big((-\Lambda)^{-\frac{1}{2}}\partial_{v}\text{\textgreek{r}}_{2}\big)\Big|+\label{eq:BoundToShowUseful}\\
 & +\Big|\log\big((-\Lambda)^{-\frac{1}{2}}\partial_{u}\text{\textgreek{r}}_{1}\big)-\log\big((-\Lambda)^{-\frac{1}{2}}\partial_{u}\text{\textgreek{r}}_{2}\big)\Big|+\sqrt{-\Lambda}\Big|\tilde{m}_{1}-\tilde{m}_{2}\Big|\Bigg\}+\nonumber \\
+ & \sup_{\bar{u}\in(u_{1},u_{2})}\int_{\{u=\bar{u}\}\cap\mathcal{W}_{u_{1};u_{2}}}\frac{\big|\bar{\text{\textgreek{t}}}_{1}(\bar{u},v)-\bar{\text{\textgreek{t}}}_{2}(\bar{u},v)\big|}{|v-v_{1}-\bar{u}|+r_{0}}\, dv+\sup_{\bar{v}\in(u_{1}+v_{1},u_{2}+v_{2})}\int_{\{v=\bar{v}\}\cap\mathcal{W}_{u_{1};u_{2}}}\frac{\big|\text{\textgreek{t}}_{1}(u,\bar{v})-\text{\textgreek{t}}_{2}(u,\bar{v})\big|}{|u+v_{1}-v|+r_{0}}\, du\nonumber \\
 & \hphantom{\sup_{\bar{u}\in(u_{1},u_{2})}\int_{\{u=\bar{u}\}\cap\mathcal{W}_{u_{1};u_{2}}}\frac{\big|\bar{\text{\textgreek{t}}}_{1}(\bar{u},v)-\bar{\text{\textgreek{t}}}_{2}(\bar{u},v)\big|}{|v-v_{1}-\bar{u}|+r_{0}}\, dv+\sup_{\bar{v}\in(u_{1}+v_{1},u_{2}+v_{2})}\int_{\{u=\bar{u}\}\cap\mathcal{W}_{u_{1};u_{2}}}blas}\le\exp\big(\exp(C_{1}^{\frac{4}{5}}(1+C_{0})\big)\text{\textgreek{d}}_{u_{1}}.\nonumber 
\end{align}

The differences $\text{\textgreek{r}}_{1}-\text{\textgreek{r}}_{2}$,
$\text{\textgreek{k}}_{1}-\text{\textgreek{k}}_{2}$, $\bar{\text{\textgreek{k}}}_{1}-\bar{\text{\textgreek{k}}}_{2}$,
$\text{\textgreek{t}}_{1}-\text{\textgreek{t}}_{2}$, $\bar{\text{\textgreek{t}}}_{1}-\bar{\text{\textgreek{t}}}_{3}$
satisfy on $\mathcal{W}_{u_{*}}$ the system (\ref{eq:DerivativeInUDirectionKappaRenormalised-1})--(\ref{eq:ConservationTau-1}),
i.\,e.:
\begin{align}
\partial_{u}\big(\log(\text{\textgreek{k}}_{2})-\log(\text{\textgreek{k}}_{1})\big)= & -\sqrt{-\frac{\Lambda}{3}}\Big(F_{1}\big(\tan\text{\textgreek{r}}_{2};\sqrt{-\frac{\Lambda}{3}}\tilde{m}_{2}\big)\bar{\text{\textgreek{k}}}_{2}^{-1}\text{\textgreek{t}}_{2}-F_{1}\big(\tan\text{\textgreek{r}}_{1};\sqrt{-\frac{\Lambda}{3}}\tilde{m}_{1}\big)\bar{\text{\textgreek{k}}}_{1}^{-1}\text{\textgreek{t}}_{1}\Big),\label{eq:DerivativeInUDirectionKappaRenormalisedCauchy}\\
\partial_{v}\big(\log(\bar{\text{\textgreek{k}}}_{2})-\log(\bar{\text{\textgreek{k}}}_{1})\big)= & \sqrt{-\frac{\Lambda}{3}}\Big(F_{1}\big(\tan\text{\textgreek{r}}_{2};\sqrt{-\frac{\Lambda}{3}}\tilde{m}_{2}\big)\bar{\text{\textgreek{k}}}_{2}^{-1}\text{\textgreek{t}}_{2}-F_{1}\big(\tan\text{\textgreek{r}}_{1};\sqrt{-\frac{\Lambda}{3}}\tilde{m}_{1}\big)\bar{\text{\textgreek{k}}}_{1}^{-1}\text{\textgreek{t}}_{1}\Big),\label{eq:DerivativeInVDirectionKappaBarRenormalisedCauchy}\\
\partial_{u}\partial_{v}\big(\text{\textgreek{r}}_{2}-\text{\textgreek{r}}_{1}\big)= & (-\frac{\Lambda}{3})\Big(F_{2}\big(\tan\text{\textgreek{r}}_{2};\sqrt{-\frac{\Lambda}{3}}\tilde{m}_{2}\big)\text{\textgreek{k}}_{2}\bar{\text{\textgreek{k}}}_{2}-F_{2}\big(\tan\text{\textgreek{r}}_{1};\sqrt{-\frac{\Lambda}{3}}\tilde{m}_{1}\big)\text{\textgreek{k}}_{1}\bar{\text{\textgreek{k}}}_{1}\Big),\label{eq:EquationRForRhoRenormalisedCauchy}\\
\partial_{u}\big(\tilde{m}_{2}-\tilde{m}_{1}\big)= & -4\pi\big(\bar{\text{\textgreek{k}}}_{2}^{-1}\text{\textgreek{t}}_{2}-\bar{\text{\textgreek{k}}}_{1}^{-1}\text{\textgreek{t}}_{1}\big),\label{eq:DerivativeInUtildeMRenormalisedCauchy}\\
\partial_{u}\big(\bar{\text{\textgreek{t}}}_{2}-\bar{\text{\textgreek{t}}}_{1}\big)= & 0,\label{eq:ConservationTauBarCauchy}\\
\partial_{v}\big(\text{\textgreek{t}}_{2}-\text{\textgreek{t}}_{1}\big)= & 0.\label{eq:ConservationTauCauchy}
\end{align}

Integrating equations \ref{eq:ConservationTauBarCauchy}--\ref{eq:ConservationTauCauchy}
and using the boundary conditions (\ref{eq:TauBoundary}) for $\text{\textgreek{t}}_{1},\text{\textgreek{t}}_{2}$
and $\bar{\text{\textgreek{t}}}_{1},\bar{\text{\textgreek{t}}}_{2}$,
we infer that, for any $(u,v)\in\mathcal{W}_{u_{1};u_{2}}$: 
\begin{equation}
\big(\text{\textgreek{t}}_{1}-\text{\textgreek{t}}_{2}\big)(u,v)=\big(\bar{\text{\textgreek{t}}}_{1}-\bar{\text{\textgreek{t}}}_{2}\big)(u_{1},u+v_{1})\label{eq:OutgoingDifferenceCauchy}
\end{equation}
and 
\begin{equation}
\big(\bar{\text{\textgreek{t}}}_{1}-\bar{\text{\textgreek{t}}}_{2}\big)(u,v)=\begin{cases}
\big(\bar{\text{\textgreek{t}}}_{1}-\bar{\text{\textgreek{t}}}_{2}\big)(u_{1},v), & v\le u_{1}+v_{2},\\
\big(\bar{\text{\textgreek{t}}}_{1}-\bar{\text{\textgreek{t}}}_{2}\big)(u_{1},v-v_{2}+v_{1}), & v\ge u_{1}+v_{2}.
\end{cases}\label{eq:IngoingDifferenceCauchy}
\end{equation}
In view of (\ref{eq:UpperBoundNonTrappingUseful}) and the initial
bound (\ref{eq:CauchyTransportedInitialDistanceUseful}), from (\ref{eq:OutgoingDifferenceCauchy})--(\ref{eq:IngoingDifferenceCauchy})
we can readily estimate: 
\begin{equation}
\sup_{\bar{u}\in(u_{1},u_{2})}\int_{\{u=\bar{u}\}\cap\mathcal{W}_{u_{1};u_{2}}}\frac{\big|\bar{\text{\textgreek{t}}}_{1}(\bar{u},v)-\bar{\text{\textgreek{t}}}_{2}(\bar{u},v)\big|}{|v-v_{1}-\bar{u}|+r_{0}}\, dv+\sup_{\bar{v}\in(u_{1}+v_{1},u_{2}+v_{2})}\int_{\{v=\bar{v}\}\cap\mathcal{W}_{u_{1};u_{2}}}\frac{\big|\text{\textgreek{t}}_{1}(u,\bar{v})-\text{\textgreek{t}}_{2}(u,\bar{v})\big|}{|u+v_{1}-v|+r_{0}}\, du\le e^{C_{1}^{\frac{1}{2}}(1+C_{0})}\text{\textgreek{d}}_{0}.\label{eq:DoneForTautauBar}
\end{equation}

Integrating equations (\ref{eq:DerivativeInUDirectionKappaRenormalisedCauchy}),
(\ref{eq:EquationRForRhoRenormalisedCauchy}) and (\ref{eq:DerivativeInUtildeMRenormalisedCauchy})
in $u$ and equations (\ref{eq:DerivativeInVDirectionKappaBarRenormalisedCauchy})
and (\ref{eq:EquationRForRhoRenormalisedCauchy}) in $v$, using the
boundary conditions (\ref{eq:RhoBoundary}), (\ref{eq:KappaBoundary})
and (\ref{eq:ConservationTildeMAxis}) on $\text{\textgreek{g}}_{0}$
and $\mathcal{I}$ as well as the bounds (\ref{eq:BoundednessF_1,F_2,F_3})
for $F_{1},F_{2}$ and (\ref{eq:UpperBoundNonTrappingUseful}), (\ref{eq:UpperBoundUsefulTauTauBar})
and (\ref{eq:CauchyTransportedInitialDistanceUseful}) for $\text{\textgreek{k}}_{i},\bar{\text{\textgreek{k}}}_{i},\text{\textgreek{r}}_{i},\tilde{m}_{i},\text{\textgreek{t}}_{i},\bar{\text{\textgreek{t}}}_{i}$,
we readily infer that, for any $(u,v)\in\mathcal{W}_{u_{1};u_{2}}$:
\begin{align}
\big|\log\big(\text{\textgreek{k}}_{2}(u,v)\big)-\log\big(\text{\textgreek{k}}_{1}(u,v)\big)\big| & \le\sqrt{-\Lambda}e^{C_{1}^{\frac{1}{2}}(1+C_{0})}\Big\{\int_{\max\{u_{1},v-v_{2}-u_{1}\}}^{u}\Big(\overline{\mathfrak{M}}\cdot\mathfrak{D}+|\text{\textgreek{t}}_{2}-\text{\textgreek{t}}_{1}|\Big)\Big|_{(\bar{u},v)}\, d\bar{u}+\label{eq:FirstDifferenceForCauchy}\\
 & \hphantom{\le C_{M,\text{\textgreek{r}}_{0}}\sqrt{-\Lambda}\Bigg\{++}+\text{\textgreek{q}}_{v>v_{2}}(v)\int_{v-v_{2}+v_{1}}^{v}\Big(\mathfrak{M}\cdot\mathfrak{D}+|\bar{\text{\textgreek{t}}}_{2}-\bar{\text{\textgreek{t}}}_{1}|\Big)\Big|_{(v-v_{2}-u_{1},\bar{v})}\, d\bar{v}+\nonumber \\
 & \hphantom{\le C_{M,\text{\textgreek{r}}_{0}}\sqrt{-\Lambda}\Bigg\{++}+\text{\textgreek{q}}_{v>v_{2}}(v)\int_{u_{1}}^{v-v_{2}-u_{1}}\Big(\overline{\mathfrak{M}}\cdot\mathfrak{D}+|\text{\textgreek{t}}_{2}-\text{\textgreek{t}}_{1}|\Big)\Big|_{(\bar{u},v_{1}+v-v_{2})}\, d\bar{u}\Big\}+e^{C_{1}^{\frac{1}{2}}(1+C_{0})}\text{\textgreek{d}}_{u_{1}},\nonumber \\
\big|\log\big(\bar{\text{\textgreek{k}}}_{2}(u,v)\big)-\log\big(\bar{\text{\textgreek{k}}}_{1}(u,v)\big)\big| & \le\sqrt{-\Lambda}e^{C_{1}^{\frac{1}{2}}(1+C_{0})}\Big\{\int_{v_{1}+u}^{v}\Big(\mathfrak{M}\cdot\mathfrak{D}+|\bar{\text{\textgreek{t}}}_{2}-\bar{\text{\textgreek{t}}}_{1}|\Big)\Big|_{(u,\bar{v})}\, d\bar{v}+\\
 & \hphantom{\le\sqrt{-\Lambda}e^{C_{1}^{\frac{1}{2}}(1+C_{0})}\Big\{}+\int_{u_{1}}^{u}\Big(\overline{\mathfrak{M}}\cdot\mathfrak{D}+|\text{\textgreek{t}}_{2}-\text{\textgreek{t}}_{1}|\Big)\Big|_{(\bar{u},v_{1}+u)}\, d\bar{u}\Big\}+e^{C_{1}^{\frac{1}{2}}(1+C_{0})}\text{\textgreek{d}}_{u_{1}},\\
\big|\partial_{v}\text{\textgreek{r}}_{2}(u,v)-\partial_{v}\text{\textgreek{r}}_{1}(u,v)\big| & \le(-\Lambda)e^{C_{1}^{\frac{1}{2}}(1+C_{0})}\Big\{\int_{\max\{u_{1},v-v_{2}-u_{1}\}}^{u}\Big(\overline{\mathfrak{M}}\cdot\mathfrak{D}+|\text{\textgreek{t}}_{2}-\text{\textgreek{t}}_{1}|\Big)\Big|_{(\bar{u},v)}\, d\bar{u}+\\
 & \hphantom{\le C_{M,\text{\textgreek{r}}_{0}}\sqrt{-\Lambda}\Bigg\{++}+\text{\textgreek{q}}_{v>v_{2}}(v)\int_{v-v_{2}+v_{1}}^{v}\Big(\mathfrak{M}\cdot\mathfrak{D}+|\bar{\text{\textgreek{t}}}_{2}-\bar{\text{\textgreek{t}}}_{1}|\Big)\Big|_{(v-v_{2}-u_{1},\bar{v})}\, d\bar{v}+\nonumber \\
 & \hphantom{\le C_{M,\text{\textgreek{r}}_{0}}\sqrt{-\Lambda}\Bigg\{++}+\text{\textgreek{q}}_{v>v_{2}}(v)\int_{u_{1}}^{v-v_{2}-u_{1}}\Big(\overline{\mathfrak{M}}\cdot\mathfrak{D}+|\text{\textgreek{t}}_{2}-\text{\textgreek{t}}_{1}|\Big)\Big|_{(\bar{u},v_{1}+v-v_{2})}\, d\bar{u}\Big\}+\sqrt{-\Lambda}e^{C_{1}^{\frac{1}{2}}(1+C_{0})}\text{\textgreek{d}}_{u_{1}},\nonumber \\
\big|\partial_{u}\text{\textgreek{r}}_{2}(u,v)-\partial_{u}\text{\textgreek{r}}_{1}(u,v)\big| & \le(-\Lambda)e^{C_{1}^{\frac{1}{2}}(1+C_{0})}\Big\{\int_{v_{1}+u}^{v}\Big(\mathfrak{M}\cdot\mathfrak{D}+|\bar{\text{\textgreek{t}}}_{2}-\bar{\text{\textgreek{t}}}_{1}|\Big)\Big|_{(u,\bar{v})}\, d\bar{v}+\\
 & \hphantom{\le\sqrt{-\Lambda}e^{C_{1}^{\frac{1}{2}}(1+C_{0})}\Big\{}+\int_{u_{1}}^{u}\Big(\overline{\mathfrak{M}}\cdot\mathfrak{D}+|\text{\textgreek{t}}_{2}-\text{\textgreek{t}}_{1}|\Big)\Big|_{(\bar{u},v_{1}+u)}\, d\bar{u}\Big\}+\sqrt{-\Lambda}e^{C_{1}^{\frac{1}{2}}(1+C_{0})}\text{\textgreek{d}}_{u_{1}},\nonumber \\
\big|\tilde{m}_{2}(u,v)-\tilde{m}_{1}(u,v)\big| & \le e^{C_{1}^{\frac{1}{2}}(1+C_{0})}\Big\{\int_{\max\{u_{1},v-v_{2}-u_{1}\}}^{u}\Big(\overline{\mathfrak{M}}\cdot\mathfrak{D}+|\text{\textgreek{t}}_{2}-\text{\textgreek{t}}_{1}|\Big)\Big|_{(\bar{u},v)}\, d\bar{u}+(-\Lambda)^{-\frac{1}{2}}e^{C_{1}^{\frac{1}{2}}(1+C_{0})}\text{\textgreek{d}}_{u_{1}},\label{eq:LastDifferenceForCauchy}
\end{align}
where 
\begin{equation}
\text{\textgreek{q}}_{v>v_{2}}(v)=\begin{cases}
1, & \mbox{ if }v>v_{2},\\
0, & \mbox{ if }v\le v_{2},
\end{cases}
\end{equation}
\begin{align}
\mathfrak{D}(u,v)\doteq & \big|\log(\text{\textgreek{k}}_{2})-\log(\text{\textgreek{k}}_{1})\big|(u,v)+\big|\log(\bar{\text{\textgreek{k}}}_{2})-\log(\bar{\text{\textgreek{k}}}_{1})\big|(u,v)+\label{eq:DifferenceToEstimate}\\
 & +\big|\text{\textgreek{r}}_{2}-\text{\textgreek{r}}_{1}\big|(u,v)+\sqrt{-\Lambda}\big|\tilde{m}_{2}-\tilde{m}_{1}\big|(u,v)\nonumber 
\end{align}
and the functions $\overline{\mathfrak{M}}(u,v)\ge0$, $\mathfrak{M}(u,v)\ge0$
satisfy the bound: 
\begin{equation}
\sup_{\bar{v}\in(u_{1}+v_{1},u_{2}+v_{2})}\int_{\{v=\bar{v}\}\cap\mathcal{W}_{u_{1};u_{2}}}\overline{\mathfrak{M}}(u,\bar{v})\, du+\sup_{\bar{u}\in(u_{1},u_{2})}\int_{\{u=\bar{u}\}\cap\mathcal{W}_{u_{1};u_{2}}}\mathfrak{M}(\bar{u},v)\, dv\le1.\label{eq:BoundMErrors}
\end{equation}

Defining the function $\mathcal{X}:(2u_{1}+v_{1},2u_{2}+v_{2})\rightarrow[0,+\infty)$
by the relation 
\begin{align}
\mathcal{X}(t)\doteq\sup_{\{u+v\le t\}\cap\mathcal{W}_{u_{1};u_{2}}}\Big\{ & \big|\log(\text{\textgreek{k}}_{2})-\log(\text{\textgreek{k}}_{1})\big|(u,v)+\big|\log(\bar{\text{\textgreek{k}}}_{2})-\log(\bar{\text{\textgreek{k}}}_{1})\big|(u,v)+\label{eq:DifferenceToEstimate-1}\\
 & +(-\Lambda)^{-\frac{1}{2}}|\partial_{v}\text{\textgreek{r}}_{2}-\partial_{v}\text{\textgreek{r}}_{1}|(u,v)+(-\Lambda)^{\frac{1}{2}}|\partial_{u}\text{\textgreek{r}}_{2}-\partial_{u}\text{\textgreek{r}}_{1}|(u,v)+\nonumber \\
 & +\big|\text{\textgreek{r}}_{2}-\text{\textgreek{r}}_{1}\big|(u,v)+\sqrt{-\Lambda}\big|\tilde{m}_{2}-\tilde{m}_{1}\big|(u,v)\Big\},\nonumber 
\end{align}
from (\ref{eq:FirstDifferenceForCauchy})--(\ref{eq:LastDifferenceForCauchy})
(using also (\ref{eq:CauchyTransportedInitialDistanceUseful}), (\ref{eq:OutgoingDifferenceCauchy}),
(\ref{eq:IngoingDifferenceCauchy}) and (\ref{eq:BoundMErrors}))
we readily obtain that, for all $t\in(2u_{1}+v_{1},2u_{2}+v_{2})$:
\begin{equation}
\mathcal{X}(t)\le e^{C_{1}^{2/3}(1+C_{0})}\int_{2u_{1}+v_{1}}^{t}\mathcal{Y}(\bar{t})\mathcal{X}(\bar{t})\, d\bar{t}+e^{C_{1}^{2/3}(1+C_{0})}\text{\textgreek{d}}_{u_{1}},\label{eq:ForGronwalX}
\end{equation}
where the function $\mathcal{Y}\ge0$ satisfies 
\begin{equation}
\int_{2u_{1}+v_{1}}^{2u_{2}+v_{2}}\mathcal{Y}(t)\, dt\le1.
\end{equation}
From (\ref{eq:ForGronwalX}), an application of Gronwall's inequality
readily yields that 
\begin{equation}
\sup_{t\in(2u_{1}+v_{1},2u_{2}+v_{2})}\mathcal{X}(t)\le\exp\big(\exp(C_{1}^{3/4}(1+C_{0}))\big)\text{\textgreek{d}}_{u_{1}}.\label{eq:AlmostDone}
\end{equation}

The bound (\ref{eq:BoundToShowUseful}) follows readily from (\ref{eq:DoneForTautauBar})
and (\ref{eq:AlmostDone}). Therefore, the proof of the Theorem \ref{thm:CauchyStability}
is complete. \qed

\subsection{\label{sub:Proof-of-Cauchy-AdS}Proof of Corollary \ref{cor:CauchyStabilityOfAdS}}

By possibly applying a suitable gauge transformation of the form $(u,v)\rightarrow(U(u),V(v))$,
such that $U(0)=0$, $V(0)=0$, $V=U$ when $v=u$ (i.\,e.~on $\text{\textgreek{g}}_{0}$)
and $V=U+V_{0}$ when $v=u+v_{0}$ (i.\,e.~on $\mathcal{I}$), we
will assume without loss of generality that the initial data satisfy
the gauge condition 
\begin{equation}
\frac{\partial_{v}r_{/}}{1-\frac{2m_{/}}{r_{/}}}=\frac{1}{2}.\label{eq:GaugeFixingInitialDataAdS}
\end{equation}

In view of (\ref{eq:GaugeFixingInitialDataAdS}), equation (\ref{eq:DerivativeInVDirectionKappaBar}),
combined with the boundary condition 
\begin{equation}
-(\partial_{u}r)_{/}(0)=\partial_{v}r_{/}(0),
\end{equation}
 imply that 
\begin{equation}
-2\frac{(\partial_{u}r)_{/}}{1-\frac{2m_{/}}{r_{/}}}(v)=\exp\Big(4\pi\int_{0}^{v}\frac{r_{/}(T_{vv})_{/}}{(\partial_{v}r)_{/}}\, d\bar{v}\Big).
\end{equation}
Therefore, (\ref{eq:RelationHawkingMass}) and (\ref{eq:GaugeFixingInitialDataAdS})
yield 
\begin{equation}
\frac{\text{\textgreek{W}}_{/}^{2}}{1-\frac{1}{3}\Lambda r_{/}^{2}}=\frac{1-\frac{2m_{/}}{r_{/}}}{1-\frac{1}{3}\Lambda r_{/}^{2}}\exp\Big(4\pi\int_{0}^{v}\frac{r_{/}(T_{vv})_{/}}{(\partial_{v}r)_{/}}\, d\bar{v}\Big).\label{eq:OmegaBoundCuchyStability}
\end{equation}

From (\ref{eq:SmallnessForCauchyStability}), (\ref{eq:GaugeFixingInitialDataAdS})
and (\ref{eq:OmegaBoundCuchyStability}) we thus infer that, for some
absolute constant $C>0$: 
\begin{equation}
\sup_{v\in[v_{1},v_{2})}\Bigg\{\Big|\log\big(\frac{\text{\textgreek{W}}_{/}^{2}}{1-\frac{1}{3}\Lambda r_{/}^{2}}\big)\Big|+\Big|\log\Big(\frac{2\partial_{v}r_{/}}{1-\frac{2m_{/}}{r_{/}}}\Big)\Big|+\Big|\log\Big(\frac{1-\frac{2m_{/}}{r_{/}}}{1-\frac{1}{3}\Lambda r_{/}^{2}}\Big)\Big|+\sqrt{-\Lambda}|\tilde{m}_{/}|\Bigg\}(v)+\int_{v_{1}}^{v_{2}}r_{/}(T_{vv})_{/}\, d\bar{v}\le C\text{\textgreek{e}}.\label{eq:UpperBoundInitialData-2}
\end{equation}
Applying Theorem \ref{thm:CauchyStability} for $v_{1}=0$, $v_{2}=v_{0}$,
$(r_{/2},\text{\textgreek{W}}_{/2}^{2},\bar{f}_{in/2},\bar{f}_{out/2})=(r_{/},\text{\textgreek{W}}_{/}^{2},\bar{f}_{in/},\bar{f}_{out/})$
and $(r_{/1},\text{\textgreek{W}}_{/1}^{2},\bar{f}_{in/1},\bar{f}_{out/1})=(r_{AdS},\text{\textgreek{W}}_{AdS}^{2},0,0)$,
where $(r_{AdS},\text{\textgreek{W}}_{AdS}^{2},0,0)$ are the trivial
initial data renormalised by the gauge condition (\ref{eq:GaugeFixingInitialDataAdS}),
and noting that, in this case, $\mathcal{U}_{1}=\{0<u<+\infty\}\cap\{u<v<u+v_{0}\}$
and $C_{0}=0$, we readily obtain (\ref{eq:InclusionInMaximalDomain})
and (\ref{eq:SmallnessCauchyStability}) in view of (\ref{eq:InclusionOtherdomain})
and (\ref{eq:UpperBoundNonTrappingForCauchyStability-1}). \qed

\appendix

\section{\label{sec:Ill-posedness}Ill-posedness of the spherically symmetric
Einstein--null dust system at $r=0$}

The aim of this Section is to establish a general ill-posedness result
for the spherically symmetric Einstein--null dust system in the presence
of a regular axis of symmetry. In order to state this result in its
strongest form, we will first introduce the notion of admissible $C^{0}$
spherically symmetric spacetimes $(\mathcal{M},g)$ in Section \ref{sub:C_0LorentzianSpacetimes}.
These are $C^{0}$ spherically symmetric spacetimes admitting a double-null
foliation. We will then examine the basic properties of the spherically
symmetric Einstein--null dust system on such spacetimes in Section
\ref{sub:Einstein--Null-dust-C0}. Finally, in Section \ref{sub:An-ill-posedness-result},
we will establish that smooth solutions to the spherically symmetric
Einstein--null dust system with a non-trivial axis of symmetry break
down (as admissible $C^{0}$ spherically symmetric solutions of the
Einstein--null dust system) in finite time.

\subsection{\label{sub:C_0LorentzianSpacetimes}The class of admissible $C^{0}$
spherically symmetric spacetimes}

In this section, we will introduce the notion of $C^{0}$ spherically
symmetric spacetimes and state their basic properties. While we will
only restrict to the case of $3+1$ dimensional spacetimes, the definitions
and results of this section can be immediately extended to arbitrary
dimensions. We will also introduce the notion of an admissible $C^{0}$
spherically symmetric spacetime, which are the $C^{0}$ spacetimes
on which the spherically symmetric Einstein--null dust can be rigorously
formulated.

A Lorentzian manifold $(\mathcal{M}^{3+1},g)$ will be called a $C^{0}$
spacetime if $\mathcal{M}$ is a $C^{1}$ manifold and $g$ is a $C^{0}$
Lorentzian metric on $\mathcal{M}$. We will define the notion of
spherical symmetry in the class of $C^{0}$ spacetimes as follows:
\begin{defn}
\label{def:C0-spherical-symmetry}\textgreek{A} $C^{0}$ spacetime
$(\mathcal{M}^{3+1},g)$ will be called $C^{0}$ \emph{spherically
symmetric }if there exists a $C^{1}$ action 
\begin{equation}
\mathcal{A}:SO(3)\times\mathcal{M}\rightarrow\mathcal{M}\label{eq:IsometricAction}
\end{equation}
with the following properties:

\begin{enumerate}

\item For any $a\in SO(3)$, the map 
\begin{equation}
\mathcal{A}(a,\cdot):\mathcal{M}\rightarrow\mathcal{M}
\end{equation}
defined by (\ref{eq:IsometricAction}) is an isometry with respect
to $g$.

\item For any $p\in\mathcal{M}$, the orbit $Orb(p)$ of $p$ under
the action (\ref{eq:IsometricAction}) is either the single point
$\{p\}$ or is a $2$-dimensional surface homeomorphic to $\mathbb{S}^{2}$.

\item For any $p\in\mathcal{M}$, there exists an open $SO(3)$-invariant
neighborhood $\mathcal{V}_{p}$ of $p$ such that:

\begin{enumerate}

\item In the case $Orb(p)\neq\{p\}$, there exists a $C^{1}$-diffeomorphism
\begin{equation}
\mathcal{F}:\mathcal{V}_{p}\rightarrow\mathcal{U}\times\mathbb{S}^{2}
\end{equation}
for some domain $\mathcal{U}$ in $\mathbb{R}^{2}$, such that $\mathcal{F}$
commutes the action (\ref{eq:IsometricAction}) of $SO(3)$ on $\mathcal{V}_{p}$
with the natural action of $SO(3)$ on $\mathcal{U}\times\mathbb{S}^{2}$
by rotations of $\mathbb{S}^{2}$.

\item For any $p\in\mathcal{M}$ such that $Orb(p)=\{p\}$, there
exists a $C^{1}$-diffeomorphism 
\begin{equation}
\mathcal{F}:\mathcal{V}_{p}\rightarrow\mathbb{D}^{3}\times(-1,1),
\end{equation}
where $\mathbb{D}^{3}$ is the unit $3$-disc, such that $\mathcal{F}$
commutes the action (\ref{eq:IsometricAction}) of $SO(3)$ on $\mathcal{V}_{p}$
with the natural action of $SO(3)$ on $\mathbb{D}^{3}\times(0,1)$
by rotations of $\mathbb{D}^{3}$.

\end{enumerate}

\end{enumerate}
\end{defn}
We will also define the axis of a $C^{0}$ spherically symmetric spacetime
$(\mathcal{M},g)$ as follows:
\begin{defn}
The \emph{axis} of the action (\ref{eq:IsometricAction}) is the set
\begin{equation}
\mathcal{Z}\doteq\big\{ p\in\mathcal{M}:\mbox{ }Orb(p)=\{p\}\big\}.
\end{equation}

\end{defn}
Definition \ref{def:C0-spherical-symmetry} implies that $\mathcal{Z}$
(if non-empty) is a $1$-dimensional $C^{1}$ submanifold of $\mathcal{M}$. 
\begin{rem*}
For the rest of this section, we will only work on $C^{0}$ spherically
symmetric spacetimes $(\mathcal{M},g)$ with $\mathcal{Z}\neq\emptyset$.
\end{rem*}
Let us define the continuous function $r:\mathcal{M}\rightarrow[0,+\infty)$
by the relation 
\begin{equation}
r(p)=\sqrt{\frac{Area(Orb(p))}{4\pi}}.\label{eq:DefinitionR-1}
\end{equation}
In view of the properties 3.a and 3.b of $C^{0}$ spherically symmetric
spacetimes (see Definition \ref{def:C0-spherical-symmetry}) and the
fact that $g$ is non-degenerate on $\mathcal{M}$, we infer that
\begin{equation}
r(p)=0\Leftrightarrow p\in\mathcal{Z}.\label{eq:R=00003D0exactlyOnAxis}
\end{equation}

For any $p\in\mathcal{Z}$, let $\mathcal{V}_{p}$ be the open neighborhood
of $p$ in $\mathcal{M}$ appearing in Definition \ref{def:C0-spherical-symmetry}.
According to 3.b in Definition \ref{def:C0-spherical-symmetry}, $\mathcal{V}_{p}$
is identified with $\mathbb{D}^{3}\times(-1,1)$ through a $C^{1}$-diffeomorphism.
In the natural $SO(3)$-invariant coordinate chart $\mathcal{U}_{p}\times\mathbb{S}^{2}$on
$\mathcal{V}_{p}\backslash\mathcal{Z}\simeq(\mathbb{D}^{3}\backslash\{0\})\times(-1,1)$
(where $\mathcal{U}_{p}\subset\mathbb{R}^{2}$ is naturally identified
with a radial slice of $(\mathbb{D}^{3}\backslash\{0\})\times(-1,1)$),
the metric $g$ splits as

\begin{equation}
g=\bar{g}+r^{2}g_{\mathbb{S}^{2}},\label{eq:SplittingOfTheMetric}
\end{equation}
where $\bar{g}$ is a $C^{0}$ Lorentzian metric on $\mathcal{U}_{p}$
\emph{extending continuously on} $\partial\mathcal{U}_{p}$. Note
that the resulting $C^{1}$ projection $\text{\textgreek{p}}:\mathcal{V}_{p}\backslash\mathcal{Z}\rightarrow\mathcal{U}_{p}\times\mathbb{S}^{2}\rightarrow\mathcal{U}_{p}$
admits a $C^{1}$ extension on $\mathcal{Z}$, mapping $\mathcal{Z}$
into $\partial\mathcal{U}_{p}.$ We will denote 
\begin{equation}
\text{\textgreek{g}}_{\mathcal{Z};p}\doteq\text{\textgreek{p}}(\mathcal{Z})\subset\partial\mathcal{U}_{p}\subset\mathbb{R}.\label{eq:ProjectionAxis-1}
\end{equation}

In view of the properties 1 and 3.b of $C^{0}$ spherically symmetric
spacetimes (see Definition \ref{def:C0-spherical-symmetry}) and the
fact that $g$ is non-degenerate, every connected component of $\mathcal{Z}$
is a timelike curve in $\mathcal{M}$. Thus, the curve (\ref{eq:ProjectionAxis-1})
in $\mathbb{R}^{2}$ is timelike with respect to $\bar{g}$. 
\begin{defn}
\label{def:ConeFoliation} A $C^{0}$ spherically symmetric spacetime
$(\mathcal{M},g)$ with non-empty axis $\mathcal{Z}$ will be called
\emph{admissible} if all of the following conditions are satisfied:

\begin{enumerate}

\item{$\mathcal{Z}$ is connected}

\item{$(\mathcal{M},g)$ has a ``simple'' topology in the following
sense: There exists an $SO(3)$-invariant $C^{1}$-diffeomorphism
$\mathcal{F}:\mathcal{M}\backslash\mathcal{Z}\rightarrow\mathcal{U}\times\mathbb{S}^{2}$,
$\mathcal{U}\subset\mathbb{R}^{2}$, with $\text{\textgreek{p}}:\mathcal{M}\backslash\mathcal{Z}\xrightarrow{\mathcal{F}}\mathcal{U}\times\mathbb{S}^{2}\rightarrow\mathcal{U}$
extending as a $C^{1}$ map on $\mathcal{Z}$ with 
\begin{equation}
\text{\textgreek{g}}_{\mathcal{Z}}\doteq\text{\textgreek{p}}(\mathcal{Z})\subset\partial\mathcal{U}.
\end{equation}
}

\item There exists a pair of $C^{0}$ functions $u,v:\mathcal{U}\cup\text{\textgreek{g}}_{\mathcal{Z}}\rightarrow\mathbb{R}$
such that:

\begin{enumerate}

\item For any $u_{0},v_{0}$ in the image of $u,v$, respectively,
the level curves $\{u=u_{0}\}$ and $\{v=v_{0}\}$ are $C^{1}$ curves
in $\mathcal{U}\cup\text{\textgreek{g}}_{\mathcal{Z}}$, either intersecting
transversally or not intersecting at all. In particular, $(u,v)$
constitute a $C^{0}$ coordinate chart on $\mathcal{U}\cup\text{\textgreek{g}}_{\mathcal{Z}}$
and the coordinate vector fields $\partial_{u},\partial_{v}$ are
well defined $C^{0}$ vector fields on $\mathcal{U}\cup\text{\textgreek{g}}_{\mathcal{Z}}$.

\item The vector fields $\partial_{u},\partial_{v}$ satisfy 
\begin{equation}
\bar{g}(\partial_{u},\partial_{u})=\bar{g}(\partial_{v},\partial_{v})=0\label{eq:OpticalFunctions}
\end{equation}
everywhere on $\mathcal{U}\cup\text{\textgreek{g}}_{\mathcal{Z}}$.

\end{enumerate}

\item Any other pair $(\bar{u},\bar{v})$ of $C^{0}$ functions on
$\mathcal{U}\cup\text{\textgreek{g}}_{\mathcal{Z}}$ satisfying the
above property is related to $(u,v)$ by a transformation of the form
$\bar{u}=U(u)$, $\bar{v}=V(v)$, for some unique and strictly monotonic,
locally bi-Lipschitz functions $U$, $V$. 

\end{enumerate}\end{defn}
\begin{rem*}
The pair $(u,v)$ will be called a \emph{double null foliation}. The
existence of a double-null foliation $(u,v)$ locally around each
point on $\mathcal{Z}$ can be readily established in the case when
$\bar{g}$ (see (\ref{eq:SplittingOfTheMetric})) is assumed, in addition,
to be of $C^{0,1}$ regularity. In general, however, when $\bar{g}$
is merely $C^{0}$, the integral curves of a congruence of null vectors
for $\bar{g}$ passing through a given point are not necessarily unique,
and it is not necessary that a continuous foliation of $\mathcal{U}$
by such curves exists (even if one restricts to open subsets of $\mathcal{U}$).

In the case when $\bar{g}$ is assumed to be $C^{0,1}$, condition
4 in Definition \ref{def:ConeFoliation} is also automatically satisfied.
\end{rem*}
Note that, if $(\mathcal{M},g)$ is an admissible $C^{0}$ spherically
symmetric spacetime, for any $p\in\mathcal{Z}$, in a $(u,v)$ coordinate
chart as in Definition \ref{def:ConeFoliation}, the metric $g$ takes
the form 
\begin{equation}
g=-\text{\textgreek{W}}^{2}dudv+r^{2}g_{\mathbb{S}^{2}},\label{eq:DoubleNullExpressionMetric}
\end{equation}
 where $\text{\textgreek{W}}>0$ is a $C^{0}$ function on $\mathcal{U}$,
extending continuously on $\text{\textgreek{g}}_{\mathcal{Z}}$ so
that 
\begin{equation}
\text{\textgreek{W}}|_{\text{\textgreek{g}}_{\mathcal{Z}}}>0.\label{eq:NonVanishingOmega}
\end{equation}

\subsection{\label{sub:Einstein--Null-dust-C0}The spherically symmetric Einstein--null
dust system on admissible $C^{0}$ spherically symmetric spacetimes}

Let $(\mathcal{M}^{3+1},g)$ be an admissible $C^{0}$ spherically
symmetric spacetime as in Section \ref{sub:C_0LorentzianSpacetimes},
with non-empty axis $\mathcal{Z}$. Recall that the metric $g$ is
expressed on $\mathcal{M}\backslash\mathcal{Z}$ as (\ref{eq:DoubleNullExpressionMetric}),
where $\text{\textgreek{W}},r\in C^{0}(\mathcal{U})$ extend continuously
on $\text{\textgreek{g}}_{\mathcal{Z}}$ such that 
\begin{equation}
\text{\textgreek{W}}>0\mbox{ on }\mathcal{U}\cup\text{\textgreek{g}}_{\mathcal{Z}}\label{eq:NonVanishingOmegaOnAxis}
\end{equation}
 and 
\begin{equation}
r|_{\text{\textgreek{g}}_{\mathcal{Z}}}=0,\mbox{ }r|_{\mathcal{U}}>0.\label{eq:ROnAxis}
\end{equation}
By possibly reparametrising the functions $u,v$ as $u\rightarrow U(u)$,
$v\rightarrow V(v)$ for some $C^{1}$ functions $U,V:\mathbb{R}\rightarrow\mathbb{R}$,
we will assume that 
\begin{equation}
u=v\mbox{ on }\text{\textgreek{g}}_{\mathcal{Z}}.\label{eq:ConditionGammaZ}
\end{equation}

The spherically symmetric Einstein--null dust system for $(r,\text{\textgreek{W}}^{2};\text{\textgreek{t}},\bar{\text{\textgreek{t}}})$
on $\mathcal{U}$, where $\text{\textgreek{t}},\bar{\text{\textgreek{t}}}$
are regular Borel measures on $\mathcal{U}$, is the following system:
\begin{align}
\partial_{u}\partial_{v}(r^{2})= & -\frac{1}{2}(1-\Lambda r^{2})\text{\textgreek{W}}^{2},\label{eq:RequationC0}\\
\partial_{u}\partial_{v}\log(\text{\textgreek{W}}^{2})= & \frac{\text{\textgreek{W}}^{2}}{2r^{2}}\big(1+4\text{\textgreek{W}}^{-2}\partial_{u}r\partial_{v}r\big),\label{eq:OmegaEquationC0}\\
\partial_{v}(\text{\textgreek{W}}^{-2}\partial_{v}r)= & -4\pi r^{-1}\text{\textgreek{W}}^{-2}\bar{\text{\textgreek{t}}},\label{eq:ConstraintVC0}\\
\partial_{u}(\text{\textgreek{W}}^{-2}\partial_{u}r)= & -4\pi r^{-1}\text{\textgreek{W}}^{-2}\text{\textgreek{t}},\label{eq:ConstraintUC0}\\
\partial_{u}\bar{\text{\textgreek{t}}}= & 0,\label{eq:IngoingC0}\\
\partial_{v}\text{\textgreek{t}}= & 0,\label{eq:OutgoingC0}
\end{align}
where $\Lambda\in\mathbb{R}$ is fixed. Notice that, for equations
(\ref{eq:RequationC0})--(\ref{eq:OutgoingC0}) to be well defined
in the sense of distributions, it suffices that $\text{\textgreek{W}},r\in C^{0}$,
$\text{\textgreek{W}},r>0$. The system (\ref{eq:RequationC0})--(\ref{eq:OutgoingC0})
is also supplemented by the following boundary conditions on $\text{\textgreek{g}}_{\mathcal{Z}}$:
\begin{equation}
r|_{\text{\textgreek{g}}_{\mathcal{Z}}}=0\label{eq:RBoundaryCondition}
\end{equation}
and: 
\begin{equation}
\text{\textgreek{t}}|_{\text{\textgreek{g}}_{\mathcal{Z}}}=\bar{\text{\textgreek{t}}}|_{\text{\textgreek{g}}_{\mathcal{Z}}}\label{eq:EqualityDustsOnAxis}
\end{equation}
Note that the condition (\ref{eq:EqualityDustsOnAxis}) is well defined
because of (\ref{eq:IngoingC0}) and (\ref{eq:OutgoingC0}). 
\begin{rem*}
The condition (\ref{eq:EqualityDustsOnAxis}) arises naturally by
requiring that the energy momentum tensor $T_{\text{\textgreek{m}\textgreek{n}}}$
on $\mathcal{M}$, defined in the $(u,v,y^{1},y^{2})$ coordinate
chart on $\mathcal{M}\backslash\mathcal{Z}\simeq\mathcal{U}\times\mathbb{S}^{2}$
by 
\begin{equation}
T_{uu}=r^{-2}\text{\textgreek{t}},\mbox{ }T_{vv}=r^{-2}\bar{\text{\textgreek{t}}},\mbox{ }T_{uv}=T_{Au}=T_{Av}=T_{AB}=0,
\end{equation}
satisies (in the weak sense) the conservation-of-energy condition
$\nabla^{\text{\textgreek{m}}}T_{\text{\textgreek{m}\textgreek{n}}}=0$
everywhere on $\mathcal{M}$ (using also the gauge condition (\ref{eq:ConditionGammaZ})).
\end{rem*}
Notice that the system (\ref{eq:RequationC0})--(\ref{eq:OutgoingC0})
is gauge invariant in the following sense: If $(r,\text{\textgreek{W}}^{2},\text{\textgreek{t}},\bar{\text{\textgreek{t}}})$
is a solution to (\ref{eq:RequationC0})--(\ref{eq:OutgoingC0}),
then, for any double null coordinate transformation of the form 
\begin{equation}
\begin{cases}
u'=U(u),\\
v'=V(v),
\end{cases}\label{eq:GaugeTransformation}
\end{equation}
for some strictly monotonic, locally bi-Lipschitz functions $U,V$,
the set of functions 
\begin{equation}
\big(r',(\text{\textgreek{W}}')^{2},\text{\textgreek{t}}',\bar{\text{\textgreek{t}}}'\big)\Big|_{(u',v')}\doteq\big(r,\text{\textgreek{W}}^{2},\frac{\text{\textgreek{t}}}{(dU/du)^{2}},\frac{\bar{\text{\textgreek{t}}}}{(dV/dv)^{2}}\big)\Big|_{(U^{-1}(u'),V^{-1}(v'))}
\end{equation}
is a solution of (\ref{eq:RequationC0})--(\ref{eq:OutgoingC0}) in
the new coordinates. In view of condition 4 in Definition \ref{def:ConeFoliation},
any double null foliation $(u',v')$ on an admissible $C^{0}$ spherically
symmetric spacetime $(\mathcal{M},g)$ is related to $(u,v)$ by a
transformation of the form (\ref{eq:GaugeTransformation}). Furthermore,
if the new coordinates $u',v'$ also satisfy the gauge condition (\ref{eq:ConditionGammaZ}),
then $\text{\textgreek{t}}',\bar{\text{\textgreek{t}}}'$ satisfy
(\ref{eq:EqualityDustsOnAxis}) (provided (\ref{eq:EqualityDustsOnAxis})
is satisfied by $\text{\textgreek{t}},\bar{\text{\textgreek{t}}}$) 

The following regularity result for (\ref{eq:RequationC0})--(\ref{eq:OutgoingC0})
can be readily established:
\begin{lem}
\label{lem:Regularity} Let $(\mathcal{M},g)$ be a $C^{0}$ spherically
symmetric spacetime, and let $\mathcal{U},\text{\textgreek{g}}_{\mathcal{Z}}$,
$(u,v)$ and $r,\text{\textgreek{W}}^{2}\in C^{0}(\mathcal{U}\cup\text{\textgreek{g}}_{\mathcal{Z}})$
be as above. Let also $\text{\textgreek{t}},\bar{\text{\textgreek{t}}}$
be regular (and non-negative) Borel measures on $\mathcal{U}$. Assume
that $(r,\text{\textgreek{W}}^{2};\text{\textgreek{t}},\bar{\text{\textgreek{t}}})$
satisfy (in the weak sense) (\ref{eq:RequationC0})--(\ref{eq:OutgoingC0})
on $\mathcal{U}$. Then, for any $u_{0},v_{0}\in\mathbb{R}$ such
that $\{u=u_{0}\}$ and $\{v=v_{0}\}$ are non-trivial curves on $\mathcal{U}$,
the derivatives $\partial_{v}r$ and $\partial_{u}r$ are defined
almost everywhere on $\{u=u_{0}\}\cap\mathcal{U}$ and $\{v=v_{0}\}\cap\mathcal{U}$,
respectively, with 
\begin{gather}
\partial_{v}r\in L_{loc}^{\infty}\big(\{u=u_{0}\}\cap\mathcal{U}\big),\label{eq:LinftylocV}\\
\partial_{u}r\in L_{loc}^{\infty}\big(\{v=v_{0}\}\cap\mathcal{U}\big).\label{eq:LinftylocU}
\end{gather}
If, in addition, $\text{\textgreek{t}},\bar{\text{\textgreek{t}}}$
are non-negative $L_{loc}^{1}$ functions on $\mathcal{U}$ (and not
merely measures), then 
\begin{equation}
\partial_{u}r,\partial_{v}r\in C^{0}(\mathcal{U}).\label{eq:ContinuityInDurDvr}
\end{equation}
\end{lem}
\begin{proof}
In view of (\ref{eq:IngoingC0}) and (\ref{eq:OutgoingC0}), we can
write 
\begin{equation}
\text{\textgreek{t}}(u,v)=\text{\textgreek{t}}(u)
\end{equation}
and 
\begin{equation}
\bar{\text{\textgreek{t}}}(u,v)=\bar{\text{\textgreek{t}}}(v)\label{eq:OneVariableMeasure}
\end{equation}
(in the sense of distributions). Theorefore, (\ref{eq:LinftylocV})
and (\ref{eq:LinftylocU}) follow readily from (\ref{eq:ConstraintVC0})
and (\ref{eq:ConstraintUC0}), using the fact that $r,\text{\textgreek{W}}\in C^{0}(\mathcal{U})$
and that $\text{\textgreek{t}},\bar{\text{\textgreek{t}}}$ are non-negative
regular Borel measures. 

In the case when $\text{\textgreek{t}},\bar{\text{\textgreek{t}}}$
are also $L_{loc}^{1}$ functions, the same procedure yields that,
for any $u_{0},v_{0}\in\mathbb{R}$ such that $\{u=u_{0}\}$ and $\{v=v_{0}\}$
are non-trivial curves on $\mathcal{U}$ 
\begin{gather}
\partial_{v}r\in C^{0}\big(\{u=u_{0}\}\cap\mathcal{U}\big),\label{eq:LinftylocV-1}\\
\partial_{u}r\in C^{0}\big(\{v=v_{0}\}\cap\mathcal{U}\big).\label{eq:LinftylocU-1}
\end{gather}
 Integrating equation (\ref{eq:RequationC0}) in $u$ and in $v$,
using the fact that the right hand side of (\ref{eq:RequationC0})
is continuous in $\mathcal{U}$, we obtain (\ref{eq:ContinuityInDurDvr}).
\end{proof}
We will also need to define the notion of a $C^{0}$ \emph{future
extension }of a solution to (\ref{eq:RequationC0})--(\ref{eq:OutgoingC0}):
\begin{defn}
\label{def:FutureExtension} Let $(\mathcal{M},g)$ be an admissible
$C^{0}$ spherically symmetric spacetime with non-empty axis $\mathcal{Z}$,
and let $\mathcal{U},\text{\textgreek{g}}_{\mathcal{Z}}$, $(u,v)$
and $r,\text{\textgreek{W}}^{2}\in C^{0}(\mathcal{U}\cup\text{\textgreek{g}}_{\mathcal{Z}})$
be as above. Let us fix a time orientation on $(\mathcal{M},g)$ by
requiring that the timelike vector field $N=\partial_{u}+\partial_{v}$
in $\mathcal{U}$ is future directed. Let also $\text{\textgreek{t}},\bar{\text{\textgreek{t}}}$
be regular (and non-negative) Borel measures on $\mathcal{U}$ and
assume that $(r,\text{\textgreek{W}}^{2};\text{\textgreek{t}},\bar{\text{\textgreek{t}}})$
satisfy (in the weak sense) (\ref{eq:RequationC0})--(\ref{eq:OutgoingC0})
on $\mathcal{U}$. 

An admissible $C^{0}$ spherically symmetric and time oriented spacetime
$(\widetilde{\mathcal{M}},\tilde{g})$ will be called a \emph{$C^{0}$
spherically symmetric future extension of $(\mathcal{M},g)$ as a
solution of (\ref{eq:RequationC0})--(\ref{eq:OutgoingC0}) }if the
following conditions are satisfied:

\begin{enumerate}

\item There exists a $C^{1}$ embedding $i:\mathcal{M}\rightarrow\widetilde{\mathcal{M}}$
which is an isometry, i.\,e.~$i^{*}\tilde{g}=g$, and preserves
time orientation.

\item There exists a point $p\in\widetilde{\mathcal{M}}\backslash\mathcal{M}$
lying to the future of $i(\mathcal{M})$.

\item There exists a double null foliation $(\tilde{u},\tilde{v})$
on $(\widetilde{\mathcal{M}},\tilde{g})$, not necessarily coinciding
with $(u,v)\circ i^{-1}$ on $\mathcal{M}$ (see Definition \ref{def:ConeFoliation}),
with the following property: Denoting with $\widetilde{\mathcal{Z}}$
the axis of $\widetilde{\mathcal{M}}$, with $\widetilde{\mathcal{U}},\text{\textgreek{g}}_{\widetilde{\mathcal{Z}}}\subset\mathbb{R}^{2}$
the sets related to $(\tilde{u},\tilde{v})$ according to Definition
\ref{def:ConeFoliation} and with $\tilde{r},\widetilde{\text{\textgreek{W}}}^{2}$
the metric components defined for $\tilde{g}$ by (\ref{eq:DoubleNullExpressionMetric}),
assuming also that (\ref{eq:ConditionGammaZ}) holds, there exists
a pair of regular, non-negative Borel measures $\tilde{\text{\textgreek{t}}},\tilde{\bar{\text{\textgreek{t}}}}$
on $\widetilde{\mathcal{U}}$ such that: 

\begin{enumerate}

\item $(\tilde{r},\widetilde{\text{\textgreek{W}}}^{2};\tilde{\text{\textgreek{t}}},\tilde{\bar{\text{\textgreek{t}}}})$
satisfy (in the weak sense) (\ref{eq:RequationC0})--(\ref{eq:OutgoingC0})
on $\widetilde{\mathcal{U}}$ and (\ref{eq:RBoundaryCondition})--(\ref{eq:EqualityDustsOnAxis})
on $\text{\textgreek{g}}_{\widetilde{\mathcal{Z}}}$. 

\item Restricted to $i(\mathcal{M})$, the pair $(\tilde{\text{\textgreek{t}}},\tilde{\bar{\text{\textgreek{t}}}})$
satisfies: 
\begin{equation}
(\tilde{\text{\textgreek{t}}},\tilde{\bar{\text{\textgreek{t}}}})\Big|_{(\tilde{u},\tilde{v})}=\big(\frac{\text{\textgreek{t}}}{(dU/du)^{2}},\frac{\bar{\text{\textgreek{t}}}}{(dV/dv)^{2}}\big)\Big|_{(U^{-1}(\tilde{u}\circ i),V^{-1}(\tilde{v}\circ i))},\label{eq:EqualTauTauBar}
\end{equation}
where the functions $U,V$ are strictly increasing, locally bi-Lipschitz
functions defining the following coordinate transformation between
$(u,v)$ and $(\tilde{u}\circ i,\tilde{v}\circ i)$ on $\mathcal{M}$:
\begin{equation}
\begin{cases}
\tilde{u}\circ i=U(u),\\
\tilde{v}\circ i=V(v),
\end{cases}
\end{equation}
(such $U,V$ exist and are unique according to Definition \ref{def:ConeFoliation}).%
\footnote{The fact that $U,V$ are increasing follows from the fact that $i$
preserves time orientation.%
}

\end{enumerate}

\end{enumerate}
\end{defn}
Finally, in the next section, we will need the notion of a smooth
characteristic initial value problem for (\ref{eq:RequationC0})--(\ref{eq:OutgoingC0}):
\begin{defn}
\label{def:Characteristic-Initial-Value-Prolem} Let $v_{0}>0$. A
\emph{smooth} characteristic initial data set for (\ref{eq:RequationC0})--(\ref{eq:OutgoingC0})
on $u=0$ for $v\in[0,v_{0})$ consists of a set of smooth functions
$(r_{/},\text{\textgreek{W}}_{/}^{2},\text{\textgreek{t}}_{/},\bar{\text{\textgreek{t}}}_{/})$
on $[0,v_{0})$ satisfying the following properties: 

\begin{enumerate}

\item $r_{/}(0)=0$ and $r_{/}|_{(0,v_{0})}>0$,

\item $\text{\textgreek{W}}_{/}>0$,

\item $\bar{\text{\textgreek{t}}}_{/},\text{\textgreek{t}}_{/}\ge0$,

\item $r_{/},\text{\textgreek{W}}_{/},\bar{\text{\textgreek{t}}}_{/}$
satisfy the constraint equation 
\begin{equation}
\partial_{v}(\text{\textgreek{W}}_{/}^{-2}\partial_{v}r_{/})=-4\pi r_{/}^{-1}\text{\textgreek{W}}_{/}^{-2}\bar{\text{\textgreek{t}}}_{/}.\label{eq:ConstraintInitialDataNullDust}
\end{equation}

\item $\text{\textgreek{t}}_{/}$ is constant, i.\,e.~satisfies
\begin{equation}
\partial_{v}\text{\textgreek{t}}_{/}=0.
\end{equation}

\end{enumerate}
\end{defn}

\subsection{\label{sub:An-ill-posedness-result}Break down for the system (\ref{eq:RequationC0})--(\ref{eq:OutgoingC0})}

In this section, we will establish two results: one related to the
well posedness of the system (\ref{eq:RequationC0})--(\ref{eq:OutgoingC0})
up to the first point when the null dust reaches the axis, and one
related to the break down of (\ref{eq:RequationC0})--(\ref{eq:OutgoingC0})
beyond that point.

Our first result is the following:

\begin{figure}[h]
\centering  
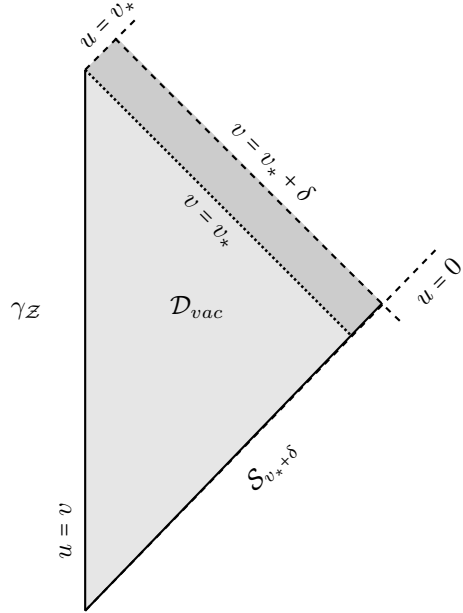
\caption{Schematic depiction of the domain $\mathcal{U}=\{0<u<v_{*}\}\cap\{u<v<v_{*}+\delta\}$ of the development $(r,\Omega^2,\tau,\bar{\tau})$ of $(r_{\slash},\Omega^2_{\slash},\tau_{\slash},\bar{\tau}_{\slash})$ in the statement of Proposition \ref{prop:Well-Posedness-Einstein-nulldust}.}
\end{figure}
\begin{prop}
\label{prop:Well-Posedness-Einstein-nulldust} For any $0<v_{*}<v_{0}$,
let $(r_{/},\text{\textgreek{W}}_{/}^{2},\text{\textgreek{t}}_{/},\bar{\text{\textgreek{t}}}_{/})$
be a smooth characteristic initial data set for (\ref{eq:RequationC0})--(\ref{eq:OutgoingC0})
on $u=0$, according to Definition \ref{def:Characteristic-Initial-Value-Prolem},
satisfying the following properties:

\begin{itemize}

\item{$(r_{/},\text{\textgreek{W}}_{/}^{2},\bar{\text{\textgreek{t}}},\bar{\text{\textgreek{t}}}_{/})$
is purely ingoing, i.\,e.: 
\begin{equation}
\text{\textgreek{t}}_{/}=0.\label{eq:PurelyIngoing}
\end{equation}
}

\item{$(r_{/},\text{\textgreek{W}}_{/}^{2},\bar{\text{\textgreek{t}}},\bar{\text{\textgreek{t}}}_{/})$
is trivial on $[0,v_{*}]$, i.\,e.: 
\begin{equation}
\bar{\text{\textgreek{t}}}_{/}|_{[0,v_{*}]}=0.\label{eq:TrivialInitialData}
\end{equation}
}

\item{ In the case $\Lambda>0$: 
\begin{equation}
\min_{[0,v_{*}]}\big(1-\frac{1}{3}\Lambda r_{/}^{2}\big)>0.\label{eq:NoTrappedInitially}
\end{equation}
}

\end{itemize}

Then, there exists some (possibly small)$0<\text{\textgreek{d}}<v_{0}-v_{*}$
depending on $(r_{/},\text{\textgreek{W}}_{/}^{2},\bar{\text{\textgreek{t}}},\bar{\text{\textgreek{t}}}_{/})$,
such that, setting 
\begin{equation}
\mathcal{U}\doteq\{0<u<v_{*}\}\cap\{u<v<v_{*}+\text{\textgreek{d}}\},\label{eq:TheDomain}
\end{equation}
\begin{equation}
\text{\textgreek{g}}_{\mathcal{Z}}\doteq\{u=v\}\cap\{0<u<v_{*}\}\label{eq:TheAxis}
\end{equation}
and 
\begin{equation}
\mathcal{S}_{v_{*}+\text{\textgreek{d}}}\doteq\{u=0\}\cap\{0<v<v_{*}+\text{\textgreek{d}}\},\label{eq:ForInitialData}
\end{equation}
there exists a unique $C^{\infty}$ quadruple $(r,\text{\textgreek{W}}^{2},\text{\textgreek{t}},\bar{\text{\textgreek{t}}})$
on $\mathcal{U}\cup\text{\textgreek{g}}_{\mathcal{Z}}\cup\mathcal{S}_{v_{*}+\text{\textgreek{d}}}$
solving (\ref{eq:RequationC0})--(\ref{eq:OutgoingC0}) on $\mathcal{U}$
and satisfying the initial data 
\begin{equation}
(r,\text{\textgreek{W}}^{2},\text{\textgreek{t}},\bar{\text{\textgreek{t}}})|_{u=0}=(r_{/},\text{\textgreek{W}}_{/}^{2},\text{\textgreek{t}}_{/},\bar{\text{\textgreek{t}}}_{/}).\label{eq:SolvesInitialData}
\end{equation}
on $\mathcal{S}_{v_{*}+\text{\textgreek{d}}}$ and the boundary conditions
(\ref{eq:ROnAxis}), (\ref{eq:EqualityDustsOnAxis}) on $\text{\textgreek{g}}_{\mathcal{Z}}$.
Furthermore, $(r,\text{\textgreek{W}}^{2},\text{\textgreek{t}},\bar{\text{\textgreek{t}}})$
satisfy the following properties: 

\begin{enumerate}

\item $(r,\text{\textgreek{W}}^{2},\text{\textgreek{t}},\bar{\text{\textgreek{t}}})$
extend as $C^{\infty}$ functions on $\partial\mathcal{U}$, satisfying
moreover 
\begin{equation}
\inf_{\mathcal{U}}\text{\textgreek{W}}>0,
\end{equation}
\begin{equation}
\sup_{\mathcal{U}}\partial_{u}r<0<\inf_{\mathcal{U}}\partial_{v}r\label{eq:NonTrappingInR-1}
\end{equation}
and, for any $k\in\mathbb{N}$
\begin{equation}
\sup_{\mathcal{U}}\frac{m}{r^{k}}<+\infty\label{eq:SmoothUpperBoundMass}
\end{equation}
(where $m$ is defined by (\ref{eq:DefinitionHawkingMass})).

\item The solution is purely ingoing, i.\,e.~the function $\text{\textgreek{t}}$
satisfies
\begin{equation}
\text{\textgreek{t}}\equiv0.\label{eq:PurelyIngoing-1}
\end{equation}

\item The solution is vacuum in the region 
\begin{equation}
\mathcal{D}_{vac}=\{0\le u\le v_{*}\}\cap\{u\le v\le v_{*}\},\label{eq:vacuumRegion}
\end{equation}
i.\,e.: 
\begin{equation}
\bar{\text{\textgreek{t}}}|_{\{0\le u\le v_{*}\}\cap\{u\le v\le v_{*}\}}=0\label{eq:TrivialTriangle}
\end{equation}

\end{enumerate}\end{prop}
\begin{rem*}
Let $(\mathcal{M},g)$ be the smooth, spherically symmetric spacetime
with boundary, obtained by equiping $(clos(\mathcal{U})\backslash\text{\textgreek{g}}_{\mathcal{Z}}\big)\times\mathbb{S}^{2}$
with the metric 
\begin{equation}
g=-\text{\textgreek{W}}^{2}dudv+r^{2}g_{\mathbb{S}^{2}}
\end{equation}
and then attaching an axis $\mathcal{Z}$ corresponding to $\text{\textgreek{g}}_{\mathcal{Z}}$.
Then, it follows from Proposition \ref{prop:Well-Posedness-Einstein-nulldust}
that $(\mathcal{M},g)$ is $C^{\infty}$ extendible as a Lorentzian
manifold beyond $\partial\mathcal{M}\simeq\big(\partial\mathcal{U}\backslash\text{\textgreek{g}}_{\mathcal{Z}}\big)\times\mathbb{S}^{2}$.
\end{rem*}

\noindent \emph{Proof.} Let us define the Hawking mass $m$ and the
renormalised Hawking mass $\tilde{m}$ by the relations (\ref{eq:DefinitionHawkingMass})
and (\ref{eq:RenormalisedHawkingMass}), respectively. The construction
of $(r,\text{\textgreek{W}}^{2},\text{\textgreek{t}},\bar{\text{\textgreek{t}}})$
will proceed by solving, instead of (\ref{eq:RequationC0})--(\ref{eq:OutgoingC0}),
the following (equivalent) system for $r,m,\text{\textgreek{t}},\bar{\text{\textgreek{t}}}$
on $\mathcal{U}$: 
\begin{align}
\partial_{u}\log\big(\frac{\partial_{v}r}{1-\frac{2m}{r}}\big)= & -4\pi r^{-1}\frac{\text{\textgreek{t}}}{-\partial_{u}r},\label{eq:DerivativeInUDirectionKappa-1}\\
\partial_{v}\log\big(\frac{-\partial_{u}r}{1-\frac{2m}{r}}\big)= & 4\pi r^{-1}\frac{\bar{\text{\textgreek{t}}}}{\partial_{v}r},\label{eq:DerivativeInVDirectionKappaBar-1}\\
\partial_{u}\tilde{m}= & -2\pi\frac{\big(1-\frac{2m}{r}\big)}{-\partial_{u}r}\text{\textgreek{t}},\label{eq:DerivativeTildeUMass-1}\\
\partial_{v}\tilde{m}= & 2\pi\frac{\big(1-\frac{2m}{r}\big)}{\partial_{v}r}\bar{\text{\textgreek{t}}},\label{eq:DerivativeTildeVMass-1}\\
\partial_{u}\bar{\text{\textgreek{t}}}= & 0,\label{eq:ConservationT_vv-1}\\
\partial_{v}\text{\textgreek{t}}= & 0.\label{eq:ConservationT_uu-1}
\end{align}
The proof of Proposition \ref{prop:Well-Posedness-Einstein-nulldust}
will conclude by showing that there exists a $0<\text{\textgreek{d}}<v_{0}-v_{*}$
and a unique smooth quadruple $(r,m,\text{\textgreek{t}},\bar{\text{\textgreek{t}}})$
on $\mathcal{U}\cup\text{\textgreek{g}}_{\mathcal{Z}}\cup\mathcal{S}_{v_{*}+\text{\textgreek{d}}}$,
solving (\ref{eq:DerivativeInUDirectionKappa-1})--(\ref{eq:ConservationT_uu-1})
on $\mathcal{U}$ with the boundary conditions (\ref{eq:ROnAxis})
and (\ref{eq:ConstraintInitialDataNullDust}) on $\text{\textgreek{g}}_{\mathcal{Z}}$
and the initial conditions 
\begin{equation}
(r,m,\bar{\text{\textgreek{t}}})|_{\mathcal{S}_{v_{*}+\text{\textgreek{d}}}}=(r_{/},m_{/},\bar{\text{\textgreek{t}}}_{/})\label{eq:InitialDataForRMT}
\end{equation}
on $\mathcal{S}_{v_{*}+\text{\textgreek{d}}}$ and such that, moreover: 

\begin{enumerate}

\item $(r,m,\text{\textgreek{t}},\bar{\text{\textgreek{t}}})$ extend
smoothly on $\partial\mathcal{U}$, satisfying (\ref{eq:NonTrappingInR-1}),
\begin{equation}
\inf_{\mathcal{U}}\big(1-\frac{2m}{r}\big)>0,\label{eq:AprioriNoTrapped}
\end{equation}
and (\ref{eq:SmoothUpperBoundMass}) (for any $k\in\mathbb{N}$). 

\item The purely ingoing condition (\ref{eq:PurelyIngoing-1}) holds.

\item The relation (\ref{eq:TrivialTriangle}) holds. 

\end{enumerate}
\begin{rem*}
In view of (\ref{eq:DefinitionHawkingMass}), the smooth extension
of $\text{\textgreek{W}}$ on $\partial\mathcal{U}$ follows from
the smooth extension of $r,m$ on $\partial\mathcal{U}$, (\ref{eq:AprioriNoTrapped})
and (\ref{eq:SmoothUpperBoundMass}).
\end{rem*}
Since $\bar{\text{\textgreek{t}}}_{/}\in C^{\infty}([0,v_{0}))$ satisfies
(\ref{eq:TrivialInitialData}), we can bound for any $k\in\mathbb{N}$:
\begin{equation}
\sup_{[0,v_{0})}\frac{\bar{\text{\textgreek{t}}}_{/}}{|v-v_{*}|^{k}}<+\infty.\label{eq:VanishingToInfiniteOrder}
\end{equation}
Since $r_{/},\text{\textgreek{W}}_{/}\in C^{\infty}([0,v_{0}))$ satisfy
the constraint equation (\ref{eq:ConstraintInitialDataNullDust})
with $r_{/}(0)=0$, $r_{/}|_{(0,v_{0})}>0$ and $\text{\textgreek{W}}_{/}>0$,
there exists some $0<\text{\textgreek{d}}<v_{0}-v_{*}$ such that
\begin{equation}
\min_{[0,v_{*}+\text{\textgreek{d}}]}\partial_{v}r_{/}>\text{\textgreek{e}}_{\Lambda},\label{eq:PositiveDvRInitially}
\end{equation}
where $0<\text{\textgreek{e}}_{\Lambda}\ll1$ is a fixed small parameter
depending on $\Lambda$ and (\ref{eq:NoTrappedInitially}). Furthermore,
in view of (\ref{eq:NoTrappedInitially}) and (\ref{eq:TrivialInitialData}),
we can assume that $\text{\textgreek{d}}>0$ is small enough so that
\begin{equation}
\inf_{[0,v_{*}+\text{\textgreek{d}}]}\Big(1-\frac{2m_{/}}{r_{/}}\Big)>\text{\textgreek{e}}_{\Lambda}\label{eq:QuantitativeNoTrappedInitially}
\end{equation}
where the initial Hawking mass $m_{/}$ is defined by the relation
\begin{equation}
\partial_{v}(m_{/}-\frac{1}{6}\Lambda r_{/}^{3})=2\pi\frac{\big(1-\frac{2m_{/}}{r_{/}}\big)}{\partial_{v}r_{/}}\bar{\text{\textgreek{t}}}_{/}
\end{equation}
under the condition $m_{/}(0)=0$; note that, equivalently, $m_{/}$
can be defined by 
\begin{equation}
m_{/}(v)\doteq\frac{1}{2}\Big(r_{/}(v)-\text{\textgreek{W}}_{/}^{-2}(v)\partial_{v}r_{/}(v)\int_{0}^{v}(1-\Lambda r_{/}^{2}(\bar{v}))\text{\textgreek{W}}_{/}^{2}(\bar{v})\, d\bar{v}\Big).
\end{equation}
 Finally, in view of (\ref{eq:TrivialInitialData}), we can assume
that $\text{\textgreek{d}}$ is small enough so that
\begin{equation}
\frac{\max_{[0,v_{*}+\text{\textgreek{d}}]}\partial_{v}\tilde{m}_{/}}{\min_{[0,v_{*}+\text{\textgreek{d}}]}\partial_{v}r_{/}}+\max_{[0,v_{*}+\text{\textgreek{d}}]}\frac{\partial_{v}\tilde{m}_{/}}{|v-v_{*}|}<\text{\textgreek{e}}_{\Lambda}.\label{eq:SmallnessDeltaRegion}
\end{equation}
From now on, we will assume that $\text{\textgreek{d}}$ has been
fixed as above.

A priori, if a smooth smooth solution $(r,m,\text{\textgreek{t}},\bar{\text{\textgreek{t}}})$
to (\ref{eq:DerivativeInUDirectionKappa-1})--(\ref{eq:ConservationT_uu-1})
on $\mathcal{U}$ satisfying (\ref{eq:ROnAxis}), (\ref{eq:ConstraintInitialDataNullDust})
and (\ref{eq:InitialDataForRMT}) exists, then (\ref{eq:TrivialTriangle})
follows immediately from (\ref{eq:TrivialInitialData}) and (\ref{eq:ConservationT_vv-1}).
Moreover, (\ref{eq:PurelyIngoing-1}) follows readily from the initial
condition (\ref{eq:TrivialInitialData}) for $\bar{\text{\textgreek{t}}}$,
the conservation equations (\ref{eq:ConservationT_vv-1})--(\ref{eq:ConservationT_uu-1})
and the form of the domain $\mathcal{U}$. Thus, the existence and
uniqueness of a smooth solution $(r,m,\text{\textgreek{t}},\bar{\text{\textgreek{t}}})$
to (\ref{eq:DerivativeInUDirectionKappa-1})--(\ref{eq:ConservationT_uu-1})
on $\mathcal{U}$ satisfying (\ref{eq:ROnAxis}), (\ref{eq:ConstraintInitialDataNullDust})
and (\ref{eq:InitialDataForRMT}) is reduced to the existence and
uniqueness of a smooth solution $(r,m,\bar{\text{\textgreek{t}}})$
of the system 
\begin{align}
\partial_{u}\big(\frac{\partial_{v}r}{1-\frac{2m}{r}}\big)= & 0,\label{eq:DerivativeInUDirectionKappaReduced}\\
\partial_{v}\log\big(\frac{-\partial_{u}r}{1-\frac{2m}{r}}\big)= & 4\pi r^{-1}\frac{\bar{\text{\textgreek{t}}}}{\partial_{v}r},\label{eq:DerivativeInVDirectionKappaBarReduced}\\
\partial_{u}\tilde{m}= & 0,\label{eq:DerivativeTildeUMassReduced}\\
\partial_{v}\tilde{m}= & 2\pi\frac{\big(1-\frac{2m}{r}\big)}{\partial_{v}r}\bar{\text{\textgreek{t}}},\label{eq:DerivativeTildeVMassReduced}\\
\partial_{u}\bar{\text{\textgreek{t}}}= & 0\label{eq:ConservationT_vvReduced}
\end{align}
on the domain 
\begin{equation}
\mathcal{U}_{\#}\doteq(0,v_{*})\times(v_{*},v_{*}+\text{\textgreek{d}})\label{eq:ReducedDomain}
\end{equation}
 satisfying the characteristic initial conditions 
\begin{equation}
(r,m,\bar{\text{\textgreek{t}}})|_{\{0\}\times[v_{*},v_{*}+\text{\textgreek{d}})}=(r_{/},m_{/},\bar{\text{\textgreek{t}}}_{/})|_{[v_{*},v_{*}+\text{\textgreek{d}})}\label{eq:InitialDataForRMTRight}
\end{equation}
and
\begin{equation}
r|_{[0,v_{*})\times\{v_{*}\}}=r_{vac},|_{[0,v_{*})\times\{v_{*}\}},\label{eq:InitialDataForRMTleft}
\end{equation}
where the function $r_{vac}$ in the vacuum region (\ref{eq:vacuumRegion})
is determined by solving 
\begin{equation}
\begin{cases}
\partial_{u}\big(\frac{\partial_{v}r_{vac}}{1-\frac{1}{3}\Lambda r_{vac}^{2}}\big)= & 0,\mbox{ on }clos(\mathcal{D}_{vac})\\
r_{vac}=r_{/} & \mbox{ on }\{0\}\times[0,v_{*}]\\
r_{vac}|_{u=v}=0
\end{cases}\label{eq:RVacuum}
\end{equation}
(note that (\ref{eq:NoTrappedInitially}) is necessary for (\ref{eq:RVacuum})
to admit a smooth solution on the whole of $clos(\mathcal{D}_{vac})$).
Given a smooth solution $(r,m,\bar{\text{\textgreek{t}}})$ to (\ref{eq:DerivativeInUDirectionKappaReduced})--(\ref{eq:ConservationT_vvReduced})
on $\mathcal{U}_{\sharp}$ satisfying (\ref{eq:InitialDataForRMTRight})--(\ref{eq:InitialDataForRMTleft}),
we can then obtain a smooth solution $(r,m,\text{\textgreek{t}},\bar{\text{\textgreek{t}}})$
of (\ref{eq:DerivativeInUDirectionKappa-1})--(\ref{eq:ConservationT_uu-1})
on $\mathcal{U}$ satisfying (\ref{eq:ROnAxis}), (\ref{eq:ConstraintInitialDataNullDust})
and (\ref{eq:InitialDataForRMT}) by extending $r,m,\bar{\text{\textgreek{t}}}$
smoothly $\mathcal{D}_{vac}$ by the relations 
\begin{equation}
r|_{\mathcal{D}_{vac}}=r_{vac},\label{eq:RrestrictedVacuum}
\end{equation}
\begin{equation}
\tilde{m}|_{\mathcal{D}_{vac}}=0\label{eq:Mrestrictedacuum}
\end{equation}
and 
\[
\bar{\text{\textgreek{t}}}|_{\mathcal{D}_{vac}}=0
\]
and setting $\text{\textgreek{t}}\equiv0$.

The existence and uniqueness of a smooth solution $(r,m,\bar{\text{\textgreek{t}}})$
to (\ref{eq:DerivativeInUDirectionKappaReduced})--(\ref{eq:ConservationT_vvReduced})
on $\mathcal{U}_{\sharp}$ satisfying (\ref{eq:InitialDataForRMTRight})--(\ref{eq:InitialDataForRMTleft})
follows readily: Equations (\ref{eq:DerivativeTildeUMassReduced})
and (\ref{eq:ConservationT_vvReduced}) imply that 
\begin{equation}
\tilde{m}(u,v)=\tilde{m}_{/}(v)\label{eq:FormulaMTilde}
\end{equation}
and 
\begin{equation}
\bar{\text{\textgreek{t}}}(u,v)=\bar{\text{\textgreek{t}}}_{/}(v)\label{eq:FormulaTauBar}
\end{equation}
and, thus, equation (\ref{eq:DerivativeInUDirectionKappaReduced})
is equivalent to 
\begin{equation}
\partial_{u}\big(\frac{\partial_{v}r}{1-\frac{2\tilde{m}}{r}-\frac{1}{3}\Lambda r^{2}}\big)=0.\label{eq:RequationReduced}
\end{equation}
Note that, in view of the initial bounds (\ref{eq:PositiveDvRInitially})--(\ref{eq:SmallnessDeltaRegion}),
applying a standard bootstrap argument for equation (\ref{eq:RequationReduced})
(the details of which will be omitted), we infer that there exists
a unique smooth function $r$ on $\mathcal{U}_{\sharp}$ solving (\ref{eq:RequationReduced})
and satisfying (\ref{eq:InitialDataForRMTRight})--(\ref{eq:InitialDataForRMTleft}),
such that, in addition: 
\begin{equation}
\inf_{\mathcal{U}_{\sharp}}\Big(1-\frac{2m}{r}\Big)>\frac{1}{2}\text{\textgreek{e}}_{\Lambda}\label{eq:BoundTrappingOnUSharp}
\end{equation}
and 
\begin{equation}
\inf_{\mathcal{U}_{\sharp}}\partial_{v}r>\frac{1}{2}\text{\textgreek{e}}_{\Lambda}.
\end{equation}

The bound (\ref{eq:AprioriNoTrapped}) on $\mathcal{U}$ follows readily
from (\ref{eq:BoundTrappingOnUSharp}), (\ref{eq:RrestrictedVacuum})
and (\ref{eq:Mrestrictedacuum}). From (\ref{eq:VanishingToInfiniteOrder})
and the transport equations (\ref{eq:DerivativeInUDirectionKappaReduced}),
(\ref{eq:DerivativeTildeUMassReduced}) and (\ref{eq:ConservationT_vvReduced}),
we readily infer (\ref{eq:SmoothUpperBoundMass}).

Finally, in view of the fact that $(r,m,\bar{\text{\textgreek{t}}})$
satisfy equations (\ref{eq:DerivativeInUDirectionKappaReduced})--(\ref{eq:ConservationT_vvReduced})
everywhere on $\mathcal{U}$, from (\ref{eq:AprioriNoTrapped}) and
(\ref{eq:SmoothUpperBoundMass}) it readily follows that $(r,m,\text{\textgreek{t}},\bar{\text{\textgreek{t}}})$
extend smoothly on $\partial\mathcal{U}$.

\qed

Our second result is the following ill-posedess theorem:
\begin{thm}
\label{thm:IllPosedness} Let $0<v_{*}<v_{0}$, $(r_{/},\text{\textgreek{W}}_{/}^{2},\text{\textgreek{t}}_{/},\bar{\text{\textgreek{t}}}_{/})$,
$\text{\textgreek{d}}$, $\mathcal{U}$ and $\text{\textgreek{g}}_{\mathcal{Z}}$
be as in Proposition \ref{prop:Well-Posedness-Einstein-nulldust}.
Let also $(\mathcal{M},g)$ be the $C^{\infty}$ spacetime with boundary,
obtained by equiping $(clos(\mathcal{U})\backslash\text{\textgreek{g}}_{\mathcal{Z}}\big)\times\mathbb{S}^{2}$
with the metric 
\begin{equation}
g=-\text{\textgreek{W}}^{2}dudv+r^{2}g_{\mathbb{S}^{2}}
\end{equation}
and then attaching an axis $\mathcal{Z}$ corresponding to $\text{\textgreek{g}}_{\mathcal{Z}}$
(see the remark below Proposition \ref{prop:Well-Posedness-Einstein-nulldust}).
Assume also that there exists a $0<\text{\textgreek{d}}_{1}\le\text{\textgreek{d}}$
such that $\bar{\text{\textgreek{t}}}_{/}$ satisfies 
\begin{equation}
\bar{\text{\textgreek{t}}}_{/}|_{(v_{*},v_{*}+\text{\textgreek{d}}_{1})}>0.\label{eq:NonTrivialSupportTauBar}
\end{equation}
 Then, there exists \underline{no} globally hyperbolic, admissible
$C^{0}$ spherically symmetric spacetime $(\widetilde{\mathcal{M}},\tilde{g})$
(see Section (\ref{sub:C_0LorentzianSpacetimes}) for the relavant
definitions) which is a $C^{0}$ spherically symmetric future extension
of $(\mathcal{M},g)$ as a solution of (\ref{eq:RequationC0})--(\ref{eq:OutgoingC0})
according to Definition \ref{def:FutureExtension}.\end{thm}
\begin{proof}
We will argue by contradiction, assuming that there exists an admissible
$C^{0}$ spherically symmetric spacetime $(\widetilde{\mathcal{M}},\tilde{g})$
which is globally hyperbolic and at the same time is a $C^{0}$ spherically
symmetric future extension of $(\mathcal{M},g)$ as a solution of
(\ref{eq:RequationC0})--(\ref{eq:OutgoingC0}). 

Let $i:\mathcal{M}\rightarrow\widetilde{\mathcal{M}}$, $\widetilde{\mathcal{Z}}$,
$\widetilde{\mathcal{U}}$, $\text{\textgreek{g}}_{\widetilde{\mathcal{Z}}}$,
$(\tilde{u},\tilde{v})$, $(\tilde{r},\widetilde{\text{\textgreek{W}}}^{2};\tilde{\text{\textgreek{t}}},\tilde{\bar{\text{\textgreek{t}}}})$
and $U,V\in C^{0}$ be as in definition \ref{def:FutureExtension}.
Recall that the double null foliations $(u,v)$ and $(\tilde{u},\tilde{v})\circ i$
on $\mathcal{M}$ are related by 
\begin{equation}
\begin{cases}
\tilde{u}\circ i=U(u),\\
\tilde{v}\circ i=V(v),
\end{cases}\label{eq:RelationBetweenGauges}
\end{equation}
where $U,V$ are strictly increasing, locally bi-Lipschitz functions.
Furthermore, $\tilde{u},\tilde{v}$ are assumed to satisfy 
\begin{equation}
\tilde{u}=\tilde{v}\mbox{ on }\text{\textgreek{g}}_{\widetilde{\mathcal{Z}}}.\label{eq:EqualNewCoordinatesGammaZ}
\end{equation}

In view of \ref{eq:TheDomain} and (\ref{eq:RelationBetweenGauges}),
the projection (under spherical symmetry) of $i(\mathcal{M})$ on
$\widetilde{\mathcal{U}}\cup\text{\textgreek{g}}_{\widetilde{\mathcal{Z}}}$
is of the form 
\begin{equation}
\tilde{\text{\textgreek{p}}}\circ i(\mathcal{M}\backslash\mathcal{Z})=\{U(0)<\tilde{u}<V(v_{*})\}\cap\{V\circ U^{-1}(\tilde{u})<\tilde{v}<V(v_{*}+\text{\textgreek{d}})\}.\label{eq:NewDomain}
\end{equation}
Since $i(\mathcal{Z})\subseteq\widetilde{\mathcal{Z}}$, it is necessary,
in view of (\ref{eq:EqualNewCoordinatesGammaZ}), that 
\begin{equation}
V\circ U^{-1}(\tilde{u})=\tilde{u}
\end{equation}
in (\ref{eq:NewDomain}). Hence, by passing to a new double null coordinate
chart $(u',v')$ on $(\widetilde{\mathcal{M}},\tilde{g})$ through
a gauge transformation of the form 
\begin{equation}
\begin{cases}
u'=\overline{U}(\tilde{u}),\\
v'=\overline{V}(\tilde{v})
\end{cases}\label{eq:NewGauge}
\end{equation}
for some strictly increasing, locally bi-Lipschitz functions $\overline{U},\overline{V}$
so that the condition 
\begin{equation}
u'=v'\mbox{ on }\text{\textgreek{g}}_{\widetilde{\mathcal{Z}}}\label{eq:EqualNewCoordinatesGammaZ-1}
\end{equation}
hods, we will assume without loss of generality that (\ref{eq:NewDomain})
is of the form 
\begin{equation}
\tilde{\text{\textgreek{p}}}\circ i(\mathcal{M}\backslash\mathcal{Z})=\{0<u'<v_{*}\}\cap\{u'<v'<v_{*}+\text{\textgreek{d}}\}.\label{eq:NewDomain-1}
\end{equation}

\medskip{}

\noindent \emph{Remark.} Note that $(u',v')$ do not necessarily coincide
with $(u,v)$ on $\tilde{\text{\textgreek{p}}}\circ i(\mathcal{M}\backslash\mathcal{Z})$.

\medskip{}

\noindent The components $r',(\text{\textgreek{W}}')^{2}$ of the
metric $\tilde{g}$ (according to the splitting (\ref{eq:DoubleNullExpressionMetric}))
satisfy 
\begin{equation}
r'|_{\text{\textgreek{g}}_{\mathcal{Z}}}=0,\mbox{ }r'|_{\widetilde{\mathcal{U}}}>0\label{eq:ROnAxis-1}
\end{equation}
and
\begin{equation}
\text{\textgreek{W}}'>0\mbox{ on }\widetilde{\mathcal{U}}\cup\text{\textgreek{g}}_{\mathcal{Z}}\label{eq:NonVanishingOmegaOnAxis-1}
\end{equation}
 (see the properties (\ref{eq:R=00003D0exactlyOnAxis}) and (\ref{eq:NonVanishingOmega})
asociated to a double null foliation of a general admissible $C^{0}$
spherically symmetric spacetime). We will also denote with $\text{\textgreek{t}}',\bar{\text{\textgreek{t}}}'$
the transformed quantities 
\begin{equation}
\big(\text{\textgreek{t}}',\bar{\text{\textgreek{t}}}'\big)\Big|_{(u',v')}\doteq\big(\frac{\tilde{\text{\textgreek{t}}}}{(d\overline{U}/d\tilde{u})^{2}},\frac{\tilde{\bar{\text{\textgreek{t}}}}}{(d\overline{V}/d\tilde{v})^{2}}\big)\Big|_{(\overline{U}^{-1}(u'),\overline{V}^{-1}(v'))}.\label{eq:TransformedTauTauBar}
\end{equation}
Recall that $\big(r',(\text{\textgreek{W}}')^{2},\text{\textgreek{t}}',\bar{\text{\textgreek{t}}}'\big)$
satisfy (\ref{eq:RequationC0})--(\ref{eq:OutgoingC0}) and (\ref{eq:EqualityDustsOnAxis})
on $\text{\textgreek{g}}_{\widetilde{\mathcal{Z}}}$. 

Since $(\widetilde{\mathcal{M}},\tilde{g})$ is globally hyperbolic
and contains a point $p\in\widetilde{\mathcal{M}}\backslash\mathcal{M}$
lying in the future of $i(\mathcal{M})$ (see Definition \ref{def:FutureExtension}),
$\widetilde{\mathcal{U}}\cup\text{\textgreek{g}}_{\widetilde{\mathcal{Z}}}\backslash\mathcal{U}\cup\text{\textgreek{g}}_{\mathcal{Z}}$
contains a point $q$ lying in the future of $\mathcal{U}\cup\text{\textgreek{g}}_{\mathcal{Z}}$.
Therefore, in view of (\ref{eq:EqualNewCoordinatesGammaZ-1}) and
the form (\ref{eq:NewDomain-1}) of $\tilde{\text{\textgreek{p}}}\circ i(\mathcal{M}\backslash\mathcal{Z})$,
we infer that $\widetilde{\mathcal{U}}$ contains a set of the form
\begin{equation}
\mathcal{V}=\{0<u'<v_{*}+\text{\textgreek{d}}_{2}\}\cap\{u'<v'<v_{*}+\text{\textgreek{d}}_{2}\}\label{eq:Vdomain}
\end{equation}
for some fixed $\text{\textgreek{d}}$ satisfying 
\begin{equation}
0<\text{\textgreek{d}}_{2}<\frac{1}{2}\min\{\text{\textgreek{d}},\text{\textgreek{d}}_{1}\}\label{eq:BoundDelta2}
\end{equation}
and moreover: 
\begin{equation}
\text{\textgreek{g}}_{\widetilde{\mathcal{Z}}}^{\prime}\doteq\{0<u'<v_{*}+\text{\textgreek{d}}_{2}\}\cap\{u'=v'\}\subset\text{\textgreek{g}}_{\widetilde{\mathcal{Z}}}.\label{eq:NeqAxis}
\end{equation}
For the rest of the proof, we will restric to $\mathcal{V}\cup\text{\textgreek{g}}_{\mathcal{Z}}^{\prime}$.

In view of (\ref{eq:EqualTauTauBar}) and (\ref{eq:TransformedTauTauBar}),
from the fact that $\text{\textgreek{t}},\bar{\text{\textgreek{t}}}$
are smooth functions of $(u,v)$ on $\mathcal{U}$ and the coordinates
$u'$,$v'$ are strictly increasing, locally bi-Lipschitz functions
of $u$, $v$, respectively, we infer that 
\begin{equation}
\text{\textgreek{t}}',\bar{\text{\textgreek{t}}}'\in L_{loc}^{1}(\mathcal{V}).
\end{equation}
Therefore, as a consequence of Lemma \ref{lem:Regularity}: 
\begin{equation}
\partial_{v'}r',\partial_{u'}r'\in C^{0}(\mathcal{V}).\label{eq:ContinuityDvR'DuR'}
\end{equation}

In view of the bound (\ref{eq:NonTrappingInR-1}) of Proposition \ref{prop:Well-Posedness-Einstein-nulldust}
and the fact that, on $\tilde{\text{\textgreek{p}}}\circ i(\mathcal{M})$,
the coordinates $u'$,$v'$ are strictly increasing, locally bi-Lipschitz
functions of $u$, $v$, respectively, we infer that, on $\tilde{\text{\textgreek{p}}}\circ i(\mathcal{M}\backslash\mathcal{Z})$:
\begin{equation}
\inf_{\tilde{\text{\textgreek{p}}}\circ i(\mathcal{M}\backslash\mathcal{Z})}\partial_{v'}r'>0\label{eq:IncreasingRFromPrevious}
\end{equation}
and
\begin{equation}
\sup_{\tilde{\text{\textgreek{p}}}\circ i(\mathcal{M}\backslash\mathcal{Z})}\partial_{u'}r'<0.\label{eq:DecreasingRFromPrevious}
\end{equation}
In view of the constraint equation (\ref{eq:ConstraintUC0}), we have
(in the sense of distributions)
\begin{equation}
\partial_{u'}\big((\text{\textgreek{W}}')^{-2}\partial_{u'}r'\big)\le0.\label{eq:NegativeConstraint}
\end{equation}
Thus, from (\ref{eq:DecreasingRFromPrevious}) and (\ref{eq:NegativeConstraint}),
using also the fact that $\text{\textgreek{W}}'$ is continuous on
$\mathcal{V}\cup\text{\textgreek{g}}_{\mathcal{Z}'}$ and satisfies
(\ref{eq:NonVanishingOmegaOnAxis-1}), we infer that:
\begin{equation}
\sup_{\bar{v}\in[v_{*},v_{*}+\frac{1}{2}\text{\textgreek{d}}_{2}]}\sup_{\{v'=\bar{v}\}\cap\mathcal{V}}\partial_{u'}r'<0.\label{eq:LowerBoundDur}
\end{equation}

In view of the transport equations (\ref{eq:IngoingC0}) and (\ref{eq:OutgoingC0})
for $\bar{\text{\textgreek{t}}}',\text{\textgreek{t}}'$ and the boundary
condition (\ref{eq:EqualityDustsOnAxis}), we obtain, for any $(u',v')\in\mathcal{V}$:
\begin{equation}
\text{\textgreek{t}}'(u',v')=\bar{\text{\textgreek{t}}}'(0,v'-u').\label{eq:RepresentationFormulaForTau'}
\end{equation}
Hence, it follows from the assumption (\ref{eq:NonTrivialSupportTauBar})
on $\bar{\text{\textgreek{t}}}_{/}$ and the bound (\ref{eq:BoundDelta2})
on $\text{\textgreek{d}}_{2}$ that 
\begin{equation}
c_{0}\doteq\mbox{ess}\inf_{\{v_{*}+\frac{1}{4}\text{\textgreek{d}}_{2}\le u\le v_{*}+\frac{1}{2}\text{\textgreek{d}}_{2}\}\cap\mathcal{V}}\text{\textgreek{t}}'>0.\label{eq:LowerBoundTau}
\end{equation}

The following relations hold for all $\bar{v}\in[v_{*}+\frac{3}{8}\text{\textgreek{d}}_{2},v_{*}+\frac{1}{2}\text{\textgreek{d}}_{2}]$:
\begin{equation}
\limsup_{u'\rightarrow\bar{v}^{-}}(-\partial_{u'}r')|_{(u',\bar{v})}=+\infty\label{eq:BoundForContradiction}
\end{equation}
and 
\begin{equation}
\lim_{u'\rightarrow\bar{v}^{-}}(-r'\partial_{u'}r')|_{(u',\bar{v})}=0.\label{eq:BoundForContradiction2}
\end{equation}
We will establish (\ref{eq:BoundForContradiction}) and (\ref{eq:BoundForContradiction2})
later. Assuming, for now, that (\ref{eq:BoundForContradiction2})
holds, we will finish the proof of Theorem \ref{thm:IllPosedness}
by reaching a contradiction with (\ref{eq:BoundForContradiction}).

Let us define the function 
\begin{equation}
y\doteq(r')^{2}\label{eq:PsiDefinition}
\end{equation}
 on 
\[
\mathcal{D}_{v_{*}}\doteq[v_{*}+\frac{3}{8}\text{\textgreek{d}}_{2},v_{*}+\frac{1}{2}\text{\textgreek{d}}_{2}]\times[v_{*}+\frac{3}{8}\text{\textgreek{d}}_{2},v_{*}+\frac{1}{2}\text{\textgreek{d}}_{2}]\cap\mathcal{V}.
\]
Since $r'\in C^{0}(\widetilde{\mathcal{U}}\cup\text{\textgreek{g}}_{\widetilde{\mathcal{Z}}})$,
we infer that $y\in C^{0}(\mathcal{V}\cup\text{\textgreek{g}}_{\mathcal{Z}'})$. 

In view of (\ref{eq:RequationC0}), (\ref{eq:ROnAxis}) and (\ref{eq:BoundForContradiction2}),
$y$ satisfies the following initial value problem on $\mathcal{D}_{v_{*}}$
with initial data on $\text{\textgreek{g}}_{\mathcal{Z}'}\cap clos(\mathcal{D}_{v_{*}})$:
\begin{equation}
\begin{cases}
\partial_{u}\partial_{v}y=-\frac{1}{2}(1-\Lambda(r')^{2})(\text{\textgreek{W}}')^{2} & \mbox{ on }\mathcal{D}_{v_{*}},\\
\lim_{u'\rightarrow\bar{v}^{-}}y|_{(u',\bar{v})}=0 & \mbox{ for all }\bar{v}\in[v_{*}+\frac{3}{8}\text{\textgreek{d}}_{2},v_{*}+\frac{1}{2}\text{\textgreek{d}}_{2}],\\
\lim_{u'\rightarrow\bar{v}^{-}}\partial_{u'}y|_{(u',\bar{v})}=0 & \mbox{ for all }\bar{v}\in[v_{*}+\frac{3}{8}\text{\textgreek{d}}_{2},v_{*}+\frac{1}{2}\text{\textgreek{d}}_{2}].
\end{cases}\label{eq:RequationPsi}
\end{equation}
Therefore, for all $(u',v')\in\mathcal{D}_{v_{*}}$, $y|_{(u',v')}$
can be uniquely represented as 
\begin{equation}
y|_{(u',v')}=-\frac{1}{2}\int_{\{u'\le\bar{u}\}\cap\{v'\ge\bar{v}\}\cap\mathcal{D}_{v_{*}}}(1-\Lambda(r')^{2})(\text{\textgreek{W}}')^{2}\big|_{(\bar{u},\bar{v})}\, d\bar{u}d\bar{v}.\label{eq:RepresentationFormulaPsi}
\end{equation}
From (\ref{eq:RepresentationFormulaPsi}), we infer that there exists
some $C>0$ depending only on $||r'||_{C^{0}(\mathcal{V})}$ and $||\text{\textgreek{W}}'||_{C^{0}(\mathcal{V})}$,
such that, for all $(u',v')\in\mathcal{D}_{v_{*}}$ 
\begin{equation}
\big|y|_{(u',v')}\big|\le C^{2}\cdot(v'-u')^{2}
\end{equation}
or, equvalently, in view of the definition (\ref{eq:PsiDefinition})
of $y$: 
\begin{equation}
r'|_{(u',v')}\le C\cdot(v'-u').\label{eq:AlostThere}
\end{equation}
The bound (\ref{eq:AlostThere}), combined with (\ref{eq:ROnAxis-1}),
(\ref{eq:LowerBoundDur}) and the fact that $(\text{\textgreek{W}}')^{-2}\in C^{0}(\mathcal{V}\cup\text{\textgreek{g}}_{\mathcal{Z}'})$
and $(\text{\textgreek{W}}')^{-2}\partial_{u'}r'$ is decreasing in
$u'$ (in view of (\ref{eq:ConstraintUC0})), yields that 
\begin{equation}
\limsup_{u'\rightarrow\bar{v}^{-}}(-\partial_{u'}r')|_{(u',\bar{v})}<+\infty,
\end{equation}
which is a contradiction in view of (\ref{eq:BoundForContradiction}).

In order to complete the proof of Theorem \ref{thm:IllPosedness},
it only remains to establish (\ref{eq:BoundForContradiction}) and
(\ref{eq:BoundForContradiction2}).

\medskip{}

\noindent \emph{Proof of (\ref{eq:BoundForContradiction}) and (\ref{eq:BoundForContradiction2}).
}Let us set: 
\begin{equation}
M_{0}\doteq\max_{v'\in[v_{*}+\frac{3}{8}\text{\textgreek{d}}_{2},v_{*}+\frac{1}{2}\text{\textgreek{d}}_{2}]}\frac{(\text{\textgreek{W}}')^{2}}{-\partial_{u'}r'}\Big|_{(0,v')},\label{eq:M_0}
\end{equation}
\begin{equation}
\text{\textgreek{m}}_{0}\doteq\min_{v'\in[v_{*}+\frac{3}{8}\text{\textgreek{d}}_{2},v_{*}+\frac{1}{2}\text{\textgreek{d}}_{2}]}\frac{(\text{\textgreek{W}}')^{2}}{-\partial_{u'}r'}\Big|_{(0,v')},\label{eq:mu_0}
\end{equation}
\begin{equation}
\text{\textgreek{r}}_{0}\doteq\min_{v'\in[v_{*}+\frac{3}{8}\text{\textgreek{d}}_{2},v_{*}+\frac{1}{2}\text{\textgreek{d}}_{2}]}r'|_{(v_{*}+\frac{1}{4}\text{\textgreek{d}}_{2},v')}\label{eq:R_0}
\end{equation}
and 
\begin{equation}
\text{\textgreek{t}}_{0}\doteq\sup_{v'\in[v_{*},v_{*}+\frac{1}{2}\text{\textgreek{d}}_{2}]}\bar{\text{\textgreek{t}}}'|_{(0,v')}\label{eq:Tau_0}
\end{equation}
 Note that, in view of (\ref{eq:ROnAxis-1}), (\ref{eq:NonVanishingOmegaOnAxis-1}),
(\ref{eq:LowerBoundDur}), (\ref{eq:EqualTauTauBar}), (\ref{eq:TransformedTauTauBar}),
(\ref{eq:NonTrivialSupportTauBar}) and the fact that $\text{\textgreek{W}}',r',\partial_{u'}r'\in C^{0}(\mathcal{U})$,
we have $0<M_{0}<+\infty$, $0<\text{\textgreek{m}}_{0}<+\infty$,
$0<\text{\textgreek{r}}_{0}<+\infty$ and $0<\text{\textgreek{t}}_{0}<+\infty$.

We will first establish (\ref{eq:BoundForContradiction}). We will
argue by contradiction, assuming that there exists a $\bar{v}\in[v_{*}+\frac{3}{8}\text{\textgreek{d}}_{2},v_{*}+\frac{1}{2}\text{\textgreek{d}}_{2}]$
such that 
\begin{equation}
\limsup_{u'\rightarrow\bar{v}^{-}}(-\partial_{u'}r')|_{(u',\bar{v})}<+\infty.\label{eq:BoundForContradiction-1}
\end{equation}
In view of (\ref{eq:LowerBoundDur}) and (\ref{eq:BoundForContradiction-1}),
the quantity 
\begin{equation}
M_{1}\doteq\sup_{u'\in[0,\bar{v})}(-\partial_{u'}r')\Big|_{(u',\bar{v})}\label{eq:UpperBoundDuR}
\end{equation}
satisfies 
\begin{equation}
0<M_{1}<+\infty.
\end{equation}

The constraint equation (\ref{eq:ConstraintUC0}) can be rewritten
as: 
\begin{equation}
\partial_{u'}\log\Big(\frac{(\text{\textgreek{W}}')^{2}}{-\partial_{u'}r'}\Big)=-\frac{4\pi}{r'}\frac{\text{\textgreek{t}}'}{(-\partial_{u'}r')}.\label{eq:RewrittenConstraint}
\end{equation}
Therefore, integrating (\ref{eq:RewrittenConstraint}) in $u'$ for
$v'=\bar{v}$ and using (\ref{eq:M_0}), (\ref{eq:R_0}), (\ref{eq:LowerBoundDur}),
(\ref{eq:LowerBoundTau}) and (\ref{eq:UpperBoundDuR}), we obtain
for any $0\le\bar{u}<\bar{v}$: 
\begin{align}
\log\Big(\frac{(\text{\textgreek{W}}')^{2}}{-\partial_{u'}r'}\Big)\Big|_{(\bar{u},\bar{v})} & =\log\Big(\frac{(\text{\textgreek{W}}')^{2}}{-\partial_{u'}r'}\Big)\Big|_{(0,\bar{v})}-\int_{0}^{\bar{u}}\frac{4\pi}{r'}\frac{\text{\textgreek{t}}'}{(-\partial_{u'}r')}\Big|_{(u,\bar{v})}\, du=\label{eq:TheImportantbound}\\
 & =\log\Big(\frac{(\text{\textgreek{W}}')^{2}}{-\partial_{u'}r'}\Big)\Big|_{(0,\bar{v})}-\int_{0}^{\bar{u}}\frac{4\pi}{r'}\frac{\text{\textgreek{t}}'}{(-\partial_{u'}r')^{2}}(-\partial_{u'}r')\Big|_{(u,\bar{v})}\, du\le\nonumber \\
 & \le\log M_{0}-4\pi M_{1}^{-2}c_{0}\int_{v_{*}+\frac{1}{4}\text{\textgreek{d}}_{2}}^{\bar{u}}\frac{1}{r'}(-\partial_{u'}r')\Big|_{(u,\bar{v})}\, du=\nonumber \\
 & \le\log M_{0}-4\pi M_{1}^{-2}c_{0}\log\Big(\frac{\text{\textgreek{r}}_{0}}{r(\bar{u},\bar{v})}\Big).\nonumber 
\end{align}
Notice that the right and side of (\ref{eq:TheImportantbound}) goes
to $-\infty$ as $\bar{u}\rightarrow\bar{v}^{-}$, which is a contradiction
in view of (\ref{eq:BoundForContradiction-1}). Therefore, (\ref{eq:BoundForContradiction})
holds.

We now proceed to establish (\ref{eq:BoundForContradiction2}). We
will argue again by contradiction. To this end, since $r',\text{\textgreek{W}}'\in C^{0}(\mathcal{V}\cup\text{\textgreek{g}}_{\mathcal{Z}'})$
and $(\text{\textgreek{W}}')^{-2}\partial_{u'}r'$ is decreasing in
$u'$ (in view of (\ref{eq:ConstraintUC0})), it suffices to reach
a contradiction based on the assumption that there exists a $\bar{v}\in[v_{*}+\frac{3}{8}\text{\textgreek{d}}_{2},v_{*}+\frac{1}{2}\text{\textgreek{d}}_{2}]$
such that 
\begin{equation}
\liminf_{u'\rightarrow\bar{v}^{-}}(-r'\partial_{u'}r')|_{(u',\bar{v})}\neq0.\label{eq:BoundForContradiction2-1}
\end{equation}
In view of (\ref{eq:LowerBoundDur}), (\ref{eq:BoundForContradiction2-1})
is equivalent to 
\begin{equation}
c_{1}\doteq\inf_{u'\in[0,\bar{v})}(-r'\partial_{u'}r')|_{(u',\bar{v})}>0.\label{eq:AnotherBoundForContradiction}
\end{equation}

Integrating (\ref{eq:RewrittenConstraint}) in $u'$ for $v'=\bar{v}$
and using (\ref{eq:RepresentationFormulaForTau'}), (\ref{eq:mu_0}),
(\ref{eq:Tau_0}), (\ref{eq:LowerBoundTau}) and (\ref{eq:AnotherBoundForContradiction}),
we obtain for any $0\le\bar{u}<\bar{v}$: 
\begin{align}
\log\Big(\frac{(\text{\textgreek{W}}')^{2}}{-\partial_{u'}r'}\Big)\Big|_{(\bar{u},\bar{v})} & =\log\Big(\frac{(\text{\textgreek{W}}')^{2}}{-\partial_{u'}r'}\Big)\Big|_{(0,\bar{v})}-\int_{0}^{\bar{u}}\frac{4\pi}{r'}\frac{\text{\textgreek{t}}'}{(-\partial_{u'}r')}\Big|_{(u,\bar{v})}\, du\ge\label{eq:TheImportantbound-1}\\
 & \ge\log\text{\textgreek{m}}_{0}-\frac{4\pi\text{\textgreek{t}}_{0}}{c_{1}}\int_{v_{*}+\frac{1}{4}\text{\textgreek{d}}_{2}}^{\bar{u}}\, du.\nonumber 
\end{align}
Thus, 
\begin{equation}
\liminf_{\bar{u}\rightarrow\bar{v}^{-}}\log\Big(\frac{(\text{\textgreek{W}}')^{2}}{-\partial_{u'}r'}\Big)\Big|_{(\bar{u},\bar{v})}>-\infty,
\end{equation}
which is a contradiction, in view of (\ref{eq:BoundForContradiction}).
Thus, (\ref{eq:BoundForContradiction2}) holds, and the proof of Theorem
\ref{thm:IllPosedness} is complete.
\end{proof}
\bibliographystyle{plain}
\bibliography{DatabaseExample}

\end{document}